\documentclass[10pt,reqno,oneside]{amsart}

\usepackage{cancel,bm, amssymb}
\usepackage{color}

\usepackage{color}
\usepackage{graphicx,framed}
\usepackage[colorlinks=true, pdfstartview=FitV, linkcolor=blue, citecolor=blue, urlcolor=blue]{hyperref}
\definecolor{shadecolor}{rgb}{0.95, 0.95, 0.86}

\renewcommand{\d}{{\mathrm d}}

\newcommand{\im}{\mathrm{i}}
\def\res{\mathop{\mathrm {res}}\limits}

\numberwithin{equation}{section}

\makeatletter
\@addtoreset{equation}{section}
\makeatother

\newtheorem{theo}{Theorem}[section]

\newtheorem{lem}[theo]{Lemma}
\newtheorem{rem}[theo]{Remark}
\newtheorem{problem}[theo]{Riemann-Hilbert Problem}
\newtheorem{remark}[theo]{Remark}
\newtheorem{prop}[theo]{Proposition} 
\newtheorem{cor}[theo]{Corollary}

\def\d{{\rm d}}
\def\z{\zeta}

\def\k{{\rm k}}

\begin{document}

\title[Transition asymptotics for the Painlev\'e II transcendent]{Transition asymptotics for the Painlev\'e II transcendent}

\author{Thomas Bothner}
\address{Centre de recherches math\'ematiques,
Universit\'e de Montr\'eal, Pavillon Andr\'e-Aisenstadt, 2920 Chemin de la tour, Montr\'eal, Qu\'ebec H3T 1J4, Canada}
\email{bothner@crm.umontreal.ca}




\thanks{The author would like to thank A. Its and P. Deift for invaluable discussions about this project. Part of this work was completed during the author's visit of SISSA, Trieste in February 2015 and we acknowledge T. Grava and M. Bertola for providing excellent working conditions.}

\begin{abstract}
We consider real-valued solutions $u=u(x|s),x\in\mathbb{R}$ of the second Painlev\'e equation $u_{xx}=xu+2u^3$ which are parametrized in terms of the monodromy data $s\equiv(s_1,s_2,s_3)\subset\mathbb{C}^3$ of the associated Flaschka-Newell system of rational differential equations.  Our analysis describes the transition, as $x\rightarrow-\infty$, between the oscillatory power-like decay asymptotics for $|s_1|<1$ (Ablowitz-Segur) to the power-like growth behavior for $|s_1|=1$ (Hastings-McLeod) and from the latter to the singular oscillatory power-like growth for $|s_1|>1$ (Kapaev). It is shown that the transition asymptotics are of Boutroux type, i.e. they are expressed in terms of Jacobi elliptic functions. As applications of our results we obtain asymptotics for the Airy kernel determinant $\det(I-\gamma K_{\textnormal{Ai}})|_{L^2(x,\infty)}$ in a double scaling limit $x\rightarrow-\infty,\gamma\uparrow 1$ as well as asymptotics for the spectrum of $K_{\textnormal{Ai}}$.
\end{abstract}

\date{\today}
\maketitle

\section{Introduction and statement of results}
One of the most impressive occurrences of Painlev\'e transcendents in non-linear mathematical physics stems from their applicability in random matrix theory. In this field Painlev\'e functions describe, for instance (cf. \cite{TW2,ASM,FW}), eigenvalue distribution functions for classical finite $n$ ensembles, they appear in the description of universal distribution functions in the large $n$ limit and are also used in the computation of gap probabilities in the large $n$ limit. One concrete example for the last two cases is given by the celebrated Tracy-Widom distribution \cite{TW1}: the distribution function for the largest eigenvalue $\lambda_{\max}$ of a $n\times n$ random matrix drawn from the Gaussian Unitary Ensemble (GUE) in the large $n$ limit equals
\begin{equation}\label{TW}
	\textnormal{F}_{\textnormal{TW}}(x)=\lim_{n\rightarrow\infty}\textnormal{Prob}\left(\lambda_{\max}\leq\sqrt{2n}+\frac{x}{\sqrt{2}\,n^{\frac{1}{6}}}\right)=\exp\left[-\int_x^{\infty}(y-x)u^2_{\scriptscriptstyle{\textnormal{HM}}}(y)\d y\right],\ \ \ x\in\mathbb{R},
\end{equation}
where $u_{\scriptscriptstyle{\textnormal{HM}}}=u_{\scriptscriptstyle{\textnormal{HM}}}(x),x\in\mathbb{R}$ is the Hastings-McLeod solution of the second Painlev\'e equation,
\begin{equation}\label{PII0}
	u_{\scriptscriptstyle{\textnormal{HM}}}''=xu_{\scriptscriptstyle{\textnormal{HM}}}+2u^3_{\scriptscriptstyle{\textnormal{HM}}},\ \ \ \ \ (')=\frac{\d}{\d x};\hspace{1cm}u_{\scriptscriptstyle{\textnormal{HM}}}(x)=\frac{x^{-\frac{1}{4}}}{2\sqrt{\pi}}\,e^{-\frac{2}{3}x^{\frac{3}{2}}}\big(1+o(1)\big),\ \ x\rightarrow+\infty.
\end{equation}
This special solution was first analyzed in the works of Hastings and McLeod \cite{HM} who showed that the boundary value problem \eqref{PII0} has a unique, monotonically decreasing, smooth solution. In addition,
\begin{equation*}
	u_{\scriptscriptstyle{\textnormal{HM}}}(x)\sim\sqrt{-\frac{x}{2}},\ \ \ x\rightarrow-\infty.
\end{equation*}
Although $\textnormal{F}_{\textnormal{TW}}(x)$ in \eqref{TW} is written as a distribution function, it is equivalent to a gap probability for the rescaled eigenvalues $\mu_j=\sqrt{2}\,n^{\frac{1}{6}}(\lambda_j-\sqrt{2n}\,)$, and as such, expressible as Fredholm determinant. More precisely,
\begin{equation}\label{Fred}
	\textnormal{F}_{\textnormal{TW}}(x)=\lim_{n\rightarrow\infty}\textnormal{Prob}\Big(\sharp\big\{j:\ \mu_j\in[x,\infty)\big\}=0\Big)=\det\big(I-K_{\textnormal{Ai}}\big)\Big|_{L^2(x,\infty)},
\end{equation}
where $K_{\textnormal{Ai}}$ is the trace-class operator in $L^2((x,\infty);\d\lambda)$ with kernel
\begin{equation*}
	K_{\textnormal{Ai}}(\lambda,\mu)=\frac{\textnormal{Ai}(\lambda)\textnormal{Ai}'(\mu)-\textnormal{Ai}'(\lambda)\textnormal{Ai}(\mu)}{\lambda-\mu},\ \ \ \lambda,\mu\in(x,\infty)
\end{equation*}
and $\textnormal{Ai}(\lambda)$ the classical Airy function. Besides the gap probability \eqref{Fred} other spectral properties of large random matrices from the GUE can also be computed through Painlev\'e transcendents, for instance for $k\in\mathbb{Z}_{\geq 0}$ we have (cf. \cite{M})
\begin{equation*}
	\lim_{n\rightarrow\infty}\textnormal{Prob}\Big(\sharp\big\{j:\ \mu_j\in[x,\infty)\big\}=k\Big) = \frac{1}{k!}\left(-\frac{\partial}{\partial\gamma}\right)^k\left[\det\big(I-\gamma K_{\textnormal{Ai}}\big)\Big|_{L^2(x,\infty)}\right]\bigg|_{\gamma=1}
\end{equation*}
and
\begin{equation}\label{Meh}
	\det\big(I-\gamma K_{\textnormal{Ai}}\big)\Big|_{L^2(x,\infty)} = \exp\left[-\int_x^{\infty}(y-x)u^2_{\scriptscriptstyle{\textnormal{AS}}}(y;\gamma)\d y\right],\ \ \gamma\geq 0.
\end{equation}
Here, $u_{\scriptscriptstyle{\textnormal{AS}}}=u_{\scriptscriptstyle{\textnormal{AS}}}(x;\gamma),x\in\mathbb{R}$ is the Ablowitz-Segur solution of the second Painlev\'e equation,
\begin{equation}\label{AS}
	u_{\scriptscriptstyle{\textnormal{AS}}}'' = xu_{\scriptscriptstyle{\textnormal{AS}}}+2u_{\scriptscriptstyle{\textnormal{AS}}}^3,\ \ \ \ \ (')=\frac{\d}{\d x};\hspace{1cm}u_{\scriptscriptstyle{\textnormal{AS}}}(x;\gamma)=\sqrt{\gamma}\,\frac{x^{-\frac{1}{4}}}{2\sqrt{\pi}}\,e^{-\frac{2}{3}x^{\frac{3}{2}}}\big(1+o(1)\big),\ \ x\rightarrow+\infty.
\end{equation}
This one-parameter family of solutions has different analytical and asymptotical properties depending on the values of $\gamma$: namely, for the values of $\gamma\in(0,1)$ fixed, the boundary value problem \eqref{AS} has a unique, bounded, smooth solution with oscillatory behavior as $x\rightarrow-\infty$, compare \eqref{known:1} below. On the other hand, if $\gamma>1$, smoothness is destroyed at finite $x$ and the solution blows up, see \eqref{known:3}. 
\begin{rem} Historically, solutions of the boundary value problem \eqref{AS} for fixed $\gamma\in(0,1)$ were first analyzed, both analytically and asymptotically, in the late 1970's, cf. \cite{AS1,AS2}. In particular for the bounded solutions, Ablowitz and Segur solved the important connection problem, i.e. determine the complete asymptotic description of $u_{\scriptscriptstyle{\textnormal{AS}}}(x;\gamma)$ as $x\rightarrow-\infty$ provided the same description is given as $x\rightarrow+\infty$ (or vice versa). The same problem was subsequently also solved for the unbounded solution $u_{\scriptscriptstyle{\textnormal{HM}}}(x)$ of \eqref{PII0}, see \cite{HM}. However, the singular asymptotic structure for $\gamma>1$ as $x\rightarrow-\infty$ remained unknown until the work of Kapaev \cite{K} in 1992. 
\end{rem}

Our modest goal in this paper is to place Ablowitz-Segur (for $\gamma\in(0,1)$ or $\gamma\in(1,\infty)$) and Hastings-McLeod solutions	on equal asymptotic footing as $x\rightarrow-\infty$. This problem is of interest to the asymptotics of the Fredholm determinant
\begin{equation*}
	\det\left(I-\gamma K_{\textnormal{Ai}}\right)\Big|_{L^2(x,\infty)}
\end{equation*}
as $x\rightarrow-\infty$ and thus to the asymptotics of the Tracy-Widom distribution itself.\smallskip

 For the gap probability let us use the refined asymptotic behavior (compare Remark \ref{Deirem} below),
\begin{equation}\label{HM:ref}
	u_{\scriptscriptstyle{\textnormal{HM}}}(x)=\sqrt{-\frac{x}{2}}\left(1+\frac{1}{8x^3}+\mathcal{O}\left(x^{-6}\right)\right),\ \ x\rightarrow-\infty
\end{equation}
and the identity
\begin{equation*}
	\int_x^{\infty}u_{\scriptscriptstyle{\textnormal{HM}}}^2(y)\d y=\big(u_{\scriptscriptstyle{\textnormal{HM}}}'(x)\big)^2-xu_{\scriptscriptstyle{\textnormal{HM}}}^2(x)-u_{\scriptscriptstyle{\textnormal{HM}}}^4(x),\ \ x\in\mathbb{R}
\end{equation*}
from which we learn in \eqref{TW} that
\begin{equation}\label{TW:asy1}
	\textnormal{F}_{\textnormal{TW}}(x) = \exp\left[\frac{x^3}{12}\right]|x|^{-\frac{1}{8}}c_0\left(1+\mathcal{O}\left(x^{-3}\right)\right),\ \ \ x\rightarrow-\infty
\end{equation}
with some constant $c_0\in\mathbb{R}$. The constant factor was first derived by Deift, Its and Krasovsky in \cite{DIK} using scaling limits for a finite Laguerre ensemble combined with universality results for the Airy kernel,
\begin{equation}\label{TW:asy2}
	c_0=\exp\left[\frac{1}{24}\ln 2+\zeta'(-1)\right],
\end{equation}
with $\zeta=\zeta(s)$ the Riemann-zeta function. Shortly afterwards Baik, Buckingham and DiFranco \cite{BBdF} gave another derivation of \eqref{TW:asy1}, \eqref{TW:asy2} based on exact integral representations of $\textnormal{F}_{\textnormal{TW}}(x)$ which, opposed to \eqref{TW}, involve an integration from $-\infty$ to $x$. In either approach the regularization of the integral in \eqref{TW} is important and thus the asymptotics \eqref{HM:ref} relevant as they are needed in the evaluation of a large gap probability in random matrix theory.\footnote{Another way of deriving \eqref{TW:asy1} without the Painlev\'e II connection was presented in \cite{BEMN}.}\smallskip

The very same asymptotic question for the behavior of $\det(I-\gamma K_{\textnormal{Ai}})$ with fixed $\gamma\neq 1$ is still open and our focus does not lie on this interesting problem. However, considering the different qualitative behaviors of $u_{\scriptscriptstyle{\textnormal{AS}}}(x;\gamma)$ as $x\rightarrow-\infty$ (see \eqref{known:1}, \eqref{known:3} below for further detail),
\begin{align*}
	u_{\scriptscriptstyle{\textnormal{AS}}}(x;\gamma)=\frac{\sqrt{-2\beta}}{(-x)^{\frac{1}{4}}}\cos&\left(\frac{2}{3}(-x)^{\frac{3}{2}}+\beta\ln\big(8(-x)^{\frac{3}{2}}\big)+\phi\right)+\mathcal{O}\left((-x)^{-\frac{7}{10}}\right),\ \ \ \gamma\in(0,1)\\
	&\beta=\frac{1}{2\pi}\ln(1-\gamma),\ \ \ \ \phi=\frac{\pi}{4}-\textnormal{arg}\,\Gamma(\im\beta);
\end{align*}
\begin{equation*}
	\ \ \ \ \ u_{\scriptscriptstyle{\textnormal{AS}}}(x;\gamma)=\frac{\sqrt{-x}}{\sin\big(\frac{2}{3}(-x)^{\frac{3}{2}}+\widehat{\beta}\ln(8(-x)^{\frac{3}{2}})+\varphi\big)+\mathcal{O}\big((-x)^{-\frac{3}{2}}\big)}+\mathcal{O}\left((-x)^{-1}\right),\ \ \ \gamma\in(1,\infty)
\end{equation*}
\begin{equation*}
	\ \ \ \ \ \ \ \ \widehat{\beta}=\frac{1}{2\pi}\ln(\gamma-1),\ \ \ \ \ \varphi=\frac{\pi}{2}-\textnormal{arg}\,\Gamma\left(\frac{1}{2}+\im\widehat{\beta}\right);
\end{equation*}
it becomes evident that $\det(I-\gamma K_{\textnormal{Ai}})|_{L^2(x,\infty)}$ experiences a phase transition near $\gamma=1$ as $x\rightarrow-\infty$: for $\gamma\in(0,1)$ the determinant will be strictly positive, decaying exponentially fast but with a slower decay rate than \eqref{TW:asy1}. Opposed to that, for $\gamma\in(1,\infty)$ the determinant will display oscillations with decreasing amplitudes and its zeros are accumulating at $x=-\infty$. This qualitative change leads to a challenging problem: Determine the asymptotic behavior of
\begin{equation}\label{prob:1}
	\exp\left[-\int_x^{\infty}(y-x)u^2_{\scriptscriptstyle{\textnormal{AS}}}(y;\gamma)\d y\right]=\det\big(I-\gamma K_{\textnormal{Ai}}\big)\Big|_{L^2(x,\infty)}\ \ \ \ \textnormal{as}\ \ x\rightarrow-\infty\ \ \ \textnormal{and simultaneously}\ \ \gamma\rightarrow 1.
\end{equation}
As a direct application of Theorem \ref{res:4} below on the transition asymptotics of $u_{\scriptscriptstyle{\textnormal{AS}}}(x;\gamma)$ we will prove a short result which shows that \eqref{TW:asy1} to leading order also describes one particular case of the transition asymptotics for $\det(I-\gamma K_{\textnormal{Ai}})$.
\begin{cor}\label{nice} As $x\rightarrow-\infty$ and $\gamma\uparrow 1$, with $c_0$ as in \eqref{TW:asy2},
\begin{equation*}
	\det\big(I-\gamma K_{\textnormal{Ai}}\big)\Big|_{L^2(x,\infty)}\!\!=\exp\left[\frac{x^3}{12}\right]|x|^{-\frac{1}{8}}c_0\big(1+o(1)\big),\ 
\end{equation*}
uniformly for 
\begin{equation*}
	\varkappa\equiv-\frac{\ln(1-\gamma)}{(-x)^{\frac{3}{2}}}>\frac{2}{3}\sqrt{2}.
\end{equation*}	
\end{cor}
For the values of $x$ and $\gamma$ such that $0<\varkappa<\frac{2}{3}\sqrt{2}$, our results in Theorem \ref{bet:1} and \ref{res:5} do in general not allow us to state an analogous expansion;  the integration in \eqref{prob:1} simply becomes too complicated. Still we can make a qualitative prediction
\begin{quote}
	The transition asymptotics of $\det(I-\gamma K_{\textnormal{Ai}})|_{L^2(x,\infty)}$ as $x\rightarrow-\infty,\gamma\uparrow 1$ and $0<\varkappa<\frac{2}{3}\sqrt{2}$ is modeled by quasi-periodic functions, to leading order.
\end{quote}
\begin{rem}\label{recentpap} We would like to point out that an asymptotic analysis similar to \eqref{prob:1} was recently carried out for the case of another prominent Fredholm determinant in random matrix theory,
\begin{equation*}
	P(s,\gamma)=\det(I-\gamma K_{\sin})\Big|_{L^2(-1,1)}\ \ \ \textnormal{with}\ \ \ \  K_{\sin}(\lambda,\mu)=\frac{\sin s(\lambda-\mu)}{\pi(\lambda-\mu)},\ \ \ s,\gamma>0.
\end{equation*}
It is well known \cite{M,TW2} that $P(s,1)$ equals a gap probability: the probability of finding no eigenvalues in the interval $(-\frac{s}{\pi},\frac{s}{\pi})$ for a random matrix chosen from the GUE, in the bulk scaling limit with mean spacing one. Similar to the Airy kernel determinant, $P(s,\gamma)$ admits a representation in terms of Painlev\'e functions, this time involving a special solution of the fifth Painlev\'e equation \cite{JMMS}. Building on previous work of Dyson \cite{Dy1}, the analysis of $P(s,\gamma)$ as $s\rightarrow+\infty,\gamma\uparrow 1$,  was essentially completed in \cite{BDIK}, however without using Painlev\'e asymptotic analysis. This lead to results completely analogous to Corollary \ref{nice} and the qualitative prediction, compare Theorems $1.4$ and $1.12$ as well as Corollary $1.13$ in \cite{BDIK}.
\end{rem}

We now move ahead and present the setup for our results on the transition asymptotics as $x\rightarrow-\infty$ between $u_{\scriptscriptstyle{\textnormal{AS}}}(x;\gamma)$ and $u_{\scriptscriptstyle{\textnormal{HM}}}(x)$. Opposed to introducing Painlev\'e II functions as solutions to boundary value problems as in \eqref{PII0} and \eqref{AS}, we choose to follow a Riemann-Hilbert point of view. In this approach \cite{FN,JMU} solutions $u=u(x)$ to
\begin{equation}\label{PII}
	u_{xx}=xu+2u^3
\end{equation}
are parametrized through the monodromy data of an associated linear system of ordinary differential equations in the complex plane. The details are as follows, compare \cite{FIKN}: Given six complex numbers $\{s_k\}_{k=1}^6$ which satisfy the relations
\begin{equation}\label{cyclic}
  s_1-s_2+s_3+s_1s_2s_3 = 0,\hspace{1cm} s_{k+3}=-s_k,\hspace{0.75cm} s_1 = \bar{s}_3,\hspace{0.5cm} s_2=\bar{s}_2,
\end{equation}
we introduce the triangular matrices
\begin{equation*}
  S_k = \begin{pmatrix}
	  1 & 0 \\
	  s_k & 1\\
        \end{pmatrix}\ \ \ \textnormal{for}\ \ k\equiv 1\mod 2,\hspace{1cm} S_k= \begin{pmatrix}
  1 & s_k\\
  0 & 1\\
  \end{pmatrix}\ \ \ \textnormal{for}\ \ k\equiv 0 \mod 2,
\end{equation*}
and the sectors
\begin{equation*}
	\Omega_k=\left\{\lambda\in\mathbb{C}:\ \frac{\pi}{6}(2k-3)<\textnormal{arg}\,\lambda<\frac{\pi}{6}(2k-1)\right\},\ \ k=1,\ldots,6
\end{equation*}
with oriented boundary rays
\begin{equation*}
	\Gamma_k =\left\{\lambda\in\mathbb{C}:\ \textnormal{arg}\,\lambda=\frac{\pi}{6}+\frac{\pi}{3}(k-1)\right\},\hspace{0.5cm}k=1,\ldots,6
\end{equation*}
as shown in Figure \ref{figure1} below. Consider now the following Riemann-Hilbert problem (RHP). 
\begin{problem}\label{masterRHP}
Determine the piecewise analytic $2\times 2$ matrix-valued function $Y(\lambda)=Y(\lambda;x,s\equiv(s_1,s_2,s_3))$ such that
\begin{itemize}
	 \item $Y(\lambda)$ is analytic for $\lambda\in\mathbb{C}\backslash \bigcup_1^6\Gamma_k$ and for every $k$, the function $Y_k(\lambda)=Y(\lambda)\big|_{\Omega_k}$ has a continuous extension on the closure $\overline{\Omega}_k$. 
\begin{figure}[tbh]
\begin{center}
\resizebox{0.6\textwidth}{!}{\includegraphics{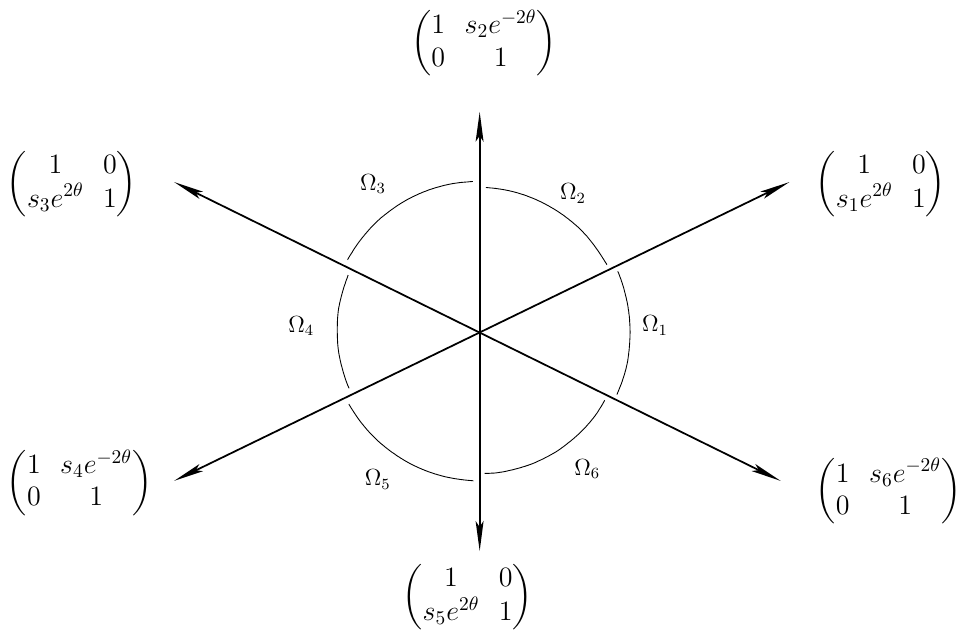}}
\caption{Jumps in the Painlev\'e II Riemann-Hilbert problem}
\label{figure1}
\end{center}
\end{figure}

  \item The boundary values $Y_+(\lambda)$ (resp. $Y_-(\lambda)$) from the left (resp. right) side of the oriented contour $\Gamma_k$ are related via the jump condition
  \begin{equation}\label{masterJ}
    Y_+(\lambda) = Y_-(\lambda) e^{-\theta(\lambda,x)\sigma_3}S_ke^{\theta(\lambda,x)\sigma_3},\hspace{0.5cm}\lambda\in\Gamma_k
  \end{equation}
  where
  \begin{equation*}
    \theta(\lambda,x) = \im\left(\frac{4}{3}\lambda^3+x\lambda\right),\hspace{0.5cm} \sigma_3=\begin{pmatrix}
                                                                                             1 & 0\\
0 & -1\\
                                                                                            \end{pmatrix}.
  \end{equation*}
  \item As $\lambda$ tends to infinity, the function $Y(\lambda)$ is normalized as follows,
  \begin{equation}\label{masterN}
    Y(\lambda) = I+\mathcal O\left(\lambda^{-1}\right),\hspace{0.5cm}\lambda\rightarrow\infty.
  \end{equation}
\end{itemize}
\end{problem}
\begin{remark} The first two of the relations in \eqref{cyclic} ensure consistency of the jump conditions \eqref{masterJ} with the continuity of $Y_k(\lambda)$ on $\overline{\Omega}_k$. Indeed, conditions \eqref{masterJ} can be rewritten as
\begin{align*}
	Y_{k+1}(\lambda)&=Y_k(\lambda)e^{-\theta(\lambda,x)\sigma_3}S_k e^{\theta(\lambda,x)\sigma_3},\ \ \lambda\in\Gamma_k,\ \ \ k=1,\ldots,5\\
	Y_1(\lambda)&=Y_6(\lambda)e^{-\theta(\lambda,x)\sigma_3}S_6e^{\theta(\lambda,x)\sigma_3},\ \ \lambda\in\Gamma_6
\end{align*}
and thus 
\begin{equation}\label{cons:1}
	Y_1(\lambda)=Y_1(\lambda)e^{-\theta(\lambda,x)\sigma_3}S_1S_2\cdot\ldots\cdot S_6e^{\theta(\lambda,x)\sigma_3},\ \ \ \lambda\in\Gamma_1.
\end{equation}
But every solution $Y(\lambda)$ of the RHP \ref{masterRHP} is invertible since $\det Y(\lambda)$ is an entire function normalized to unity at infinity. Hence \eqref{cons:1} yields upon cancellation of $Y_1(\lambda)$ the cyclic constraint
\begin{equation}\label{cyclic:2}
	S_1S_2\cdot\ldots\cdot S_6=I
\end{equation}
which is equivalent to the first two conditions in \eqref{cyclic}. We also mention that by triangularity of $S_k$, we have $e^{-\theta(\lambda,x)\sigma_3}S_ke^{\theta(\lambda,x)\sigma_3}\rightarrow I$ as $\lambda\rightarrow\infty$ along the rays $\Gamma_k$, i.e. \eqref{masterJ} is consistent with \eqref{masterN}.
\end{remark}

The connection between the RHP \ref{masterRHP} and \eqref{PII} is as follows. It is known \cite{BIK} that for any set of parameters $s\equiv (s_1,s_2,s_3)$ satisfying \eqref{cyclic} the RHP for $Y(\lambda)$ is meromorphically (with respect to $x$) solvable. Moreover its solution determines
the Painlev\'e II transcendent via
\begin{equation}\label{FN}
  u(x)\equiv u(x|s) = 2\lim_{\lambda\rightarrow\infty}\Big[\lambda\big(Y(\lambda;x,s)\big)_{12}\Big],
\end{equation}
and in addition we have $\bar{u}(x)=u(\bar{x})$. Conversely every real for real $x$ solution of \eqref{PII} has a unique Riemann-Hilbert representation \eqref{FN} for suitable
$s$ and we therefore adopt the notation $u(x)\equiv u(x|s)$, indicating the parametrization of solutions to PII equation \eqref{PII} in terms of the data $s$.
\begin{rem}\label{connect} The previously considered cases \eqref{AS} and \eqref{PII0} are special cases of \eqref{FN},
\begin{equation*}
	u_{\scriptscriptstyle{\textnormal{AS}}}(x;\gamma)=u\big(x|(-\im\sqrt{\gamma},0,\im\sqrt{\gamma})\big),\hspace{1cm}u_{\scriptscriptstyle{\textnormal{HM}}}(x)=u\big(x|(-\im,0,\im)\big).
\end{equation*}
\end{rem}
The data $s$ is  equal to the monodromy data of the associated linear $2\times 2$ matrix ODE
\begin{equation}\label{ODE:1}
	\frac{\d\Psi}{\d\lambda}=\left[-4\im\lambda^2\sigma_3+4\im\lambda\begin{pmatrix}
	0 & u\\
	-u & 0\\
	\end{pmatrix}+\begin{pmatrix}
	-\im x-2\im u^2 & -2u_x\\ 
	-2u_x & \im x+2\im u^2\\
	\end{pmatrix}\right]\Psi
\end{equation}
where $u$ satisfies \eqref{PII}. In the language of classical monodromy theory of differential equations (cf. \cite{I}) the entire functions $\Psi_k(\lambda)=Y_k(\lambda)e^{-\theta(\lambda,x)\sigma_3},k=1,\ldots,6$ and $\Psi_7(\lambda)\equiv \Psi_1(\lambda)$ form the seven canonical solutions of the ODE \eqref{ODE:1} and the triangular matrices $S_k$ in \eqref{masterJ} are the Stokes matrices characterizing the Stokes phenomenon of the irregular singular point $\lambda=\infty$. The matrices $S_k$ are $x$-independent and thus the second Painlev\'e transcendent $u=u(x|s)$ describes the isomonodromy deformations of system \eqref{ODE:1} with respect to the (isomonodromic) deformation parameter $x$.\bigskip

As we are interested in the asymptotic behavior of real solutions $u(x|s)$ as $x\rightarrow-\infty$, it will be useful to recall the following three different cases which have been singled out over the past $40$ years.
\begin{enumerate}
	\item[(A)] If $|s_1|<1$ is fixed, then, as $x\rightarrow-\infty$,
	\begin{equation}\label{known:1}
      u(x|s) = (-x)^{-\frac{1}{4}}\sqrt{-2\beta}\cos\left(\frac{2}{3}(-x)^{\frac{3}{2}}+\beta\ln\left(8(-x)^{\frac{3}{2}}\right)+\phi\right)+\mathcal O\left((-x)^{-\frac{7}{10}}\right),
    \end{equation}
    with
    \begin{equation*}
      \beta=\frac{1}{2\pi}\ln\left(1-|s_1|^2\right),\hspace{0.5cm} \phi=-\frac{\pi}{4}-\textnormal{arg}\,\Gamma(\im\beta)-\textnormal{arg}\,s_1,
    \end{equation*}
    and $\Gamma(z)$ is the Euler gamma function.
\end{enumerate}
\begin{remark} Expansion \eqref{known:1} appeared first in \cite{AS1,AS2} and was partially proven by Hastings, McLeod and Clarkson in \cite{HM,CM} using Gelfand-Levitan type integral equations. In the late $80$'s, Its, Kapaev, Suleimanov and Kitaev \cite{IK,S,Ki}, using the ``isomondromy method", obtained the full leading term. Another derivation of \eqref{known:1} was given by Deift and Zhou \cite{DZ2} in the mid $90$'s based on a direct asymptotic analysis of the RHP \ref{masterRHP} with the help of the Deift-Zhou nonlinear steepest descent method \cite{DZ1}. We refer the reader to the monograph \cite{FIKN} for a detailed exposition of the methods  used in the rigorous derivation of \eqref{known:1}. 
\end{remark}
\begin{enumerate}
	\item[(B)] If $|s_1|=1$, then with $s_1=\im\epsilon,\epsilon=\textnormal{sgn}(\Im\,s_1)\in\{\pm 1\}$ (compare here \eqref{cyclic}), as $x\rightarrow-\infty$,
  \begin{equation}\label{known:2}
    u(x|s) = -\epsilon\sqrt{-\frac{x}{2}}+\mathcal O \left(x^{-\frac{5}{2}}\right).
  \end{equation}
\end{enumerate}
\begin{remark}\label{Deirem} The leading order expansion in \eqref{known:2} is originally due to Hastings and McLeod \cite{HM} but was later on also derived by Deift and Zhou \cite{DZ1}, who extended \eqref{known:2} to a full asymptotic series in which all coefficients can be computed recursively. 
\end{remark}
\begin{enumerate}
	\item[(C)] If $|s_1|>1$ is fixed, then, as $x\rightarrow-\infty$,
	\begin{equation}\label{known:3}
    u(x|s) = \frac{\sqrt{-x}}{\sin\left(\frac{2}{3}(-x)^{\frac{3}{2}}+\widehat{\beta}\ln\left(8(-x)^{\frac{3}{2}}\right)+\varphi\right)+\mathcal O \left((-x)^{-\frac{3}{2}}\right)} +
    \mathcal O\left((-x)^{-1}\right)
  \end{equation}
  with
  \begin{equation*}
    \widehat{\beta}=\frac{1}{2\pi}\ln\left(|s_1|^2-1\right),\hspace{0.5cm}\varphi=-\textnormal{arg}\,\Gamma\left(\frac{1}{2}+\im\widehat{\beta}\right)-\textnormal{arg}\,s_1
  \end{equation*}
  where the error terms are uniform outside some neighborhoods of the singularities of the trigonometric function appearing in the denominator of the leading order.
 \end{enumerate}
\begin{remark} The singular asymptotic formula \eqref{known:3} was first obtained by Kapaev in \cite{K} through the isomonodromy method. In \cite{BI}, the authors rederived \eqref{known:3} based on nonlinear steepest descent techniques applied to the RHP \ref{masterRHP}.
\end{remark}
In this paper we derive new asymptotic expansions of $u(x|s)$ as $x\rightarrow-\infty$ with the help of a nonlinear steepest descent analysis applied to RHP \ref{masterRHP}. These expansions explain how the leading oscillatory power-like decay in case (A) is transformed to the leading power-like growth in case (B) as $|s_1|\uparrow 1$. We shall refer to the transition between regimes (A) and (B) as {\it regular transition}. In addition, we will also explain the transformation of the leading singular oscillatory power-like growth in case (C) to the leading power-like growth asymptotics in (B) as $|s_1|\downarrow 1$. It is natural to refer to the transition between (B) and (C) as {\it singular transition}. The results are as follows.

\subsection{Statement of results for regular transition}
Let 
\begin{equation}\label{dspara}
	\varkappa=\frac{v}{t}\in(0,\infty);\ \ \  t=(-x)^{\frac{3}{2}}>0,\ \ 0<v=-\ln\big|1-|s_1|^2\big|= -\ln\big(1-|s_1|^2\big)\equiv -2\pi\beta,\ 0<|s_1|<1
\end{equation}
and $\k\in(0,1)$ be implicitly determined from the transcendental equation
\begin{equation}\label{inteq:1}
	\varkappa=\frac{2}{3}\sqrt{\frac{2}{1+\k^2}}\left[E'-\frac{2\k^2}{1+\k^2}K'\right],
\end{equation}
in which $K$ and $E$ are standard complete elliptic integrals
\begin{equation*}
	K=K(\k)=\int_0^1\frac{\d\mu}{\sqrt{(1-\mu^2)(1-\k^2\mu^2)}},\ \ K'=K(\k');\hspace{0.75cm}E=E(\k)=\int_0^1\sqrt{\frac{1-\k^2\mu^2}{1-\mu^2}}\,\d\mu,\ \ E'=E(\k')
\end{equation*}
with modulus $\k\in(0,1)$ and complementary modulus $\k'=\sqrt{1-\k^2}$. It is shown in Proposition \ref{modex} that for any $\varkappa\in(0,\frac{2}{3}\sqrt{2})$, equation \eqref{inteq:1} determines $\k=\k(\varkappa)$ uniquely. In addition, set
\begin{equation}\label{Vfancy}
	V=V(\varkappa)=-\frac{2}{3\pi}\sqrt{\frac{2}{1+\k^2}}\left[E-\frac{1-\k^2}{1+\k^2}K\right];\hspace{1cm}\tau=\tau(\varkappa)=2\im\frac{K}{K'},
\end{equation}
and let
\begin{equation*}
	\theta_3(z,q)=1+2\sum_{m=1}^{\infty}q^{m^2}\cos(2\pi mz);\ \ \ \ \ \theta_2(z,q)=2\sum_{m=0}^{\infty}q^{(m+\frac{1}{2})^2}\cos\big((2m+1)\pi z\big),\ \ z\in\mathbb{C}
\end{equation*}
denote the Jacobi theta functions with nome $q=e^{\im\pi\tau}$. The theta functions in turn determine the Jacobi elliptic function 
\begin{equation}\label{Jell:1}
	\textnormal{cd}\left(2z\,K\left(\frac{1-\k}{1+\k}\right),\,\frac{1-\k}{1+\k}\right)=\frac{\theta_3(0,q)}{\theta_2(0,q)}\frac{\theta_2(z,q)}{\theta_3(z,q)},\ \ \ \ z\in\mathbb{C}\big\backslash\left\{\frac{1}{2}+\frac{\tau}{2}\mod \mathbb{Z}+\tau\mathbb{Z}\right\}
\end{equation}
and we list properties of the elliptic integrals and theta functions which are relevant to our analysis in Appendices \ref{appa} and \ref{app:theta}. The main result concerning the regular transition asymptotics of $u(x|s)$ as $x\rightarrow-\infty$ is formulated in the following Theorem.
\begin{theo}\label{bet:1} For any fixed $\delta\in(0,\frac{1}{3}\sqrt{2}),f_1\in(0,\infty)$, there exist positive constants $t_0=t_0(\delta),t_1=t_1(f_1),v_1=v_1(f_1)$ and $c_0=c_0(\delta),c_1=c_1(f_1)$ such that
\begin{equation}\label{ellip:1}
	u(x|s)=-\epsilon\sqrt{-\frac{x}{2}}\,\frac{1-\k(\varkappa)}{\sqrt{1+\k^2(\varkappa)}}\,\textnormal{cd}\left(2(-x)^{\frac{3}{2}}V(\varkappa)\,K\left(\frac{1-\k(\varkappa)}{1+\k(\varkappa)}\right),\,\frac{1-\k(\varkappa)}{1+\k(\varkappa)}\right)+J_1(x,s),
\end{equation}
with
\begin{equation}\label{bd:1}
	\big|J_1(x,s)\big|\leq c_0t^{-\frac{1}{15}} \ \ \ \ \forall\, t\geq t_0,\ \ 0<v\leq t\left(\frac{2}{3}\sqrt{2}-\delta\right),
\end{equation}
and
\begin{equation}\label{bd:2}
	\big|J_1(x,s)\big|\leq \frac{c_1}{\ln t}\ \ \ \ \forall\, t\geq t_1,\ v\geq v_1,\ \ \frac{2}{3}\sqrt{2}\,t-f_1\leq v<\frac{2}{3}\sqrt{2}\,t.
\end{equation}
\end{theo}
Theorem \ref{bet:1} describes the transition of the leading order asymptotics between case (A) for fixed $|s_1|<1$ given in \eqref{known:1} and case (B) for $|s_1|=1$ in \eqref{known:2}. Indeed, using the limiting behavior of $V(\varkappa),\tau(\varkappa)$ and $\k(\varkappa)$ in the Jacobi elliptic function \eqref{ellip:1} as $\varkappa\downarrow 0$ and $\varkappa\uparrow\frac{2}{3}\sqrt{2}$ (compare Corollary \ref{cor1} below), we obtain directly
\begin{cor}\label{mat} As $x\rightarrow-\infty$,
\begin{equation}\label{cor:as}
	u(x|s)=(-x)^{-\frac{1}{4}}\sqrt{-2\beta}\cos\left(\frac{2}{3}(-x)^{\frac{3}{2}}+\beta\ln\left(8(-x)^{\frac{3}{2}}\right)+\phi\right)+\mathcal{O}\left((-x)^{-\frac{1}{10}}\right),
\end{equation}
uniformly for $0<\varkappa\leq t^{-\frac{4}{5}}$. In addition, as $x\rightarrow-\infty,|s_1|\uparrow 1$,
\begin{equation}\label{cor:hm}
	u(x|s)=-\epsilon\sqrt{-\frac{x}{2}}+\mathcal{O}\left(\frac{1}{\ln(-x)}\right),
\end{equation}
uniformly for $\frac{2}{3}\sqrt{2}-\frac{1}{t}\leq\varkappa<\frac{2}{3}\sqrt{2}$.
\end{cor}
\begin{figure}[tbh]
\begin{center}
\hspace{2cm}\resizebox{0.53\textwidth}{!}{\includegraphics{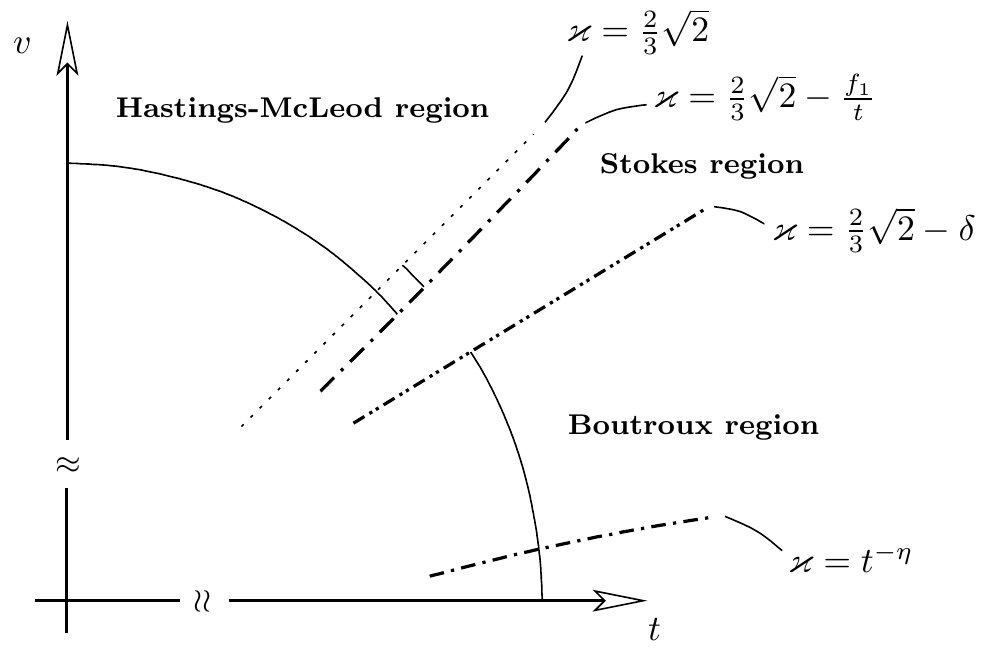}}
\caption{Depiction of transition asymptotics. The region captured by Theorems \ref{bet:1} and \ref{bet:2} is coined Boutroux region since the asymptotics are described in terms of Jacobi elliptic functions. Above and slightly below the separating line $\varkappa=\frac{2}{3}\sqrt{2}$ we observe to leading order Hastings-McLeod asymptotics, compare Theorems \ref{res:4} and \ref{res:9} below. In between we find Stokes lines, see Theorem \ref{res:5}.}
\label{diareg}
\end{center}
\end{figure}

Our second result addresses the asymptotic regime $x\rightarrow-\infty,|s_1|\uparrow 1$ for $\varkappa\geq\frac{2}{3}\sqrt{2}-\frac{f_2}{t}$ with fixed $f_2\in\mathbb{R}$. In this case the leading behavior is already of type \eqref{known:2} and therefore connects the Boutroux behavior \eqref{ellip:1}, \eqref{bd:2} completely to the Hastings-McLeod asymptotics \eqref{known:2} for $|s_1|=1$.
\begin{theo}\label{res:4} Given $f_2\in\mathbb{R}$, there exist positive constants $t_0=t_0(f_2),v_0=v_0(f_2)$ and $c=c(f_2)$ such that
\begin{equation}\label{theo:3}
	u(x|s)=-\epsilon\sqrt{-\frac{x}{2}}\left(1-\frac{1}{2\pi}\frac{e^{\sigma t}}{(t\sqrt{2})^{\frac{1}{2}}}+J_2(x,s)\right),\ \ \ \ \sigma=\frac{2}{3}\sqrt{2}-\varkappa,\ \ t=(-x)^{\frac{3}{2}}
\end{equation}
with
\begin{equation*}
	\big|J_2(x,s)\big|\leq ct^{-1}\ \ \ \ \forall\, t\geq t_0,\ v\geq v_0,\ \ v\geq\frac{2}{3}\sqrt{2}\,t-f_2.
\end{equation*}
\end{theo}
Theorems \ref{bet:1} and \ref{res:4} describe the behavior of $u(x|s)$ as $x\rightarrow-\infty$ for the values of $\varkappa$ in
\begin{equation*}
	\left(0,\frac{2}{3}\sqrt{2}-\delta\right]\bigsqcup\left[\frac{2}{3}\sqrt{2}-\frac{f_2}{t},\infty\right)\subset(0,\infty).
\end{equation*}
The corresponding regions in the double scaling diagram, i.e. {\it Boutroux region} for Theorem \ref{bet:1} and {\it Hastings-McLeod region} for Theorem \ref{res:4}, are shown in Figure \ref{diareg}. Another region, the {\it Stokes region}, corresponds to 
\begin{equation}\label{dead}
	\varkappa\in\left(\frac{2}{3}\sqrt{2}-\delta,\frac{2}{3}\sqrt{2}-\frac{f_1}{t}\right),\ \ \delta\in\left(0,\frac{2}{3}\sqrt{2}\right),\ f_1\in(0,\infty).
\end{equation}
We will not provide a full asymptotic description of $u(x|s)$ for the latter values of $\varkappa$ as the nonlinear steepest descent analysis becomes increasingly difficult. Instead we focus on the scale
\begin{equation*}
	t\geq t_0,\ v\geq v_0:\ \ \ \varkappa\geq \frac{2}{3}\sqrt{2}-f_3\frac{\ln t}{t},\ \ f_3\in\mathbb{R},
\end{equation*}
which eventually motivates the term {\it Stokes region}: For $f_3\in(-\infty,\frac{1}{6})$ the leading order behavior of $u(x|s)$ turns out to be unchanged from \eqref{theo:3}, however once we cross the {\it Stokes lines}
\begin{equation}\label{Stokeline}
	\mathcal{S}_k:\ \ v=\frac{2}{3}\sqrt{2}\,t-\frac{6k+1}{6}\ln t,\ \ k\in\mathbb{Z}_{\geq 0},
\end{equation}
once we go deeper into \eqref{dead}, additional $(k+1)$ terms will contribute to the leading behavior. This is in sharp contrast to Theorems \ref{bet:1} and \ref{res:4} where the fixed choices of $f_1\in(0,\infty)$ and $f_2\in\mathbb{R}$ had no effect on the leading orders. In the present paper, we shall prove the following estimation.
\begin{theo}\label{res:5} Given $f_3\in(-\infty,\frac{7}{6})$ there exists positive constants $t_0=t_0(f_3),v_0=v_0(f_3)$ and $c=c(f_3)$ such that
\begin{equation}\label{Stoke:res}
	u(x|s)=-\epsilon\sqrt{-\frac{x}{2}}\,\left(\frac{1+\frac{\epsilon p}{\sqrt{2}}}{1-\frac{\epsilon p}{\sqrt{2}}}\right)+J_3(x,s), \ \ \ \frac{\epsilon p}{\sqrt{2}}=-\frac{1}{2\pi}\sqrt{\frac{\pi}{2}}\frac{e^{\sigma t}}{(2\sqrt{2}\,t)^{\frac{1}{2}}},\ \ \sigma=\frac{2}{3}\sqrt{2}-\varkappa
\end{equation}
with
\begin{equation*}
	\big|J_3(x,s)\big|\leq ct^{-\min\{\frac{7}{6}-f_3,\frac{2}{3}\}}\ \ \ \ \forall\, t\geq t_0,\ v\geq v_0,\ \ v\geq\frac{2}{3}\sqrt{2}\,t-f_3\ln t.
\end{equation*}
\end{theo}
This rigorously clarifies the appearance of the first Stokes line $\mathcal{S}_1$: Fix $f_3\in(-\infty,\frac{1}{6})$ in \eqref{Stoke:res}, then 
\begin{equation*}
	p=\mathcal{O}\left(t^{-(\frac{1}{2}-f_3)}\right)=o(1)
\end{equation*}
and subsequently by geometric progression, as $x\rightarrow-\infty,|s_1|\uparrow 1$,
\begin{equation*}
	u(x|s)= -\epsilon\sqrt{-\frac{x}{2}}+\mathcal{O}\left((-x)^{-\min\{\frac{3}{2}(\frac{1}{6}-f_3),1\}}\right),\ \ \ \textnormal{uniformly for}\ \ \varkappa\geq\frac{2}{3}\sqrt{2}-f_3\frac{\ln t}{t}.
\end{equation*}

\subsection{Statement of results for singular transition}
In this case we also work with the double scaling parameter
\begin{equation*}
	\varkappa=\frac{v}{t}\in(0,\infty);\ \ \  t=(-x)^{\frac{3}{2}}>0,\ \ \ 0<v=-\ln\big|1-|s_1|^2\big|=-\ln\big(|s_1|^2-1\big)\equiv-2\pi\widehat{\beta},\ \ 1<|s_1|<\sqrt{2}
\end{equation*}
but notice that opposed to \eqref{dspara}, we have now 
\begin{equation*}
	0<-\ln\big(|s_1|^2-1\big)=-2\pi\widehat{\beta}=v\neq -2\pi\beta=-\ln\big(1-|s_1|^2\big).
\end{equation*}	
Furthermore we choose an upper constraint $1<|s_1|<\sqrt{2}$ which is important to our analysis below, but at the same time is also consistent with the anticipated transition asymptotics for $|s_1|\downarrow 1$.\footnote{For fixed $|s_1|\in[\sqrt{2},\infty)$, the leading order behavior of $u(x|s)$ as $x\rightarrow-\infty$ is already known through \eqref{known:3}.} The module $\k\in(0,1)$ is again determined via \eqref{inteq:1} and we require $V=V(\varkappa)$ and $\tau=\tau(\varkappa)$ as in \eqref{Vfancy}. However, instead of the Jacobi elliptic function $\textnormal{cd}(z)$ in\eqref{Jell:1}, we need
\begin{equation}\label{Jell:2}
 	\textnormal{dc}\left(2zK\left(\frac{1-\k}{1+\k}\right),\frac{1-\k}{1+\k}\right)=\frac{\theta_2(0,q)}{\theta_3(0,q)}\frac{\theta_3(z,q)}{\theta_2(z,q)},\ \ \ z\in\mathbb{C}\big\backslash\left\{\frac{1}{2}\mod \mathbb{Z}+\tau\mathbb{Z}\right\}.
\end{equation}
The analogue of Theorem \ref{bet:1} for the singular transition is contained in the following Theorem.
\begin{theo}\label{bet:2} For any fixed $\delta\in(0,\frac{1}{3}\sqrt{2})$, there exist positive constants $t_0=t_0(\delta)$ and $c_0=c_0(\delta)$ such that
\begin{equation}\label{singf}
	u(x|s)=-\epsilon\sqrt{-\frac{x}{2}}\frac{1+\k(\varkappa)}{\sqrt{1+\k^2(\varkappa)}}\,\textnormal{dc}\left(2(-x)^{\frac{3}{2}}K\left(\frac{1-\k(\varkappa)}{1+\k(\varkappa)}\right),\frac{1-\k(\varkappa)}{1+\k(\varkappa)}\right)+J_4(x,s)
\end{equation}
with
\begin{equation}\label{sing:est1}
	\big|J_4(x,s)\big|\leq c_0t^{-\frac{1}{15}}\ \ \ \ \forall\,t\geq t_0,\ \ 0<v\leq t\left(\frac{2}{3}\sqrt{2}-\delta\right),
\end{equation}
provided $(x,s_1)$ is uniformly bounded away from the discrete set
\begin{equation}\label{ex}
	\mathcal{Z}_n=\left\{(x,s_1):\ 2(-x)^{\frac{3}{2}}V(\varkappa)=n\in\mathbb{Z}\backslash\{0\}\right\}.
\end{equation}
On the other hand, for any fixed $f_1\in(0,\infty)$, there exist constants $t_1=t_1(f_1),v_1=v_1(f_1)$ and $c_1=c_1(f_1)$ such that \eqref{singf} holds true with
\begin{equation*}
	\big|J_4(x,s)\big|\leq \frac{c_1}{\ln t}\ \ \ \ \forall\, t\geq t_1,\ v\geq v_1,\ \ \frac{2}{3}\sqrt{2}\, t-f_1\leq v<\frac{2}{3}\sqrt{2}\,t
\end{equation*}
and no further constraint placed on $(x,s_1)$.
\end{theo}
Here, Theorem \ref{bet:2} describes the transition to leading order between case (C) for fixed $|s_1|>1$ in \eqref{known:3} and case (B) with $|s_1|=1$ in \eqref{known:2}. Indeed, through the limiting behavior of $V(\varkappa)$ and $\tau(\varkappa)$ as $\varkappa\downarrow 0$ and $\varkappa\uparrow\frac{2}{3}\sqrt{2}$ in \eqref{singf} we obtain
\begin{cor}\label{singmat} As $x\rightarrow-\infty$,
\begin{equation}\label{AS:mat}
	u(x|s)=\frac{\sqrt{-x}}{\sin\left(\frac{2}{3}(-x)^{\frac{3}{2}}+\widehat{\beta}\ln\big(8(-x)^{\frac{3}{2}}\big)+\varphi\right)+\mathcal{O}\left((-x)^{-\frac{3}{10}}\right)}+\mathcal{O}\left((-x)^{-\frac{1}{10}}\right)
\end{equation}
uniformly for $0<\varkappa\leq t^{-\frac{4}{5}}$ and away from the zeros of the trigonometric function appearing in the denominator of the leading order. In addition, as $x\rightarrow-\infty,|s_1|\downarrow 1$,
\begin{equation}\label{HS:mat}
	u(x|s)=-\epsilon\sqrt{-\frac{x}{2}}+\mathcal{O}\left(\frac{1}{\ln(-x)}\right),
\end{equation}
uniformly for $\frac{2}{3}\sqrt{2}-\frac{1}{t}\leq\varkappa<\frac{2}{3}\sqrt{2}$.
\end{cor}
Our last result is the direct analogue of Theorem \ref{res:4} to the singular transition. 
\begin{theo}\label{res:9} Given $f_2\in\mathbb{R}$, there exist positive constants $t_0=t_0(f_2),v_0=v_0(f_2)$ and $c=c(f_2)$ such that
\begin{equation*}
	u(x)=-\epsilon\sqrt{-\frac{x}{2}}\left(1+\frac{1}{2\pi}\frac{e^{\sigma t}}{(t\sqrt{2})^{\frac{1}{2}}}+J_5(x,s)\right)
\end{equation*}
with
\begin{equation*}
	\big|J_5(x,s)\big|\leq ct^{-1}\ \ \ \ \forall\,t\geq t_0,\ \ v\geq v_0,\ \ \ v\geq\frac{2}{3}\sqrt{2}\, t-f_2.
\end{equation*}
\end{theo}
The regions in the $(t,v)$-plane captured by Theorems \ref{bet:2} and \ref{res:9} are again shown in Figure \ref{diareg} and we use the same terminology, {\it Boutroux} and {\it Hastings-McLeod}. We do not address the Stokes region for the singular case, a full analysis for all Stokes lines \eqref{Stokeline} in both regular and singular transition will be postponed to a forthcoming publication.
\subsection{Further applications and outline of paper} The transition asymptotics are derived through an application of the nonlinear steepest descent method \cite{DZ1,DVZ} to the initial RHP \ref{masterRHP}. This analysis depends crucially on the values of the double scaling parameter
\begin{equation*}
	\varkappa=\frac{v}{t}\in(0,\infty);\ \ \ \ t=(-x)^{\frac{3}{2}}>0,\ \ \ 0<v=-\ln\big|1-|s_1|^2\big|=\begin{cases}
	-\ln\big(1-|s_1|^2\big),& 0<|s_1|<1\\
	-\ln\big(|s_1|^2-1\big),&1<|s_1|<\sqrt{2}.
	\end{cases}
\end{equation*}
For instance, the proof of Theorem \ref{bet:1} consists of several steps.
\begin{enumerate}
	\item First, the master RHP \ref{masterRHP} is analyzed with the constraint $t^{1-\eta}\leq v\leq t(\frac{2}{3}\sqrt{2}-\delta),\eta\in(0,1)$ in place. All details are worked out in Sections \ref{sec:reg} and \ref{relax}, which include a $g$-function transformation, the construction of local model problems via Jacobi theta and Airy functions and iterative solutions of singular integral equations. This part resembles in some of its aspects the approach carried out in \cite{BDIK}.
	\item Second, we analyze RHP \ref{masterRHP} subject to the constraint $t\geq v^{k+1}>0,t\geq t_0$ with $k\in\mathbb{Z}_{\geq 3}$. The techniques used now are very different from (1), in fact they are close to the ones we would use in the asymptotic analysis of RHP \ref{masterRHP} for fixed $|s_1|<1$, i.e. no need for a $g$-function transformation but different model problems involving parabolic cylinder functions. We work out the necessary details in Section \ref{flower} and combining error estimations from (1) and (2), we obtain \eqref{ellip:1} with estimation \eqref{bd:1}. In addition, expansion \eqref{cor:as} follows as well.
	\item Third, estimations \eqref{bd:2} and \eqref{cor:hm} are derived by applying modular transformations to the Jacobi theta functions used in (1) and introducing a new model function in a vicinity of the origin. The details are presented in Section \ref{upextension}.
\end{enumerate}
After that, we address the regime $t\geq t_0,v\geq v_0,v\geq \frac{2}{3}\sqrt{2}\,t-f_2$ in Theorem \ref{res:4}. Here the nonlinear steepest techniques are again different from the ones used in the derivation of Theorem \ref{bet:1}, we use a new $g$-function and different model problems, compare Section \ref{near}. It is worth mentioning that these techniques resemble the ones used in the derivation of \eqref{known:2} with $|s_1|=1$. Our final result for the regular transition is contained in Theorem \ref{res:5} which we derive in Section \ref{stokes}. The appearance of Stokes lines is already observed earlier in Section \ref{upextension}, see Remark \ref{stokefirst}, however we give a quantitative analysis of this phenomenon only subject to the constraint
\begin{equation*}
	t\geq t_0,\ v\geq v_0,\ \varkappa\geq\frac{2}{3}\sqrt{2}-f_3\frac{\ln t}{t},\ \ f_3\in\mathbb{R}.
\end{equation*}
In \cite{BDIK}, Theorem $1.12$, the authors also discovered a Stokes region while analyzing $P(s,\gamma)$ as $s\rightarrow\infty,\gamma\uparrow 1$, recall Remark \ref{recentpap}. Starting from the known large $s$ asymptotics of the eigenvalues of $K_{\sin}$ all Stokes lines were classified by an application of Lidskii's Theorem to the Fredholm determinant $P(s,\gamma)$. We will now argue that we can reverse this approach and derive asymptotics of the eigenvalues $\{\lambda_j(x)\}_{j=0}^{\infty}$ of $K_{\textnormal{Ai}}$ via Theorem \ref{res:5} and identity \eqref{Meh}. First, we state a slight improvement of Corollary \ref{nice}.
\begin{cor}\label{eig0} As $x\rightarrow-\infty$ and $\gamma\uparrow 1$, with $c_0$ as in \eqref{TW:asy2} and $\sigma$ in \eqref{Stoke:res},
\begin{equation}\label{eig01}
	\det\big(I-\gamma K_{\textnormal{Ai}}\big)\Big|_{L^2(x,\infty)}=\exp\left[\frac{x^3}{12}\right]|x|^{-\frac{1}{8}}b(\gamma)\left(1+\frac{1}{2^{\frac{9}{4}}\sqrt{\pi}}\frac{e^{\sigma t}}{t^{\frac{1}{2}}}\right)\left(1+\mathcal{O}\left(t^{-\min\{\frac{2}{3},\frac{3}{2}-f_3\}}\right)\right),
\end{equation}
uniformly for
\begin{equation*}
	\varkappa\geq\frac{2}{3}\sqrt{2}-f_3\frac{\ln t}{t},\ \ f_3\in\left(-\infty,\frac{7}{6}\right),
\end{equation*}
where $b=b(\gamma)$ is positive, bounded in $\gamma$, such that
\begin{equation*}
	b(\gamma)\rightarrow c_0,\ \ \gamma\uparrow 1.
\end{equation*}
\end{cor}
This result is a direct consequence of Theorem \ref{res:5} and we state its proof in Section \ref{eigproof}. Compared to Corollary \ref{nice}, an additional contribution to the leading order appears in \eqref{eig01} and the extra factor is in general not close to unity. In other words, also the Fredholm determinant displays an asymptotic Stokes phenomenon, and we have identified the first Stokes line,
\begin{equation*}
	\widetilde{S}_1:\ \ \ v=\frac{2}{3}\sqrt{2}\,t-\frac{1}{2}\ln t.
\end{equation*}
Now back to the spectrum of $K_{\textnormal{Ai}}$, we let $\{\lambda_j(x)\}_{j=0}^{\infty}$ denote the eigenvalues of the trace class operator $K_{\textnormal{Ai}}:L^2((x,\infty);\d\lambda)\circlearrowleft$. These have been, for instance, analyzed in \cite{TW1}, were it was proven that the spectrum of $K_{\textnormal{Ai}}$ is simple, we have $1>\lambda_0(x)>\lambda_1(x)>\ldots$ and for fixed $j\in\mathbb{Z}_{\geq 0}$,
\begin{equation}\label{TWproof}
	\lambda_j(x) \sim 1,\ \ \ x\rightarrow-\infty.
\end{equation}
In fact, Tracy and Widom in loc. cit. also gave an asymptotic formula,
\begin{equation}\label{conj}
	1-\lambda_j(x)=\frac{\sqrt{\pi}}{j!}2^{\frac{7}{2}j+\frac{9}{4}}t^{j+\frac{1}{2}}e^{-\frac{2}{3}\sqrt{2}\,t}\big(1+o(1)\big),\ \ t=(-x)^{\frac{3}{2}}\rightarrow+\infty,
\end{equation}
but the derivation is not fully rigorous. We will now outline a proof for the stated behavior, using Corollary \ref{eig0}. Choose $t\geq t_0,v\geq v_0$ such that
\begin{equation}\label{sc}
	\varkappa\geq\frac{2}{3}\sqrt{2}-f_3\frac{\ln t}{t},\ \ f_3\in\left(-\infty,\frac{7}{6}\right);\ \ \ v=-\ln(1-\gamma).
\end{equation}
Apply Lidskii's Theorem,
\begin{align}
	\frac{\det(I-\gamma K_{\textnormal{Ai}})}{\det(I-K_{\textnormal{Ai}})}\bigg|_{L^2(x,\infty)}&= \det\big(I+e^{-v}K_{\textnormal{Ai}}(I-K_{\textnormal{Ai}})^{-1}\big)\Big|_{L^2(x,\infty)} = \prod_{j=0}^{\infty}\left(1+e^{-v}\frac{\lambda_j(x)}{1-\lambda_j(x)}\right)\nonumber\\
	&=\left(1+e^{-v}\frac{\lambda_0(x)}{1-\lambda_0(x)}\right)\det\big(I+e^{-v}K_{1}(I-K_{1})^{-1}\big)\Big|_{L^2(x,\infty)},\label{good}
\end{align}
with $K_1=K_{\textnormal{Ai}}\cdot P_{\textnormal{Ai}}^1$ and $P_{\textnormal{Ai}}^1$ is the projection onto the eigenspace of $K_{\textnormal{Ai}}$ with eigenvalues $\{\lambda_j\}_{j=1}^{\infty}$. This allows us to compare \eqref{TW:asy1}, \eqref{eig01} to \eqref{good}, i.e. as $x\rightarrow-\infty,|s_1|\uparrow 1$ subject to \eqref{sc},
\begin{equation*}\label{comp}
	\left(1+e^{-v}\frac{\lambda_0(x)}{1-\lambda_0(x)}\right)\det\big(I+e^{-v}K_1(I-K_1)^{-1}\big)=\left(\frac{b(\gamma)}{c_0}+\frac{b(\gamma)}{2^{\frac{9}{4}}\sqrt{\pi}c_0}\frac{e^{\sigma t}}{t^{\frac{1}{2}}}\right)\left(1+\mathcal{O}\left(t^{-\min\{\frac{2}{3},\frac{3}{2}-f_3\}}\right)\right).
\end{equation*}
We multiply though with the second summand in the first factor in the right hand side of the last estimation, 
\begin{eqnarray}
	\left(2^{\frac{9}{4}}\sqrt{\pi}\,t^{\frac{1}{2}}e^{-\sigma t}+\frac{2^{\frac{9}{4}}\sqrt{\pi}\,t^{\frac{1}{2}}\lambda_0(x)}{e^{\frac{2}{3}\sqrt{2}t}(1-\lambda_0(x))}\right)\det\big(I+e^{-v}K_1(I-K_1)^{-1}\big)&=&\frac{b(\gamma)}{c_0}\left(1+2^{\frac{9}{4}}\sqrt{\pi}\,t^{\frac{1}{2}}e^{-\sigma t}\right)\label{rig:0}\\
	&&\times\left(1+\mathcal{O}\left(t^{-\min\{\frac{2}{3},\frac{3}{2}-f_3\}}\right)\right)\nonumber
\end{eqnarray}
and here, both factors in the left hand side of \eqref{rig:0} are positive. But for $t\geq t_0,v\geq v_0$ such that
\begin{equation*}
	\frac{2}{3}\sqrt{2}>\varkappa\geq\frac{2}{3}\sqrt{2}-f_3\frac{\ln t}{t},\ \ \ f_3\in\left(0,\frac{7}{6}\right)\ \ \ \ \ \Leftrightarrow\ \ \ \ 0<\sigma=\frac{2}{3}\sqrt{2}-\varkappa\leq f_3\frac{\ln t}{t}
\end{equation*}
all summands in the right hand side of \eqref{rig:0} are bounded. Thus, by positivity, all summands in the left hand side of \eqref{rig:0} have to be bounded. So, using also \eqref{TWproof},
\begin{eqnarray}
	\frac{t^{\frac{1}{2}}}{e^{\frac{2}{3}\sqrt{2}t}(1-\lambda_0(x))}\det\big(I+e^{-v}K_1(I-K_1)^{-1}\big)&=&\mathcal{O}(1),\nonumber\\
	t^{\frac{1}{2}}e^{-\sigma t}\det\big(I+e^{-v}K_1(I-K_1)^{-1}\big)&=&\mathcal{O}(1).\label{rig:00}
\end{eqnarray}
Since $\det(I+e^{-v}K_1(I-K_1)^{-1})\geq 1$, we deduce from the first estimation,
\begin{equation*}
	\frac{t^{\frac{1}{2}}}{e^{\frac{2}{3}\sqrt{2}t}(1-\lambda_0(x))}=\mathcal{O}(1),
\end{equation*}
which is consistent with \eqref{conj} for $j=0$. In fact, \eqref{rig:0} allows us to deduce the inequality,
\begin{equation}\label{rig:000}
	1-\lambda_0(x)\geq \sqrt{\pi}\,2^{\frac{9}{4}}t^{\frac{1}{2}}e^{-\frac{2}{3}\sqrt{2}\,t}\big(1+o(1)\big),
\end{equation}
but in order to achieve equality our analysis in the Stokes region has to be extended; so far \eqref{rig:00} only gives
\begin{equation*}
	\frac{t^{\frac{1}{2}}}{e^{\frac{2}{3}\sqrt{2}t}(1-\lambda_1(x))}=\mathcal{O}(1),
\end{equation*}
which is not sufficient yet to deduce equality in \eqref{rig:000}. Summarizing, provided Theorem \ref{res:5} is extended to the full Stokes region and simultaneously also Corollary \ref{eig0}, the results on the transition asymptotics of $u(x|s)$ would enable us to derive \eqref{conj} rigorously.\smallskip

The derivation of Theorem \ref{bet:2} follows largely its regular counterpart.
\begin{enumerate}
	\item[(4)] The master RHP \ref{masterRHP} is analyzed asymptotically subject to the constraint $t^{1-\eta}\leq v\leq t(\frac{2}{3}\sqrt{2}-\delta),\eta\in(0,1)$ in Sections \ref{singsec:1} and \ref{singsec:2}. Opposed to the regular case $(1)$ we require a different outer parametrix which has the singular structure and corresponding exceptional set \eqref{ex} encoded.
	\item[(5)] After that, we address the scale $t\geq v^{k+1}>0,t\geq t_0$ following \cite{BI}. The details are summarized in Section \ref{singsec:3}.
	\item[(6)] Again, by modular transformations, plus an additional argument related to the singular structure, we complete the proof of Theorem \ref{bet:2} in Section \ref{singsec:4} and in addition obtain Corollary \ref{singmat}.
\end{enumerate}
Section \ref{singsec:5}  concludes the manuscript by deriving Theorem \ref{res:9} and we list a few identities for complete elliptic integrals and Jacobi theta functions in Appendices \ref{appa} and \ref{app:theta}.

\section{Preliminary steps for regular and singular transition analysis}\label{sec:3}
We start with the scaling transformation $X(\lambda) = Y\left(\lambda\sqrt{-x}\right),\lambda\in\mathbb{C}\backslash \bigcup_1^6\Gamma_k$ and are lead from the initial RHP \ref{masterRHP} to a RHP for the function $X(\lambda)$. This problem is formulated on the same jump contour $\bigcup_1^6\Gamma_k$ with jump
\begin{equation*}
  X_+(\lambda)=X_-(\lambda)e^{-t\vartheta(\lambda)\sigma_3}S_ke^{t\vartheta(\lambda)\sigma_3},\hspace{0.5cm}\lambda\in\Gamma_k
\end{equation*}
and jump exponent
\begin{equation*}
  \vartheta(\lambda) = \im\left(\frac{4}{3}\lambda^3-\lambda\right),\hspace{0.5cm} t =(-x)^{\frac{3}{2}}.
\end{equation*}
Next we deform the original jump contours, to obtain a RHP for a function $Z(\lambda)$ as shown in Figure \ref{figure2} - at this stage the endpoints $\lambda=\pm\lambda^{\ast}$ are yet to be determined. 
\begin{problem} The function $Z(\lambda)$ has the following analytical properties:
\begin{itemize}
 \item $Z(\lambda)$ is analytic for $\lambda\in\mathbb{C}\backslash([-\lambda^{\ast},\lambda^{\ast}]\cup\bigcup_1^6\gamma_k)$
  \item We have the following jumps, compare Figure \ref{figure2},
\begin{align*}
  Z_+(\lambda) &= Z_-(\lambda)e^{-t\vartheta(\lambda)\sigma_3}S_ke^{t\vartheta(\lambda)\sigma_3}, \hspace{1.7cm} \lambda\in\gamma_k,\ k=1,\ldots,6;\\
  Z_+(\lambda) &= Z_-(\lambda)e^{-t\vartheta(\lambda)\sigma_3}S_5^{-1}S_4^{-1}S_3^{-1}e^{t\vartheta(\lambda)\sigma_3}, \ \ \lambda\in[-\lambda^{\ast},0];\\
  Z_+(\lambda) &= Z_-(\lambda)e^{-t\vartheta(\lambda)\sigma_3}\sigma_2 S_3S_4S_5\sigma_2e^{t\vartheta(\lambda)\sigma_3}, \ \ \lambda\in[0,\lambda^{\ast}],
\end{align*}
where in fact
\begin{equation}\label{Zjump}
  S_5^{-1}S_4^{-1}S_3^{-1} = \sigma_2S_3S_4S_5\sigma_2,\hspace{1cm}\sigma_2=\begin{pmatrix}
  0 & -\im\\
  \im & 0\\
  \end{pmatrix}.
\end{equation}
  \item As $\lambda\rightarrow\infty$, the function $Z(\lambda)$ is normalized as
  \begin{equation*}
    Z(\lambda) = I+\mathcal O\left(\lambda^{-1}\right).
  \end{equation*}
\end{itemize}
\end{problem}

\begin{figure}[tbh]
\begin{minipage}{0.4\textwidth} 
\begin{center}
\resizebox{1.1\textwidth}{!}{\includegraphics{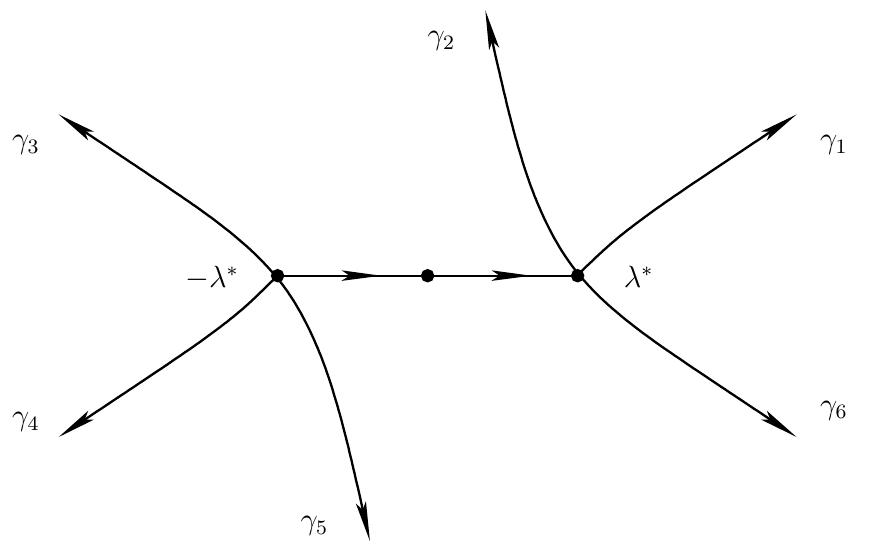}}
\caption{Deformation of original jump contours, $Y(\lambda)\mapsto Z(\lambda)$. The three dots represent the points $\lambda=-\lambda^{\ast},0,\lambda^{\ast}$.}
\label{figure2}
\end{center}
\end{minipage}
\begin{minipage}{0.5\textwidth}
\begin{center}
\resizebox{0.85\textwidth}{!}{\includegraphics{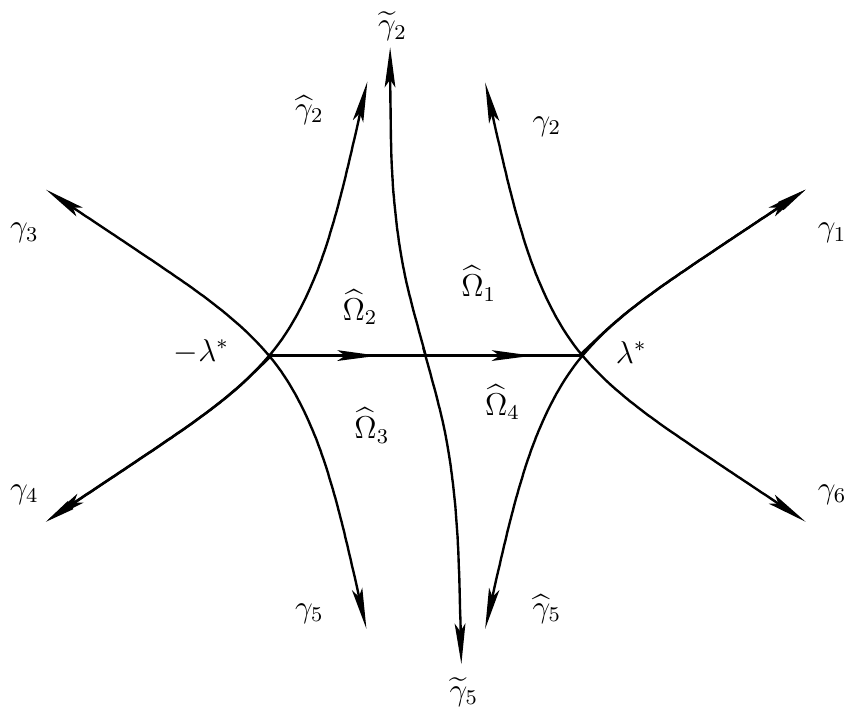}}
\caption{First opening of lens, $Z(\lambda)\mapsto T(\lambda)$.}
\label{figure4}
\end{center}
\end{minipage}
\end{figure}
Next, referring to the domains $\widehat{\Omega}_j,j=1,\ldots,4$ shown in Figure \ref{figure4}, we set
\begin{equation}\label{Tfunc}
  T(\lambda) = Z(\lambda)\begin{cases}
                e^{-t\vartheta(\lambda)\sigma_3}S_2^{-1}e^{t\vartheta(\lambda)\sigma_3},&\lambda\in\widehat{\Omega}_1\\
e^{-t\vartheta(\lambda)\sigma_3}S_4^{-1}e^{t\vartheta(\lambda)\sigma_3},&\lambda\in\widehat{\Omega}_2\\
e^{-t\vartheta(\lambda)\sigma_3}S_5^{-1}e^{t\vartheta(\lambda)\sigma_3},&\lambda\in\widehat{\Omega}_3\\
e^{-t\vartheta(\lambda)\sigma_3}S_1^{-1}e^{t\vartheta(\lambda)\sigma_3},&\lambda\in\widehat{\Omega}_4\\
I,&\textnormal{else}.
               \end{cases}
\end{equation}
and are lead to a RHP for $T(\lambda)$ with jumps (compare Figure \ref{figure4} for the jump contours)
\begin{align*}
  T_+(\lambda)&=T_-(\lambda)e^{-t\vartheta(\lambda)\sigma_3}S_ke^{t\vartheta(\lambda)\sigma_3},\ \ \lambda\in\gamma_k,\ k=1,3,4,6\\
  T_+(\lambda)&=T_-(\lambda)e^{-t\vartheta(\lambda)\sigma_3}S_4e^{t\vartheta(\lambda)\sigma_3}, \ \ \lambda\in\widehat{\gamma}_2;\hspace{1cm}
  T_+(\lambda)=T_-(\lambda)e^{-t\vartheta(\lambda)\sigma_3}S_1e^{t\vartheta(\lambda)\sigma_3}, \ \ \lambda\in\widehat{\gamma}_5\\
\end{align*}
and
\begin{align*}
  T_+(\lambda)&=T_-(\lambda)e^{-t\vartheta(\lambda)\sigma_3}S_2S_4^{-1}e^{t\vartheta(\lambda)\sigma_3},\ \ \lambda\in\widetilde{\gamma}_2;\hspace{0.75cm}
  T_+(\lambda)=T_-(\lambda)e^{-t\vartheta(\lambda)\sigma_3}\sigma_2 S_2S_4^{-1}\sigma_2e^{t\vartheta(\lambda)\sigma_3},\ \ \lambda\in\widetilde{\gamma}_5\\
  T_+(\lambda)&=T_-(\lambda)e^{-t\vartheta(\lambda)\sigma_3}S_4^{-1}S_3^{-1}S_4^{-1}e^{t\vartheta(\lambda)\sigma_3};\hspace{1.25cm}
  T_+(\lambda)=T_-(\lambda)e^{-t\vartheta(\lambda)\sigma_3}\sigma_2S_4S_3S_4\sigma_2e^{t\vartheta(\lambda)\sigma_3},
\end{align*}
for $\lambda\in[-\lambda^{\ast},0]$ and $\lambda\in[0,\lambda^{\ast}]$ in the last two cases. Notice that 
\begin{equation}\label{st}
  S_4^{-1}S_3^{-1}S_4^{-1} = \begin{pmatrix}
                              1-s_1s_3 & s_1+s_1(1-s_1s_3)\\
-s_3 & 1-s_1s_3\\
                             \end{pmatrix},\ \ 
\sigma_2 S_4S_3S_4\sigma_2 =\begin{pmatrix}
                               1-s_1s_3 & -s_3\\
s_1+s_1(1-s_1s_3) & 1-s_1s_3\\
                              \end{pmatrix}.
\end{equation}
The presence of the factors $1-s_1s_3=1-|s_1|^2$ on the diagonal motivates the use of a $g$-function in order to address the transition asymptotics $x\rightarrow-\infty,|s_1|\rightarrow1$. Subsequently we shall first work out the necessary details in case of the regular transition.

\section{Regular transition analysis for $\varkappa\in\big[\delta,\frac{2}{3}\sqrt{2}-\delta\big]$ with $0<\delta<\frac{1}{3}\sqrt{2}$ fixed}\label{sec:reg}
 We introduce the double scaling parameter
\begin{equation}\label{dbpara}
  \varkappa =\frac{v}{t}\equiv -\frac{\ln\big(1-|s_1|^2\big)}{(-x)^{\frac{3}{2}}} \equiv -\frac{2\pi\beta}{t}\in(0,\infty);\hspace{0.75cm}t=(-x)^{\frac{3}{2}},\ \ \beta=\frac{1}{2\pi}\ln\big(1-|s_1|^2\big)
\end{equation}
and keep $\varkappa\in[\delta,\frac{2}{3}\sqrt{2}-\delta]$ with $0<\delta<\frac{1}{3}\sqrt{2}$ fixed throughout. Based on the cyclic constraints \eqref{cyclic} the parameter $\varkappa$ allows us to derive expansions of the Stokes multipliers in terms of $\varkappa$ which are used later on.
\subsection{Preliminary expansions}
 First observe from \eqref{cyclic} and \eqref{dbpara} the exact identity $\Re(s_1)=\frac{s_2}{2}e^{-\varkappa t}$. Second, by definition \eqref{dbpara},
\begin{equation*}
	|s_1| = 1-\frac{1}{2}e^{-\varkappa t}+\mathcal{O}\left(e^{-2\varkappa t}\right),\ \ t\rightarrow\infty;\ \ \ \ \textnormal{arg}(s_1)=\epsilon\arccos\left(\frac{\Re(s_1)}{|s_1|}\right),\ \ \ \arccos:[-1,1]\rightarrow[0,\pi],
\end{equation*}	
and we obtain for $s_1=|s_1|e^{\im\textnormal{arg}(s_1)}$,
\begin{prop} As $t\rightarrow\infty,|s_1|\uparrow 1$ with $e^{\im\frac{\pi}{2}\epsilon}=\im\epsilon$,
\begin{equation}\label{approx:3}
	s_1=\im\epsilon\left(1-\frac{1}{2}(1+\im\epsilon s_2)e^{-\varkappa t}+\mathcal{O}\left(e^{-2\varkappa t}\right)\right),\ \ s_3=-\im\epsilon\left(1-\frac{1}{2}(1-\im\epsilon s_2)e^{-\varkappa t}+\mathcal{O}\left(e^{-2\varkappa t}\right)\right),
\end{equation}
and
\begin{equation}\label{approx:4}
	s_1+s_1(1-s_1s_3)=\im\epsilon\left(1+\frac{1}{2}(1-\im\epsilon s_2)e^{-\varkappa t}+\mathcal{O}\left(e^{-2\varkappa t}\right)\right),
\end{equation}
uniformly for $\varkappa\in[\delta,\frac{2}{3}\sqrt{2}-\delta]$ with $0<\delta<\frac{1}{3}\sqrt{2}$ fixed.
\end{prop}
Let us now move ahead and introduce one of the key ingredients in the nonlinear steepest descent analysis, the $g$-function.
\subsection{Introduction of g-function} Set $J=(-1,-\k)\cup(1,\k)\subset\mathbb{R}$ where $\k\in(0,1)$ is determined implicitly in \eqref{inteq:1},
\begin{equation}\label{inteq:2}
	\varkappa=\frac{2}{3}\sqrt{\frac{2}{1+\k^2}}\left[E'-\frac{2\k^2}{1+\k^2}K'\right] = \left(\frac{2}{1+\k^2}\right)^{\frac{3}{2}}\int_{\k}^1\sqrt{(1-\mu^2)(\mu^2-\k^2)}\,\d\mu.
\end{equation}
The following Proposition addresses the bijective correspondence between $\varkappa\in(0,\frac{2}{3}\sqrt{2})$ and $\k\in(0,1)$.
\begin{prop}\label{modex} The modulus $\k$ is uniquely determined via \eqref{inteq:2} in case $\varkappa\in(0,\frac{2}{3}\sqrt{2})$. Moreover, as $\varkappa\downarrow 0$,
\begin{equation}\label{I:1}
	\k=1-2\sqrt{\frac{\varkappa}{\pi}}+\frac{2\varkappa}{\pi}-\frac{29}{8}\left(\frac{\varkappa}{\pi}\right)^{\frac{3}{2}}+\mathcal{O}\left(\varkappa^2\right),
\end{equation}
and, as $\varkappa\uparrow\frac{2}{3}\sqrt{2}$, with $\sigma=\sigma(\varkappa)=\frac{2}{3}\sqrt{2}-\varkappa\downarrow 0$,
\begin{equation}\label{I:2}
	\k=\sqrt{\frac{\sqrt{2}\,\sigma}{|\ln\sigma|}}\left(1+\frac{\ln|\ln\sigma|}{2\ln\sigma}+\frac{2+7\ln 2}{4\ln\sigma}+\mathcal{O}\left(\frac{\ln|\ln\sigma|}{\ln^2\sigma}\right)\right).
\end{equation}
\end{prop}
\begin{proof} Consider the function
\begin{equation*}
	\mathrm{F}(\k,\varkappa)=\varkappa-\mathrm{I}(\k),\ \ \k\in[0,1],
\end{equation*}
where
\begin{equation*}
	\mathrm{I}(\k)=\left(\frac{2}{1+\k^2}\right)^{\frac{3}{2}}\int_{\k}^1\sqrt{(1-\mu^2)(\mu^2-\k^2)}\,\d\mu.
\end{equation*}
With the help of standard expansions for the complete elliptic integrals as $\k\downarrow 0$ and $\k\uparrow 1$ (compare Appendix \ref{appa}), we get
\begin{eqnarray}
	\mathrm{I}(\k)&=&\frac{2}{3}\sqrt{2}\left[1-\frac{3}{2}\k^2|\ln \k|-\frac{3}{4}(1+4\ln 2)\k^2+\frac{39}{16}\k^4|\ln \k|+\mathcal{O}\left(\k^4\right)\right],\ \ \ \k\downarrow 0,\label{I:k1}\\
	\mathrm{I}(\k)&=&\frac{\pi}{4}(1-\k)^2\left[1+(1-\k)+\frac{11}{32}(1-\k)^2+\mathcal{O}\left((1-\k)^3\right)\right],\ \ \ \k\uparrow 1.\label{I:k2}
\end{eqnarray}
Thus
\begin{equation*}
	\lim_{\k\downarrow 0}\mathrm{F}(\k,\varkappa)=\varkappa-\frac{2}{3}\sqrt{2}<0,\hspace{1cm}\lim_{\k\uparrow 1}\mathrm{F}(\k,\varkappa)=\varkappa>0,
\end{equation*}
but $\mathrm{F}(\k,\varkappa)$ is real analytic in a neighborhood of $(\k\in(0,1],\varkappa\in[0,\infty))$ with the first partial derivatives equal to
\begin{eqnarray*}
	\mathrm{F}_{\k}(\k,\varkappa)&=&\k\left(\frac{2}{1+\k^2}\right)^{\frac{3}{2}}\left[\frac{3}{1+\k^2}\int_{\k}^1\sqrt{(1-\mu^2)(\mu^2-\k^2)}\,\d\mu+\int_{\k}^1\sqrt{\frac{1-\mu^2}{\mu^2-\k^2}}\,\d\mu\right]>0,\\
	\mathrm{F}_{\varkappa}(\k,\varkappa)&=&1.
\end{eqnarray*}
Hence the implicit function theorem guarantees existence of a unique real analytic solution $\k=\k(\varkappa)$ of the equation $\mathrm{F}(\k,\varkappa)=0$ near the point $(\k,\varkappa)$. We use \eqref{I:k2} and \eqref{I:k1} to derive \eqref{I:1} and \eqref{I:2}.
\end{proof}
Besides the expansions \eqref{I:1} and \eqref{I:2} for $\k=\k(\varkappa)$ itself, we will later on also require expansions of the frequency $V=V(\varkappa)$ and module $\tau=\tau(\varkappa)$ introduced in \eqref{Vfancy}. These follow directly from \eqref{I:1} and \eqref{I:2} and are summarized in the Corollary below.
\begin{cor}\label{cor1} As $\varkappa\downarrow 0$,
\begin{equation}\label{l:1}
	V(\varkappa)=-\frac{2}{3\pi}-\frac{\varkappa}{2\pi^2}\ln\varkappa+\frac{\varkappa}{2\pi^2}(1+\ln 16\pi)+\mathcal{O}\left(\varkappa^2\right),\ \ \ 
	\tau(\varkappa)=-\frac{\im}{\pi}\ln\left(\frac{\varkappa}{16\pi}\right)-\frac{17\im}{8\pi}\frac{\varkappa}{\pi}+\mathcal{O}\big(\varkappa^{\frac{3}{2}}\big).
\end{equation}
Secondly, as $\sigma=\frac{2}{3}\sqrt{2}-\varkappa\downarrow 0$,
\begin{equation}\label{l:2}
	V(\varkappa)=-\frac{\sigma}{|\ln\sigma|}\left(1+\frac{\ln|\ln\sigma|}{\ln\sigma}+\frac{2+7\ln 2}{2\ln\sigma}+\mathcal{O}\left(\frac{\ln^2|\ln\sigma|}{\ln^2\sigma}\right)\right),\ \ \ \ \ \ 
\end{equation}
and in the same limit also
\begin{equation}\label{l:3}
	\tau'\equiv-\frac{1}{\tau(\varkappa)}=-\frac{|\ln\sigma|}{2\pi\im}\left(1-\frac{\ln|\ln\sigma|}{\ln\sigma}-\frac{7\ln 2}{2\ln\sigma}+\frac{\ln|\ln\sigma|}{\ln^2\sigma}+\mathcal{O}\left(\frac{1}{\ln^2\sigma}\right)\right).
\end{equation}
\end{cor}
We now define the $g$-function,
\begin{equation}\label{gf:1}
	g(z)=4\im\int_M^z\big((\mu^2-M^2)(\mu^2-m^2)\big)^{\frac{1}{2}}\,\d\mu,\ \ \ \ z\in\mathbb{C}\backslash[-M,M]
\end{equation}
where
\begin{equation}\label{b:point}
	0<m=\frac{1}{\sqrt{2}}\frac{\k}{\sqrt{1+\k^2}}<\frac{1}{2}<M=\frac{1}{\sqrt{2}}\frac{1}{\sqrt{1+\k^2}}<\frac{1}{\sqrt{2}}
\end{equation}
and the contour of integration is chosen in the simply connected domain $\mathbb{CP}^1\backslash[-M,M]$. Moreover, we fix
\begin{equation*}
	-\pi<\textnormal{arg}\,\left(\big(\mu^2-M^2\big)\big(\mu^2-m^2\big)\right)\leq\pi\ \ \ \ \textnormal{such that}\ \ \ \sqrt{(\mu^2-M^2)(\mu^2-m^2)}>0\ \ \textnormal{for}\ \ \mu>M,
\end{equation*}
and thus, $g(z)$ is single-valued and analytic in $\mathbb{CP}^1\backslash[-M,M]$. As $z$ tends to infinity
\begin{equation*}
  g(z) = \vartheta(z)+\ell+\frac{\im}{8z}\left(\frac{1-\k^2}{1+\k^2}\right)^2+\mathcal O\left(z^{-3}\right),\hspace{0.5cm}z\rightarrow\infty,
\end{equation*}
where
\begin{equation}\label{gf:2}
  \ell=-\vartheta(M)+4\im\int_M^{\infty}\left[\sqrt{(\mu^2-M^2)(\mu^2-m^2)}-\mu^2+\frac{1}{4}\right]\d\mu.
\end{equation}
Further steps in the analysis require certain analytical properties of $g(z)$.
\begin{prop}\label{gprop} For $z\in\mathbb{R}$ introduce
\begin{equation*}
  \Omega(z) = \im\big(g_+(z)+g_-(z)\big),\ \ \ \textnormal{and}\ \ \ \Pi(z) = g_+(z)-g_-(z),\hspace{0.5cm}\textnormal{with}\ \ \ g_{\pm}(z)=\lim_{\varepsilon\downarrow 0}g(z\pm\im\varepsilon).
\end{equation*}
The functions $\Omega(z)$ and $\Pi(z)$ are real-valued on the real line, in fact
\begin{eqnarray*}
  \Omega(z) &=&-8\int_M^z\sqrt{\big(\mu^2-M^2\big)\big(\mu^2-m^2\big)}\,\d\mu,\hspace{0.25cm}z\in(M,+\infty),\hspace{0.75cm} \Omega(z) = 0,\hspace{0.25cm}z\in(m,M),\\
  \Omega(z) &=& -8\int_z^m\sqrt{\big(M^2-\mu^2\big)\big(m^2-\mu^2\big)}\,\d\mu,\hspace{0.25cm}z\in(-m,m),\\
  \Omega(z) &=& -8\int_{-m}^m\sqrt{\big(M^2-\mu^2\big)\big(m^2-\mu^2\big)}\,\d\mu\equiv 2\pi V(\varkappa),\hspace{0.25cm}z\in(-M,-m),\\
  \Omega(z) &=& -8\int_{-m}^m\sqrt{\big(M^2-\mu^2\big)\big(m^2-\mu^2)}\,\d\mu+8\int_z^{-M}\sqrt{\big(\mu^2-M^2\big)\big(\mu^2-m^2\big)}\,\d\mu,\hspace{0.25cm}z\in(-\infty,-M).
\end{eqnarray*}
Moreover
\begin{eqnarray*}
  \Pi(z) &=& 0,\hspace{0.25cm} z\in(-\infty,-M)\cup(M,+\infty),\hspace{0.5cm}\Pi(z) = 8\int_z^M\sqrt{\big(M^2-\mu^2\big)\big(\mu^2-m^2\big)}\,\d\mu,\hspace{0.25cm} z\in(m,M),\\
  \Pi(z) &=& 8\int_m^M\sqrt{\big(M^2-\mu^2\big)\big(\mu^2-m^2\big)}\,\d\mu\equiv\varkappa,\hspace{0.25cm} z\in(-m,m),\\
  \Pi(z) &=& \varkappa-8\int_z^{-m}\sqrt{\big(M^2-\mu^2\big)\big(\mu^2-m^2\big)}\,\d\mu,\hspace{0.25cm} z\in(-M,-m).
\end{eqnarray*}
\end{prop}
\subsection{The g-function transformation}
We first go back to the RHP for $T(\lambda)$ as defined in \eqref{Tfunc} with jump contour $\Sigma_T$ shown in Figure \ref{figure4}. Now fix the endpoint
\begin{equation*}
 	\lambda^{\ast}=\frac{1}{\sqrt{2}}\frac{1}{\sqrt{1+\k^2}}=M
\end{equation*}	
and employ the following transformation,
\begin{equation}\label{gtraf:1}
  S(\lambda) = e^{-t\ell\sigma_3}T(\lambda)e^{t(g(\lambda)-\vartheta(\lambda))\sigma_3},\hspace{0.5cm}\lambda\in\mathbb{C}\backslash\Sigma_T,\hspace{0.25cm}
\Sigma_T=\left([-M,M]
  \cup\bigcup_{k=1}^6 \gamma_k\cup\hat{\gamma}_2\cup\hat{\gamma}_5\cup\tilde{\gamma}_2\cup\tilde{\gamma}_5\right),
\end{equation}
with $g=g(z)$ as in \eqref{gf:1} and $\ell$ in \eqref{gf:2}. Since all jumps in the $T$-RHP display the structure
\begin{equation*}
  T_+(\lambda) = T_-(\lambda) e^{-t\vartheta(\lambda)\sigma_3}G_T(\lambda)e^{t\vartheta(\lambda)\sigma_3},\hspace{0.5cm}\lambda\in\mathbb{C}\backslash\Sigma_T,
\end{equation*}
we are lead to the following RHP
\begin{problem} Determine the $2\times 2$ matrix valued function $S(\lambda)=S(\lambda;x,s)$ such that
\begin{itemize}
	\item $S(\lambda)$ is analytic for $\lambda\in\mathbb{C}\backslash\Sigma_S$ where the jump contour $\Sigma_S$ is identical to the contour $\Sigma_T$ shown in Figure \ref{figure4}.
	\item The function $S(\lambda)$ satisfies the jump condition
\begin{equation*}
  S_+(\lambda) = S_-(\lambda) \underbrace{e^{-tg_-(\lambda)\sigma_3}G_T(\lambda)e^{tg_+(\lambda)\sigma_3}}_{=G_S(\lambda)},\ \ \lambda\in\Sigma_S\equiv\Sigma_T.
\end{equation*}
	\item As $\lambda\rightarrow\infty$,
	\begin{equation*}
		S(\lambda)=I+\mathcal{O}\left(\lambda^{-1}\right).
	\end{equation*}
\end{itemize}
\end{problem}
Recall at this point \eqref{st}: for $\lambda\in(m,M)$ we have thus
\begin{equation*}
  S_+(\lambda) = S_-(\lambda)\begin{pmatrix}
                              e^{-t(\varkappa-\Pi(\lambda))} & -s_3\\
 s_1+s_1(1-s_1s_3)& (1-s_1s_3)e^{-t\Pi(\lambda)}\\
                             \end{pmatrix},\hspace{0.5cm}\lambda\in(m,M)
\end{equation*}
with $\varkappa$ from \eqref{dbpara}. But (compare Proposition \ref{gprop})
\begin{equation*}
  \varkappa-\Pi(\lambda) = 8\int_m^{\lambda}\sqrt{\big(M^2-\mu^2\big)\big(\mu^2-m^2\big)}\,\d\mu>0,\hspace{0.5cm}\lambda\in(m,M)
\end{equation*}
and $\Pi(\lambda)>0$ in the right slit. Hence, using also the expansions \eqref{approx:3},\eqref{approx:4},
\begin{equation}\label{ob1}
  G_S(\lambda)\begin{pmatrix}
  	0 & -e^{-\im\frac{\pi}{2}\epsilon}\\
	e^{\im\frac{\pi}{2}\epsilon} & 0\\
	\end{pmatrix}^{-1}\rightarrow I,\hspace{0.5cm} x\rightarrow-\infty,|s_1|\uparrow 1:\ \ \varkappa\in\left[\delta,\frac{2}{3}\sqrt{2}-\delta\right],\ \ 0<\delta<\frac{1}{3}\sqrt{2}
\end{equation}
uniformly in $\lambda$ chosen from any compact subset of the right slit $(m,M)$. In the left slit $(-M,-m)$ a similar situation occurs,
\begin{equation*}
  S_+(\lambda)=S_-(\lambda)\begin{pmatrix}
                            e^{-t(\varkappa-\Pi(\lambda))} & (s_1+s_1(1-s_1s_3))e^{\im t\Omega(\lambda)}\\
-s_3e^{-\im t\Omega(\lambda)} & (1-s_1s_3)e^{-t\Pi(\lambda)}
                           \end{pmatrix},\hspace{0.5cm}\lambda\in(-M,-m).
\end{equation*}
But here
\begin{equation*}
 \varkappa-\Pi(\lambda) = 8\int_{\lambda}^{-m}\sqrt{\big(M^2-\mu^2\big)\big(\mu^2-m^2\big)}\,\d\mu>0,\hspace{0.5cm}\lambda\in(-M,-m)
\end{equation*}
and also
\begin{equation*}
  \Pi(\lambda) = 8\int_{-\lambda}^M\sqrt{\big(M^2-\mu^2\big)\big(\mu^2-m^2\big)}\,\d\mu>0,\hspace{0.5cm}\lambda\in(-M,-m).
\end{equation*}
Hence combined with \eqref{approx:3},\eqref{approx:4},
\begin{equation}\label{ob4}
  G_S(\lambda)\begin{pmatrix}
               0 & e^{\im\frac{\pi}{2}\epsilon+\im t\Omega(\lambda)}\\
-e^{-\im\frac{\pi}{2}\epsilon-\im t\Omega(\lambda)} & 0\\
              \end{pmatrix}^{-1}\rightarrow I,\hspace{0.5cm} x\rightarrow-\infty,|s_1|\uparrow 1:\ \ \varkappa\in\left[\delta,\frac{2}{3}\sqrt{2}-\delta\right],\ 0<\delta<\frac{1}{3}\sqrt{2}
\end{equation}
uniformly in $\lambda$ chosen from any compact subset of the left slit $(-M,-m)$. Next, we consider the jump contours which extend to infinity, i.e. 
\begin{equation*}
  \Sigma_{S_{\infty}}=\hat{\gamma}_2\cup\hat{\gamma}_5\cup\tilde{\gamma}_2\cup\tilde{\gamma}_5\cup\bigcup\gamma_k.
\end{equation*}
Along these contours the jumps in the $S$-RHP are given by
\begin{equation*}
  S_+(\lambda)=S_-(\lambda) e^{-tg(\lambda)\sigma_3}G_T(\lambda)e^{tg(\lambda)\sigma_3},\hspace{0.5cm}\lambda\in\Sigma_{S_{\infty}}
\end{equation*}
Hence, by triangularity and the sign chart of $\Re(g(\lambda))$ (compare Figure \ref{figure6}), the jumps on the infinite contours approach the identity matrix exponentially fast in the limit
$x\rightarrow-\infty,|s_1|\uparrow 1$ with $\varkappa\in[\delta,\frac{2}{3}\sqrt{2}-\delta],\delta>0$, provided we stay away from the endpoints $\lambda=\pm m,\pm M$ and the origin $\lambda=0$. In the remaining gap $(-m,m)\subset\mathbb{R}$ two cases need to be distinguished: First for $\lambda\in(0,m)$,
\begin{eqnarray}
  S_+(\lambda) &=& S_-(\lambda)\begin{pmatrix}
                              1 & -s_3 e^{\im t\Omega(\lambda)}\\
(s_1+s_1(1-s_1s_3))e^{-\im t\Omega(\lambda)} & (1-s_1s_3)^2\\
                             \end{pmatrix}\nonumber\\
&=&S_-(\lambda)\begin{pmatrix}
1 & 0\\
(s_1+s_1(1-s_1s_3))e^{-\im t\Omega(\lambda)}& 1\\
\end{pmatrix}\begin{pmatrix}
1 & -s_3e^{\im t\Omega(\lambda)}\\
0 & 1\\
\end{pmatrix}\nonumber\\
&\equiv&S_-(\lambda) S_{L_1}(\lambda)S_{U_1}(\lambda),\hspace{0.5cm}\lambda\in(0,m)\label{ob2}
\end{eqnarray}
and secondly for $\lambda\in(-m,0)$,
\begin{eqnarray}
  S_+(\lambda) &=&S_-(\lambda)\begin{pmatrix}
                               1 & (s_1+s_1(1-s_1s_3))e^{\im t\Omega(\lambda)}\\
-s_3e^{-\im t\Omega(\lambda)} & (1-s_1s_3)^2\\
                              \end{pmatrix}\nonumber\\
&=&S_-(\lambda)\begin{pmatrix}
                1 & 0\\
-s_3e^{-\im t\Omega(\lambda)} & 1\\
               \end{pmatrix}\begin{pmatrix}
1 & (s_1+s_1(1-s_1s_3))e^{\im t\Omega(\lambda)} \\
0 & 1\\
\end{pmatrix}\nonumber\\
&\equiv&S_-(\lambda)S_{L_2}(\lambda)S_{U_2}(\lambda),\hspace{0.5cm}\lambda\in(-m,0).\label{ob3}
\end{eqnarray}
In \eqref{ob2} and \eqref{ob3} all off-diagonal entries are fast oscillating as $x\rightarrow-\infty$, see Proposition \ref{gprop}. We now transform this behavior on $(-m,m)$ to exponential decay with the help of contour deformations ({\it opening of lens}) tailored to the factorizations written in \eqref{ob2} and \eqref{ob3}. The key to this explicit transformation is the following Proposition.
\begin{prop}\label{openup} Introduce for $z\in(-m,m)$ the functions
 \begin{equation*}
    H_1(z) = \im\Omega(z),\hspace{0.5cm} H_2(z) = -\im\Omega(z),
 \end{equation*}
with $\Omega=\Omega(z)$ as in Proposition \ref{gprop}. Then $H_1(z)$ admits local analytical continuation into a neighborhood of the gap $(-m,m)$ into the upper half-plane such that
\begin{equation*}
  \Re \big(H_1(z)\big)<0,\hspace{0.75cm}\Im z>0,\hspace{0.25cm} \Re z\in(-m,m).
\end{equation*}
Similarly, the function $H_2(z)$ admits local analytical continuation into a neighborhood of the gap $(-m,m)$ into the lower half-plane such that
\begin{equation*}
  \Re\big(H_2(z)\big)<0,\hspace{0.75cm}\Im z<0,\hspace{0.25cm} \Re z\in(-m,m)
\end{equation*}
\end{prop}
\begin{proof} 
 We follow the standard line of argument. First
\begin{equation*}
  \im\Omega(z) = -8\im\int_z^m\sqrt{\big(M^2-\mu^2\big)\big(m^2-\mu^2\big)}\,\d\mu,\hspace{0.5cm} z\in(-m,m)
\end{equation*}
and thus
\begin{eqnarray*}
  \frac{\d}{\d y}H_1(x+\im y)\bigg|_{y=0} &=&-8\sqrt{\big(M^2-x^2\big)\big(m^2-x^2\big)}<0,\hspace{0.5cm}x\in(-m,m),\\
  \frac{\d}{\d y}H_2(x-\im y)\bigg|_{y=0} &=&-8\sqrt{\big(M^2-x^2\big)\big(m^2-x^2\big)}<0,\hspace{0.5cm}x\in(-m,m),
\end{eqnarray*}
which implies the claim via the Cauchy-Riemann equations.
\end{proof}
\begin{figure}[tbh]
\begin{center}
\resizebox{0.8\textwidth}{!}{\includegraphics{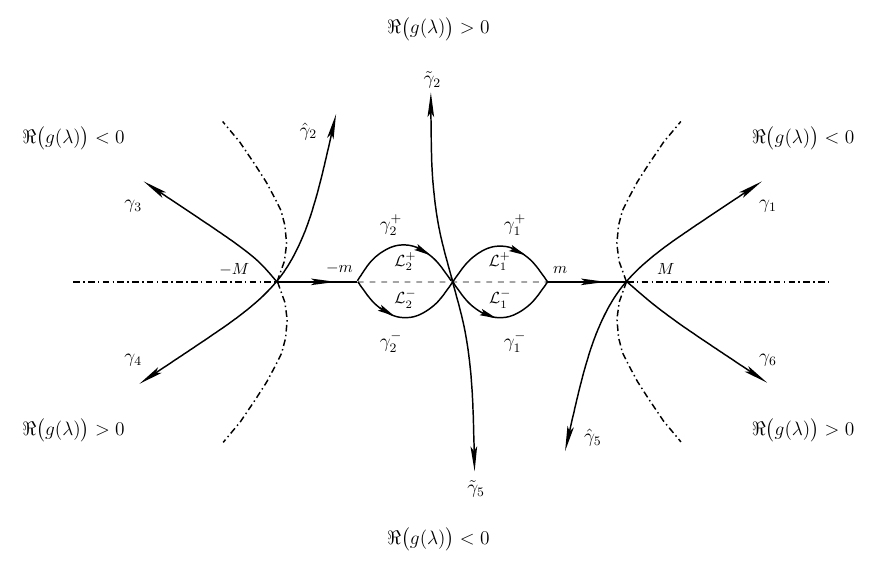}}
\caption{Opening of lens, the jump contours for $L(\lambda)$ consists of the solid black lines. In addition, along the dotted lines $\Re(g(\lambda))=0$.}
\label{figure6}
\end{center}
\end{figure}
\subsection{Opening of lens} Suppose $\mathcal{L}_j^{\pm},j=1,2$ denotes the lens shaped regions shown in Figure \ref{figure6}. With the help of the analytical continuations of $H_1(z)$ and $H_2(z)$ as discussed in Proposition \ref{openup}, we set
\begin{equation}\label{ou:1}
  L(\lambda) = \begin{cases}
                S(\lambda)S_{U_1}^{-1}(\lambda),&\lambda\in\mathcal{L}_1^+\\
S(\lambda)S_{L_1}(\lambda),&\lambda\in\mathcal{L}_1^-\\
S(\lambda)S_{U_2}^{-1}(\lambda),&\lambda\in\mathcal{L}_2^+\\
S(\lambda)S_{L_2}(\lambda),&\lambda\in\mathcal{L}_2^-\\
S(\lambda),&\textnormal{otherwise}
               \end{cases}
\end{equation}
so that $L(\lambda)$ solves the following RHP
\begin{problem}\label{openRHP} Determine the $2\times 2$ piecewise analytic function $L(\lambda)$ such that
\begin{itemize}
 \item $L(\lambda)$ is analytic for $\lambda\in\mathbb{C}\backslash\left([-M,-m]\cup[m,M]\cup\gamma_1^{\pm}\cup\gamma_2^{\pm}\cup\Sigma_{S_{\infty}}\right)$
 \item We have the following jump behavior, with orientation as indicated in Figure \ref{figure6},
\begin{equation*}
  L_+(\lambda) = L_-(\lambda)\begin{cases}
G_S(\lambda),&\lambda\in(-M,-m)\cup(m,M)\cup\Sigma_{S_{\infty}}\\
                              S_{U_1}(\lambda),&\lambda\in\gamma_1^+\\
S_{L_1}(\lambda),&\lambda\in\gamma_1^-\\
S_{U_2}(\lambda),&\lambda\in\gamma_2^+\\
S_{L_2}(\lambda),&\lambda\in\gamma_2^-
                             \end{cases}
\end{equation*}
  \item As $\lambda\rightarrow\infty$,
\begin{equation*}
 L(\lambda)=I+\mathcal O\left(\lambda^{-1}\right).
\end{equation*}
\end{itemize}
\end{problem}
After employing the explicit transformation \eqref{ou:1}, it is now time to focus on the local model problems near the points $\lambda=\pm m,\pm M$ and $\lambda=0$ as well as on the slit segment $J=(-M,-m)\cup(m,M)$.
\subsection{The outer parametrix}\label{outer1para}
The RHP associated to the outer parametrix is motivated by \eqref{ob1},\eqref{ob4} and thus consists in 
\begin{problem}\label{oup:2} Find a $2\times 2$ matrix-valued piecewise analytic function $N(\lambda)=N(\lambda;t,\epsilon)$ such that
\begin{itemize}
	\item $N(\lambda)$ is analytic for $\lambda\in\mathbb{C}\backslash\overline{J}$
  \item Along the branch cuts, with orientation as in \eqref{figure6},
\begin{eqnarray}
  N_+(\lambda) &=& N_-(\lambda)\begin{pmatrix}
                              0 & e^{\im\frac{\pi}{2}\epsilon+\im t\Omega(\lambda)}\\
-e^{-\im\frac{\pi}{2}\epsilon-\im t\Omega(\lambda)} & 0\\
                             \end{pmatrix},\hspace{0.5cm}\lambda\in(-M,-m)\label{ou:1}\\
  N_+(\lambda)&=&N_-(\lambda)\begin{pmatrix}
                              0 & -e^{-\im\frac{\pi}{2}\epsilon}\\
  e^{\im\frac{\pi}{2}\epsilon} & 0\\
                             \end{pmatrix},\hspace{0.5cm}\lambda\in(m,M)\label{ou:2}
\end{eqnarray}
  \item $N(\lambda)$ is square integrable on $\overline{J}$
  \item As $\lambda\rightarrow\infty$,
  \begin{equation*}
    N(\lambda) = I+\mathcal O\left(\lambda^{-1}\right)
  \end{equation*}
\end{itemize}
\end{problem}
Observe that the function
\begin{equation*}
	\widetilde{N}(\lambda)=e^{-\im\frac{\pi}{4}\epsilon\sigma_3}N(\lambda)e^{\im\frac{\pi}{4}\epsilon\sigma_3},\ \ \lambda\in\mathbb{C}\backslash\overline{J}
\end{equation*}
solves a similar RHP as the one posed for $N(\lambda)$, but with jumps on $J$ given by
\begin{eqnarray*}
	\widetilde{N}_+(\lambda)&=&\widetilde{N}_-(\lambda)\begin{pmatrix}
	0 & e^{\im t\Omega(\lambda)}\\
	-e^{-\im t\Omega(\lambda)} & 0\\
	\end{pmatrix},\ \ \lambda\in(-M,-m),\\
	\widetilde{N}_+(\lambda)&=&\widetilde{N}_-(\lambda)\begin{pmatrix}
	0 & 1\\
	-1 & 0\\
	\end{pmatrix},\ \ \lambda\in(m,M).
\end{eqnarray*}
This is precisely the same jump behavior which appeared in \cite{BDIK}, Section $3.1$. The solution method is therefore analogous, we only summarize the relevant steps and refer to \cite{BDIK} for further details.\smallskip

The solution to RHP \ref{oup:2} is derived in terms of Jacobi theta functions defined on the elliptic curve
\begin{equation*}
  \Gamma = \big\{(z,w):\ w^2=p(z)\big\},\hspace{1cm} p(z) = \left(z^2-m^2\right)\big(z^2-M^2\big)
\end{equation*}
of genus one. We view $\Gamma$ as two sheeted covering of the Riemann sphere, glued together in the standard way. For definitness, let $\sqrt{p(z)}\sim z^2$ as $z\rightarrow\infty$ on the first sheet and $\sqrt{p(z)}\sim -z^2$ in the similar limit on the second sheet. 
\begin{figure}[tbh]
\begin{center}
\resizebox{0.7\textwidth}{!}{\includegraphics{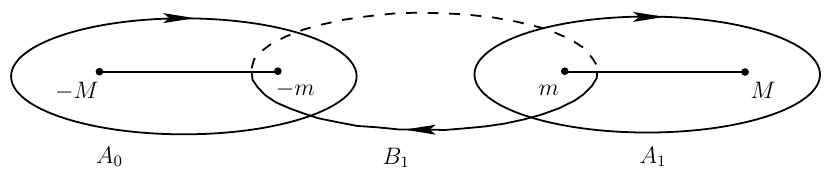}}
\caption{Standard homology basis for $\Gamma$}
\label{figure5}
\end{center}
\end{figure}

We fix a homology basis for $\Gamma$ as indicated in Figure \ref{figure5}, and let
\begin{equation}\label{norm}
  \omega = \frac{c\,\d z}{w},\hspace{0.5cm} c=c(\varkappa)=\frac{\im}{2}\left[\int_m^M\frac{\d\mu}{\sqrt{(M^2-\mu^2)(\mu^2-m^2)}}\right]^{-1} = \frac{\im}{2}\frac{M}{K'}
\end{equation}
be the unique holomorphic one form on $\Gamma$ with normalization
\begin{equation*}
	\oint_{A_1}\omega =1,
\end{equation*}
and $B$-period (compare \eqref{Vfancy})
\begin{equation*}
	\tau=\oint_{B_1}\omega=2c\int_{-m}^m\frac{\d z}{\sqrt{(m^2-z^2)(M^2-z^2)}}=2\im\frac{K}{K'},\hspace{0.5cm} -\im\tau>0.
\end{equation*}
Also, define the Abel (type) map
\begin{equation*}
	u:\mathbb{CP}^1\backslash[-M,M]\rightarrow\mathbb{C},\ \ z\mapsto u(z)=\int_M^z\omega=u(\infty)-\frac{c}{z}+\mathcal{O}\left(z^{-3}\right),\ z\rightarrow\infty
\end{equation*}
and collect the following properties
\begin{prop}\label{pout:1} The Abelian integral $u(z)$ is single-valued and analytic for $z\in\mathbb{CP}^1\backslash[-M,M]$, we have in addition
\begin{equation*}
	u_+(z)-u_-(z)=\begin{cases}
	0,&z\in(-\infty,-M)\cup(M,\infty)\\
	-1,&z\in(-m,m)
	\end{cases},\hspace{0.5cm}u_+(z)+u_-(z)=\begin{cases}
	0,&z\in(m,M)\\
	\tau,&z\in(-M,-m),
	\end{cases}
\end{equation*}
and
\begin{equation}\label{Abel:1}
	u(\infty)=\frac{\tau}{4} \ \ \ (\textnormal{in}\,\, \mathbb{CP}^1).
\end{equation}
\begin{remark} Equation \eqref{Abel:1} is a simple identity between complete elliptic integrals,
\begin{equation*}
	u(\infty)=\frac{\im}{2}\frac{K}{K'},\ \ \ \ \tau=2\im\frac{K}{K'}.
\end{equation*}
\end{remark}
\end{prop}
Our construction requires furthermore the functions
\begin{equation*}
	\omega(z)=\left(\frac{(z+m)(z-M)}{(z+M)(z-m)}\right)^{\frac{1}{4}},\hspace{0.5cm}\phi(z)=\frac{1}{2}\left(\omega(z)+\big(\omega(z)\big)^{-1}\right),\ \ \hat{\phi}(z)=\frac{1}{2\im}\left(\omega(z)-\big(\omega(z)\big)^{-1}\right),
\end{equation*}
defined and analytic for $z\in\mathbb{CP}^1\backslash[-M,M]$ such that $\omega(z)>0$ for $z>M$. To complete the derivation, we use the Jacobi theta function
\begin{equation*}
	\theta(z|\tau)\equiv \theta_3(z|\tau)=\sum_{k\in\mathbb{Z}}\exp\left[\im\pi k^2\tau+2\pi\im kz\right],\ \ z\in\mathbb{C},
\end{equation*}
and define
\begin{equation*}
	N^{(\pm)}(z)=\left(\frac{\theta(u(z)+tV\pm d)}{\theta(u(z)\pm d)},\frac{\theta(-u(z)+tV\pm d)}{\theta(-u(z)\pm d)}\right)\equiv \left(N_1^{(\pm)}(z),N_2^{(\pm)}(z)\right),
\end{equation*}
where (see \eqref{Vfancy})
\begin{equation}\label{Vchoice}
	V(\varkappa)\equiv V=-\frac{4}{\pi}\int_{-m}^m\sqrt{\left(M^2-\mu^2\right)\left(m^2-\mu^2\right)}\,\d\mu
	\equiv\frac{1}{2\pi}\Omega(z),\ z\in(-M,-m);\hspace{1.5cm} d=-\frac{\tau}{4}.
\end{equation}
With this, a solution to RHP \ref{oup:2} is given in the next Proposition.
\begin{prop}[cf. \cite{BDIK}, Section $3.1$]\label{pout:3} The function
\begin{equation}\label{ou:3}
	N(\lambda)=e^{\im\frac{\pi}{4}\epsilon\sigma_3}\frac{\theta(0)}{\theta(tV)}\begin{pmatrix}
	N_1^{(+)}(\lambda)\phi(\lambda) & N_2^{(+)}(\lambda)\hat{\phi}(\lambda)\\
	-N_1^{(-)}(\lambda)\hat{\phi}(\lambda) & N_2^{(-)}(\lambda)\phi(\lambda)\\
	\end{pmatrix}e^{-\im\frac{\pi}{4}\epsilon\sigma_3}
\end{equation}
is single-valued and analytic in $\mathbb{C}\backslash\overline{J}$. Its jumps are stated in \eqref{ou:1} and \eqref{ou:2}, furthermore, as $\lambda\rightarrow\infty$,
\begin{equation*}
	N(\lambda)=I+\frac{1}{\lambda}\begin{pmatrix}
	-c\,\frac{\theta'(tV)}{\theta(tV)} & -\frac{\theta(0)}{\theta(tV)}\frac{\theta(u(\infty)-tV-d)}{\theta(u(\infty)-d)}\frac{M-m}{2\im}e^{\im\frac{\pi}{2}\epsilon}\smallskip\\
	\frac{\theta(0)}{\theta(tV)}\frac{\theta(u(\infty)+tV-d)}{\theta(u(\infty)-d)}\frac{M-m}{2\im}e^{-\im\frac{\pi}{2}\epsilon} & c\,\frac{\theta'(tV)}{\theta(tV)}\\
	\end{pmatrix}+\mathcal{O}\left(\lambda^{-2}\right).
\end{equation*}
\end{prop}
Our next move focuses on the construction of model functions near the branch points and the origin.
\subsection{Parametrix near the origin}  Near the origin, we first observe that
\begin{equation*}
	2g(\lambda)=\varkappa-\im\Omega(\lambda),\ \ \lambda\in\tilde{\gamma}_2;\hspace{1cm}2g(\lambda)=-\varkappa-\im\Omega(\lambda),\ \ \lambda\in\tilde{\gamma}_5.
\end{equation*}
Thus the jumps on the quasi-vertical segments $\tilde{\gamma}_2\cup\tilde{\gamma}_5$ displayed in Figure \ref{figure6} approach the identity matrix exponentially fast along the entire segment, more precisely
\begin{eqnarray*}
	e^{-tg(\lambda)\sigma_3}S_2S_4^{-1}e^{tg(\lambda)\sigma_3} &=&I+(s_1+s_2)e^{-t\varkappa}e^{\im t\Omega(\lambda)}\sigma_+,\ \ \lambda\in\tilde{\gamma}_2\\
	e^{-tg(\lambda)\sigma_3}\sigma_2S_2S_4^{-1}\sigma_2e^{tg(\lambda)\sigma_3}&=&I-(s_1+s_2)e^{-t\varkappa}e^{-\im t\Omega(\lambda)}\sigma_-,\ \ \lambda\in\tilde{\gamma}_5.
\end{eqnarray*}
Using also \eqref{approx:3} and \eqref{approx:4} we shall consider the following RHP.
\begin{problem}
Find a $2\times 2$ piecewise analytic function $H(\lambda)$ such that
\begin{itemize}
	\item $H(\lambda)$ is analytic for $\lambda\in(\mathbb{C}\cap D(0,r))\backslash(\gamma_1^{\pm}\cup\gamma_2^{\pm})$ where $D(0,r)=\{\lambda\in\mathbb{C}:\ |\lambda|<r\}$ and we fix $0<r<\frac{m}{2}$.
	\item The boundary values are related via the equations (compare Figure \ref{figure6} for orientation)
	\begin{eqnarray}
		H_+(\lambda)&=&H_-(\lambda)\begin{pmatrix}
	1& \im\epsilon\, e^{\im t\Omega(\lambda)}\\
	0 & 1\\
	\end{pmatrix},\ \ \ \lambda\in D(0,r)\cap(\gamma_1^+\cup \gamma_2^+);\label{o:1}\\
	H_+(\lambda)&=&H_-(\lambda)\begin{pmatrix}
	1 & 0\\
	\im\epsilon\, e^{-\im t\Omega(\lambda)} & 1\\
	\end{pmatrix},\ \ \ \lambda\in D(0,r)\cap(\gamma_1^-\cup\gamma_2^-)\label{o:2}
	\end{eqnarray}
	\item $H(\lambda)$ is bounded at the origin $\lambda=0$.
	\item As $x\rightarrow-\infty,|s_1|\uparrow 1$ with $\varkappa\in[\delta,\frac{2}{3}\sqrt{2}-\delta],\delta\in(0,\frac{1}{3}\sqrt{2})$ fixed, we have 
	\begin{equation*}
		H(\lambda)=\big(I+o(1)\big)N(\lambda)
	\end{equation*}
	uniformly for $\lambda\in\partial D(0,r)$.\bigskip
\end{itemize}
\end{problem}
The solution to this problem is very elementary: assemble
\begin{equation*}
	H^o(\z)=\begin{cases}
	e^{\im\z\sigma_3},&\textnormal{arg}\,\z\in(-\frac{\pi}{4},\frac{\pi}{4})\cup(\frac{3\pi}{4},\frac{5\pi}{4})\\
	e^{\im\z\sigma_3}\begin{pmatrix}
	1 & \im\epsilon\,e^{\im t\Omega(0)}\\
	0 & 1\\
	\end{pmatrix},&\textnormal{arg}\,\z\in(\frac{\pi}{4},\frac{3\pi}{4})\smallskip\\
	e^{\im\z\sigma_3}\begin{pmatrix}
	1 & 0\\
	-\im\epsilon\,e^{-\im t\Omega(0)} & 1\\
	\end{pmatrix},&\textnormal{arg}\,\z\in(-\frac{3\pi}{4},-\frac{\pi}{4})
	\end{cases}
\end{equation*}
with $\Omega(0)=\pi V$ (see Proposition \ref{gprop}) and observe that $H^o(\z)$ solves a ``bare" RHP
\begin{itemize}
	\item $H^o(\z)$ is analytic for $\z\in\mathbb{C}\backslash\{\z:\,\textnormal{arg}\,\z=\pm\frac{\pi}{4},\pm\frac{3\pi}{4}\}$
	\item Along the four rays oriented from zero to infinity,
	\begin{align*}
		H_+^o(\z)&=H_-^o(\z)\begin{pmatrix}
		1 & \im\epsilon\,e^{\im t\Omega(0)}\\
		0 & 1\\
		\end{pmatrix},\ \textnormal{arg}\,\z=\frac{\pi}{4};\hspace{0.5cm}&H_+^o(\z)&=H_-^o(\z)\begin{pmatrix}
		1& -\im\epsilon\,e^{\im t\Omega(0)}\\
		0 & 1\\
		\end{pmatrix},\  \textnormal{arg}\,\z=\frac{3\pi}{4};\\
		H_+^o(\z)&=H_-^o(\z)\begin{pmatrix}
		1 & 0\\
		\im\epsilon\,e^{-\im t\Omega(0)} & 1\\
		\end{pmatrix},\ \textnormal{arg}\,\z=-\frac{\pi}{4};\hspace{0.5cm}&H_+^o(\z)&=H_-^o(\z)\begin{pmatrix}
		1 & 0\\
		-\im\epsilon\,e^{-\im t\Omega(0)} & 1\\
		\end{pmatrix},\ \ \textnormal{arg}\,\z=-\frac{3\pi}{4}.
	\end{align*}
	\item $H^o(\z)$ is bounded as $\z\rightarrow 0$.
	\item As $\z\rightarrow\infty$, we have
	\begin{equation}\label{o:3}
		H^o(\z)=\left(I+\mathcal{O}\left(\z^{-\infty}\right)\right)e^{\im\z\sigma_3}
	\end{equation}
	uniformly in a full neighborhood of $\z=\infty$.
\end{itemize}
Referring to the locally analytic change of variables,
\begin{equation}\label{o:4}
	\z(\lambda)=\frac{t}{2}\left(\Omega(\lambda)-\Omega(0)\right)=4tMm\lambda\left(1+\mathcal{O}\left(\lambda^2\right)\right),\ \ \lambda\in D(0,r),\ 0<r<\frac{m}{2}
\end{equation}
the origin parametrix is then given by
\begin{equation}\label{o:5}
	H(\lambda)=N(\lambda)H^o\big(\z(\lambda)\big)e^{-\im\z(\lambda)\sigma_3},\ \ \lambda\in D(0,r),
\end{equation}
with $N(\lambda)$ as in \eqref{ou:3}. Since $N(\lambda)$ is analytic in the disk $D(0,r)$, compare Proposition \ref{pout:3}, we check directly that the jumps of $H(\lambda)$ are indeed as in \eqref{o:1} and \eqref{o:2} with the orientation of the contours near the origin like in Figure \ref{figure6}. Note also that the four rays in the bare RHP can always be locally deformed to match the contours $\gamma_j^{\pm}$. Furthermore with \eqref{o:3} and \eqref{o:4} we obtain from \eqref{o:5} the desired matching between the model functions: as $t=(-x)^{\frac{3}{2}}\rightarrow+\infty,|s_1|\uparrow 1$ such that $\varkappa\in[\delta,\frac{2}{3}\sqrt{2}-\delta]$ is fixed,
\begin{equation}\label{paraes:0}
	H(\lambda)=\left(1+\mathcal{O}\left(t^{-\infty}\right)\right)N(\lambda)
\end{equation}
uniformly for $0<r_1\leq|z|\leq r_2<\frac{m}{2}$. This completes the construction of the origin parametrix.
\subsection{Parametrices near the inner branch points} Our construction for the model function near the right inner branch point $\lambda=m$ is motivated by the local expansions
\begin{eqnarray*}
	\im\Omega(\lambda)&=&c_0(\lambda-m)^{\frac{3}{2}}+\mathcal{O}\left((\lambda-m)^{\frac{5}{2}}\right),\ \ \lambda\in\gamma_1^+\cap D(m,r)\\
	-\im\Omega(\lambda)&=&-c_0(\lambda-m)^{\frac{3}{2}}+\mathcal{O}\left((\lambda-m)^{\frac{5}{2}}\right),\ \ \lambda\in\gamma_1^-\cap D(m,r)\\
	\Pi(\lambda)&=&\varkappa-c_0(\lambda-m)^{\frac{3}{2}}+\mathcal{O}\left((\lambda-m)^{\frac{5}{2}}\right),\ \ \lambda\in(m,M)\cap D(m,r)
\end{eqnarray*}
where
\begin{equation*}
	c_0=\frac{16}{3}\sqrt{2m(M^2-m^2)}>0,
\end{equation*}
we fix $0<r<\min\{\frac{m}{2},\frac{1}{2}(M-m)\}$ and the function $(\lambda-m)^{\frac{3}{2}}$ is defined for $\lambda\in\mathbb{C}\backslash[m,\infty)$ with the branch fixed by the requirement $\textnormal{arg}\,(\lambda-m)=\pi$ for $\lambda<m$. This motivates the use of Airy functions and our construction follows from now on \cite{BDIK}, Section $3.2$. with minor modifications. We define
\begin{equation}\label{bare:1}
	A^{RH}(\z)=A_0(\z)\begin{cases}
	I,&\textnormal{arg}\,\z\in(0,\frac{2\pi}{3})\\
	\begin{pmatrix}
	1 & -1\\
	0 & 1\\
	\end{pmatrix},&\textnormal{arg}\,\z\in(\frac{2\pi}{3},\frac{4\pi}{3})\smallskip\\
	\begin{pmatrix}
	1 & -1\\
	0 & 1\\
	\end{pmatrix}\begin{pmatrix}
	1 & 0\\
	1 & 1\\
	\end{pmatrix},&\textnormal{arg}\,\z\in(\frac{4\pi}{3},2\pi)
	\end{cases}
\end{equation}
where $A_0(\z)$ denotes the unimodular entire function
\begin{equation}\label{bareAiry}
	A_0(\z)=\im\sqrt{\pi}e^{-\im\frac{\pi}{6}\sigma_3}\begin{pmatrix}
	e^{-\im\frac{\pi}{6}} & 0\\
	0 & e^{\im\frac{\pi}{2}} \\
	\end{pmatrix}\begin{pmatrix}
	\textnormal{Ai}\left(e^{-\im\frac{2\pi}{3}}\z\right) & \textnormal{Ai}(\z)\\
	e^{-\im\frac{2\pi}{3}}\textnormal{Ai}'\left(e^{-\im\frac{2\pi}{3}}\z\right)&\textnormal{Ai}'(\z)
	\end{pmatrix}e^{\im\frac{\pi}{6}\sigma_3},\ \ \z\in\mathbb{C}
\end{equation}
which is constructed with the help of the Airy function $w=\textnormal{Ai}(z)$, the unique solution to the boundary value problem
\begin{equation*}
	w''=zw;\hspace{0.7cm} \textnormal{Ai}(z)=\frac{z^{-\frac{1}{4}}}{2\sqrt{\pi}}e^{-\frac{2}{3}z^{\frac{3}{2}}}\left(1-\frac{5}{48}z^{-\frac{3}{2}}+\mathcal{O}\left(z^{-\frac{6}{2}}\right)\right),\ \ z\rightarrow\infty,\ \ -\pi<\textnormal{arg}\,z<\pi.
\end{equation*}
Through the standard properties of $w=\textnormal{Ai}(z)$ (cf. \cite{N}), the model function $A^{RH}(\z)$ in \eqref{bare:1} has jumps on the contour shown in Figure \ref{Airy1}, more precisely
\begin{itemize}
	\item $A^{RH}(\z)$ is analytic for $\z\in\mathbb{C}\backslash\{\textnormal{arg}\,\z=0,\frac{2\pi}{3},\frac{4\pi}{3}\}$
	\item The boundary values are related via the equations
	\begin{eqnarray*}
		A^{RH}_+(\z)&=&A^{RH}_-(\z)\begin{pmatrix}
		1 & 1\\
		0 & 1\\
		\end{pmatrix},\ \ \textnormal{arg}\,\z=\frac{2\pi}{3},\\
		A^{RH}_+(\z)&=&A^{RH}_-(\z)\begin{pmatrix}
		1 & 0\\
		-1 & 1\\
		\end{pmatrix},\ \ \textnormal{arg}\,\z=\frac{4\pi}{3},\\
		A^{RH}_+(\z)&=&A^{RH}_-(\z)\begin{pmatrix}
		1 & 1\\
		-1 & 0\\
		\end{pmatrix},\ \ \textnormal{arg}\,\z=0.
	\end{eqnarray*}
	\item As $\z\rightarrow\infty$, we have in a full neighborhood of infinity,
	\begin{equation}\label{bare:2}
		A^{RH}(\z)=\z^{-\frac{1}{4}\sigma_3}\frac{\im}{2}\begin{pmatrix}
		1 & -\im\\
		-1 & -\im\\
		\end{pmatrix}\left[I+\frac{1}{48\z^{\frac{3}{2}}}\begin{pmatrix}
		-1 & 6\im\\
		6\im & 1\\
		\end{pmatrix}+\mathcal{O}\left(\z^{-\frac{6}{2}}\right)\right]e^{\frac{2}{3}\z^{\frac{3}{2}}\sigma_3}.
	\end{equation}
\end{itemize}
\begin{figure}
\begin{minipage}{0.35\textwidth} 
\begin{center}
\resizebox{0.8\textwidth}{!}{\includegraphics{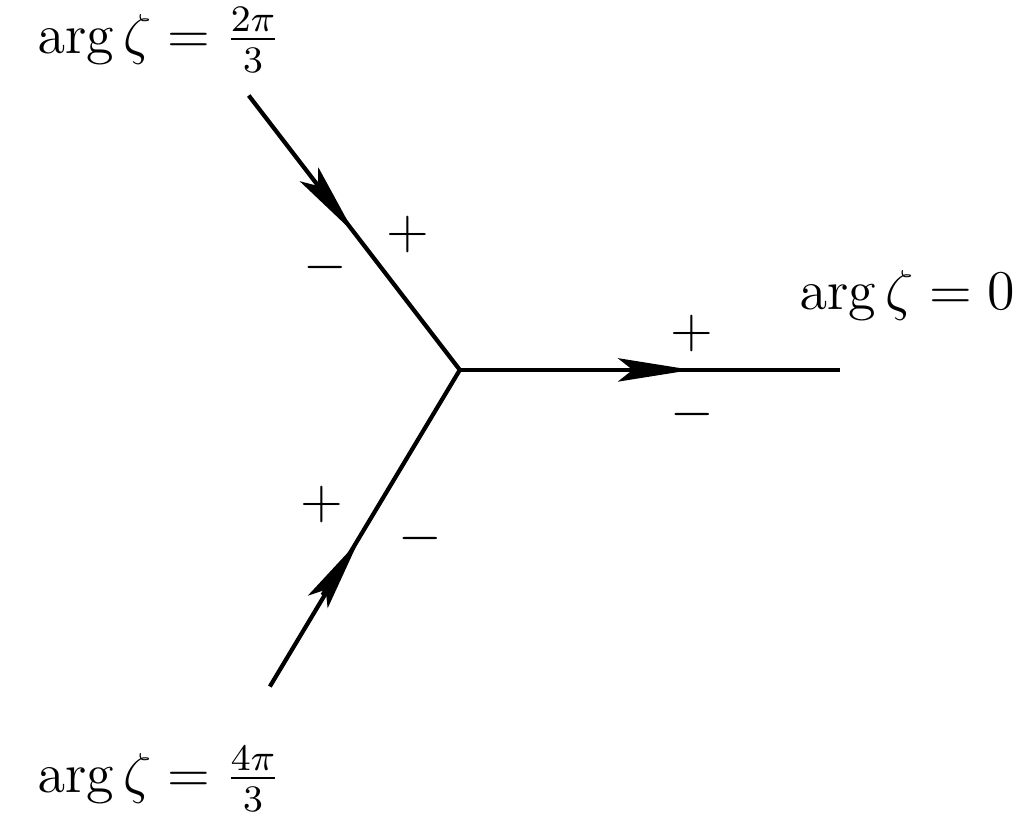}}
\caption{Jump contour for the bare parametrix $A^{RH}(\zeta)$.}
\label{Airy1}
\end{center}
\end{minipage}
\begin{minipage}{0.45\textwidth}
\begin{center}
\resizebox{0.85\textwidth}{!}{\includegraphics{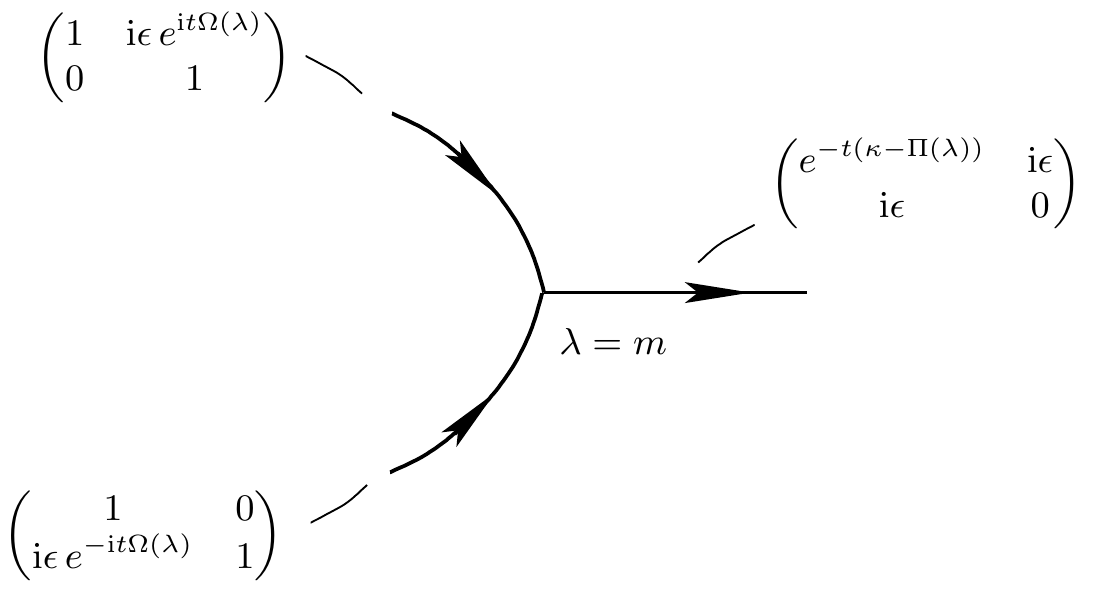}}
\caption{Jumps of the model function $U(\lambda)$ near $\lambda=m$.}
\label{figure7}
\end{center}
\end{minipage}
\end{figure}
With the (locally conformal) change of variables
\begin{equation}\label{mat:0}
	\z(\lambda)=\left[6t\,e^{\im\frac{\pi}{2}}\int_m^{\lambda}\left(\left(M^2-\mu^2\right)\left(m^2-\mu^2\right)\right)^{\frac{1}{2}}\,\d\mu\right]^{\frac{2}{3}}\sim\left(4t\sqrt{2m(M^2-m^2)}\,\right)^{\frac{2}{3}}(\lambda-m),\hspace{0.5cm} |\lambda-m|<r,
\end{equation}
we define
\begin{equation}\label{param}
	U(\lambda)=B_{r_1}(\lambda)A^{RH}\big(\z(\lambda)\big)e^{-\frac{2}{3}\z^{\frac{3}{2}}(\lambda)\sigma_3}e^{-\im\frac{\pi}{4}\epsilon\sigma_3},\hspace{0.5cm}|\lambda-m|<r
\end{equation}
which involves the (locally) analytic function
\begin{equation*}
	B_{r_1}(\lambda)=N(\lambda)e^{\im\frac{\pi}{4}\epsilon\sigma_3}\begin{pmatrix}
	-\im & \im\\
	1 & 1\\
	\end{pmatrix}\delta(\lambda)^{\sigma_3}\left(\z(\lambda)\frac{\lambda-M}{\lambda-m}\right)^{\frac{1}{4}\sigma_3},\ \ \delta(\lambda)=\left(\frac{\lambda-m}{\lambda-M}\right)^{\frac{1}{4}}\rightarrow 1.
\end{equation*}
\begin{remark} Analyticity of $B_{r_1}(\lambda)$ near $\lambda=m$ follows from the simple observation that for $\lambda\in(m,m+r)$,
\begin{align*}
	\big(B_{r_1}(\lambda)\big)_+&=N_-(\lambda)\begin{pmatrix}
	0 & -e^{-\im\frac{\pi}{2}\epsilon}\\
	e^{\im\frac{\pi}{2}\epsilon} & 0\\
	\end{pmatrix}e^{\im\frac{\pi}{4}\epsilon\sigma_3}\begin{pmatrix}
	-\im & \im\\
	1 & 1\\
	\end{pmatrix}\big(\delta_-(\lambda)\big)^{\sigma_3}e^{-\im\frac{\pi}{2}\sigma_3}\left(\z(\lambda)\frac{\lambda-M}{\lambda-m}\right)^{\frac{1}{4}\sigma_3}\\
	&=\big(B_{r_1}(\lambda)\big)_-,
\end{align*}
thus $B_{r_1}(\lambda)$ can only have an isolated singularity at $\lambda=m$. But this singularity is at worst of square root type, hence has to be removable.
\end{remark}
As we are allowed to locally deform the jump contours in \eqref{bare:1} and match those in the initial $L$-RHP, the local parametrix $U(\lambda)$ has jumps in $D(m,r)$ as shown in Figure \ref{figure7}. Moreover, with the help of \eqref{bare:2},
	\begin{eqnarray}
		U(\lambda)&=&N(\lambda)e^{\im\frac{\pi}{4}\epsilon\sigma_3}\left[I+\frac{1}{48\z^{\frac{3}{2}}}\begin{pmatrix}
		-1 & 6\im\\
		6\im & 1\\
		\end{pmatrix} +\mathcal{O}\left(\z^{-\frac{6}{2}}\right)\right]e^{-\im\frac{\pi}{4}\epsilon\sigma_3}\nonumber\\
		&=&\left[I+N(\lambda)\left\{\frac{1}{48\z^{\frac{3}{2}}}\begin{pmatrix}
		-1 & -6\epsilon\\
		6\epsilon & 1\\
		\end{pmatrix} +\mathcal{O}\left(\z^{-\frac{6}{2}}\right)\right\}\big(N(\lambda)\big)^{-1}\right]N(\lambda)\label{mat:1}
	\end{eqnarray}
	so that (as $t\rightarrow+\infty,|s_1|\uparrow 1$ with $\varkappa\in[\delta,\frac{2}{3}\sqrt{2}-\delta]$ fixed)
	\begin{equation}\label{pares:1}
		U(\lambda)=\big(I+\mathcal{O}\left(t^{-1}\right)\big)N(\lambda),
	\end{equation}
	uniformly for $0<r_1\leq|\lambda-m|\leq r_2<\min\{\frac{m}{2},\frac{1}{2}(M-m)\}$.\bigskip

The parametrix near the remaining inner branch point $\lambda=-m$ is constructed similarly: introduce
\begin{equation}\label{bare:3}
	\widetilde{A}^{RH}(\z)=\widetilde{A}_0(\z)\begin{cases}
	I,&\textnormal{arg}\,\z\in(-\pi,-\frac{\pi}{3})\\
	\begin{pmatrix}
	1 & 0\\
	-e^{-\im\pi(1-\gamma)} & 1\\
	\end{pmatrix},&\textnormal{arg}\,\z\in(-\frac{\pi}{3},\frac{\pi}{3})\smallskip\\
	\begin{pmatrix}
	1 & 0\\
	-e^{-\im\pi(1-\gamma)} & 1\\
	\end{pmatrix}\begin{pmatrix}
	1 & e^{\im\pi(1-\gamma)}\\
	0 & 1\\
	\end{pmatrix},&\textnormal{arg}\,\z\in(\frac{\pi}{3},\pi)
	\end{cases}
\end{equation}
where
\begin{equation*}
	\gamma=1+\frac{8t}{\pi}\int_{-m}^m\sqrt{\left(M^2-\mu^2\right)\left(m^2-\mu^2\right)}\,\d\mu\equiv1-2tV(\varkappa),
\end{equation*}
with
\begin{align}
	\widetilde{A}_0(\z)&=\im\sqrt{\pi}e^{\im\frac{\pi}{2}(1-\gamma)\sigma_3}e^{-\im\frac{\pi}{3}\sigma_3}\begin{pmatrix}
	e^{-\im\frac{\pi}{4}}e^{-\im\pi(1-\gamma)} & 0\\
	0 & e^{-\im\frac{5\pi}{12}}\\
	\end{pmatrix}\begin{pmatrix}
	e^{\im\pi}\textnormal{Ai}'\left(e^{\im\pi}\z\right) & -e^{\im\frac{\pi}{3}}\textnormal{Ai}'\left(e^{\im\frac{\pi}{3}}\z\right)\smallskip\\
	-\textnormal{Ai}\left(e^{\im\pi}\z\right)&\textnormal{Ai}\left(e^{\im\frac{\pi}{3}}\z\right)
	\end{pmatrix}\nonumber\\
	&\,\times e^{\im\frac{\pi}{3}\sigma_3}e^{-\im\frac{\pi}{2}(1-\gamma)\sigma_3},\hspace{0.5cm}\z\in\mathbb{C},\label{bareAiry2}
\end{align}
and collect the following properties
\begin{itemize}
	\item $\widetilde{A}^{RH}(\z)$ is analytic for $\z\in\mathbb{C}\backslash\{\textnormal{arg}\,\z=-\frac{\pi}{3},\frac{\pi}{3},\pi\}$
	\item The jumps are as follows, compare Figure \ref{Airy2} for orientation,
	\begin{eqnarray*}
		\widetilde{A}^{RH}_+(\z)&=&\widetilde{A}^{RH}_-(\z)\begin{pmatrix}
		1 & 0\\
		-e^{-\im\pi(1-\gamma)} & 1\\
		\end{pmatrix},\ \ \ \textnormal{arg}\,\z=-\frac{\pi}{3}\\
		\widetilde{A}^{RH}_+(\z)&=&\widetilde{A}^{RH}_-(\z)\begin{pmatrix}
		1 & e^{\im\pi(1-\gamma)}\\
		0& 1\\
		\end{pmatrix},\ \ \ \textnormal{arg}\,\z=\frac{\pi}{3}\\
		\widetilde{A}^{RH}_+(\z)&=&\widetilde{A}^{RH}_-(\z)\begin{pmatrix}
		1 & e^{\im\pi(1-\gamma)}\\
		-e^{-\im\pi(1-\gamma)} & 0\\
		\end{pmatrix},\ \ \ \textnormal{arg}\,\z=\pi
	\end{eqnarray*}
	\item As $\z\rightarrow\infty$, valid in a full neighborhood of infinity,
	\begin{align}\label{bare:4}
		\widetilde{A}^{RH}(\z) &=\z^{\frac{1}{4}\sigma_3}\frac{\im}{2}e^{-\im\pi(1-\gamma)}\begin{pmatrix}
		1 & -\im e^{\im\pi(1-\gamma)}\\
		-1 & -\im e^{\im\pi(1-\gamma)}\\
		\end{pmatrix}\bigg[I+\frac{\im}{48\z^{\frac{3}{2}}}e^{\im\frac{\pi}{2}(1-\gamma)\sigma_3}\begin{pmatrix}
		1 & 6\im\\
		6\im & -1\\
		\end{pmatrix}e^{-\im\frac{\pi}{2}(1-\gamma)\sigma_3}\nonumber\\
		&+\mathcal{O}\left(\z^{-\frac{6}{2}}\right)\bigg]e^{\frac{2}{3}\im\z^{\frac{3}{2}}\sigma_3}
	\end{align}
\end{itemize}
\begin{figure}
\begin{minipage}{0.35\textwidth} 
\begin{center}
\resizebox{0.8\textwidth}{!}{\includegraphics{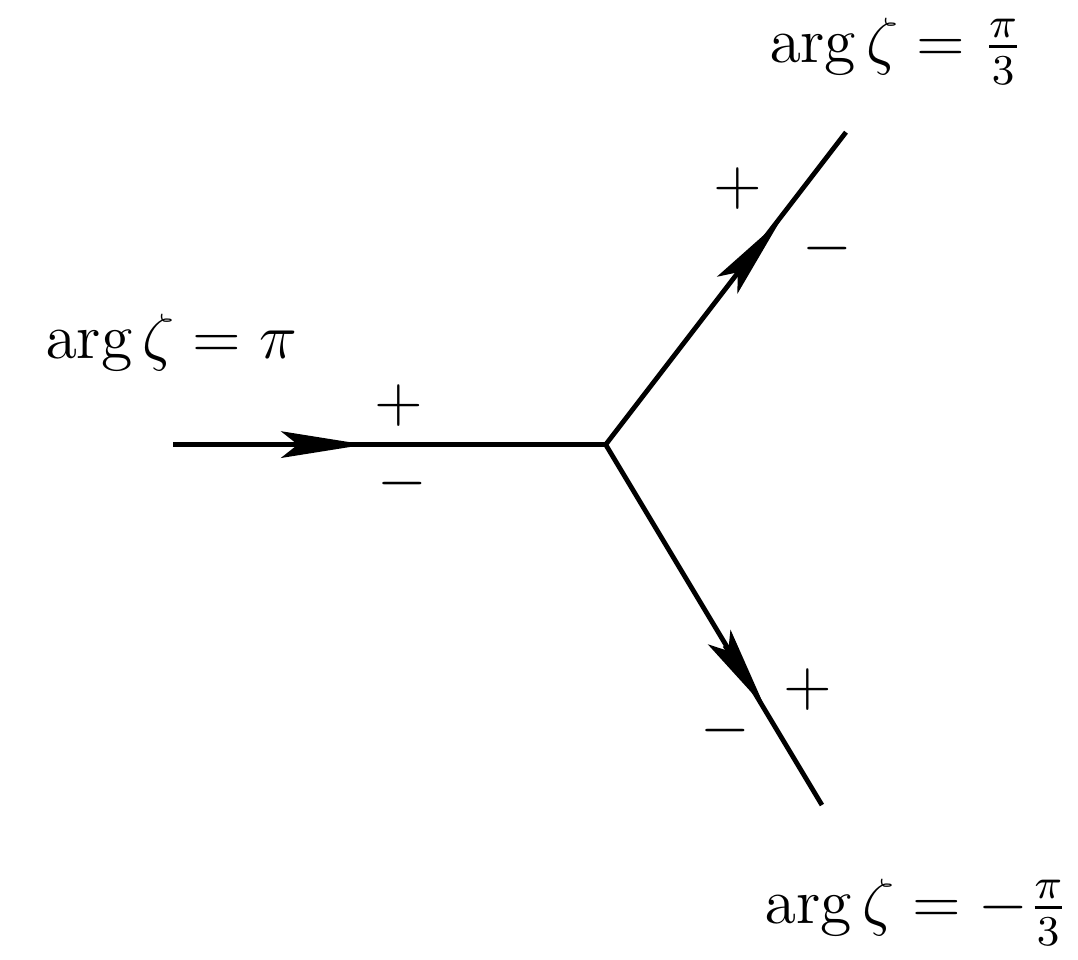}}
\caption{Jump contour for the bare parametrix $\widetilde{A}^{RH}(\zeta)$.}
\label{Airy2}
\end{center}
\end{minipage}
\ \ \ \ \ \ \begin{minipage}{0.35\textwidth}
\begin{center}
\resizebox{1.3\textwidth}{!}{\includegraphics{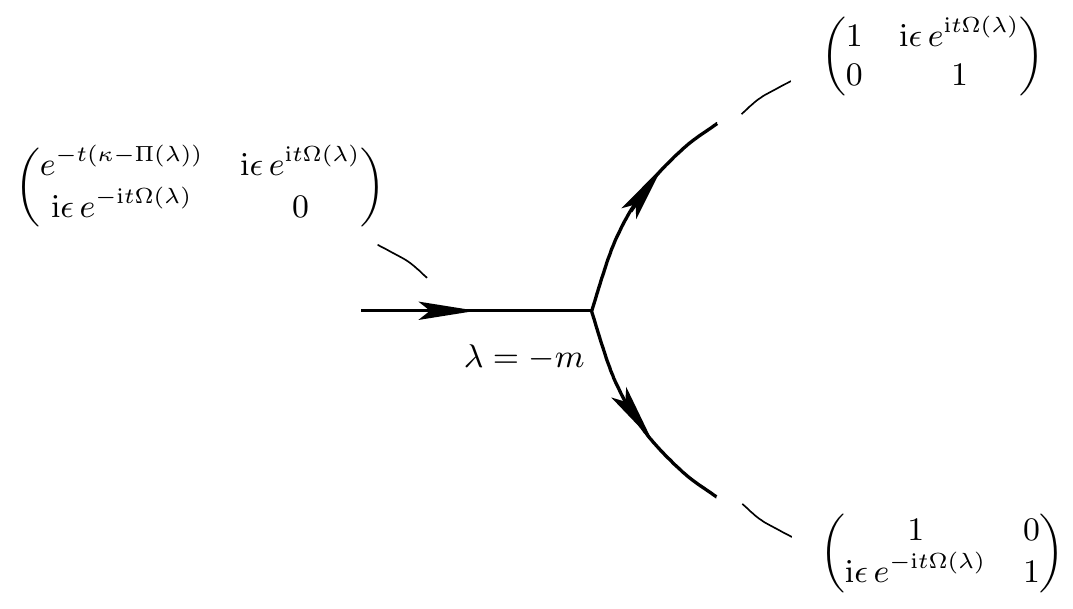}}
\caption{Jumps of the model function $V(\lambda)$ near $\lambda=-m$.}
\label{figure8}
\end{center}
\end{minipage}
\end{figure}

With the change of variables
\begin{equation*}
	\z(\lambda)=\left[6t\int_{-m}^{\lambda}\left(\left(M^2-\mu^2\right)\left(m^2-\mu^2\right)\right)^{\frac{1}{2}}\,\d\mu\right]^{\frac{2}{3}}\sim\left(4t\sqrt{2m(M^2-m^2)}\,\right)^{\frac{2}{3}}(\lambda+m),\ \ \ |\lambda+m|<r,
\end{equation*}
we move ahead and define the parametrix near $\lambda=-m$,
\begin{equation}\label{paramm}
	V(\lambda)=B_{\ell_1}(\lambda)\widetilde{A}^{RH}\big(\z(\lambda)\big)e^{-\frac{2}{3}\im\z^{\frac{3}{2}}(\lambda)\sigma_3}e^{-\im\frac{\pi}{4}\epsilon\sigma_3},\ \ \ |\lambda+m|<r,
\end{equation}
where
\begin{equation*}
	B_{\ell_1}(\lambda)=N(\lambda)e^{\im\frac{\pi}{4}\epsilon\sigma_3}\begin{pmatrix}
	-\im e^{\im\pi(1-\gamma)} & \im e^{\im\pi(1-\gamma)}\\
	1 & 1\\
	\end{pmatrix}\widetilde{\delta}(\lambda)^{-\sigma_3}\left(\z(\lambda)\frac{\lambda+M}{\lambda+m}\right)^{-\frac{1}{4}\sigma_3},
\end{equation*}
with $\widetilde{\delta}(\lambda)=\big(\frac{\lambda+m}{\lambda+M}\big)^{\frac{1}{4}}\rightarrow 1,\,\lambda\rightarrow\infty$.
\begin{remark} Again, the multiplier $B_{\ell_1}(\lambda)$ is analytic at $\lambda=-m$, since
\begin{align*}
	\big(B_{\ell _1}(\lambda)\big)_+&=N_-(\lambda)\begin{pmatrix}
	0 & e^{\im\frac{\pi}{2}\epsilon+2\pi\im tV}\\
	-e^{-\im\frac{\pi}{2}\epsilon-2\pi\im tV} & 0\\
	\end{pmatrix}e^{\im\frac{\pi}{4}\epsilon\sigma_3}\begin{pmatrix}
	-\im e^{\im\pi(1-\gamma)} & \im e^{\im\pi(1-\gamma)}\\
	1 & 1\\
	\end{pmatrix}\big(\,\widetilde{\delta}_-(\lambda)\big)^{-\sigma_3}\\
	&\times\,e^{-\im\frac{\pi}{2}\sigma_3}\left(\,\z(\lambda)\frac{\lambda+M}{\lambda+m}\right)^{-\frac{1}{4}\sigma_3}=\big(B_{\ell_1}(\lambda)\big)_-,\ \ \ \ \lambda\in(-m-r,-m).
\end{align*}
\end{remark}
Moreover, we check that $V(\lambda)$ (after employing a local contour deformation) has jumps inside $D(-m,r)$ as shown in Figure \ref{figure8} and through \eqref{bare:4},
	\begin{eqnarray*}
		V(\lambda)&=&N(\lambda)e^{\im\frac{\pi}{4}\epsilon\sigma_3}\left[I+\frac{\im}{48\z^{\frac{3}{2}}}e^{\im\frac{\pi}{2}(1-\gamma)\sigma_3}\begin{pmatrix}
		1 & 6\im\\
		6\im & -1\\
		\end{pmatrix}e^{-\im\frac{\pi}{2}(1-\gamma)\sigma_3}+\mathcal{O}\left(\z^{-\frac{6}{2}}\right)\right]e^{-\im\frac{\pi}{4}\epsilon\sigma_3}\\
		&=&\left[I+N(\lambda)\left\{\frac{\im}{48\z^{\frac{3}{2}}}\begin{pmatrix}
		1 & -6\epsilon\, e^{2\pi\im tV}\\
		6\epsilon\, e^{-2\pi\im tV} & -1\\
		\end{pmatrix}+\mathcal{O}\left(\z^{-\frac{6}{2}}\right)\right\}\big(N(\lambda)\big)^{-1}\right]N(\lambda).
	\end{eqnarray*}
	Thus for $t\rightarrow+\infty,|s_1|\uparrow 1$ such that $\varkappa\in[\delta,\frac{2}{3}\sqrt{2}-\delta],\delta\in(0,\frac{1}{3}\sqrt{2})$,
	\begin{equation}\label{paraes:2}
		V(\lambda)=\big(I+\mathcal{O}\left(t^{-1}\right)\big)N(\lambda)
	\end{equation}
	uniformly for $0<r_1\leq|\lambda+m|\leq r_2<\min\{\frac{m}{2},\frac{1}{2}(M-m)\}$.
This completes the construction of the parametrices at the inner branch points $\lambda=\pm m$.
\subsection{Parametrices near the outer branch points}\label{outpara}
Once more Airy functions are used in the constructions, first near $\lambda=M$: Define the bare parametrix,
\begin{equation}\label{bare:5}
	\widehat{A}^{RH}(\z)=\widehat{A}_0(\z)\begin{cases}
	\begin{pmatrix}
	1 & 0\\
	1 & 1\\
	\end{pmatrix},&\textnormal{arg}\,\z\in(-\pi,-\frac{2\pi}{3})\\
	I,&\textnormal{arg}\,\z\in(-\frac{2\pi}{3},-\frac{\pi}{3})\\
	\begin{pmatrix}
	1 & 1\\
	0 & 1\\
	\end{pmatrix},&\textnormal{arg}\,\z\in(-\frac{\pi}{3},\frac{\pi}{3})\smallskip\\
	\begin{pmatrix}
	1 & 1\\
	0 & 1\\
	\end{pmatrix}\begin{pmatrix}
	1 & 0\\
	-1 & 1\\
	\end{pmatrix},&\textnormal{arg}\,\z\in(\frac{\pi}{3},\pi)
	\end{cases}
\end{equation}
with
\begin{equation*}
	\widehat{A}_0(\z)=\sigma_2\widetilde{A}_0(\z)\Big|_{\gamma\equiv1}\sigma_2=\im\sqrt{\pi}\,e^{\im\frac{\pi}{3}\sigma_3}\begin{pmatrix}
	e^{-\im\frac{5\pi}{12}} & 0\\
	0 & e^{-\im\frac{\pi}{4}}\\
	\end{pmatrix}\begin{pmatrix}
	\textnormal{Ai}\left(e^{\im\frac{\pi}{3}}\z\right) & \textnormal{Ai}\left(e^{\im\pi}\z\right)\smallskip\\
	e^{\im\frac{\pi}{3}}\textnormal{Ai}'\left(e^{\im\frac{\pi}{3}}\z\right) & e^{\im\pi}\textnormal{Ai}'\left(e^{\im\pi}\z\right)
	\end{pmatrix}e^{-\im\frac{\pi}{3}\sigma_3},\ \ \z\in\mathbb{C}.
\end{equation*}
This model function has the properties listed below
\begin{itemize}
	\item $\widehat{A}^{RH}(\z)$ is analytic for $\z\in\mathbb{C}\backslash\{\textnormal{arg}\,\z=-\frac{2\pi}{3},-\frac{\pi}{3},\frac{\pi}{3},\pi\}$
	\item We have the following jump behavior on the contour shown in Figure \ref{figure9}, 
	\begin{eqnarray*}
		\widehat{A}^{RH}_+(\z)&=&\widehat{A}^{RH}_-(\z)\begin{pmatrix}
		1 & 0\\
		-1 & 1\\
		\end{pmatrix},\ \ \ \textnormal{arg}\,\z=-\frac{2\pi}{3},\frac{\pi}{3}\\
		\widehat{A}^{RH}_+(\z)&=&\widehat{A}^{RH}_-(\z)\begin{pmatrix}
		1 & 1\\
		0 & 1\\
		\end{pmatrix},\ \ \ \textnormal{arg}\,\z=-\frac{\pi}{3}\\
		\widehat{A}^{RH}_+(\z)&=&\widehat{A}^{RH}_-(\z)\begin{pmatrix}
		0 & 1\\
		-1 & 0\\
		\end{pmatrix},\ \ \ \textnormal{arg}\,\z=\pi
	\end{eqnarray*}
\begin{figure}[tbh]
\begin{minipage}{0.35\textwidth} 
\begin{center}
\resizebox{0.85\textwidth}{!}{\includegraphics{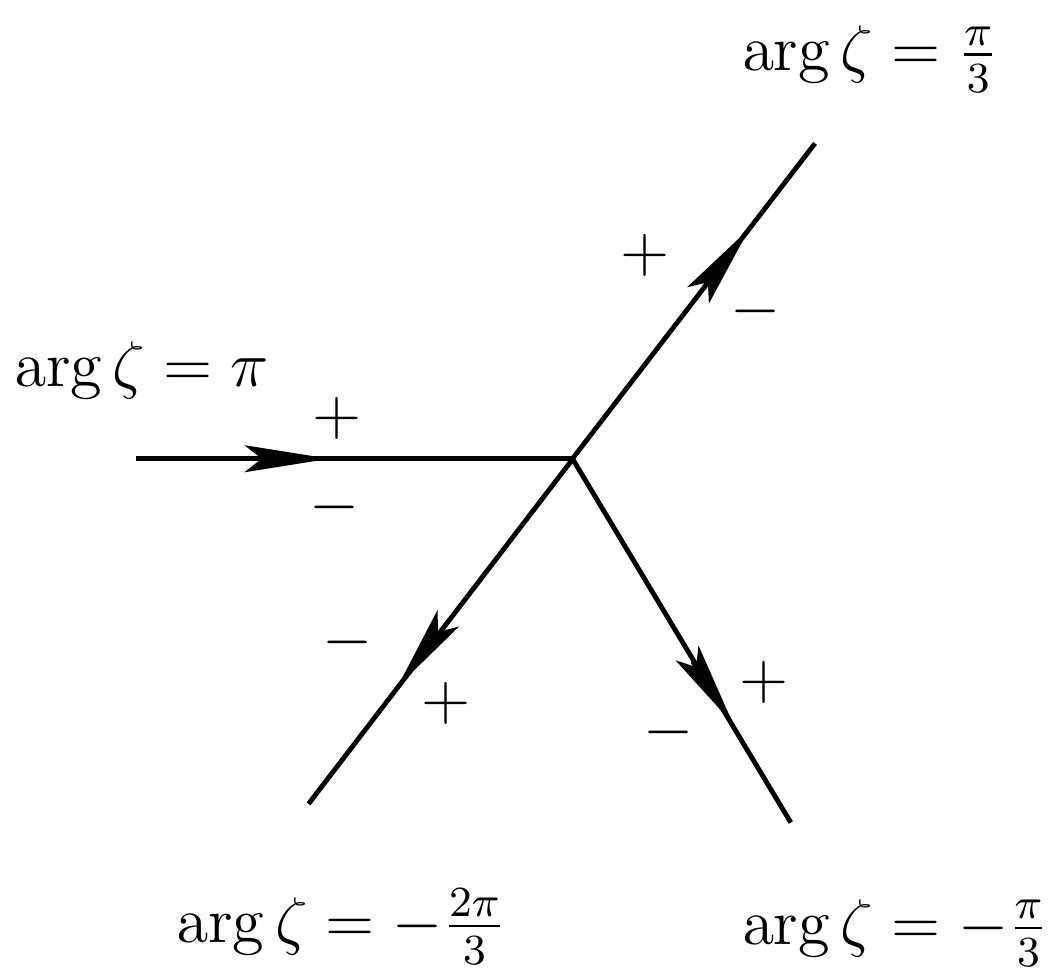}}
\caption{Jump contour for the bare parametrix $\widehat{A}^{RH}(\zeta)$.}
\label{figure9}
\end{center}
\end{minipage}
\begin{minipage}{0.5\textwidth}
\begin{center}
\resizebox{0.85\textwidth}{!}{\includegraphics{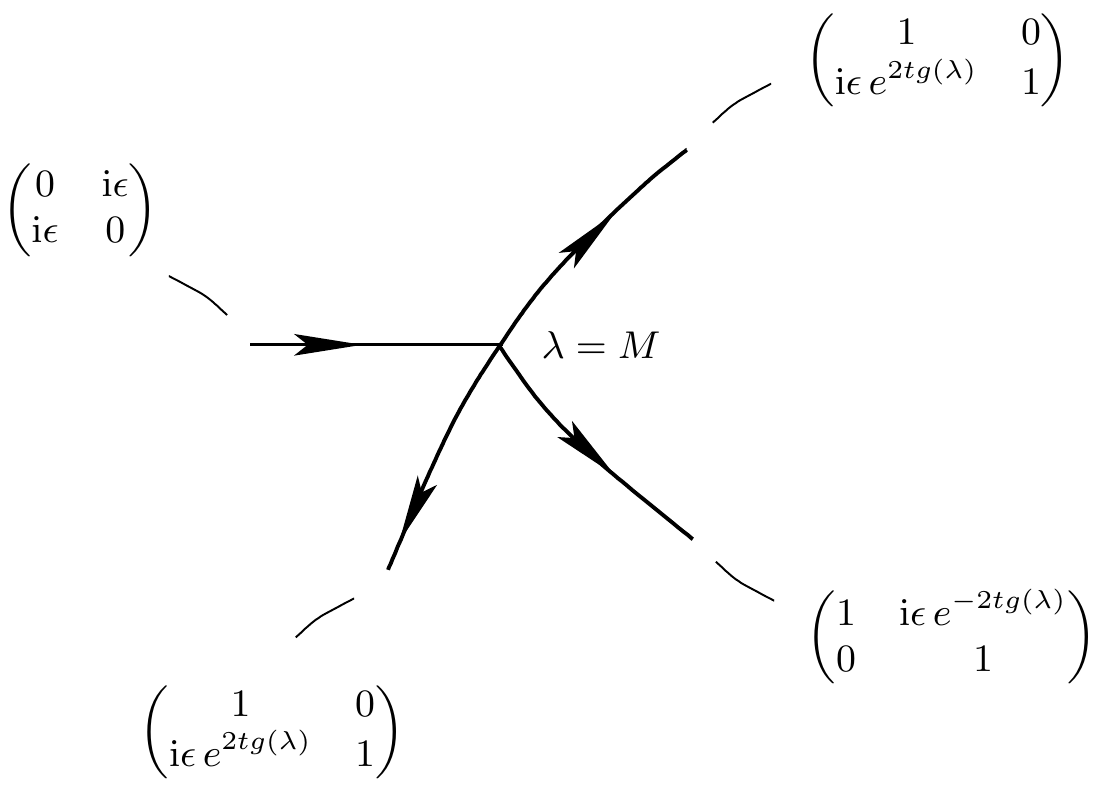}}
\caption{Jumps of the model function $P(\lambda)$ near $\lambda=M$.}
\label{figure10}
\end{center}
\end{minipage}
\end{figure}		
	\item As $\z\rightarrow\infty$,
	\begin{equation}\label{bare:6}
		\widehat{A}^{RH}(\z)=\z^{-\frac{1}{4}\sigma_3}\frac{\im}{2}\begin{pmatrix}
		-\im & 1\\
		\im & 1\\
		\end{pmatrix}\left[I+\frac{\im}{48\z^{\frac{3}{2}}}\begin{pmatrix}
		-1 & -6\im\\
		-6\im & 1\\
		\end{pmatrix}+\mathcal{O}\left(\z^{-\frac{6}{2}}\right)\right]e^{-\frac{2}{3}\im\z^{\frac{3}{2}}\sigma_3}
	\end{equation}
\end{itemize}
This time, we use the change of variables,
\begin{equation*}		
	\z(\lambda)=\left[6t\int_M^{\lambda}\left(\left(\mu^2-M^2\right)\left(\mu^2-m^2\right)\right)^{\frac{1}{2}}\d\mu\right]^{\frac{2}{3}}\sim\left(4t\sqrt{2M(M^2-m^2)}\right)^{\frac{2}{3}}(\lambda-M),\ \ \ |\lambda-M|<r
\end{equation*}
and the local parametrix is defined as
\begin{equation}\label{parao}
	P(\lambda)=B_{r_2}(\lambda)\widehat{A}^{RH}\big(\z(\lambda)\big)e^{\frac{2}{3}\im\z^{\frac{3}{2}}(\lambda)\sigma_3}e^{-\im\frac{\pi}{4}\epsilon\sigma_3},\ \ \ |\lambda-M|<r.
\end{equation}
Here,
\begin{equation*}
	B_{r_2}(\lambda)=N(\lambda)e^{\im\frac{\pi}{4}\epsilon\sigma_3}\begin{pmatrix}
	1 & -1\\
	-\im & -\im\\
	\end{pmatrix}\widehat{\delta}(\lambda)^{\sigma_3}\left(\z(\lambda)\frac{\lambda-m}{\lambda-M}\right)^{\frac{1}{4}\sigma_3},\hspace{0.5cm}\widehat{\delta}(\lambda)=\left(\frac{\lambda-M}{\lambda-m}\right)^{\frac{1}{4}}\rightarrow 1,\ \lambda\rightarrow\infty
\end{equation*}
is once more analytic near $\lambda=M$.  The relevant jump properties of $P(\lambda)$ are displayed in Figure \ref{figure10} and with \eqref{bare:6},
	\begin{eqnarray*}
		P(\lambda)&=&N(\lambda)e^{\im\frac{\pi}{4}\epsilon\sigma_3}\left[I+\frac{\im}{48\z^{\frac{3}{2}}}\begin{pmatrix}
		-1 & -6\im\\
		-6\im & 1\\
		\end{pmatrix}+\mathcal{O}\left(\z^{-\frac{6}{2}}\right)\right]e^{-\im\frac{\pi}{4}\epsilon\sigma_3}\\
		&=&\left[I+N(\lambda)\left\{\frac{\im}{48\z^{\frac{3}{2}}}\begin{pmatrix}
		-1 & 6\epsilon\\
		-6\epsilon & 1\\
		\end{pmatrix}+\mathcal{O}\left(\z^{-\frac{6}{2}}\right)\right\}\big(N(\lambda)\big)^{-1}\right]N(\lambda).
	\end{eqnarray*}
	Hence for $t\rightarrow+\infty,|s_1|\uparrow 1$ such that $\varkappa\in[\delta,\frac{2}{3}\sqrt{2}-\delta],\delta\in(0,\frac{1}{3}\sqrt{2})$,
	\begin{equation}\label{paraes:3}
		P(\lambda)=\left(I+\mathcal{O}\left(t^{-1}\right)\right)N(\lambda)
	\end{equation}
	uniformly for $0<r_1\leq|\lambda-M|\leq r_2<\frac{1}{2}(M-m)$.\bigskip		

We are left with the parametrix near the left most branch point $\lambda=-M$: assemble
\begin{equation}\label{bare:7}
	\bar{A}^{RH}(\z)=\bar{A}_0(\z)\begin{cases}
	\begin{pmatrix}
	1 & e^{\im\pi(1-\gamma)}\\
	0 & 1\\
	\end{pmatrix},&\textnormal{arg}\,\z\in(0,\frac{\pi}{3})\\
	I,&\textnormal{arg}\,\z\in(\frac{\pi}{3},\frac{2\pi}{3})\\
	\begin{pmatrix}
	1 & 0\\
	e^{-\im\pi(1-\gamma)} & 1\\
	\end{pmatrix},&\textnormal{arg}\,\z\in(\frac{2\pi}{3},\frac{4\pi}{3})\smallskip\\
	\begin{pmatrix}
	1 & 0\\
	e^{-\im\pi(1-\gamma)} & 1\\
	\end{pmatrix}\begin{pmatrix}
	1 & -e^{\im\pi(1-\gamma)}\\
	0 & 1\\
	\end{pmatrix},&\textnormal{arg}\,\z\in(\frac{4\pi}{3},2\pi)
	\end{cases}
\end{equation}
where
\begin{align*}
	\bar{A}_0(\z)&=e^{\im\frac{\pi}{2}(1-\gamma)\sigma_3}\sigma_2A_0(\z)\sigma_2e^{-\im\frac{\pi}{2}(1-\gamma)\sigma_3}\\
	&=\im\sqrt{\pi}e^{\im\frac{\pi}{6}\sigma_3}e^{\im\frac{\pi}{2}(1-\gamma)\sigma_3}\begin{pmatrix}
	e^{\im\frac{\pi}{2}} & 0\\
	0 & e^{-\im\frac{\pi}{6}}\\
	\end{pmatrix}\begin{pmatrix}
	\textnormal{Ai}'(\z) & -e^{-\im\frac{2\pi}{3}}\textnormal{Ai}'\big(e^{-\im\frac{2\pi}{3}}\z\big)\smallskip\\
	-\textnormal{Ai}(\z) & \textnormal{Ai}\big(e^{-\im\frac{2\pi}{3}}\z\big)\\
	\end{pmatrix}e^{-\im\frac{\pi}{6}\sigma_3}e^{-\im\frac{\pi}{2}(1-\gamma)\sigma_3},\ \ \z\in\mathbb{C},
\end{align*}
and, as before, $\gamma=1-2tV(\varkappa)$. We have
\begin{itemize}
	\item $\bar{A}^{RH}(\z)$ is analytic for $\z\in\mathbb{C}\backslash\{\textnormal{arg}\,\z=0,\frac{\pi}{3},\frac{2\pi}{3},\frac{4\pi}{3}\}$
	\item Fixing orientations as shown in Figure \ref{figure11} below, the relevant boundary values are related as follows
	\begin{eqnarray*}
		\bar{A}^{RH}_+(\z)&=&\bar{A}^{RH}_-(\z)\begin{pmatrix}
		1 & -e^{\im\pi(1-\gamma)}\\
		0 & 1\\
		\end{pmatrix},\ \ \ \textnormal{arg}\,\z=\frac{\pi}{3},\frac{4\pi}{3}\\
		\bar{A}^{RH}_+(\z)&=&\bar{A}^{RH}_-(\z)\begin{pmatrix}
		1 & 0\\
		e^{-\im\pi(1-\gamma)} & 1\\
		\end{pmatrix},\ \ \ \textnormal{arg}\,\z=\frac{2\pi}{3}\\
		\bar{A}^{RH}_+(\z)&=&\bar{A}^{RH}_-(\z)\begin{pmatrix}
		0 & e^{\im\pi(1-\gamma)}\\
		-e^{-\im\pi(1-\gamma)} & 0\\
		\end{pmatrix},\ \ \ \textnormal{arg}\,\z=0
	\end{eqnarray*}
	\begin{figure}[tbh]
\begin{minipage}{0.35\textwidth} 
\begin{center}
\resizebox{1\textwidth}{!}{\includegraphics{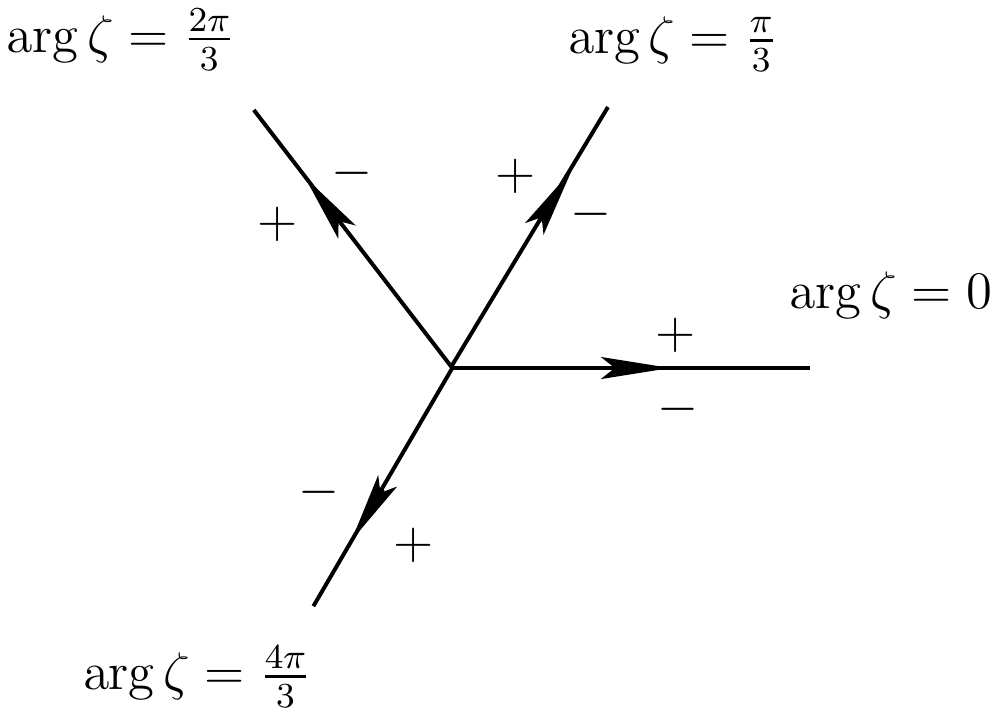}}
\caption{Jump contour for the bare parametrix $\bar{A}^{RH}(\zeta)$.}
\label{figure11}
\end{center}
\end{minipage}
\ \ \ \ \ \ \ \begin{minipage}{0.5\textwidth}
\begin{center}
\resizebox{1\textwidth}{!}{\includegraphics{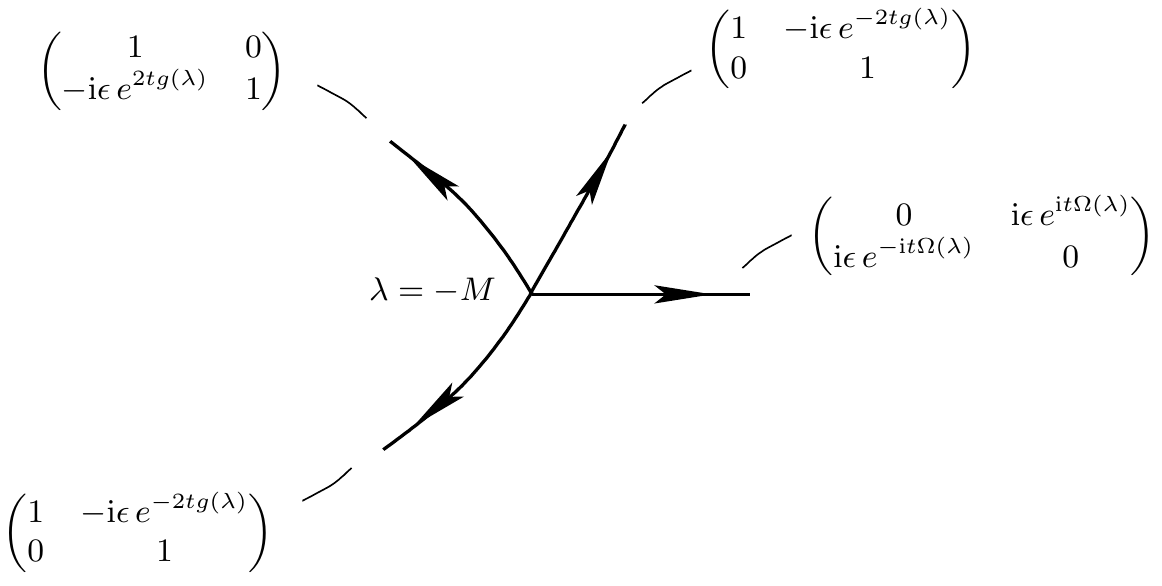}}
\caption{Jumps of the model function $Q(\lambda)$ near $\lambda=-M$.}
\label{figure12}
\end{center}
\end{minipage}
\end{figure}			
	\item As $\z\rightarrow\infty$,
	\begin{eqnarray}\label{bare:8}
		\bar{A}^{RH}(\z)&=&\z^{\frac{1}{4}\sigma_3}\frac{\im}{2}e^{\im\frac{\pi}{2}(1-\gamma)\sigma_3}\begin{pmatrix}
		-\im & 1\\
		\im & 1\\
		\end{pmatrix}e^{-\im\frac{\pi}{2}(1-\gamma)\sigma_3}\bigg[I+\frac{1}{48\z^{\frac{3}{2}}}\begin{pmatrix}
		1 & 6\im e^{-\im\pi\gamma}\\
		6\im e^{\im\pi\gamma} & -1\\
		\end{pmatrix}\nonumber\\
		&&+\mathcal{O}\left(\z^{-\frac{6}{2}}\right)\bigg]e^{-\frac{2}{3}\z^{\frac{3}{2}}\sigma_3}
	\end{eqnarray}
\end{itemize}
and with the change of variables
\begin{equation*}
	\z(\lambda)=\left[6te^{\im\frac{\pi}{2}}\int_{-M}^{\lambda}\left(\left(\mu^2-M^2\right)\left(\mu^2-m^2\right)\right)^{\frac{1}{2}}\d\mu\right]^{\frac{2}{3}}\sim\left(4t\sqrt{2M(M^2-m^2)}\,\right)^{\frac{2}{3}}(\lambda+M),\ \ |\lambda+M|<r,
\end{equation*}
the required local parametrix is given by
\begin{equation}\label{paraom}
	Q(\lambda)=B_{\ell_2}(\lambda)\bar{A}^{RH}\big(\z(\lambda)\big)e^{\frac{2}{3}\z^{\frac{3}{2}}(\lambda)\sigma_3}e^{-\im\frac{\pi}{4}\epsilon\sigma_3},\ \ |\lambda+M|<r,
\end{equation}
where	
\begin{equation*}
	B_{\ell_2}(\lambda)=N(\lambda)e^{\im\frac{\pi}{4}\epsilon\sigma_3}e^{\im\frac{\pi}{2}(1-\gamma)\sigma_3}\begin{pmatrix}
	1 & -1\\
	-\im & -\im\\
	\end{pmatrix}e^{-\im\frac{\pi}{2}(1-\gamma)\sigma_3}\bar{\delta}(\lambda)^{-\sigma_3}\left(\z(\lambda)\frac{\lambda+m}{\lambda+M}\right)^{-\frac{1}{4}\sigma_3},
\end{equation*}
and $\bar{\delta}(\lambda)=\big(\frac{\lambda+M}{\lambda+m}\big)^{\frac{1}{4}}\rightarrow 1$ as $\lambda\rightarrow\infty$. Since $B_{\ell_2}(\lambda)$ is analytic at $\lambda=-M$, the model function $Q(\lambda)$ has the jumps shown in Figure \ref{figure12} and with \eqref{bare:8},
	\begin{eqnarray*}
		Q(\lambda)&=&N(\lambda)e^{\im\frac{\pi}{4}\sigma_3}\left[I+\frac{1}{48\z^{\frac{3}{2}}}\begin{pmatrix}
		1 & 6\im e^{-\im\pi\gamma}\\
		6\im e^{\im\pi\gamma} & -1\\
		\end{pmatrix}+\mathcal{O}\left(\z^{-\frac{6}{2}}\right)\right]e^{-\im\frac{\pi}{4}\epsilon\sigma_3}\\
		&=&\left[I+N(\lambda)\left\{\frac{1}{48\z^{\frac{3}{2}}}\begin{pmatrix}
		1 & 6\epsilon\,e^{2\pi\im tV}\\
		-6\epsilon\,e^{-2\pi\im tV} & -1\\
		\end{pmatrix}+\mathcal{O}\left(\z^{-\frac{6}{2}}\right)\right\}\big(N(\lambda)\big)^{-1}\right]N(\lambda).
	\end{eqnarray*}
	This implies for $t\rightarrow+\infty,|s_1|\uparrow 1$ such that $\varkappa\in[\delta,\frac{2}{3}\sqrt{2}-\delta],\delta\in(0,\frac{1}{3}\sqrt{2})$,
	\begin{equation}\label{paraes:4}
		Q(\lambda)=\left(I+\mathcal{O}\left(t^{-1}\right)\right)N(\lambda)
	\end{equation}
	uniformly for $0<r_1\leq|\lambda+M|\leq r_2<\frac{1}{2}(M-m)$.		
\subsection{Final transformation - ratio problem}\label{ratiosec1}
In this final step, we use the explicit model functions $H(\lambda)$ in \eqref{o:5}, $U(\lambda)$ in \eqref{param}, $V(\lambda)$ in \eqref{paramm}, $P(\lambda)$ in \eqref{parao}, $Q(\lambda)$ in \eqref{paraom} and $N(\lambda)$ in \eqref{ou:3}. We transform the previous RHP \ref{openRHP} for $L(\lambda)$ by setting
\begin{equation}\label{ratio:1}
	R(\lambda)=L(\lambda)\begin{cases}
		\big(H(\lambda)\big)^{-1},&|\lambda|<r\\
		\big(U(\lambda)\big)^{-1},&|\lambda-m|<\hat{r}\\
		\big(V(\lambda)\big)^{-1},&|\lambda+m|<\hat{r}\\
		\big(P(\lambda)\big)^{-1},&|\lambda-M|<\bar{r}\\
		\big(Q(\lambda)\big)^{-1},&|\lambda+M|<\bar{r}\\
		\big(N(\lambda)\big)^{-1},&|\lambda|>r, |\lambda\mp m|>\hat{r}, |\lambda\mp M|>\bar{r}
		\end{cases}
\end{equation}
where $0<r<\frac{m}{2},0<\hat{r}<\min\{\frac{m}{2},\frac{1}{2}(M-m)\}$ and $0<\bar{r}<\frac{1}{2}(M-m)$ remain fixed. Definition \eqref{ratio:1} implies that the RHP for $L(\lambda)$ shown in Figure \ref{figure6} is transformed to the RHP for $R(\lambda)$ posed on the contour $\Sigma_R$ shown in Figure \ref{figure13}.
\begin{figure}[tbh]
\begin{center}
\resizebox{0.7\textwidth}{!}{\includegraphics{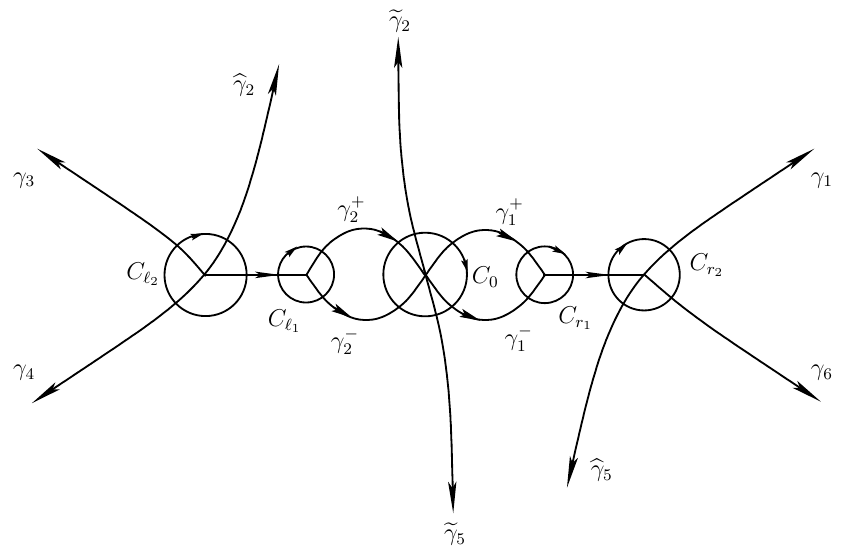}}
\caption{The jump contour $\Sigma_R$ of the ratio function $R(\lambda)$ defined in \eqref{ratio:1}.}
\label{figure13}
\end{center}
\end{figure}
\begin{itemize}
	\item $R(\lambda)$ is analytic for $\lambda\in\mathbb{C}\backslash\Sigma_R$ where $\Sigma_R=[-M,-m]\cup[m,M]\cup\gamma_1^{\pm}\cup\gamma_2^{\pm}\cup C_0\cup C_{r_1}\cup C_{r_2}\cup C_{\ell_1}\cup C_{\ell_2}\cup\Sigma_{S_{\infty}}$
	\item For the oriented contour $\Sigma_R$ shown in Figure \ref{figure13}, the jumps read as follows. First inside the five circles $C_j$,
	\begin{equation*}
		R_+(\lambda)=R_-(\lambda)G_R(\lambda;t,|s_1|),\ \ \ \ \ G_R(\lambda)\Xi_-(\lambda)=\Xi_-(\lambda)C(\lambda)
	\end{equation*}
	where
	\begin{equation*}
		\Xi(\lambda)=H(\lambda),\ \ |\lambda|<r;\ \ \ \ \Xi(\lambda)=\begin{cases}
			U(\lambda),&|\lambda-m|<\hat{r}\\
			V(\lambda),&|\lambda+m|<\hat{r}
			\end{cases};\ \ \ \ \Xi(\lambda)=\begin{cases}
			P(\lambda),&|\lambda-M|<\bar{r}\\
			Q(\lambda),&|\lambda+M|<\bar{r}
			\end{cases}
	\end{equation*}
	and (inside the far most left circle $C_{\ell_2}$),
	\begin{equation*}
		C(\lambda)=\begin{cases}
			\begin{pmatrix}
		1 & \big(\im\epsilon-s_1\big)e^{-2tg(\lambda)}\\
		0 & 1\\
		\end{pmatrix},&\lambda\in(\hat{\gamma}_2\cup\gamma_4)\cap D(-M,\bar{r})\smallskip\\
		\begin{pmatrix}
		1 & 0\\
		\big(s_3+\im\epsilon\big)e^{2tg(\lambda)} & 1\\
		\end{pmatrix},&\lambda\in\gamma_3\cap D(-M,\bar{r})\smallskip\\
		\begin{pmatrix}
		-\im\epsilon\,s_1(1+e^{-\varkappa t}) & -\im\epsilon\,e^{\im t\Omega(\lambda)}e^{-t(\varkappa-\Pi(\lambda))}\\
		-\im\epsilon\,e^{-\im t\Omega(\lambda)}e^{-t(\varkappa+\Pi(\lambda))}& \im\epsilon\,s_3\\
		\end{pmatrix},&\lambda\in(-M,-M+\bar{r})
		\end{cases}
	\end{equation*}
	followed by (inside the left circle $C_{\ell_1}$),
	\begin{equation*}
		C(\lambda)=\begin{cases}
		\begin{pmatrix}
		1 & \big(s_1+s_1(1-s_1s_3)-\im\epsilon\big)e^{\im t\Omega(\lambda)}\\
		0 & 1\\
		\end{pmatrix},&\!\!\!\!\lambda\in\gamma_2^+\cap D(-m,\hat{r})\smallskip\\
		\begin{pmatrix}
		1 & 0\\
		-\big(s_3+\im\epsilon)e^{-\im t\Omega(\lambda)} & 1\\
		\end{pmatrix},&\!\!\!\!\lambda\in\gamma_2^-\cap D(-m,\hat{r})\smallskip\\
		\begin{pmatrix}
		-\im\epsilon\,s_1(1+e^{-\varkappa t}) & (s_1+s_1(1-s_1s_3)-\im\epsilon)e^{\im t\Omega(\lambda)}e^{-t(\varkappa-\Pi(\lambda))}\\
		-\im\epsilon\,e^{-\im t\Omega(\lambda)}e^{-t(\varkappa+\Pi(\lambda))} & \im\epsilon\,s_3+e^{-2t\varkappa}\\
		\end{pmatrix},&\!\!\!\!\lambda\in (-m-\hat{r},-m)
		\end{cases}
	\end{equation*}
	as well as for the corresponding circles on the right,
	\begin{equation*}
		C(\lambda)=\begin{cases}
		\begin{pmatrix}
			1 & -\big(s_3+\im\epsilon\big)e^{\im t\Omega(\lambda)}\\
			0 & 1\\
			\end{pmatrix},&\lambda\in\gamma_1^+\cap D(m,\hat{r})\smallskip\\
			\begin{pmatrix}
			1 & 0\\
			\big(s_1+s_1(1-s_1s_3)-\im\epsilon\big)e^{-\im t\Omega(\lambda)} & 1\\
			\end{pmatrix},&\lambda\in\gamma_1^-\cap D(m,\hat{r})\smallskip\\
			\begin{pmatrix}
				\im\epsilon\, s_3 & -\im\epsilon\,e^{-t(\varkappa-\Pi(\lambda))}\\
				-\im\epsilon\,e^{-t(\varkappa+\Pi(\lambda))} & -\im\epsilon\,s_1(1+e^{-\varkappa t})+e^{-2t\varkappa}\\
				\end{pmatrix},& \lambda\in (m,m+\hat{r})
			\end{cases}
	\end{equation*}
	and
	\begin{equation*}
		C(\lambda)=\begin{cases}
			\begin{pmatrix}
			1 & 0\\
			\big(s_1-\im\epsilon\big)e^{2tg(\lambda)} & 1\\
			\end{pmatrix},&\lambda\in(\gamma_1\cup\hat{\gamma}_5)\cap D(M,\bar{r})\smallskip\\
			\begin{pmatrix}
			1 & -\big(s_3+\im\epsilon\big)e^{-2tg(\lambda)}\\
			0 & 1\\
			\end{pmatrix},&\lambda\in\gamma_6\cap D(M,\bar{r})\smallskip\\
			\begin{pmatrix}
			\im\epsilon\,s_3 & -\im\epsilon\,e^{-t(\varkappa-\Pi(\lambda))}\\
			-\im\epsilon\,e^{-t(\varkappa+\Pi(\lambda))} & -\im\epsilon\,s_1(1+e^{-\varkappa t})\\
			\end{pmatrix},&\lambda\in(M-\bar{r},M).
			\end{cases}
	\end{equation*}
	Inside the remaining circle centered at the origin,
	\begin{equation*}
		C(\lambda)=\begin{cases}
			\begin{pmatrix}
		1 & (s_1+s_2)e^{-\varkappa t}e^{\im t\Omega(\lambda)}\\
		0 & 1\\
		\end{pmatrix},&\lambda\in\tilde{\gamma}_2\cap D(0,r)\smallskip\\
		\begin{pmatrix}
		1 & 0\\
		-(s_1+s_2)e^{-\varkappa t}e^{-\im t\Omega(\lambda)} & 1\\
		\end{pmatrix},&\lambda\in\tilde{\gamma}_5\cap D(0,r)\smallskip\\
		\begin{pmatrix}
		1 & -(s_3+\im\epsilon)e^{\im t\Omega(\lambda)}\\
		0 & 1\\
		\end{pmatrix},&\lambda\in\gamma_1^+\cap D(0,r)\smallskip\\
		\begin{pmatrix}
		1 & (s_1+s_1e^{-\varkappa t}-\im\epsilon)e^{\im t\Omega(\lambda)}\\
		0 & 1\\
		\end{pmatrix},&\lambda\in\gamma_2^+\cap D(0,r)\smallskip\\
		\begin{pmatrix}
		1 & 0\\
		-(s_3+\im\epsilon)e^{-\im t\Omega(\lambda)} & 1\\
		\end{pmatrix},&\lambda\in\gamma_2^-\cap D(0,r)\smallskip\\
		\begin{pmatrix}
		1 & 0\\
		(s_1+s_1e^{-\varkappa t}-\im\epsilon)e^{-\im t\Omega(\lambda)} & 1\\
		\end{pmatrix},&\lambda\in\gamma_1^-\cap D(0,r)
		\end{cases}
	\end{equation*}
	Recalling our previous expansions for $s_1=\bar{s}_3$ in \eqref{approx:3},\eqref{approx:4}, we see that for $\lambda\in\Sigma_R$ inside the circles with radii $r,\hat{r}$ and $\bar{r}$, we have 
	\begin{equation*}
		C(\lambda)=I+\mathcal{O}\left(e^{-\varkappa t}\right),\ \ t\rightarrow+\infty,|s_1|\uparrow 1:\ \ \varkappa\in\left[\delta,\frac{2}{3}\sqrt{2}-\delta\right],
	\end{equation*}
	which is uniform with respect to the chosen radii. But since $\Xi_-(\lambda)$ is bounded on the same contours, we obtain the estimation
	\begin{equation}\label{DZ:1}
		\|G_R(\cdot\,;t,|s_1|)-I\|_{L^2\cap L^{\infty}(\Sigma_R\cap D_j)}\leq d_1e^{-\varkappa t},\ \ \ d_1>0
	\end{equation}
	where $D_j$ denotes any of the five open disks centered at $\lambda=0$ or $\lambda=\pm m,\pm M$. Secondly we estimate the jumps on the remaining finite branches of $\Sigma_R$: for the horizontal line segments, first for $\lambda\in(-M+\bar{r},-m-\hat{r})$, followed then by $\lambda\in(m+\hat{r},M-\bar{r})$,
	\begin{align}
		R_+(\lambda)&=R_-(\lambda)N_-(\lambda)\begin{pmatrix}
		-\im\epsilon\,s_1(1+e^{-\varkappa t}) & -\im\epsilon\,e^{\im t\Omega(\lambda)}e^{-t(\varkappa-\Pi(\lambda))}\\
		-\im\epsilon\,e^{-\im t\Omega(\lambda)}e^{-t(\varkappa+\Pi(\lambda))} & \im\epsilon\,s_3\\
		\end{pmatrix}\big(N_-(\lambda)\big)^{-1},\nonumber\\
		R_+(\lambda)&=R_-(\lambda)N_-(\lambda)\begin{pmatrix}
		\im\epsilon\,s_3 & -\im\epsilon\,e^{-t(\varkappa-\Pi(\lambda))}\\
		-\im\epsilon\,e^{-t(\varkappa+\Pi(\lambda))} & -\im\epsilon\,s_1(1+e^{-\varkappa t})\\
		\end{pmatrix}\big(N_-(\lambda)\big)^{-1}.\nonumber
	\end{align}
	Notice that for $\lambda\in(m+\hat{r},M-\bar{r})$,
	\begin{align*}
		\varkappa-\Pi(\lambda)&=8\int_m^{\lambda}\sqrt{\left(M^2-\mu^2\right)\left(\mu^2-m^2\right)}\,\d\mu\geq 8\sqrt{2m(M+m)}\int_m^{\lambda}\sqrt{(M-\mu)(\mu-m)}\,\d\mu\\
		&\geq c\sqrt{\k(1-\k)}\,\hat{r}^{\frac{3}{2}},\ c>0\ \ \textnormal{universal},
	\end{align*}
	with a similar estimate also holding for $\lambda\in(-M+\bar{r},-m-\hat{r})$. As the outer parametrix $N_{-}(\lambda)$ remains bounded on the line segments for fixed radii, we have with universal $d_j>0$ (using again \eqref{approx:3},\eqref{approx:4}),
	\begin{equation}\label{DZ:2}
		\|G_R(\cdot\,;t,|s_1|)-I\|_{L^2\cap L^{\infty}((-M+\bar{r},-m-\hat{r})\cup(m+\hat{r},M-\bar{r}))}\leq d_2e^{-d_3t\min\{\varkappa,\sqrt{\k(1-\k)}\hat{r}^{\frac{3}{2}}\}}.
	\end{equation}
	\begin{rem} Although we have chosen fixed radii in \eqref{ratio:1} we shall always indicate the dependency of the estimations of $G_R(\cdot;t,|s_1|)-I$ on the latter. This will be useful later on when we partially drop the constraint on $\varkappa\in[\delta,\frac{2}{3}\sqrt{2}-\delta]$.
	\end{rem}
	Next, on the lens boundaries, for $j=1,2$,
	\begin{eqnarray*}
		R_+(\lambda)&=&R_-(\lambda)N(\lambda)S_{U_j}(\lambda)\big(N(\lambda)\big)^{-1},\ \lambda\in\gamma_j^+\cap\{|\lambda\mp m|>\hat{r}\},\\
		R_+(\lambda)&=&R_-(\lambda)N(\lambda)S_{L_j}(\lambda)\big(N(\lambda)\big)^{-1},\ \lambda\in\gamma_j^-\cap\{|\lambda\mp m|>\hat{r}\},
	\end{eqnarray*}
	and therefore, with fixed radii which ensures the boundedness of $N(\lambda)$,
	\begin{equation}\label{DZ:3}
		\|G_R(\cdot\,;t,|s_1|)-I\|_{L^2\cap L^{\infty}(\gamma_j^{\pm}\cap\{|\lambda\mp m|>\hat{r}\})}\leq d_4e^{-d_5t\sqrt{\k(1-\k)}\,\hat{r}^{\frac{3}{2}}}.
	\end{equation}
	The pieces of $\Sigma_R$ now left are all infinite branches and circle boundaries, more precisely
	\begin{equation*}
		R_+(\lambda)=R_-(\lambda)Q(\lambda)\big(N(\lambda)\big)^{-1},\ \lambda\in C_{\ell_2};\hspace{0.5cm}R_+(\lambda)=R_-(\lambda)V(\lambda)\big(N(\lambda)\big)^{-1},\ \lambda\in C_{\ell_1}
	\end{equation*}
	\begin{equation*}
		R_+(\lambda)=R_-(\lambda)U(\lambda)\big(N(\lambda)\big)^{-1},\ \lambda\in C_{r_1};\hspace{0.5cm}R_+(\lambda)=R_-(\lambda)P(\lambda)\big(N(\lambda)\big)^{-1},\ \lambda\in C_{r_2}
	\end{equation*}
	and
	\begin{equation*}
		R_+(\lambda)=R_-(\lambda)H(\lambda)\big(N(\lambda)\big)^{-1},\ \ \lambda\in C_0.
	\end{equation*}
	We conclude with \eqref{pares:1},\eqref{paraes:2},\eqref{paraes:3},\eqref{paraes:4} and \eqref{paraes:0}, (fixed radii so that $N(\lambda)$ is bounded)
	\begin{equation}\label{DZ:4}
		\|G_R(\cdot\,;t,|s_1|)-I\|_{L^2\cap L^{\infty}(\cup C_j)}\leq \frac{d_6}{t}\big(\k(1-\k)\big)^{-\frac{1}{2}}\max\left\{\hat{r}^{-\frac{3}{2}},\sqrt{\k}\,\bar{r}^{-\frac{3}{2}}\right\}
	\end{equation}
	Also, for the infinite branches,
	\begin{equation*}
		R_+(\lambda)=R_-(\lambda)N(\lambda)e^{-tg(\lambda)\sigma_3}G_T(\lambda)e^{tg(\lambda)\sigma_3}\big(N(\lambda))^{-1},
	\end{equation*}
	where $G_T(\lambda)$ depends only on the Stokes multipliers and $G_T(\lambda)$ are piecewise constant triangular matrices. Due to this triangularity and the sign chart of $\Re\big(g(\lambda)\big)$, compare Figure \ref{figure6}, we have that
	\begin{equation}\label{DZ:5}
		\|G_R(\cdot\,;t,|s_1|)-I\|_{L^2\cap L^{\infty}(\textnormal{infinite})}\leq d_7e^{-d_8t\min\{\sqrt{1-\k}\,\bar{r}^{\frac{3}{2}},\sqrt{\k}\,r\}},\ \ d_j>0.
	\end{equation}
	\item As $\lambda\rightarrow\infty$, 
	\begin{equation*}
		R(\lambda)=I+\mathcal{O}\left(\lambda^{-1}\right).
	\end{equation*}
\end{itemize}	
The latter RHP for $R(\lambda)$ can be solved iteratively as $t\rightarrow+\infty,|s_1|\uparrow 1$ provided $\varkappa\in[\delta,\frac{2}{3}\sqrt{2}-\delta]$.
\subsection{Iterative solution}
We collect estimations \eqref{DZ:1}-\eqref{DZ:5}, so that
\begin{equation}\label{DZ:6}
	\|G_R(\cdot\,;t,|s_1|)-I\|_{L^2\cap L^{\infty}(\Sigma_R)}\leq\frac{d}{t},\ \ t\rightarrow+\infty,|s_1|\uparrow 1:\ \ \varkappa\in\left[\delta,\frac{2}{3}\sqrt{2}-\delta\right],\,\,\,0<\delta<\frac{1}{3}\sqrt{2}.
\end{equation}
Here we used that for the values of $\varkappa\in[\delta,\frac{2}{3}\sqrt{2}-\delta]$ we only work with fixed radii in \eqref{ratio:1}. Since the RHP at hand is equivalent to the singular integral equation
\begin{equation*}
	R_-(\lambda)=I+\frac{1}{2\pi\im}\int_{\Sigma_R}R_-(w)\big(G_R(w)-I\big)\frac{\d w}{w-\lambda},\ \ \lambda\in\Sigma_R,
\end{equation*}
estimation \eqref{DZ:6} guarantees (cf. \cite{DZ1}) the unique existence of an asymptotic solution of the RHP in $L^2(\Sigma_R)$, this solution satisfies moreover
\begin{equation}\label{DZ:7}
	\|R_-(\cdot\,;t,|s_1|)-I\|_{L^2(\Sigma_R)}\leq \frac{d}{t}.
\end{equation}
\subsection{Transition asymptotics}\label{transasy:1}
We now extract the asymptotics through \eqref{FN}, i.e. we apply
\begin{equation*}
	u(x|s)=2\lim_{\lambda\rightarrow\infty}\Big[\lambda\big(Y(\lambda;x,s)\big)_{12}\Big]
\end{equation*}
and for this have to trace back the sequence of transformations
\begin{equation*}
	Y(\lambda)\mapsto X(\lambda)\mapsto Z(\lambda)\mapsto T(\lambda)\mapsto S(\lambda)\mapsto L(\lambda)\mapsto R(\lambda).
\end{equation*}
First, we get that
\begin{align*}
	u(x|s)&=2\lim_{\lambda\rightarrow\infty}\Big[\lambda\big(Y(\lambda)\big)_{12}\Big]=2\sqrt{-x}\lim_{z\rightarrow\infty}\big[zX(z)\big]_{12}=2\sqrt{-x}\lim_{z\rightarrow\infty}\big[zT(z)\big]_{12}\\
	&=2\sqrt{-x}\lim_{z\rightarrow\infty}\left[ze^{t\ell\sigma_3}S(z)e^{-t(g(z)-\vartheta(z))\sigma_3}\right]_{12}=2\sqrt{-x}\lim_{z\rightarrow\infty}\left[ze^{t\ell\sigma_3}L(z)e^{-t(g(z)-\vartheta(z))\sigma_3}\right]_{12}\\
	&=2\sqrt{-x}\lim_{z\rightarrow\infty}\left[ze^{t\ell\sigma_3}N(z)R(z)e^{-t(g(z)-\vartheta(z))\sigma_3}\right]_{12}.
\end{align*}
But since
\begin{equation*}
	g(z)-\vartheta(z)=\ell+\frac{\im}{8z}\left(\frac{1-\k^2}{1+\k^2}\right)^2+\mathcal{O}\left(z^{-3}\right),\ \ z\rightarrow\infty,
\end{equation*}
we have
\begin{equation*}
	e^{-t(g(z)-\vartheta(z))\sigma_3}=e^{-t\ell\sigma_3}\left(I-\frac{\im t}{8z}\left(\frac{1-\k^2}{1+\k^2}\right)^2\sigma_3+\mathcal{O}\left(z^{-2}\right)\right).
\end{equation*}
Also, as $z\rightarrow\infty$,
\begin{equation*}
	R(z)=I+\frac{1}{2\pi\im}\int_{\Sigma_R}R_-(w)\big(G_R(w)-I\big)\frac{\d w}{w-z}=I+\frac{\im}{2\pi z}\int_{\Sigma_R}R_-(w)\big(G_R(w)-I\big)\d w+\mathcal{O}\left(z^{-2}\right)
\end{equation*}
and 
\begin{equation*}
	N(z)=I+\frac{N_1}{z}+\mathcal{O}\left(z^{-2}\right),\ \ z\rightarrow\infty
\end{equation*}
with (compare Proposition \ref{pout:3}),
\begin{equation*}
	(N_1)_{12}=-\frac{\theta(0)}{\theta(tV)}\frac{\theta(u(\infty)-tV-d)}{\theta(u(\infty)-d)}\frac{M-m}{2\im}e^{\im\frac{\pi}{2}\epsilon}=-\frac{\theta_3(0|\tau)}{\theta_2(0|\tau)}\frac{\theta_2(tV|\tau)}{\theta_3(tV|\tau)}e^{\im\pi tV}\frac{M-m}{2\im}e^{\im\frac{\pi}{2}\epsilon}
\end{equation*}
where we used \eqref{Abel:1} and \eqref{Vchoice}, i.e. 
\begin{equation*}
	u(\infty)-d=\frac{\tau}{2},
\end{equation*}
and the well known transformation identities between Jacobi theta functions, compare Appendix \ref{app:theta}. So far, we have thus the exact identity,
\begin{equation}\label{f1}
	u(x)=\sqrt{-x}\,e^{2t\ell}\bigg[-\epsilon(M-m)\frac{\theta_3(0|\tau)}{\theta_2(0|\tau)}\frac{\theta_2(tV|\tau)}{\theta_3(tV|\tau)}e^{\im\pi tV}+\mathcal{E}(\varkappa)\bigg],
\end{equation}
with
\begin{equation*}
	\mathcal{E}(\varkappa)=\frac{\im}{\pi}\int_{\Sigma_R}\Big(R_-(w)\big(G_R(w)-I\big)\Big)_{12}\,\d w=\mathcal{O}\left(t^{-1}\right),\ \ t\rightarrow+\infty,|s_1|\uparrow 1:\ \ \varkappa\in\left[\delta,\frac{2}{3}\sqrt{2}-\delta\right]
\end{equation*}
In order to further simplify expansion \eqref{f1}, we note that
\begin{prop}\label{helpful}
\begin{equation*}
	e^{2t(\ell+\im\frac{\pi}{2}V)} = 1.
\end{equation*}
\end{prop}
\begin{proof} Let us treat 
\begin{align*}
	f(M)\equiv\,&\ell+\im\frac{\pi}{2}V = -\vartheta(M)+4\im\int_M^{\infty}\left[\sqrt{\left(\mu^2-M^2\right)\left(\mu^2-m^2\right)}-\mu^2+\frac{1}{4}\right]\,\d\mu\\
	&-2\im\int_{-m}^m\sqrt{\left(M^2-\mu^2\right)\left(m^2-\mu^2\right)}\,\d\mu,\hspace{0.7cm} m=m(M)=\sqrt{\frac{1}{2}-M^2}
\end{align*}
as a function of one variable $M\in[\frac{1}{2},\frac{1}{\sqrt{2}}]$. It is not hard to verify that
\begin{equation*}
	\frac{\d}{\d M}f(M)=2\im M(m^2-M^2)\left[2\int_M^{\infty}\frac{\d\mu}{\sqrt{(\mu^2-M^2)(\mu^2-m^2)}}-\int_{-m}^m\frac{\d\mu}{\sqrt{(M^2-\mu^2)(m^2-\mu^2)}}\right].
\end{equation*}
But from the identity $u(\infty)=\frac{\tau}{4}$, compare \eqref{Abel:1}, we see that the last difference of integrals vanishes, hence $f(M)$ is in fact $M$ independent. Letting $M\downarrow \frac{1}{2}$, we have
\begin{equation*}
	f(M)=\frac{\im}{3}+o(1)-\frac{\im}{3}=o(1),\ \ \ M\downarrow \frac{1}{2}
\end{equation*}
and thus the stated identity follows.
\end{proof}
Summarizing our previous simplifications and noting that $\ell\in\im\mathbb{R}$, we have 
\begin{equation}\label{final1}
	u(x|s)=-\epsilon\sqrt{-x}\,(M-m)\frac{\theta_3(0|\tau)}{\theta_2(0|\tau)}\frac{\theta_2(tV|\tau)}{\theta_3(tV|\tau)}+J_1(x,s),\ \ \ t=(-x)^{\frac{3}{2}}
\end{equation}
and
\begin{prop}\label{error:1} For any given $\delta\in(0,\frac{1}{3}\sqrt{2})$ there exist positive constants $t_0=t_0(\delta),v_0=v_0(\delta)$ and $c=c(\delta)$ such that
\begin{equation*}
	\big|J_1(x,s)\big|\leq ct^{-\frac{2}{3}},\ \ \ \forall\,t\geq t_0(\delta),\ \ v\geq v_0(\delta),\ \ t\delta\leq v\leq t\left(\frac{2}{3}\sqrt{2}-\delta\right).
\end{equation*}
\end{prop}
The leading term in \eqref{final1}, although currently expressed in terms of ratios of Jacobi theta functions, can be rewritten with the help of the Jacobi elliptic function \eqref{Jell:1} using Appendix \ref{app:theta},
\begin{equation*}
	\theta_j(z|\tau) \equiv \theta_j(z,q),\ \ j=2,3;\hspace{1cm}q=e^{\im\pi\tau}=\exp\left[-2\pi\frac{K}{K'}\right],
\end{equation*}
and the identity (see \eqref{b:point})
\begin{equation*}
	M-m=\frac{1}{\sqrt{2}}\frac{1-\k}{\sqrt{1+\k^2}}.
\end{equation*}
Hence,
\begin{equation}\label{final:11}
	u(x|s)=-\epsilon\sqrt{-\frac{x}{2}}\,\frac{1-\k}{\sqrt{1+\k^2}}\,\textnormal{cd}\left(2(-x)^{\frac{3}{2}}V\,K\left(\frac{1-\k}{1+\k}\right),\,\frac{1-\k}{1+\k}\right)+J_1(x,s)
\end{equation}
and the error term $J_1(x,s)$ is estimated in Proposition \ref{error:1}. We have already control over the error in the domain
\begin{equation*}
	\Big\{\big(t,v\big):\ t\geq t_0,\ v\geq v_0:\ t\delta\leq v\leq t\left(\frac{2}{3}\sqrt{2}-\delta\right)\Big\},
\end{equation*}
but in order to complete the proof of Theorem \ref{bet:1} this domain will be extended in the following sections.
\section{Extension at the lower end for regular transition}\label{relax}
Assume that both $t=(-x)^{\frac{3}{2}}\geq t_0$ and $v=-\ln(1-|s_1|^2)\geq v_0$ are sufficiently large such that
\begin{equation}\label{shrink0}
	0<t^{-\eta}\leq\varkappa\leq\frac{2}{3}\sqrt{2}-\delta,\ \ \ \ \ \ \textnormal{with}\ \ \ \ 0<\eta<1\ \ \ \textnormal{fixed}.
\end{equation}
From \eqref{I:1} and \eqref{b:point}, we see that
\begin{equation*}
	m=\frac{1}{2}-\frac{\hat{c}}{t^{\frac{\eta}{2}}}+\mathcal{O}\left(t^{-\eta}\right),\ \ \ M=\frac{1}{2}+\frac{\hat{c}}{t^{\frac{\eta}{2}}}+\mathcal{O}\left(t^{-\eta}\right)
\end{equation*}
as $t\geq t_0$ for some $\hat{c}>0$. Hence, in the definition of the ratio problem \eqref{ratio:1} for $R(\lambda)$, we cannot keep the radii $\hat{r}$ and $\bar{r}$ fixed since $M-m\downarrow 0$. This issue can be resolved by scaling both radii with $t$, in fact we shall choose in \eqref{ratio:1}
\begin{equation}\label{shrink1}
	\hat{r}=c_1t^{-\frac{\eta}{2}},\hspace{0.5cm} \bar{r}=c_2t^{-\frac{\eta}{2}},\hspace{1cm}0<c_1,c_2:\  c_1+c_2<\hat{c}.
\end{equation}
In order to estimate the jumps in the ratio problem with the latter choice in place, we need to derive estimations for the outer parametrix $N(\lambda)$ and $N_-(\lambda)$ on the jump contours. First with Corollary \ref{cor1}, expansion \eqref{l:1},
\begin{equation*}
	\theta_3\big(z|\tau(\varkappa)\big)=1+\mathcal{O}(\varkappa),\ \ \varkappa \downarrow 0
\end{equation*}
uniformly for $z$ chosen from compact subsets of $\mathbb{C}$. Next, with the contracting choice of radii,
\begin{equation*}
	\omega(z) = \mathcal{O}(1),\ \ \ \omega^{-1}(z)=\mathcal{O}(1),\hspace{1cm}0<t^{-\eta}\leq\varkappa\leq\frac{2}{3}\sqrt{2}-\delta
\end{equation*}
uniformly for $z\in C_{\ell_2}\cup C_{\ell_1}\cup C_{r_1}\cup C_{r_2}$. The same estimations for $\omega^{\pm 1}(z)$ are also valid on the infinite branches $\Sigma_{S_{\infty}}$ as well as on all remaining finite branches outside the five disks $D_j$ centered at $\lambda=0$ and $\lambda=\pm m,\pm M$. Together
\begin{equation*}
	N_-(\lambda)=\mathcal{O}(1)=\big(N_-(\lambda)\big)^{-1},\ \ \varkappa\downarrow 0
\end{equation*}
uniformly on the circle boundaries and all other finite or infinite branches outside $\bigcup D_j$. We can now simply go back to \eqref{DZ:1}-\eqref{DZ:4} and substitute \eqref{shrink0},\eqref{shrink1} and \eqref{I:1} into the estimations. We obtain immediately
\begin{equation}\label{DZ:8}
	\|G_R(\cdot;t,|s_1|)-I\|_{L^2\cap L^{\infty}(\Sigma_R)}\leq\frac{d_9}{t^{1-\eta}},\ \ \forall\, t\geq t_0,\ v\geq v_0:\ \ 0<t^{-\eta}\leq\varkappa\leq\frac{2}{3}\sqrt{2}-\delta,\ \ \eta\in(0,1),
\end{equation}
and thus, repeating the steps in the previous section (with adjusted error terms),
\begin{equation}\label{final2}
	u(x|s)=-\epsilon\sqrt{-\frac{x}{2}}\,\frac{1-\k}{\sqrt{1+\k^2}}\,\textnormal{cd}\left(2(-x)^{\frac{3}{2}}V\,K\left(\frac{1-\k}{1+\k}\right),\,\frac{1-\k}{1+\k}\right)+\mathcal{O}\left((-x)^{-1+\frac{3}{2}\eta}\right),
\end{equation}
uniformly as $x\rightarrow-\infty,|s_1|\uparrow 1$ such that $0<t^{-\eta}\leq\varkappa\leq\frac{2}{3}\sqrt{2}-\delta$ for any $0<\eta<\frac{2}{3}$. In estimation \eqref{final2}, we obtained control over the error term through the a-priori knowledge 
\begin{equation}\label{apriori}
	\|R_-(\cdot\,;t,|s_1|)-I\|_{L^2\cap L^{\infty}(\Sigma_R)}\leq \frac{d_{10}}{t^{1-\eta}}
\end{equation}
and a direct estimation of
\begin{equation*}
	\sqrt{-x}\,\mathcal{E}(\varkappa)=\mathcal{O}\left((-x)^{-1+\frac{3}{2}\eta}\right).
\end{equation*}
We can get a better error estimation by explicitly computing the first terms in
\begin{align*}
	\mathcal{E}(\varkappa)=\frac{\im}{\pi}\sum_{k=1}^2\bigg(\oint_{C_{\ell_k}}\big(G_R(w)-I\big)_{12}\,\d w+&\oint_{C_{r_k}}\big(G_R(w)-I\big)_{12}\,\d w\bigg)\\
	&+\frac{\im}{\pi}\int_{\Sigma_R}\Big(\big(G_R(w)-I\big)\big(R_-(w)-I\big)\Big)_{12}\,\d w.
\end{align*}
Two of the four contour integrals along the circles $C_{\ell_k}\cup C_{r_k}$ have been explicitly computed in \cite{BDIK}, Section $4.2$. It was shown in loc. cit that the aformentioned two contour integrals are in fact $\mathcal{O}\left(t^{-1}\right)$ even with contracting radii. The same result also applies to the contour integrals over $C_{\ell_2}$ and $C_{r_2}$ and without reproducing the lengthy computations of \cite{BDIK}, we simply conclude 
\begin{equation}\label{Ek:1}
	\mathcal{E}(\varkappa)=\frac{\im}{\pi}\int_{\Sigma_R}\Big(\big(G_R(w)-I\big)\big(R_-(w)-I\big)\Big)_{12}\,\d w+\mathcal{O}\left(t^{-1}\right).
\end{equation}
\begin{rem}\label{bettererror} Even without referring to the explicit computations in \cite{BDIK} we see directly from, say \eqref{mat:0} and\eqref{mat:1}, that the evaluation of the contour integral
\begin{equation*}
	 \oint_{C_{r_1}}\big(G_R(w)-I\big)\d w
\end{equation*}
by residue theorem leads to an asymptotic series in reciprocal nonnegative integer powers of $t$.
\end{rem}

For further improvement, we use that we have slightly better $L^2$-estimations than \eqref{DZ:8},
\begin{equation*}
	\|G_R(\cdot;t,|s_1|)-I\|_{L^2(\Sigma_R)}\leq \frac{d_{11}}{t^{1-\frac{3}{4}\eta}}
\end{equation*}
so that by general small norm theory of \cite{DZ1},
\begin{equation*}
	\|R_-(\cdot;t,|s_1|)-I\|_{L^2(\Sigma_R)}\leq \frac{d_{12}}{t^{1-\frac{3}{4}\eta}}.
\end{equation*}
Back to \eqref{Ek:1}, for any $\eta\in(0,1)$,
\begin{equation*}
	\mathcal{E}(\varkappa)=\mathcal{O}\left(t^{-2+\frac{3}{2}\eta}\right)+\mathcal{O}\left(t^{-1}\right)
\end{equation*}
and we have therefore derived \eqref{final:11} with an error term $J_1(x,s)$ such that
\begin{prop}\label{imp:1} For any given $\delta\in(0,\frac{2}{3}\sqrt{2}),\eta\in(0,1)$ there exist positive constants $t_0=t_0(\delta,\eta),v_0=v_0(\delta,\eta)$ and $c=c(\delta,\eta)$ such that
\begin{equation*}
	\big|J_1(x,s)\big|\leq ct^{-\min\{\frac{2}{3},\frac{5}{3}-\frac{3}{2}\eta\}},\ \ \ \forall\, t\geq t_0,\ \ v\geq v_0,\ \ t^{1-\eta}\leq v\leq t\left(\frac{2}{3}\sqrt{2}-\delta\right).
\end{equation*}
\end{prop}

Notice that for $t\geq t_0$ such that $0<\varkappa\leq t^{-\frac{4}{5}}$ we obtain from \eqref{I:1} and \eqref{l:1},
\begin{equation*}
	-\epsilon\sqrt{-\frac{x}{2}}\,\frac{1-\k}{\sqrt{1+\k^2}}\,\textnormal{cd}\left(2(-x)^{\frac{3}{2}}V\,K\left(\frac{1-\k}{1+\k}\right),\,\frac{1-\k}{1+\k}\right)=-\epsilon(-x)^{-\frac{1}{4}}\sqrt{-2\beta}\,\cos(\pi tV(\varkappa))+\mathcal{O}\left((-x)^{-\frac{7}{10}}\right)
\end{equation*}
with
\begin{equation*}
	 \beta=\frac{1}{2\pi}\ln\left(1-|s_1|^2\right).
\end{equation*}
Also from \eqref{l:1}, as $0<\varkappa\leq t^{-\frac{4}{5}}$,
\begin{equation}\label{m:2}
	\pi tV(\varkappa)=-\frac{2}{3}(-x)^{\frac{3}{2}}-\beta\ln\left(8(-x)^{\frac{3}{2}}\right)+\beta\ln\big|\ln(1-|s_1|^2)\big|-\beta(1+\ln 2\pi)+\mathcal{O}\left((-x)^{-\frac{9}{10}}\right).
\end{equation}
But from Stirling's approximation,
\begin{equation}\label{m:3}
	\textnormal{arg}\,\Gamma(\im\beta)=\textnormal{arg}\,\Gamma\left(-\frac{\im}{2\pi}\varkappa t\right)=-\beta(1+\ln 2\pi)+\beta\ln\big|\ln(1-|s_1|^2)\big|+\frac{\pi}{4}+\mathcal{O}\left((\varkappa t)^{-1}\right),
\end{equation}
since $\varkappa t=\mathcal{O}\big(t^{\frac{1}{5}}\big)\rightarrow+\infty$. Using in addition \eqref{approx:3},
\begin{equation}\label{m:4}
	\textnormal{arg}(s_1)=\frac{\epsilon\pi}{2}+\mathcal{O}\left(t^{-\infty}\right),
\end{equation}
expansion \eqref{m:4},\eqref{m:3} combined in \eqref{m:2} gives therefore
\begin{equation}\label{m:5}
	\pi tV(\varkappa)=-\left[\frac{2}{3}(-x)^{\frac{3}{2}}-\beta\ln\left(8(-x)^{\frac{3}{2}}\right)+\phi\right]-\frac{\pi}{2}(1+\epsilon)+\mathcal{O}\left((-x)^{-\frac{3}{10}}\right),
\end{equation}
where (compare \eqref{known:1})
\begin{equation*}
	\phi=-\frac{\pi}{4}-\textnormal{arg}\,\Gamma(\im\beta)-\textnormal{arg}\, s_1.
\end{equation*}
Substituting \eqref{m:5} back into the expansion of the Jacobi elliptic function, we use the addition theorem for $\cos z$ and the identity 
\begin{equation*}
	-\epsilon\cos\big(z+\frac{\pi}{2}(1+\epsilon)\big)=\cos z,\ \  z\in\mathbb{C},
\end{equation*}
so that finally
\begin{align}
	-\epsilon\sqrt{-\frac{x}{2}}\,\frac{1-\k}{\sqrt{1+\k^2}}\,\textnormal{cd}\left(2(-x)^{\frac{3}{2}}V\,K\left(\frac{1-\k}{1+\k}\right),\,\frac{1-\k}{1+\k}\right)=&(-x)^{-\frac{1}{4}}\sqrt{-2\beta}\cos\left(\frac{2}{3}(-x)^{\frac{3}{2}}+\beta\ln\left(8(-x)^{\frac{3}{2}}\right)+\phi\right)\nonumber\\
	&+\mathcal{O}\left((-x)^{-\frac{2}{5}}\right),\ x\rightarrow-\infty,\ \ \label{mat:6}
\end{align}
uniformly for $0<\varkappa\leq t^{-\frac{4}{5}}$. This short computation shows that the Jacobi elliptic function leading term in \eqref{final:11} reproduces the known oscillatory leading behavior of \eqref{known:1} for the values of $(t,v)$ such that $t\geq t_0$ and $0<\varkappa\leq t^{-\frac{4}{5}}$. However in the same region of the double scaling diagram we do not yet have control over the error $J_1(x,s)$, compare Proposition \ref{imp:1}. This will be achieved below using different nonlinear steepest descent techniques applied to the RHP \ref{masterRHP}.

\section{Further extension at the lower end for regular transition}\label{flower}
Suppose throughout that $t=(-x)^{\frac{3}{2}}\geq t_0$ is sufficiently large such that
\begin{equation}\label{lowerscale}
	t\geq v^{k+1}>0,\ \ t\geq t_0,\ \ k\in\mathbb{Z}_{\geq 0}.
\end{equation}
With this constraint in place we employ a different approach to the nonlinear steepest descent analysis of RHP \ref{masterRHP}. In fact we now choose steps which are close in certain points to the ones used in the derivation of \eqref{known:1} for fixed $|s_1|<1$, see \cite{FIKN}, chapter $9$, $\S 4$.\smallskip

Start from the Riemann-Hilbert problem for $Z(\lambda)$ with its jump contour shown in Figure \ref{figure2} or Figure \ref{lower1} below and where we fix 
\begin{equation*}
	\lambda^{\ast}=\frac{1}{2}.
\end{equation*}	
The jump matrix on the full segment $[-\lambda^{\ast},\lambda^{\ast}]$, compare \eqref{Zjump}, is factorized
\begin{equation*}
	S_5^{-1}S_4^{-1}S_3^{-1}=\sigma_2 S_3S_4S_5\sigma_2 = \begin{pmatrix}
	1-s_1s_3 & s_1\\
	s_1 & 1-s_1s_3\\
	\end{pmatrix} = \begin{pmatrix}
	1 & 0\\
	s_1e^{\varkappa t} & 1\\
	\end{pmatrix}e^{-\varkappa t\sigma_3}\begin{pmatrix}
	1 & s_1e^{\varkappa t}\\
	0 & 1\\
	\end{pmatrix}\equiv S_LS_DS_U
\end{equation*}
and we observe that with \eqref{lowerscale},
\begin{equation}\label{obs:1}
	0<e^{\varkappa t}\leq \exp\left(t^{\frac{1}{k+1}}\right),\ \ \ \ \ t\geq t_0,\ \ t\geq v^{k+1}>0.
\end{equation}	
\begin{figure}[tbh]
\begin{center}
\resizebox{0.6\textwidth}{!}{\includegraphics{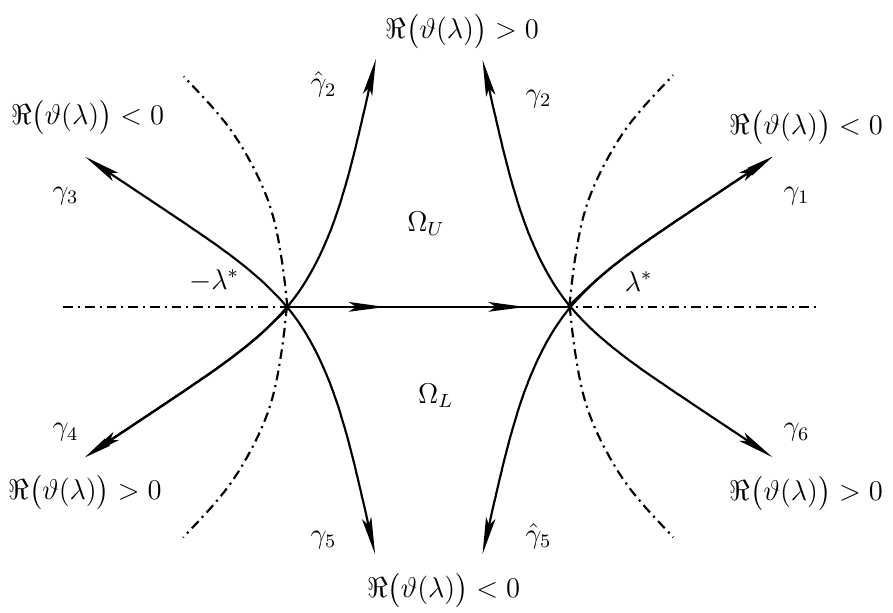}}
\caption{Opening of lens subject to the scale \eqref{lowerscale}, the jump contours for $\Phi(\lambda)$ are the solid black lines. In addition, along the dotted lines $\Re\big(\vartheta(\lambda)\big)=0$.}
\label{lower1}
\end{center}
\end{figure}

Next define (compare Figure \ref{lower1}),
\begin{equation*}
	\Phi(\lambda)=\begin{cases}
		Z(\lambda)S_U^{-1},&\lambda\in\Omega_U\\
		Z(\lambda)S_L,&\lambda\in\Omega_L\\
		Z(\lambda),&\textnormal{otherwise}
		\end{cases}
\end{equation*}
and obtain the following RHP
\begin{itemize}
	\item $\Phi(\lambda)$ is analytic for $\lambda\in\mathbb{C}\backslash(\hat{\gamma}_2\cup\hat{\gamma}_5\cup\bigcup_1^6\gamma_k)$
	\item The boundary values are related via the following jump conditions
	\begin{eqnarray*}
		\Phi_+(\lambda)&=&\Phi_-(\lambda)e^{-t\vartheta(\lambda)\sigma_3}S_ke^{t\vartheta(\lambda)\sigma_3},\ \  \ \ \ \ \ \lambda\in\gamma_k,\ k=1,3,4,6\\
		\Phi_+(\lambda)&=&\Phi_-(\lambda)e^{-t\vartheta(\lambda)\sigma_3}S_2S_U^{-1}e^{t\vartheta(\lambda)\sigma_3},\ \ \lambda\in\gamma_2\\
		\Phi_+(\lambda)&=&\Phi_-(\lambda)e^{-t\vartheta(\lambda)\sigma_3}S_Ue^{t\vartheta(\lambda)\sigma_3},\ \ \ \ \ \ \,\lambda\in\hat{\gamma}_2\\
		\Phi_+(\lambda)&=&\Phi_-(\lambda)e^{-t\vartheta(\lambda)\sigma_3}S_5S_Le^{t\vartheta(\lambda)\sigma_3},\ \ \ \,\lambda\in\gamma_5\\
		\Phi_+(\lambda)&=&\Phi_-(\lambda)e^{-t\vartheta(\lambda)\sigma_3}S_L^{-1}e^{t\vartheta(\lambda)\sigma_3},\ \ \ \ \ \lambda\in\hat{\gamma}_5\\
		\Phi_+(\lambda)&=&\Phi_-(\lambda)S_D,\hspace{3cm}\ \ \lambda\in[-\lambda^{\ast},\lambda^{\ast}]
	\end{eqnarray*}
	\item The function $\Phi(\lambda)$ is normalized,
	\begin{equation*}
		\Phi(\lambda)=I+\mathcal{O}\left(\lambda^{-1}\right),\ \ \lambda\rightarrow\infty.
	\end{equation*}
\end{itemize}
Notice that the jumps on the infinite branches $\gamma_k$ for $k=1,3,4,6$ are exponentially close to the unit matrix in the double scaling limit \eqref{lowerscale} as long as we stay away from $\lambda=\pm\lambda^{\ast}$. On the remaining infinite branches we have for $\lambda\in\partial D(\lambda^{\ast},r)\cup\partial D(-\lambda^{\ast},r),r>0$ with our choice $\lambda^{\ast}=\frac{1}{2}$,
\begin{equation*}
	G_{\Phi}(\lambda)=I+\mathcal{O}\left(e^{\varkappa t}e^{-\hat{c}tr^2}\right)
\end{equation*}
with a universal constant $\hat{c}>0$. Thus recalling \eqref{obs:1}, we could choose a contracting radius $r=r(t)$, compare Section \ref{rloweresti} below, and obtain again sufficiently fast decay for the jumps on the latter infinite branches. We will work out all necessary details after the following two subsections.
\subsection{The outer parametrix}
The required outer parametrix is given by
\begin{equation}\label{crucialpara}
	\Psi^D(\lambda)=\left(\frac{\lambda+\lambda^{\ast}}{\lambda-\lambda^{\ast}}\right)^{\nu\sigma_3},\ \ \lambda\in\mathbb{C}\backslash[-\lambda^{\ast},\lambda^{\ast}];\hspace{1cm}\nu=\frac{\varkappa t}{2\pi\im}
\end{equation}
and appeared in \cite{FIKN}, chapter $9$, $\S 4$. Its relevant analytical properties are summarized below
\begin{itemize}
	\item $\Psi^D(\lambda)$ is analytic for $\lambda\in\mathbb{C}\backslash[-\lambda^{\ast},\lambda^{\ast}]$
	\item The boundary values are related via the equation
	\begin{equation*}
		\Psi^D_+(\lambda)=\Psi^D_-(\lambda)e^{-\varkappa t\sigma_3},\ \ \lambda\in(-\lambda^{\ast},\lambda^{\ast})
	\end{equation*}
	\item $\Psi^D(\lambda)$ is square integrable on the closed interval $[-\lambda^{\ast},\lambda^{\ast}]$
	\item As $\lambda\rightarrow\infty$, we have
	\begin{equation*}
		\Psi^D(\lambda)=I+\frac{\nu\sigma_3}{\lambda}+\mathcal{O}\left(\lambda^{-2}\right).
	\end{equation*}
\end{itemize}
\subsection{The edge point parametrices} Near the end points $\lambda=\pm\lambda^{\ast}$ the required parametrices are constructed in terms of parabolic cylinder functions, see again \cite{FIKN}, chapter $9$, $\S 4$. We briefly summarize the results:\smallskip

Let $D_{\nu}(\z)$ denote the unique solution to Weber's equation
\begin{equation*}
	w''+\left(\nu+\frac{1}{2}-\frac{\z^2}{4}\right)w=0,\ \ w=w(\z)
\end{equation*}
subject to the boundary condition
\begin{equation*}
	D_{\nu}(\z)=\z^{\nu}e^{-\frac{1}{4}\z^2}\left(1-\frac{\nu(\nu-1)}{2\z^2}+\mathcal{O}\left(\z^{-4}\right)\right),\ \ -\frac{3\pi}{4}<\textnormal{arg}\,\z<\frac{3\pi}{4}.
\end{equation*}
Observe that $w=D_{\nu}(\z)$ is entire in $\z\in\mathbb{C}$. We define for $\z\in\mathbb{C}$ the Wronskian type matrix
\begin{equation*}
	Z_0(\z)=2^{-\frac{1}{2}\sigma_3}\begin{pmatrix}
	D_{-\nu-1}(\im\z) & D_{\nu}(\z)\\
	\im D_{-\nu-1}'(\im\zeta) & D_{\nu}'(\z)\\
	\end{pmatrix}\begin{pmatrix}
	e^{\im\frac{\pi}{2}(\nu+1)} & 0\\
	0 & 1\\
	\end{pmatrix}
\end{equation*}
and assemble
\begin{equation*}
	Z^{RH}(\z)=Z_0(\z)\begin{cases}
	I,&\textnormal{arg}\,\z\in(-\frac{\pi}{4},0)\\
	H_0,&\textnormal{arg}\,\z\in(0,\frac{\pi}{2})\\
	H_0H_1,&\textnormal{arg}\,\z\in(\frac{\pi}{2},\pi)\\
	H_0H_1H_2,&\textnormal{arg}\,\z\in(\pi,\frac{3\pi}{2})\\
	H_0H_1H_2H_3,&\textnormal{arg}\,\z\in(\frac{3\pi}{2},\frac{7\pi}{4})
	\end{cases}
\end{equation*}
where
\begin{equation*}
	H_0=\begin{pmatrix}
	1 & 0\\
	h_0 & 1\\
	\end{pmatrix},\ \ H_1=\begin{pmatrix}
	1 & h_1\\
	0 & 1\\
	\end{pmatrix},\hspace{1cm} H_{n+2}=e^{\im\pi(\nu+\frac{1}{2})\sigma_3}H_ne^{-\im\pi(\nu+\frac{1}{2})\sigma_3},\ n=0,1,2
\end{equation*}
and
\begin{equation*}
	h_0=-\frac{\im\sqrt{2\pi}}{\Gamma(1+\nu)},\hspace{0.5cm} h_1=\frac{\sqrt{2\pi}}{\Gamma(-\nu)}e^{\im\pi\nu},\ \ \ \ 1+h_0h_1=e^{2\pi\im\nu}.
\end{equation*}
Using standard properties of parabolic cylinder functions, cf. \cite{N} this setup leads to the following RHP for the bare parametrix $Z^{RH}(\z)$
\begin{itemize}
	\item $Z^{RH}(\z)$ is analytic for $\z\in\mathbb{C}\backslash\{\textnormal{arg}\,\z=0,\frac{\pi}{2},\pi,\frac{3\pi}{2},\frac{7\pi}{4}\}$
	\item The jumps along the contours shown in Figure \ref{cylindbare} are as follows
	\begin{eqnarray*}
		Z_+^{RH}(\z)&=&Z_-^{RH}(\z)H_k,\ \ \textnormal{arg}\,\z=\frac{\pi k}{2},\ k=0,1,2,3\\
		Z_+^{RH}(\z)&=&Z_-^{RH}(\z)e^{2\pi\im\nu\sigma_3},\ \ \textnormal{arg}\,\z=\frac{7\pi}{4}
	\end{eqnarray*}
	\item As $\z\rightarrow\infty$, valid in a full neighborhood of infinity,
	\begin{align*}
		Z^{RH}(\z)\sim&\frac{1}{\sqrt{2}}\z^{-\frac{1}{2}\sigma_3}\begin{pmatrix}
		1 & 1\\
		1 & -1\\
		\end{pmatrix}\left[I+\sum_{s=1}^{\infty}\begin{pmatrix}
		(\nu)_{2s} & (-)^s\big((-\nu)_{2s}-(-\nu-1)_{2s}\big)\\
		(\nu+1)_{2s}-(\nu)_{2s} & (-)^s(-\nu-1)_{2s}\\
		\end{pmatrix}\frac{\z^{-2s}}{s!\,2^s}\right]\\
		&\times\,e^{(\frac{\z^2}{4}-(\nu+\frac{1}{2})\ln\z)\sigma_3}
	\end{align*}
	where we used the full asymptotic series
	\begin{equation*}
		D_{\nu}(\z)\sim e^{-\frac{\z^2}{4}}\z^{\nu}\sum_{s=0}^{\infty}(-)^s\frac{(-\nu)_{2s}}{s!\,2^s}\z^{-2s},\ \ \z\rightarrow\infty,\ \ -\frac{3\pi}{4}<\textnormal{arg}\,\z<\frac{3\pi}{4}
	\end{equation*}
	combined with the identity
	\begin{equation*}
		D_{\nu}'(\z)=\frac{\z}{2}D_{\nu}(\z)-D_{\nu+1}(\z),\ \ \z,\nu\in\mathbb{C}.
	\end{equation*}
\end{itemize}
Next, we employ the locally conformal change of variables
\begin{equation*}
	\z(\lambda)=2\sqrt{t}\big(\vartheta(\lambda^{\ast})-\vartheta(\lambda)\big)^{\frac{1}{2}} = 2\sqrt{2t}\,e^{\im\frac{3\pi}{4}}(\lambda-\lambda^{\ast})\left(1+\frac{2}{3}\big(\lambda-\lambda^{\ast}\big)\right)^{\frac{1}{2}},\ \ |\lambda-\lambda^{\ast}|<r,
\end{equation*}
and define the edge point parametrix near $\lambda=\lambda^{\ast}$ as
\begin{equation*}
	\Psi^r(\lambda)=B(\lambda)Z^{RH}\big(\z(\lambda)\big)e^{-\frac{1}{4}\z^2(\lambda)\sigma_3}e^{t\vartheta(\lambda^{\ast})\sigma_3}\left(-\frac{h_1}{s_3}\right)^{\frac{1}{2}\sigma_3}
\end{equation*}
where
\begin{equation*}
	B(\lambda)=\left(\z(\lambda)\frac{\lambda+\lambda^{\ast}}{\lambda-\lambda^{\ast}}\right)^{\nu\sigma_3}\left(-\frac{h_1}{s_3}\right)^{-\frac{1}{2}\sigma_3}e^{-t\vartheta(\lambda^{\ast})\sigma_3}2^{-\frac{1}{2}\sigma_3}\begin{pmatrix}
	\z(\lambda) & 1\\
	1 & 0\\
	\end{pmatrix},\ \ \ |\lambda-\lambda^{\ast}|<r
\end{equation*}
is analytic at $\lambda=\lambda^{\ast}$. After a local contour deformation we check that the jumps of $\Psi^r(\lambda)$ near $\lambda=\lambda^{\ast}$ are identical to the ones in the initial $\Phi$-RHP, i.e. (compare Figure \ref{parabo})
\begin{figure}[tbh]
\begin{minipage}{0.35\textwidth} 
\begin{center}
\resizebox{0.9\textwidth}{!}{\includegraphics{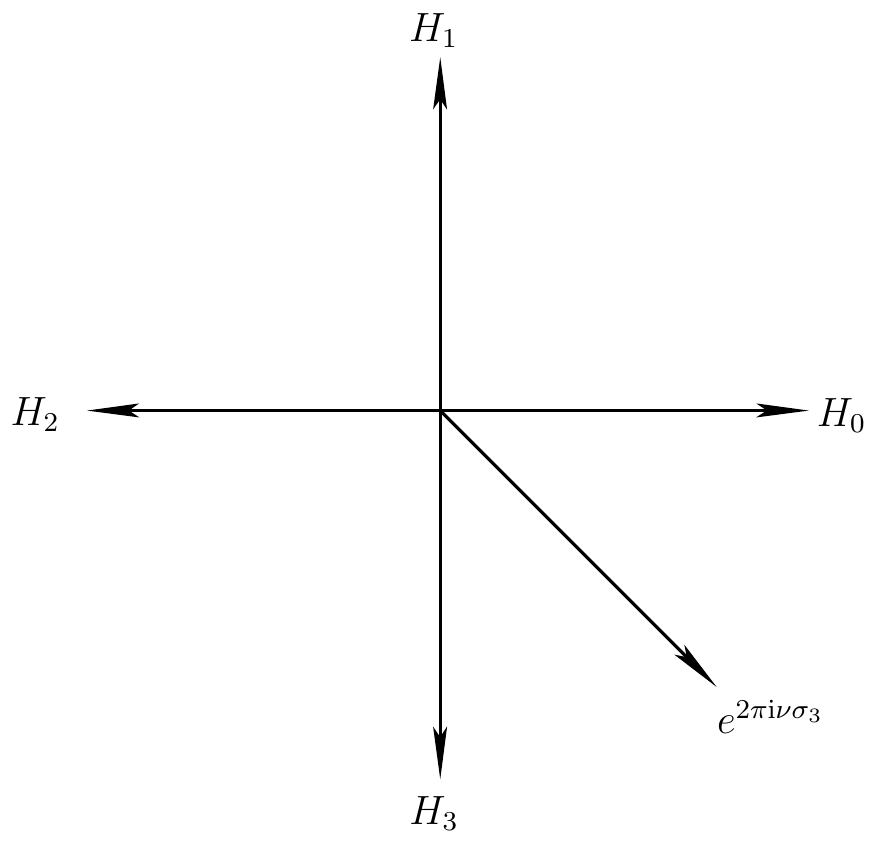}}
\caption{Jump contour for the bare parametrix $Z^{RH}(\zeta)$.}
\label{cylindbare}
\end{center}
\end{minipage}
\begin{minipage}{0.5\textwidth}
\begin{center}
\resizebox{0.5\textwidth}{!}{\includegraphics{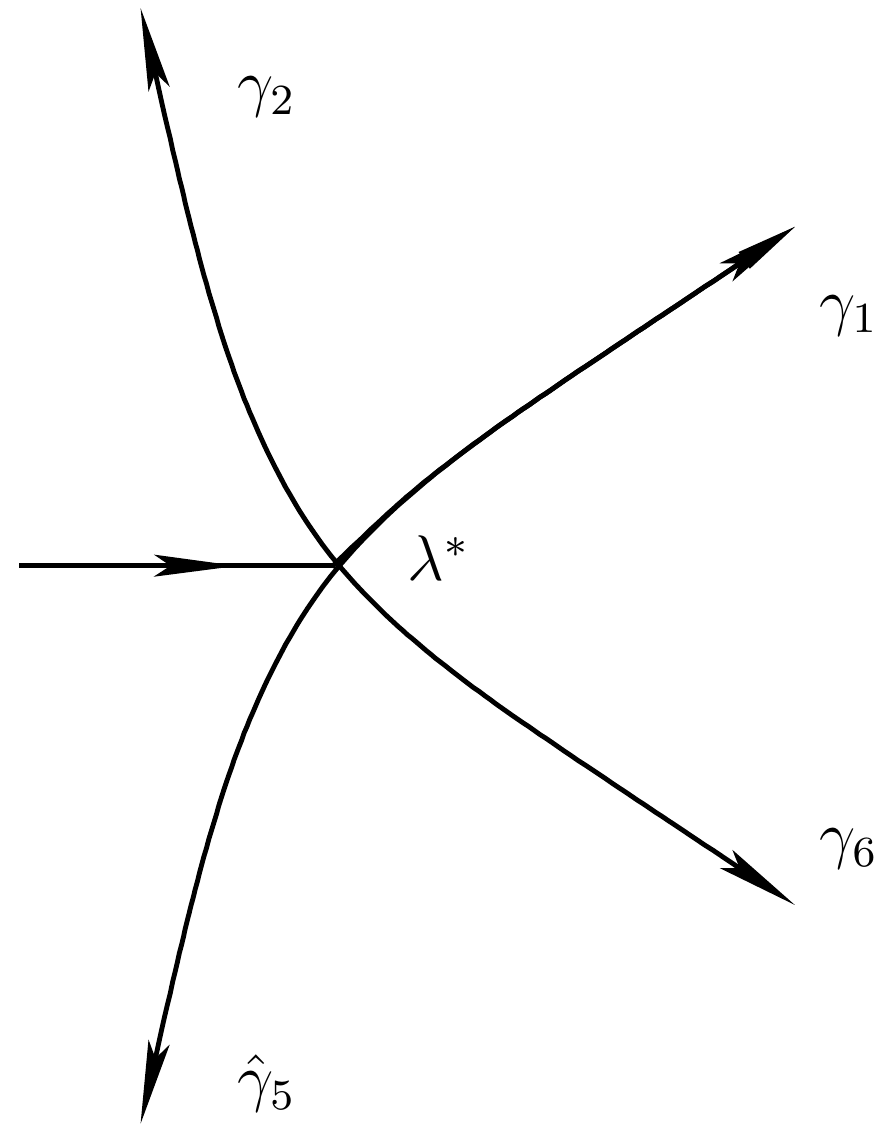}}
\caption{Jump contour of the model function $\Psi^r(\lambda)$ near $\lambda=\lambda^{\ast}$.}
\label{parabo}
\end{center}
\end{minipage}
\end{figure}	
\begin{align*}
	\Psi^r_+(\lambda)&=\Psi^r_-(\lambda)e^{-t\vartheta(\lambda)\sigma_3}S_1e^{t\vartheta(\lambda)\sigma_3},\ \lambda\in\gamma_{1r};\ \ \ \
	\Psi^r_+(\lambda)=\Psi^r_-(\lambda)e^{-t\vartheta(\lambda)\sigma_3}S_2S_U^{-1}e^{t\vartheta(\lambda)\sigma_3},\ \ \lambda\in\gamma_{2r}\\
	\Psi^r_+(\lambda)&=\Psi^r_-(\lambda)e^{-t\vartheta(\lambda)\sigma_3}S_6e^{t\vartheta(\lambda)\sigma_3},\ \ \lambda\in\gamma_{6r};\ \ \
	\Psi^r_+(\lambda)=\Psi^r_-(\lambda)e^{-t\vartheta(\lambda)\sigma_3}S_L^{-1}e^{t\vartheta(\lambda)\sigma_3},\ \ \lambda\in\hat{\gamma}_{5r}\\
	\Psi^r_+(\lambda)&=\Psi^r_-(\lambda)S_D,\ \ \lambda\in(\lambda^{\ast}-r,\lambda^{\ast}),
\end{align*}
where $\gamma_{jr}=\gamma_j\cap D(\lambda^{\ast},r)$. Moreover, with the previous asymptotic expansion of $Z^{RH}(\z)$,
\begin{align}
	\Psi^r&(\lambda)\sim\big(\beta(\lambda)\big)^{\sigma_3}\left(-\frac{h_1}{s_3}\right)^{-\frac{1}{2}\sigma_3}\z^{\frac{1}{2}\sigma_3}\begin{pmatrix}
	1 & 0\\
	1 & 1\\
	\end{pmatrix}\left[I+\sum_{s=1}^{\infty}C_{2s}(\nu)\frac{\z^{-2s}}{s!\,2^s}\right]\z^{-\frac{1}{2}\sigma_3}\left(-\frac{h_1}{s_3}\right)^{\frac{1}{2}\sigma_3}\big(\beta(\lambda)\big)^{-\sigma_3}\Psi^D(\lambda)\nonumber\\
	&=\bigg[I+\frac{\sigma_-}{\z(\lambda)}\big(\beta(\lambda)\big)^{-2}\left(-\frac{h_1}{s_3}\right)\sum_{m=0}^{\infty}(\nu+1)_{2m}\frac{\z(\lambda)^{-2m}}{m!\,2^m}+\sigma_+\big(\beta(\lambda)\big)^2\left(-\frac{s_3}{h_1}\right)\nonumber\\
	&\times\,\sum_{m=1}^{\infty}(-)^m\big((-\nu)_{2m}-(-\nu-1)_{2m}\big)\frac{\z(\lambda)^{-2m+1}}{m!\,2^m}+\sum_{m=1}^{\infty}\begin{pmatrix}
	(\nu)_{2m} & 0\\
	0 & (-)^m(-\nu)_{2m}\\
	\end{pmatrix}\frac{\z(\lambda)^{-2m}}{m!\,2^m}\bigg]\Psi^D(\lambda),\label{mup:1}
\end{align}
as $t\rightarrow+\infty$ (and thus $|\z|\rightarrow\infty$). Here we used the abbreviation
\begin{equation*}
	\beta(\lambda)=\left(\z(\lambda)\frac{\lambda+\lambda^{\ast}}{\lambda-\lambda^{\ast}}\right)^{\nu}e^{-t\vartheta(\lambda^{\ast})}.
\end{equation*}
Expansion \eqref{mup:1} establishes the matching relation between the model functions $\Psi^r(\lambda)$ and $\Psi^D(\lambda)$ valid as $t\rightarrow\infty$ for $0<r_1\leq|\lambda-\lambda^{\ast}|\leq r_2<\frac{1}{4}$.\smallskip

Near the other edge point $\lambda=-\lambda^{\ast}$ we can use the symmetry of the problem and define the required local parametrix as
\begin{equation}\label{lowerparaleft}
	\Psi^{\ell}(\lambda)=\sigma_2\Psi^r(-\lambda)\sigma_2,\ \ \ |\lambda+\lambda^{\ast}|<r
\end{equation}
which replaces the matchup \eqref{mup:1} by
\begin{align}
	\Psi^{\ell}&(\lambda)\sim\bigg[I+\frac{\sigma_+}{\z(-\lambda)}\big(\beta(-\lambda)\big)^{-2}\left(\frac{h_1}{s_3}\right)\sum_{m=0}^{\infty}(\nu+1)_{2m}\frac{\z(-\lambda)^{-2m}}{m!\,2^m}+\sigma_-\big(\beta(-\lambda)\big)^2\left(\frac{s_3}{h_1}\right)\nonumber\\
	&\times\,\sum_{m=1}^{\infty}(-)^m\big((-\nu)_{2m}-(-\nu-1)_{2m}\big)\frac{\z(-\lambda)^{-2m+1}}{m!\,2^m}+\sum_{m=1}^{\infty}\begin{pmatrix}
	(-)^m(-\nu)_{2m} & 0\\
	0 & (\nu)_{2m}\\
	\end{pmatrix}\frac{\z(-\lambda)^{-2m}}{m!\,2^m}\bigg]\Psi^D(\lambda).\label{mup:2}
\end{align}
Collecting the model functions $\Psi^D(\lambda),\Psi^r(\lambda)$ and $\Psi^{\ell}(\lambda)$ we move on to the final transformation.

\subsection{Ratio problem}\label{rloweresti}
Assemble
\begin{equation}\label{ratiolower}
	\chi(\lambda)=\Phi(\lambda)\begin{cases}
	\big(\Psi^r(\lambda)\big)^{-1},&|\lambda-\lambda^{\ast}|<r\\
	\big(\Psi^{\ell}(\lambda)\big)^{-1},&|\lambda+\lambda^{\ast}|<r\\
	\big(\Psi^D(\lambda)\big)^{-1},&|\lambda\pm \lambda^{\ast}|>r
	\end{cases}
\end{equation}
with $0<r<\frac{1}{4}$ which we will choose more specifically below. This leads to the RHP formulated on the contours as shown in Figure \ref{ratiolowerfig}
\begin{figure}[tbh]
\begin{center}
\resizebox{0.5\textwidth}{!}{\includegraphics{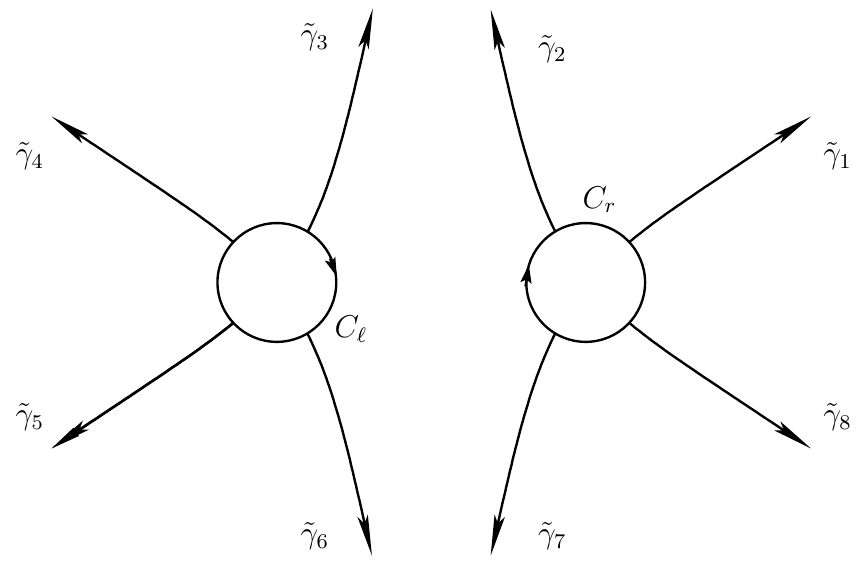}}
\caption{Jump contour for the ratio problem as introduced in \eqref{ratiolower}.}
\label{ratiolowerfig}
\end{center}
\end{figure}
\begin{itemize}
	\item $\chi(\lambda)$ is analytic for $\lambda\in\mathbb{C}\backslash(C_r\cup C_{\ell}\cup\{\textnormal{infinite branches}\,\tilde{\gamma}_j\})$
	\item The following jumps are valid:
	\begin{equation*}
		\chi_+(\lambda)=\chi_-(\lambda)\Psi^D(\lambda)\tilde{S}_j\big(\Psi^D(\lambda)\big)^{-1},\ \ \ \lambda\in\tilde{\gamma}_j,\ \ j=1,\ldots,8
	\end{equation*}
	along the infinite branches $\tilde{\gamma}_k$ with piecewise constant matrices $\tilde{S}_j$ which are given in the $\Phi$-RHP. There are no jumps inside the circles and along the line segment $[-\lambda^{\ast},\lambda^{\ast}]$ since we constructed ``exact" local parametrices $\Psi^D(\lambda),\Psi^r(\lambda)$ and $\Psi^{\ell}(\lambda)$. However on the circle boundaries,
	\begin{equation*}
		\chi_+(\lambda)=\chi_-(\lambda)\Psi^r(\lambda)\big(\Psi^D(\lambda)\big)^{-1},\ \ \lambda\in C_r;\hspace{1cm}\chi_+(\lambda)=\chi_-(\lambda)\Psi^{\ell}(\lambda)\big(\Psi^D(\lambda)\big)^{-1},\ \ \lambda\in C_{\ell}
	\end{equation*}
	\item As $\lambda\rightarrow\infty$, we have that
	\begin{equation*}
		\chi(\lambda)=I+\mathcal{O}\left(\lambda^{-1}\right).
	\end{equation*}
\end{itemize}
For the jumps on the infinite branches $\tilde{\gamma}_j,j=1,\ldots,8$ we obtain immediately
\begin{equation*}
	G_{\chi}(\lambda)=I+\mathcal{O}\left(e^{d_1\varkappa t-d_2tr^2}\right)
\end{equation*}
with universal constants $d_1,d_2>0$. If we choose a (in general) contracting radius $r$ such that
\begin{equation*}
	\frac{1}{4}>r=\frac{1}{v}\geq t^{-\frac{1}{k+1}}>0,
\end{equation*}
then by \eqref{lowerscale} for $t\geq v^{k+1}>0,t\geq t_0$,
\begin{equation*}
	d_1\varkappa t-d_2tr^2\leq d_1t^{\frac{1}{k+1}}-d_2t^{\frac{k-1}{k+1}}\leq -d_3t^{\frac{k-1}{k+1}},\ \ d_3>0,\ k\geq 3;
\end{equation*}
hence all contributions from the infinite branches are approaching the identity matrix exponentially fast, i.e.
\begin{equation}\label{loesti:1}
	\|G_{\chi}(\cdot;t,|s_1|)-I\|_{L^2\cap L^{\infty}(\cup\tilde{\gamma}_j)}\leq d_4e^{-d_5t^{\frac{k-1}{k+1}}},\ \ k\in\mathbb{Z}_{\geq 3}
\end{equation}
for $t\geq v^{k+1}>0,t\geq t_0$. For the circle boundaries we use \eqref{mup:1} and \eqref{mup:2}, for instance for $\lambda\in C_r$ we use $|\z(\lambda)|\geq 4\sqrt{2}\,t^{\frac{k-1}{2(k+1)}}$ as well as
\begin{equation*}
	h_1=\mathcal{O}\left((\varkappa t)^{\frac{1}{2}}e^{\frac{3}{4}\varkappa t}\right)
\end{equation*}
which follows from Stirling's approximation. Thus for $\lambda\in C_r$,
\begin{equation*}
	\big|G_{\chi}(\lambda;t,|s_1|)-I\big|\leq c\left\{\frac{(\varkappa t)^{\frac{1}{2}}}{r\sqrt{t}}\begin{pmatrix}
	0 & e^{-\frac{3}{4}\varkappa t}e^{\varphi(r)}\\
	e^{\frac{3}{4}\varkappa t}e^{-\varphi(r)} & 0\\
	\end{pmatrix}+\hat{\mathcal{E}}\big(\lambda;t,|s_1|\big)\right\}
\end{equation*}
where
\begin{equation*}
	\varphi(r)=\frac{\varkappa t}{\pi}\textnormal{arg}\left(\z(\lambda)\frac{\lambda+\lambda^{\ast}}{\lambda-\lambda^{\ast}}\right)\ \ \Rightarrow\ \ \left|\frac{3}{4}\varkappa t-\varphi(r)\right|\leq \frac{\varkappa t}{\pi}\arctan(r)=\frac{\varkappa t}{\pi}r\left(1+\mathcal{O}\left(r^3\right)\right)=\frac{1}{\pi}\left(1+o(1)\right).
\end{equation*}
Bounds for the error term $\hat{\mathcal{E}}(\lambda;t,|s_1|)$ can be found for instance in \cite{N}: there exists $t_0=t_0(k)>0$ and constants $d_j=d_j(t_0)$ such that
\begin{equation*}
	\|\hat{\mathcal{E}}(\lambda;t,|s_1|)\|\leq d_6e^{d_7\varkappa t-d_8\sqrt{t}r}\leq d_6e^{-d_9t^{\frac{k-1}{2(k+1)}}},\ \ \forall\,t\geq v^{k+1}>0,\ \ t\geq t_0,\ \ k\in\mathbb{Z}_{\geq 3}.
\end{equation*}
For the latter estimation to hold for $k=3$, we have to adjust the radius $r\mapsto \frac{c}{v}$ with a sufficiently large, but $k,t,|s_1|$ independent, constant $c>0$. Notice that with \eqref{lowerparaleft} completely similar estimations will also hold for $G_{\chi}(\lambda;t,|s_1|)$ when $\lambda\in C_{\ell}$. Thus, summarizing
\begin{equation}\label{loesti:2}
	\|G_{\chi}(\cdot\,;t,|s_1|)-I\|_{L^2\cap L^{\infty}(C_r\cup C_{\ell})}\leq d_{10}t^{-\frac{k-2}{2(k+1)}},\ \ \forall\,t\geq v^{k+1}>0,\ \ t\geq t_0,\ \ k\in\mathbb{Z}_{\geq 3}.
\end{equation}
The latter estimation combined with \eqref{loesti:1}, we obtain from general theory \cite{DZ1} the unique solvability of the ratio problem for the values of $(t,v)$ satisfying \eqref{lowerscale}. Moreover its unique solution satisfies
\begin{equation}\label{loesti:3}
	\|\chi_-(\cdot\,;t,|s_1|)-I\|_{L^2(\Sigma_R)}\leq d_{11}t^{-\frac{k-2}{2(k+1)}},\ \ \forall\,t\geq v^{k+1}>0,\  \ t\geq t_0,\ \ k\in\mathbb{Z}_{\geq 3}.
\end{equation}
\subsection{Derivation of leading order asymptotics subject to \eqref{lowerscale}}
We need to trace back the transformations
\begin{equation*}
	Y(\lambda)\mapsto X(\lambda)\mapsto Z(\lambda)\mapsto \Phi(\lambda)\mapsto \chi(\lambda)
\end{equation*}
in order to obtain with \eqref{FN},
\begin{align*}
	u(x|s)&=2\lim_{\lambda\rightarrow\infty}\Big[\lambda\big(Y(\lambda)\big)_{12}\Big]=2\sqrt{-x}\lim_{z\rightarrow\infty}\big[z\,\Phi(z)\big]_{12} = 2\sqrt{-x}\lim_{z\rightarrow\infty}\big[z\Psi^D(z)\chi(z)\big]_{12}\\
	&=2\sqrt{-x}\left[\nu\sigma_3+\frac{\im}{2\pi}\int_{\Sigma_{\chi}}\chi_-(w)\big(G_{\chi}(w)-I\big)\,\d w\right]_{12}\\
	&=\frac{\im\sqrt{-x}}{\pi}\oint_{C_R\cup C_{\ell}}\big(G_{\chi}(w)-I\big)_{12}\,\d w+\mathcal{O}\left((-x)^{-\frac{2k-7}{2(k+1)}}\right),\ \ k\in\mathbb{Z}_{\geq 4}
\end{align*}
where we already neglected all exponentially small contributions from the infinite branches, compare \eqref{loesti:1}. Also the latter error estimation follows from \eqref{loesti:3} and the following computation: observe that by residue theorem and \eqref{mup:1} as well as \eqref{mup:2},
\begin{align*}
	\oint_{C_r}\big(G_{\chi}(w)-I\big)_{12}\d w&\sim \sum_{m=1}^{\infty}(-)^m\big((-\nu)_{2m}-(-\nu-1)_{2m}\big)\frac{1}{m!\,2^m}\left(-\frac{s_3}{h_1}\right)\oint_{C_r}\beta^2(w)\z(w)^{-2m+1}\,\d w\\
	&=\nu\left(-\frac{s_3}{h_1}\right)(-2\pi\im)e^{-2t\vartheta(\lambda^{\ast})}\left(2\sqrt{2t}\,e^{\im\frac{3\pi}{4}}\right)^{2\nu-1}+\mathcal{O}\left(t^{-\frac{3k-2}{2(k+1)}}\right);\\
	\oint_{C_\ell}\big(G_{\chi}(w)-I\big)_{12}\d w&\sim\sum_{m=0}^{\infty}(\nu+1)_{2m}\frac{1}{m!2^m}\left(\frac{h_1}{s_3}\right)\oint_{C_{\ell}}\beta^{-2}(-w)\z(-w)^{-2m-1}\,\d w\\
	&=\left(-\frac{h_1}{s_3}\right)(-2\pi\im)e^{2t\vartheta(\lambda^{\ast})}\left(2\sqrt{2t}\,e^{\im\frac{3\pi}{4}}\right)^{-2\nu-1}+\mathcal{O}\left(t^{-\frac{3k-2}{2(k+1)}}\right).
\end{align*}
Recalling the definition of $h_1$ and applying standard symmetrization techniques, we obtain with $k=4$,
\begin{equation}\label{lowerfinal}
	u(x|s)=(-x)^{-\frac{1}{4}}\sqrt{-2\beta}\cos\left(\frac{2}{3}(-x)^{\frac{3}{2}}+\beta\ln\left(8(-x)^{\frac{3}{2}}\right)+\phi\right)+\mathcal{O}\left((-x)^{-\frac{1}{10}}\right),
\end{equation}
uniformly as $x\rightarrow-\infty$ subject to the scale \eqref{lowerscale} with $k=4$. Estimation \eqref{lowerfinal} appears in Corollary \ref{mat}, expansion \eqref{cor:as}. Moreover, with \eqref{mat:6}, we have now shown that
\begin{equation*}
	u(x|s)=-\epsilon\sqrt{-\frac{x}{2}}\,\frac{1-\k}{\sqrt{1+\k^2}}\,\textnormal{cd}\left(2(-x)^{\frac{3}{2}}V\,K\left(\frac{1-\k}{1+\k}\right),\,\frac{1-\k}{1+\k}\right)+J_1(x,s),
\end{equation*}
where
\begin{prop}\label{imp:2} There exists positive constants $t_0$ and $c$ such that
\begin{equation*}
	\big|J_1(x,s)\big|\leq ct^{-\frac{1}{15}},\ \ \ \forall\, t\geq t_0,\ \ 0<v\leq t^{\frac{1}{5}}
\end{equation*}
\end{prop}
Propositions \ref{imp:2} and \ref{imp:1} combined together (for $\eta=\frac{4}{5}$ in the latter) form the content of Theorem \ref{bet:1}, estimation \eqref{bd:1}. In order to complete the proof of Theorem \ref{bet:1} we have to derive the outstanding estimation \eqref{bd:2} by extending the region
\begin{equation*}
	\left\{\big(t,v\big):\ t\geq t_0:\ 0<v\leq t\left(\frac{2}{3}\sqrt{2}-\delta\right)\right\}
\end{equation*}
at the upper end. This will be achieved in the next section.
\section{Extension at the upper end for regular transition}\label{upextension}
In the analysis of the lower end extension presented in Section \ref{relax} we used the crucial fact that for $\varkappa\downarrow 0$, we have $-\im\tau\rightarrow+\infty$, compare \eqref{l:1}. Thus at the lower end, the oscillations in the Jacobi theta functions become strictly periodic and their amplitudes vanish exponentially fast. In turn we obtained
\begin{equation*}
	N_-(\lambda)=\mathcal{O}(1)=\big(N_-(\lambda)\big)^{-1},\ \ \varkappa\downarrow 0,\ \ \lambda\in\Sigma_R\backslash\bigcup D_j
\end{equation*}
which lead us to the central estimation \eqref{DZ:8}. At the upper end the oscillations in the theta functions still vanish, however this time not through vanishing amplitudes but increasing periods: as $\varkappa\uparrow \frac{2}{3}\sqrt{2}$ we get with \eqref{l:2},
\begin{equation*} 
	-\im\tau=\frac{2\pi}{|\ln\sigma|}+\mathcal{O}\left(\frac{\ln|\ln\sigma|}{\ln^2\sigma}\right)=o(1),\ \ \sigma=\frac{2}{3}\sqrt{2}-\varkappa\downarrow 0.
\end{equation*}
This implies the necessity of using modular transformations for the Jacobi theta functions appearing in the construction of the outer parametrix \eqref{ou:3}, i.e. we use (cf. \cite{N})
\begin{equation*}
	\theta_3(z|\tau)=\sqrt{\frac{\im}{\tau}}\,\theta_3(z\tau'|\tau')e^{\im\pi\tau'z^2},\hspace{0.5cm}\theta_2(z|\tau)=\sqrt{\frac{\im}{\tau}}\,\theta_4(z\tau'|\tau')e^{\im\pi\tau'z^2},\ \ z\in\mathbb{C},\ \ \tau'=-\frac{1}{\tau}.
\end{equation*}
Going back to \eqref{ou:3} the transformation identities imply for $\lambda\in\mathbb{C}\backslash\overline{J}$,
\begin{equation*}
	N(\lambda)=e^{\im\frac{\pi}{4}\epsilon\sigma_3}\frac{\theta(0|\tau')}{\theta(\tau'tV|\tau')}e^{-2\pi\im\tau'tVu(\infty)\sigma_3}\begin{pmatrix}
	N_1'^{(+)}(\lambda)\phi(\lambda) & N_2'^{(+)}(\lambda)\hat{\phi}(\lambda)\\
	-N_1'^{(-)}(\lambda)\hat{\phi}(\lambda) & N_2'^{(-)}(\lambda)\phi(\lambda)
	\end{pmatrix}e^{2\pi\im\tau'tVu(\lambda)\sigma_3}e^{-\im\frac{\pi}{4}\epsilon\sigma_3}
\end{equation*}
where
\begin{equation*}
	\left(N_1'^{(\pm)}(z),N_2'^{(\pm)}(z)\right)=\left(\frac{\theta(\tau'(u(z)+tV\pm d)|\tau')}{\theta(\tau'(u(z)\pm d)|\tau')},\frac{\theta(\tau'(-u(z)+tV\pm d)|\tau')}{\theta(\tau'(-u(z)\pm d)|\tau')}\right).
\end{equation*}
Notice that from \eqref{l:2}, for $\lambda\in\mathbb{CP}^1\backslash[-M,M]$ chosen from compact subsets,
\begin{equation}\label{e:1s}
	e^{\pm 2\pi\im\tau'tVu(\lambda)}\sim\exp\left[\mp\sigma t\left(1+\mathcal{O}\left(\frac{\ln^2|\ln\sigma|}{\ln^2\sigma}\right)\right) u(\lambda)\,\right],\ \ \sigma\downarrow 0
\end{equation}
and for the same values of $\lambda$ we have at the same time
\begin{equation}\label{e:2s}
	u(\lambda)=\mathcal{O}\left(\frac{1}{\ln\sigma}\right)=u(\infty)
\end{equation}
since with \eqref{norm},
\begin{equation*}
	c(\varkappa)=\frac{\im}{\sqrt{2}\,|\ln\sigma|}\left(1+\frac{\ln|\ln\sigma|}{\ln\sigma}+\mathcal{O}\left(\frac{1}{\ln\sigma}\right)\right),\ \ \sigma\downarrow 0.
\end{equation*}
Along the line segment $[-M,M]$, by direct computation
\begin{equation*}
	u_-(\lambda)=\mathcal{O}\left(\frac{1}{\ln\sigma}\right),\ \ \sigma\downarrow 0,\ \ \ \ \lambda\in(-M,-m-r]\cup[m+r,M),\ \ \ 0<r<\frac{1}{2}\ \ \ \textnormal{fixed},
\end{equation*}
and hence all together
\begin{equation*}
	e^{\pm 2\pi\im\tau'tVu(\lambda)}=\mathcal{O}(1),\ \ \ \ \ \ \ \sigma\downarrow 0,\ \  \lambda\in\mathbb{CP}^1\backslash D(0,\hat{r}),
\end{equation*}
for any fixed $\hat{r}:m<\hat{r}<M$, provided we ensure that 
\begin{equation}\label{sc:1}
	\frac{\sigma t}{\ln\sigma}=o(1),\ \ \sigma\downarrow 0.
\end{equation}
We shall in fact impose a stronger condition on $\varkappa$ than implied by \eqref{sc:1}: suppose subsequently that $t=(-x)^{\frac{3}{2}}$ and $v=-\ln\big(1-|s_1|^2\big)$ are sufficiently large such that
\begin{equation}\label{sc:2}
	\frac{2}{3}\sqrt{2}-\frac{f_1}{t}\leq\varkappa<\frac{2}{3}\sqrt{2}\ \ \ \ \ \Leftrightarrow\ \ \ \ \ 0<\sigma\leq\frac{f_1}{t},\ \ \ \ f_1>0.
\end{equation}
Subject to this constraint we obtain directly with \eqref{l:2},
\begin{equation}\label{modup}
	N_1'^{(\pm)}(\lambda)=1+\mathcal{O}\left(\sqrt{\frac{\sigma}{|\ln\sigma|}}\,\right) = N_2'^{(\pm)}(\lambda),
\end{equation}
uniformly for $\lambda\in\mathbb{CP}^1\backslash\overline{J}$. Moreover
\begin{prop}\label{ou:match} As $x\rightarrow-\infty,|s_1|\uparrow 1$ subject to \eqref{sc:2} we have
\begin{equation}\label{outparamatch}
	N(\lambda)=\left(I+\mathcal{O}\left(\sqrt{\frac{\sigma}{|\ln\sigma|}}\,\right)\right)e^{-2\pi\im\tau'tVu(\infty)\sigma_3}\Upsilon(\lambda)e^{2\pi\im\tau'tVu(\lambda)\sigma_3}=\left(I+\mathcal{O}\left(\frac{1}{\ln t}\right)\right)\Upsilon(\lambda)
\end{equation}
uniformly for $\lambda$ chosen from $\partial D(0,\hat{r})$ for any $\hat{r}:m<\hat{r}<M$. Here, $\Upsilon(\lambda)$ equals
\begin{equation}\label{nouter}
	\Upsilon(\lambda)=e^{\im\frac{\pi}{4}\epsilon\sigma_3}\alpha(\lambda)^{\sigma_2}e^{-\im\frac{\pi}{4}\epsilon\sigma_3}=\frac{1}{2}e^{\im\frac{\pi}{4}\epsilon\sigma_3}\begin{pmatrix}
	\alpha(\lambda)+\alpha^{-1}(\lambda) & -\im(\alpha(\lambda)-\alpha^{-1}(\lambda))\\
	\im(\alpha(\lambda)-\alpha^{-1}(\lambda)) & \alpha(\lambda)+\alpha^{-1}(\lambda)\\
	\end{pmatrix}e^{-\im\frac{\pi}{4}\epsilon\sigma_3},
\end{equation}
where
\begin{equation*}
	\alpha(\lambda)=\left(\frac{\lambda-\frac{1}{\sqrt{2}}}{\lambda+\frac{1}{\sqrt{2}}}\right)^{\frac{1}{4}}\rightarrow 1,\ \ \lambda\rightarrow\infty
\end{equation*}
is defined and analytic on $\mathbb{C}\backslash\big[-\frac{1}{\sqrt{2}},\frac{1}{\sqrt{2}}\big]$.
\end{prop}
At this point let us get back to the nonlinear steepest descent analysis. We start from the initial RHP \ref{masterRHP} and use the transformation sequence
\begin{equation*}
	Y(\lambda)\mapsto X(\lambda)\mapsto Z(\lambda)\mapsto T(\lambda)\mapsto S(\lambda)
\end{equation*}
just like it occurred in Section \ref{sec:3} and parts of Section \ref{sec:reg} in the analysis for $\varkappa\in[\delta,\frac{2}{3}\sqrt{2}-\delta]$. However, instead of using the ``opening of lens" transformation \eqref{ou:1} we will not split the jump contours near $(-m,m)$ but instead introduce a new model problem on the (contracting) line segment $(-\hat{r},\hat{r})\subset(-M,M)$ with
\begin{equation}\label{ridea}
	\hat{r}=m+\frac{d}{t^{\frac{1}{3}}},\ \ \ d>0.
\end{equation}
In more detail
\begin{problem}\label{nmodel} Find a $2\times 2$ matrix-valued piecewise analytic function $J(\lambda)=J(\lambda;\varkappa)$ such that
\begin{itemize}
	\item $J(\lambda)$ is analytic for $\lambda\in\mathbb{C}\backslash(-\hat{r},\hat{r})$
	\item The limiting values of $J(\lambda)$ are related via the jump condition
	\begin{equation*}
		J_+(\lambda)=J_-(\lambda)\begin{pmatrix}
		e^{-t(\varkappa-\Pi(\lambda))} & \im\epsilon\\
		\im\epsilon & 0
		\end{pmatrix},\ \ \lambda\in(-\hat{r},\hat{r})
	\end{equation*}
	where $\varkappa-\Pi(\lambda)\in C^1(-M,M)$ has been explicitly computed in Proposition \ref{gprop}. Also, $J_{\pm}(\lambda)$ are bounded on $(-\hat{r},\hat{r})$.
	\item As $x\rightarrow-\infty,|s_1|\uparrow 1$ subject to \eqref{sc:2},
	\begin{equation*}
		J(\lambda)=\big(I+o(1)\big)\Upsilon(\lambda)
	\end{equation*}
	uniformly for $\lambda\in\partial D(0,\hat{r})$ and $\Upsilon(\lambda)$ has been introduced in \eqref{nouter}.
\end{itemize}
\end{problem}
We choose a solution to this model problem in the following form,
\begin{equation}\label{ou:idea}
	J(\lambda)=B_u(\lambda)e^{\im\frac{\pi}{4}\epsilon\sigma_3}\begin{pmatrix}
	1 & \frac{1}{2\pi\im}\int_{-M}^{M}e^{-t(\varkappa-\Pi(\mu))}\frac{\d\mu}{\lambda-\mu}\\
	0 & 1\\
	\end{pmatrix}\begin{cases}
	\begin{pmatrix}
	0 & 1\\
	-1 & 0\\
	\end{pmatrix}e^{-\im\frac{\pi}{4}\epsilon\sigma_3},&\lambda\in D(0,\hat{r}):\ \Im\lambda>0\\
	e^{-\im\frac{\pi}{4}\epsilon\sigma_3},&\lambda\in D(0,\hat{r}):\ \Im\lambda<0,
	\end{cases}
\end{equation}
where $B_u(\lambda)$ is locally analytic,
\begin{equation*}
	B_u(\lambda)=\Upsilon(\lambda)e^{\im\frac{\pi}{4}\epsilon\sigma_3}\begin{cases}
		\begin{pmatrix}
		0 & -1\\
		1 & 0\\
		\end{pmatrix}e^{-\im\frac{\pi}{4}\epsilon\sigma_3},&\lambda\in D(0,\hat{r}):\,\Im\lambda>0\\
		e^{-\im\frac{\pi}{4}\epsilon\sigma_3},&\lambda\in D(0,\hat{r}):\,\Im\lambda<0.
		\end{cases}
\end{equation*}
Let us quickly verify that \eqref{ou:idea} indeed satisfies the properties of RHP \ref{nmodel}. First, analyticity of the multiplier $B_u(\lambda)$ follows from the jump of $\Upsilon(\lambda)$,
\begin{equation*}
	\Upsilon_+(\lambda)=\Upsilon_-(\lambda)e^{\im\frac{\pi}{4}\epsilon\sigma_3}\begin{pmatrix}
	0 & 1\\
	-1 & 0
	\end{pmatrix}e^{-\im\frac{\pi}{4}\epsilon\sigma_3},\ \ \ \lambda\in(-\hat{r},\hat{r})\subset\left(-\frac{1}{\sqrt{2}},\frac{1}{\sqrt{2}}\right).
\end{equation*}
Thus the jump behavior of $J(\lambda)$ is a consequence of the Plemelj-Sokhotskii formula,
\begin{equation*}
	e^{\im\frac{\pi}{4}\epsilon\sigma_3}\begin{pmatrix}
	1 & -e^{-t(\varkappa-\Pi(\lambda))}\\
	0 & 1\\
	\end{pmatrix}\begin{pmatrix}
	0 & 1\\
	-1 & 0
	\end{pmatrix}e^{-\im\frac{\pi}{4}\epsilon\sigma_3} = \begin{pmatrix}
		e^{-t(\varkappa-\Pi(\lambda))} & \im\epsilon\\
		\im\epsilon & 0
		\end{pmatrix},\ \ \lambda\in(-\hat{r},\hat{r}).
\end{equation*}
Secondly, as $x\rightarrow-\infty,|s_1|\uparrow 1$ with \eqref{sc:2}, 
\begin{align*}
	\int_{-M}^{M}\frac{e^{-t(\varkappa-\Pi(\mu))}}{\lambda-\mu}\frac{\d\mu}{2\pi\im}&=\frac{\im}{2\pi}\ln\left(\frac{\lambda-m}{\lambda+m}\right)+\int_{m}^{M}\frac{e^{-t(\varkappa-\Pi(\mu))}}{\lambda-\mu}\frac{\d\mu}{2\pi\im}+\int_{-M}^{-m}\frac{e^{-t(\varkappa-\Pi(\mu))}}{\lambda-\mu}\frac{\d\mu}{2\pi\im}\\
	&=\frac{\im}{2\pi}\ln\left(\frac{\lambda-m}{\lambda+m}\right)+\mathcal{O}\left(\frac{\ln^{\frac{1}{6}}t}{t^{\frac{1}{2}}}\right)=\mathcal{O}\left(\frac{1}{t^{\frac{1}{6}}\ln^{\frac{1}{2}}t}\right)+
	\mathcal{O}\left(\frac{\ln^{\frac{1}{6}}t}{t^{\frac{1}{2}}}\right)
\end{align*}
uniformly for $\lambda\in\partial D(0,\hat{r})$. Here we used a Laplace-type argument for the two remaining integrals based on the behavior
\begin{eqnarray*}
	\varkappa-\Pi(\lambda)&\sim&\frac{16}{3}\sqrt{2m(M^2-m^2)}\,(\lambda-m)^{\frac{3}{2}},\ \ \lambda\downarrow m\\
	\varkappa-\Pi(\lambda)&\sim&\frac{16}{3}\sqrt{2m(M^2-m^2)}\,(-\lambda-m)^{\frac{3}{2}},\ \ \lambda\uparrow -m
\end{eqnarray*}
as well as the expansions
\begin{equation*}
	m=\sqrt{\frac{\sigma}{\sqrt{2}|\ln\sigma|}}\left(1+\frac{\ln|\ln\sigma|}{2\ln\sigma}+\mathcal{O}\left(\frac{1}{\ln\sigma}\right)\right),\ \ \ M=\frac{1}{\sqrt{2}}\left(1-\frac{\sigma}{\sqrt{2}|\ln\sigma|}+\mathcal{O}\left(\frac{\sigma\ln|\ln\sigma|}{\ln^2\sigma}\right)\right);\ \ \sigma\downarrow 0.
\end{equation*}
Thus
\begin{equation}\label{nmodelmat}
	J(\lambda)=\left(I+\mathcal{O}\left(\frac{1}{t^{\frac{1}{6}}\ln^{\frac{1}{2}}t}\right)\right)\Upsilon(\lambda)
\end{equation}
as $t\rightarrow+\infty,|s_1|\uparrow 1$ subject to \eqref{sc:2} uniformly for $\lambda\in\partial D(0,\hat{r})$ with $\hat{r}$ as in \eqref{ridea}.\smallskip

Instead of the ratio transformation \eqref{ratio:1}, introduce for \eqref{upperscale},
\begin{equation}\label{ratio:3}
	R_u(\lambda)=S(\lambda)\begin{cases}
		\big(J(\lambda)\big)^{-1},&|\lambda|<\hat{r}\\
		\big(P(\lambda)\big)^{-1},&|\lambda-M|<\bar{r}\\
		\big(Q(\lambda)\big)^{-1},&|\lambda+M|<\bar{r}\\
		\big(N(\lambda)\big)^{-1},&|\lambda|>\hat{r},|\lambda\mp M|>\bar{r}
		\end{cases}
\end{equation}
where $0<\hat{r}=m+dt^{-\frac{1}{3}}<\frac{1}{3}$ and $0<\bar{r}<\frac{1}{4}$ remains fixed. The model functions $N(\lambda),P(\lambda)$ and $Q(\lambda)$ are (as before) given in \eqref{ou:3}, \eqref{parao} and \eqref{paraom}.
\begin{figure}[tbh]
\begin{center}
\resizebox{0.7\textwidth}{!}{\includegraphics{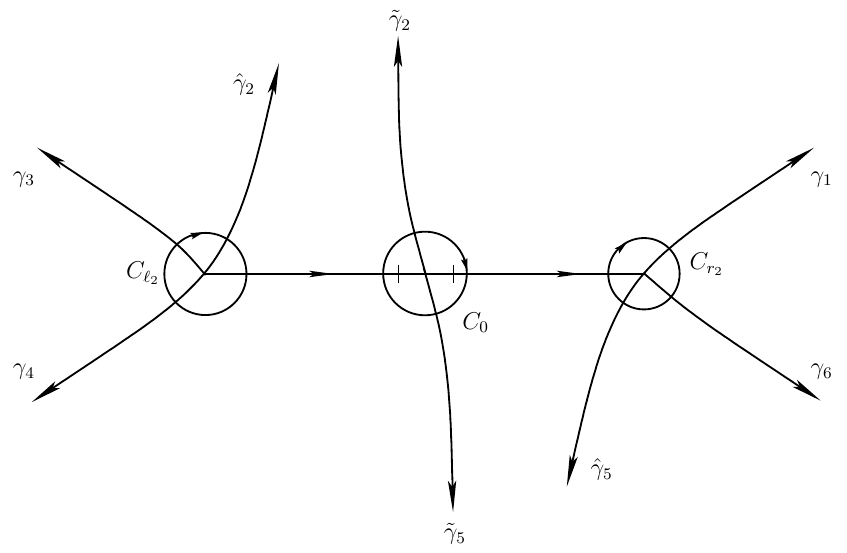}}
\caption{The jump contour $\Sigma_{R_u}$ of the ratio function $R_u(\lambda)$ defined in \eqref{ratio:3}. The two vertical lines represent the branch points $\lambda=\pm m$.}
\label{figureratio3}
\end{center}
\end{figure}

The RHP for $R_u(\lambda)$ is formulated on the contour shown in Figure \ref{figureratio3}. Compared to the ratio problem for $R(\lambda)$ the only difference in the jump behavior occurs inside the circle $C_0$ centered at the origin,
\begin{equation*}
	\big(R_u(\lambda)\big)_+=\big(R_u(\lambda)\big)_-\left[I+J_-(\lambda)\begin{pmatrix}
	\im\epsilon\,s_3e^{\im t\Omega(\lambda)}-1 & e^{-t(\varkappa-\Pi(\lambda))}(-\im\epsilon-s_3e^{\im t\Omega(\lambda)})\\
	-\im\epsilon\,e^{-t(\varkappa+\Pi(\lambda))} & -\im\epsilon\,s_1(1+e^{-\varkappa t})e^{-\im t\Omega(\lambda)} -1+ e^{-2t\varkappa}\\
	\end{pmatrix}\big(J_-(\lambda)\big)^{-1}\right],
\end{equation*}
for $\lambda\in(-\hat{r},\hat{r})$ as well as
\begin{align*}
	\big(R_u(\lambda)\big)_+&=\big(R_u(\lambda)\big)_-J(\lambda)\begin{pmatrix}
	1 & (s_1+s_2)e^{-\varkappa t}e^{\im t\Omega(\lambda)} \\
	0 & 1\\
	\end{pmatrix}\big(J(\lambda)\big)^{-1},\ \ \lambda\in\tilde{\gamma}_2\cap D(0,\hat{r})\\
	\big(R_u(\lambda)\big)_+&=\big(R_u(\lambda)\big)_-J(\lambda)\begin{pmatrix}
	1 & 0\\
	-(s_1+s_2)e^{-\varkappa t}e^{-\im t\Omega(\lambda)} & 1\\
	\end{pmatrix}\big(J(\lambda)\big)^{-1},\ \ \lambda\in\tilde{\gamma}_5\cap D(0,\hat{r}),
\end{align*}
and on the circle boundary $\partial D(0,\hat{r})$,
\begin{equation*}
	\big(R_u(\lambda)\big)_+=\big(R_u(\lambda)\big)_-J(\lambda)\big(N(\lambda)\big)^{-1},\ \ |\lambda|=\hat{r}.
\end{equation*}
Notice that (since $\hat{r}=\mathcal{O}(t^{-\frac{1}{3}})\rightarrow 0$)
\begin{equation*}
	t\Omega(\lambda)=\mathcal{O}\left(\frac{\sigma t}{|\ln\sigma|}\right)=\mathcal{O}\left(\frac{1}{\ln t}\right),\ \ \ \lambda\in[-\hat{r},\hat{r}]
\end{equation*}
subject to \eqref{upperscale} as well as
\begin{equation*}
	\lambda\in(m+dt^{-\frac{1}{3}},M-\bar{r}):\ \ t(\varkappa-\Pi(\lambda))\geq 8t\int_m^{m+dt^{-\frac{1}{3}}}\sqrt{(M^2-\mu^2)(\mu^2-m^2)}\,\d\mu\geq c\sqrt{mt}\rightarrow+\infty.
\end{equation*}
The latter combined with \eqref{outparamatch},\eqref{ridea},\eqref{nmodelmat}, using also the boundedness of $N_-(\lambda)$ away from the three disks $D_j$ centered at $\lambda=0$ and $\lambda=\pm M$  and recalling at the same time our estimations for the jumps given in Section \ref{ratiosec1}, we obtain
\begin{equation}\label{upperesti:1}
	\|G_{R_u}(\cdot;t,|s_1|)-I\|_{L^{\infty}(\Sigma_{R_u})}\leq \frac{d}{\ln t},\ \ t\rightarrow+\infty,\ \ |s_1|\uparrow 1:\ \ \frac{2}{3}\sqrt{2}-\frac{f_2}{t}\leq\varkappa<\frac{2}{3}\sqrt{2}.
\end{equation}
Opposed to the this $L^{\infty}$-estimation we can obtain better $L^2$-estimations through
\begin{equation*}
	 \int_{-\hat{r}}^{\hat{r}}\left|1-e^{\im t\Omega(\lambda)}\right|^2\,\d\lambda=\frac{2d}{t^{\frac{1}{3}}}\big(1-\cos(2\pi tV)\big)+2\int_{-m}^m\big(1-\cos(t\Omega(\lambda))\big)\,\d\lambda=\mathcal{O}\left(\frac{1}{t^{\frac{1}{3}}\ln^2t}\right)
\end{equation*}
and
\begin{equation*}
	\oint_{\partial D(0,\hat{r})}\left|u(\lambda)\right|^2|\d\lambda|\leq \frac{c\hat{r}}{\ln^2 t}=\mathcal{O}\left(\frac{1}{t^{\frac{1}{3}}\ln^2t}\right).
\end{equation*}
The last estimation is required for the $L^2$-estimation of the jump on $\partial D(0,\hat{r})$:
\begin{align*}
	J(\lambda)N^{-1}(\lambda)=&\,J(\lambda)\big(\Upsilon(\lambda)\big)^{-1}\Upsilon(\lambda)N^{-1}(\lambda)\stackrel{\eqref{nmodelmat}}{=}\left(I+\mathcal{O}\left(\frac{1}{t^{\frac{1}{6}}\ln^{\frac{1}{2}}t}\right)\right)\Upsilon(\lambda)\big(N(\lambda)\big)^{-1}\\
	\stackrel{\eqref{outparamatch}}{=}&\left(I+\mathcal{O}\left(\frac{1}{t^{\frac{1}{6}}\ln^{\frac{1}{2}}t}\right)\right)\Upsilon(\lambda)e^{-2\pi\im\tau'tVu(\lambda)\sigma_3}\big(\Upsilon(\lambda)\big)^{-1}e^{2\pi\im\tau'tVu(\infty)\sigma_3}\left(I+\mathcal{O}\left(\frac{1}{t^{\frac{1}{2}}\ln^{\frac{1}{2}}t}\right)\right)\\
	=&\left(I-2\pi\im\tau'tVu(\lambda)\Upsilon(\lambda)\sigma_3\big(\Upsilon(\lambda)\big)^{-1}+\ldots\right)e^{2\pi\im\tau'tVu(\infty)\sigma_3}\left(I+\mathcal{O}\left(\frac{1}{t^{\frac{1}{6}}\ln^{\frac{1}{2}}t}\right)\right).
\end{align*}
We summarize,
\begin{equation}\label{upperesti:2}
	\|G_{R_u}(\cdot;t,|s_1|)-I\|_{L^{2}(\Sigma_{R_u})}\leq \frac{d}{t^{\frac{1}{6}}\ln^{\frac{1}{2}} t},\ \ t\rightarrow+\infty,\ \ |s_1|\uparrow 1:\ \ \frac{2}{3}\sqrt{2}-\frac{f_1}{t}\leq\varkappa<\frac{2}{3}\sqrt{2},\ \ \ f_1>0
\end{equation}
and combining \eqref{upperesti:1}, \eqref{upperesti:2} this implies from general theory \cite{DZ1} the unique solvability of the $R_u$-RHP in $L^2(\Sigma_{R_u})$, moreover its unique solution satisfies
\begin{equation*}
	\|(R_u)_-(\cdot;t,|s_1|)-I\|_{L^2(\Sigma_{R_u})}\leq \frac{d}{t^{\frac{1}{6}}\ln^{\frac{1}{2}} t}.
\end{equation*}
In the same way as we derived \eqref{f1}, we obtain
\begin{equation*}
	u(x|s)=\sqrt{-x}\,e^{2t\ell}\left[-\epsilon(M-m)\frac{\theta_3(0|\tau)}{\theta_2(0|\tau)}\frac{\theta_2(tV|\tau)}{\theta_3(tV|\tau)}e^{\im\pi tV}+\mathcal{E}_u(\varkappa)\right]
\end{equation*}
with
\begin{equation*}
	\mathcal{E}_u(\varkappa)=\frac{\im}{\pi}\int_{\Sigma_{R_u}}\!\!\!\!\big(G_{R_u}(w)-I\big)_{12}\,\d w+\mathcal{O}\left(\frac{1}{t^{\frac{1}{3}}\ln t}\right)=\frac{\im}{\pi}\left[\oint_{\partial D(0,\hat{r})}\!\!\!+\int_{-\hat{r}}^{\hat{r}}\right]\big(G_{R_u}(w)-I\big)_{12}\,\d w+\mathcal{O}\left(\frac{1}{t^{\frac{1}{3}}\ln t}\right).
\end{equation*}
However, by direct computation,
\begin{equation*}
	\int_{-\hat{r}}^{\hat{r}}\big(G_{R_u}(w)-I\big)_{12}\,\d w = \mathcal{O}\left(\frac{1}{t^{\frac{1}{3}}\ln t}\right)=\oint_{\partial D(0,\hat{r})}\big(G_{R_u}(w)-I\big)_{12}\,\d w
\end{equation*}
and thus 
\begin{equation}\label{final3}
	u(x|s)=-\epsilon\sqrt{-\frac{x}{2}}\,\frac{1-\k}{\sqrt{1+\k^2}}\,\textnormal{cd}\left(2(-x)^{\frac{3}{2}}V\,K\left(\frac{1-\k}{1+\k}\right),\,\frac{1-\k}{1+\k}\right)+J_1(x,s),
\end{equation}
with
\begin{prop}\label{imp:3} For any given $f_1>0$ there exist positive constants $t_0=t_0(f_1),v_0=v_0(f_1)$ and $c=c(f_1)$ such that
\begin{equation*}
	\big|J_1(x,s)\big|\leq\frac{c}{\ln t},\ \ \ \forall\,t\geq t_0,\ v\geq v_0:\ \ \frac{2}{3}\sqrt{2}\,t-f_1\leq v<\frac{2}{3}\sqrt{2}\,t.
\end{equation*}
\end{prop}
Together with Propositions \ref{imp:1} and \ref{imp:2}, Proposition \ref{imp:3} completes the proof of Theorem \ref{bet:1}. In order to derive \eqref{cor:hm} in Corollary \ref{mat}, we choose $t\geq t_0,v\geq v_0$ such that $0<\sigma\leq\frac{1}{t}$. Since
\begin{equation}\label{modtrafo}
	\frac{\theta_3(0|\tau)}{\theta_3(tV|\tau)}\frac{\theta_2(tV|\tau)}{\theta_2(0|\tau)} = \frac{\theta_3(0|\tau')}{\theta_3(tV\tau'|\tau')}\frac{\theta_4(tV\tau'|\tau')}{\theta_4(0|\tau')}
\end{equation}
we get with \eqref{l:2} at once
\begin{equation*}
	\frac{\theta_3(0|\tau')}{\theta_4(0|\tau')} = 1+\mathcal{O}\left(\sqrt{\frac{\sigma}{|\ln\sigma|}}\right)=1+\mathcal{O}\left(\frac{1}{\sqrt{t}\ln^{\frac{1}{2}}t}\right).
\end{equation*}
Similarly,
\begin{equation*}
	\theta_3(tV\tau'|\tau')=1+\mathcal{O}\left(\frac{1}{\sqrt{t}\ln^{\frac{1}{2}}t}\right)=\theta_4(tV\tau'|\tau').
\end{equation*}
Moving ahead (recall \eqref{I:2}),  
\begin{equation}\label{stok:1}
	-\epsilon\sqrt{-\frac{x}{2}}\,\frac{1-\k}{\sqrt{1+\k^2}}\,\frac{\theta_3(0|\tau)}{\theta_3(tV|\tau)}\frac{\theta_2(tV|\tau)}{\theta_2(0|\tau)} =-\epsilon\sqrt{-\frac{x}{2}}+\mathcal{O}\left(\frac{1}{(-x)^{\frac{1}{4}}\ln^{\frac{1}{2}}|x|}\right)
\end{equation}
uniformly as $x\rightarrow-\infty,|s_1|\uparrow 1$ such that $0<\sigma\leq\frac{1}{t}$. In view of Proposition \ref{imp:3}, the latter estimation completes the proof of Corollary \ref{mat}.
\begin{rem}\label{stokefirst} The introduction of $J(\lambda)$ in \eqref{ou:idea} together with bounds on $N_-(\lambda)$ in \eqref{modup},\eqref{outparamatch} allowed us to control the error $J_1(x,s)$ in the region \eqref{sc:2}. The analysis of $J_1(x,s)$ inside the region
\begin{equation*}
	\left\{\big(t,v\big):\ \ t\geq t_0,\ v\geq v_0:\ \ \frac{f_1}{t}<\sigma\leq\delta,\ \ \ f_1>0,\ 0<\delta<\frac{1}{3}\sqrt{2}\right\}
\end{equation*}
which is not covered by Propositions \ref{imp:1}, \ref{imp:2} and \ref{imp:3} is much more difficult. In this case \eqref{sc:1} is violated in general and thus $N_-(\lambda)$ may become unbounded. For instance, if we were to choose 
\begin{equation*}
	t\geq t_0,\ \ v\geq v_0:\ \ \frac{1}{t}\leq\sigma\leq\chi\frac{\ln t}{t}<\delta,\ \ \  \chi>0,
\end{equation*}
then from \eqref{I:2} and \eqref{l:2}, using again \eqref{modtrafo},
\begin{equation*}
	-\epsilon\sqrt{-\frac{x}{2}}\,\frac{1-\k}{\sqrt{1+\k^2}}\,\frac{\theta_3(0|\tau)}{\theta_3(tV|\tau)}\frac{\theta_2(tV|\tau)}{\theta_2(0|\tau)}=-\epsilon\sqrt{-\frac{x}{2}}\left(1+\mathcal{O}\left(t^{\chi-\frac{1}{2}}\right)\right) = -\epsilon\sqrt{-\frac{x}{2}}+\mathcal{O}\left((-x)^{\frac{3}{2}(\chi-\frac{1}{6})}\right).
\end{equation*}
Hence we only restore the leading Hastings-McLeod asymptotics for $\chi<\frac{1}{6}$. For the values of $\chi\geq\frac{1}{6}$ more terms will contribute to the leading order, or equivalently, the leading asymptotics changes once the ``Stokes line" 
\begin{equation*}
	v=\frac{2}{3}\sqrt{2}\,t-\frac{1}{6}\ln t
\end{equation*}
is crossed. We provide further detail on this interesting feature in Section \ref{stokes} below.
\end{rem}	

\section{Regular transition analysis near and above the separating line $\varkappa=\frac{2}{3}\sqrt{2}$}\label{near}

In this section we assume that both $t=(-x)^{\frac{3}{2}}$ and $v=-\ln(1-|s_1|^2)$ are sufficiently large such that
\begin{equation}\label{upperscale}
	\varkappa\geq \frac{2}{3}\sqrt{2}-\frac{f_2}{t},\ \ f_2\in\mathbb{R}.
\end{equation}
Again, we use the transformation sequence
\begin{equation*}
	Y(\lambda)\mapsto X(\lambda)\mapsto T(\lambda)
\end{equation*}
but subsequently the nonlinear steepest descent analysis of the $T$-RHP in \eqref{Tfunc} is very different from the one used in Sections \ref{sec:reg} and \ref{flower}. In fact, the steps carried out below are a modification of the ones used in the analysis for $|s_1|=1$, see \cite{FIKN}, chapter $11$, $\S 4$. The details are as follows.
\subsection{g-function transformation} We fix 
\begin{equation*}
	\lambda^{\ast}=\frac{1}{\sqrt{2}}
\end{equation*}
in \eqref{Tfunc} and define for $\lambda\in\mathbb{C}\backslash\Sigma_T$ with $\Sigma_T=[-\lambda^{\ast},\lambda^{\ast}]\cup\hat{\gamma}_2\cup\hat{\gamma}_5\cup\tilde{\gamma}_2\cup\tilde{\gamma}_5\bigcup\gamma_k$,
\begin{equation}\label{newg}
	\varpi(\lambda)=T(\lambda)e^{t(g(\lambda)-\vartheta(\lambda))\sigma_3};\hspace{0.5cm}\textnormal{with}\ \ \ \ 
	g(\lambda)=\frac{4\im}{3}\left(\lambda^2-\frac{1}{2}\right)^{\frac{3}{2}}=\vartheta(\lambda)+\mathcal{O}\left(\lambda^{-1}\right),\ \ \lambda\rightarrow\infty.
\end{equation}
This transforms the $T$-RHP to a RHP with jumps on $[-\lambda^{\ast},0]\cup[0,\lambda^{\ast}]$ such that
\begin{eqnarray}
  \varpi_+(\lambda) &=&\varpi_-(\lambda)\begin{pmatrix}
                               (1-s_1s_3)e^{t\Pi(\lambda)} & s_1+s_1(1-s_1s_3)\\
-s_3 & (1-s_1s_3)e^{-t\Pi(\lambda)}\\
                              \end{pmatrix},\hspace{0.5cm}\lambda\in[-\lambda^{\ast},0]\label{s:1}\\
\varpi_+(\lambda) &=&\varpi_-(\lambda)\begin{pmatrix}
                             (1-s_1s_3)e^{t\Pi(\lambda)} & -s_3\\
s_1+s_1(1-s_1s_3) & (1-s_1s_3)e^{-t\Pi(\lambda)}\\
                            \end{pmatrix},\hspace{0.5cm}\lambda\in[0,\lambda^{\ast}]\label{s:2}
\end{eqnarray}
where
\begin{equation*}
  \Pi(\lambda) = g_+(\lambda)-g_-(\lambda) = \frac{8}{3}\left(\frac{1}{2}-\lambda^2\right)^{\frac{3}{2}}>0;\hspace{0.5cm}\lambda\in\left(-\frac{1}{\sqrt{2}},\frac{1}{\sqrt{2}}\right).
\end{equation*}
For the $(11)$-entries in \eqref{s:1} and \eqref{s:2} we use \eqref{dbpara} and observe that
\begin{equation*}
	\varkappa-\Pi(\lambda)>0\ \ \Leftrightarrow\ \ |\lambda|>\lambda_0\equiv\frac{1}{\sqrt{2}}\left(1-\left(\frac{3\varkappa}{2\sqrt{2}}\right)^{\frac{2}{3}}\right)^{\frac{1}{2}}.
\end{equation*}
Hence the strict inequality for $\varkappa-\Pi(\lambda)$ will fail on the (shrinking) interval $[-\lambda_0,\lambda_0]$. We thus face the necessity to consider an additional model problem at the origin. Besides this new feature, we however still expect (and will prove) that the outer model problem is given by a problem in which we neglect the diagonal entries in \eqref{s:1},\eqref{s:2} and apply estimations \eqref{approx:3} and \eqref{approx:4} for the off-diagonal, i.e. on the full line segment $[-\lambda^{\ast},\lambda^{\ast}]$ we require a function $\Upsilon(\lambda)$ such that
\begin{equation}\label{idea}
	\Upsilon_+(\lambda)=\Upsilon_-(\lambda)\begin{pmatrix}
	0&\im\epsilon\\
	\im\epsilon & 0\\
	\end{pmatrix},\ \ \lambda\in\left(-\frac{1}{\sqrt{2}},\frac{1}{\sqrt{2}}\right).
\end{equation}
Furthermore, along the remaining infinite branches, i.e. along
\begin{equation*}
	\Sigma_{\varpi_{\infty}}=\hat{\gamma}_2\cup\hat{\gamma}_5\cup\tilde{\gamma}_2\cup\tilde{\gamma}_5\cup\bigcup\gamma_k,
\end{equation*}
the jumps in the $S$-RHP are once more of the form
\begin{equation*}
	\varpi_+(\lambda)=\varpi_-(\lambda)e^{-tg(\lambda)\sigma_3}G_T(\lambda)e^{tg(\lambda)\sigma_3},\ \ \lambda\in\Sigma_{\varpi_{\infty}}.
\end{equation*}
Hence by triangularity and the sign of $\Re(g(\lambda))$ these jumps are asymptotically close to the unit matrix away from the endpoints $\lambda=\pm\lambda^{\ast}$. We proceed now with the analysis of the relevant model problems.

\subsection{The outer parametrix}\label{outer:2} The outer problem consists in finding a $2\times 2$ piecewise analytic function $\Upsilon(\lambda)=\Upsilon(\lambda;\epsilon)$ such that
\begin{itemize}
	\item $\Upsilon(\lambda)$ is analytic for $\lambda\in\mathbb{C}\backslash[-\frac{1}{\sqrt{2}},\frac{1}{\sqrt{2}}]$
	\item If the cut is oriented from $-\frac{1}{\sqrt{2}}$ to $\frac{1}{\sqrt{2}}$, the boundary values of $\Upsilon(\lambda)$ are related via
	\begin{equation}\label{nout:1}
		\Upsilon_+(\lambda)=\Upsilon_-(\lambda)\begin{pmatrix}
		0 & \im\epsilon\\
		\im\epsilon & 0\\
		\end{pmatrix},\ \ \lambda\in\left(-\frac{1}{\sqrt{2}},\frac{1}{\sqrt{2}}\right)
	\end{equation}
	\item $\Upsilon(\lambda)$ is square integrable on the closed interval $[-\frac{1}{\sqrt{2}},\frac{1}{\sqrt{2}}]$
	\item As $\lambda\rightarrow\infty$,
	\begin{equation*}
		\Upsilon(\lambda)=I+\mathcal{O}\left(\lambda^{-1}\right)
	\end{equation*}
\end{itemize}
A solution to this {\it quasi-permutation} RHP appeared already in Proposition \ref{ou:match}, formula \eqref{nouter}. Note also that
\begin{equation}\label{nout:2}
	\Upsilon(\lambda)=I-\frac{\sqrt{2}\epsilon}{4\lambda}\sigma_1+\mathcal{O}\left(\lambda^{-2}\right),\ \ \lambda\rightarrow\infty.
\end{equation}
\subsection{Origin parametrix} Our construction is motived by the local expansion
\begin{equation}\label{oriidea}
	\varkappa-\Pi(\lambda)=\varkappa-\frac{2}{3}\sqrt{2}+2\sqrt{2}\lambda^2-\frac{\lambda^4}{2\sqrt{2}}+\mathcal{O}\left(\lambda^6\right),\ \ \lambda\rightarrow 0,
\end{equation}
and solves the following problem.
\begin{problem} Determine a $2\times 2$ piecewise analytic function $F(\lambda)$ such that
\begin{itemize}
	\item $F(\lambda)$ is analytic inside the disk $D(0,r)\backslash(-r,r)$ for some fixed $0<r<\frac{1}{\sqrt{2}}$
	\item Orienting the segment $(-r,r)$ from left to right, we have the jump condition
	\begin{equation}\label{orij}
		F_+(\lambda)=F_-(\lambda)\begin{pmatrix}
		e^{-t(\varkappa-\Pi(\lambda))} & \im\epsilon\\
		\im\epsilon & 0\\
		\end{pmatrix},\ \ \lambda\in\left(-r,r\right);\hspace{0.5cm}\Pi(\lambda)=\frac{8}{3}\left(\frac{1}{2}-\lambda^2\right)^{\frac{3}{2}}
	\end{equation}
	\item As $x\rightarrow-\infty,|s_1|\uparrow 1$ subject to \eqref{upperscale} we have
	\begin{equation*}
		F(\lambda)=\left(I+o(1)\right)\Upsilon(\lambda)
	\end{equation*}
	uniformly for $|\lambda|=r>0$. Here, $\Upsilon(\lambda)$ is the outer model function introduced in Section \ref{outer:2}.
\end{itemize}
\end{problem}
The difference between this problem and RHP \ref{nmodel} is only marginal: Recall that $\sigma=\frac{2}{3}\sqrt{2}-\varkappa$ and set
\begin{equation}\label{FRHmodel}
	F^{RH}(\z)=e^{\im\frac{\pi}{4}\epsilon\sigma_3}\begin{pmatrix}
	1 & \frac{e^{\sigma t}}{2\pi\im}\int_{\mathbb{R}}\frac{e^{-\mu^2}}{\z-\mu}\d\mu\\
	0 & 1\\
	\end{pmatrix}\begin{cases}
	\begin{pmatrix}
	0 & 1\\
	-1 & 0\\
	\end{pmatrix}e^{-\im\frac{\pi}{4}\epsilon\sigma_3},&\Im\,\z>0\\
	e^{-\im\frac{\pi}{4}\epsilon\sigma_3},&\Im\,\z<0.
	\end{cases}
\end{equation}
We require the following Lemma.
\begin{lem}\label{princip} As $\z\rightarrow\infty$, we have 
\begin{equation*}
	I(\z)=\frac{1}{2\pi\im}\int_{\mathbb{R}}\frac{e^{-\mu^2}}{\mu-\z}\,\d\mu=\begin{cases}
	-\frac{\sqrt{\pi}}{2\pi\im\z}\left(1+\mathcal{O}\left(\z^{-1}\right)\right),&\z\notin\mathbb{R}\\
	-\frac{\sqrt{\pi}}{2\pi\im\z}\left(1+\mathcal{O}\left(\z^{-1}\right)\right)+e^{-\z^2}\begin{cases}
	1,&\Im\z>0\\
	0,&\Im\z<0
	\end{cases},&\z\in\mathbb{R}
	\end{cases}
\end{equation*}

\end{lem}
\begin{proof} Assume first that $\z\notin\mathbb{R}$, hence
\begin{equation*}
	I(\z)=-\frac{1}{2\pi\im\z}\int_{\mathbb{R}}\frac{e^{-\mu^2}}{1-\frac{\mu}{\z}}\,\d\z=-\frac{1}{2\pi\im\z}\left[\int_{\mathbb{R}}e^{-\mu^2}\,\d\mu+\mathcal{O}\left(\z^{-1}\right)\right]=-\frac{\sqrt{\pi}}{2\pi\im\z}\left(1+\mathcal{O}\left(\z^{-1}\right)\right),\ \ \z\rightarrow\infty.
\end{equation*}		
Secondly for $\z\in\mathbb{R}$ we apply Plemelj formula
\begin{equation*}
	I_{\pm}(\z)=\pm\frac{1}{2}e^{-\z^2}+\frac{1}{2\pi\im}\,\textnormal{v.p.}\int_{\mathbb{R}}\frac{e^{-\mu^2}}{\mu-\z}\d\mu.
\end{equation*}
The principal value integral can be simplified as follows: let $\gamma_{R,\epsilon}$ denote the contour shown in Figure \ref{box} with $0<\epsilon<R$.
\begin{figure}[tbh]
\begin{center}
\resizebox{0.6\textwidth}{!}{\includegraphics{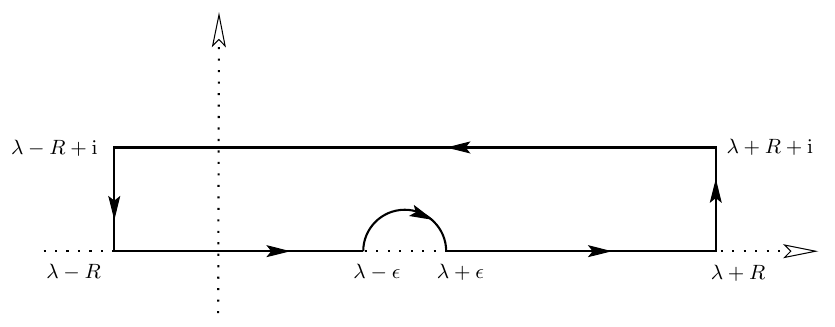}}
\caption{Integration contour $\gamma_{R,\epsilon}$ for a possible choice $\lambda>0$.}
\label{box}
\end{center}
\end{figure}

From Cauchy's theorem
\begin{equation*}
	\forall\, R>\epsilon>0:\ \ \oint_{\gamma_{R,\epsilon}}\frac{e^{-\mu^2}}{\mu-\z}\d\mu=0.
\end{equation*}
On the other hand
\begin{equation*}
	\int_{C_\epsilon}\frac{e^{-\mu^2}}{\mu-\z}\d\mu=-\im\int_0^{\pi}e^{-(\z+\epsilon e^{\im\theta})^2}\d\theta\stackrel{\epsilon\downarrow 0}{\longrightarrow} -\im\pi e^{-\z^2},
\end{equation*}
as well as
\begin{equation*}
	\left(\int_{\z+R}^{\z+R+\im}+\int_{\z-R+\im}^{\z-R}\right)\frac{e^{-\mu^2}}{\mu-\z}\d\mu=\im\int_0^1e^{-(\z+R+\im t)^2}\frac{\d t}{R+\im t}+\im\int_0^1e^{-(\z-R+\im t)^2}\frac{\d t}{R-\im t}
	\stackrel{R\rightarrow\infty}{\longrightarrow} 0,
\end{equation*}
and
\begin{equation*}
	\lim_{\substack{R\rightarrow\infty \\ \epsilon\downarrow 0}}\left(\int_{\z-R}^{\z-\epsilon}+\int_{\z+\epsilon}^{\z+R}\right)\frac{e^{-\mu^2}}{\mu-\z}\d\mu=\textnormal{v.p.}\,\int_{\mathbb{R}}\frac{e^{-\mu^2}}{\mu-\z}\d\mu.
\end{equation*}
Thus
\begin{equation*}
	I_{\pm}(\z)=\frac{1}{2\pi\im}\int_{\mathbb{R}+\im}\frac{e^{-\mu^2}}{\mu-\z}\d\mu+e^{-\z^2}\begin{cases}
	1,&\Im\z>0\\
	0,&\Im\z<0
	\end{cases}=-\frac{\sqrt{\pi}}{2\pi\im\z}\left(1+\mathcal{O}\left(\z^{-1}\right)\right)+e^{-\z^2}\begin{cases}
	1,&\Im\z>0\\
	0,&\Im\z<0
	\end{cases}
\end{equation*}
which completes the proof.		
\end{proof}
The solution to the model problem is now given by
\begin{equation*}
	F(\lambda)=B_0(\lambda)F^{RH}\big(\z(\lambda)\big),\ |\lambda|<r;\hspace{0.15cm}\z(\lambda)=\sqrt{t}\left(\sigma+\varkappa-\Pi(\lambda)\right)^{\frac{1}{2}}=(2\sqrt{2}\,t)^{\frac{1}{2}}\lambda\left(1-\frac{\lambda^2}{4}-\frac{11}{96}\lambda^4+\mathcal{O}\left(\lambda^6\right)\right)
\end{equation*}
where we introduce the (locally analytic) multiplier
\begin{equation*}
	B_0(\lambda)=\Upsilon(\lambda)e^{\im\frac{\pi}{4}\epsilon\sigma_3}\begin{cases}
		\begin{pmatrix}
		0 & -1\\
		1 & 0\\
		\end{pmatrix}e^{-\im\frac{\pi}{4}\epsilon\sigma_3},&\lambda\in D(0,r):\,\Im\lambda>0\\
		e^{-\im\frac{\pi}{4}\epsilon\sigma_3},&\lambda\in D(0,r):\,\Im\lambda<0
		\end{cases};\hspace{0.5cm}B_0(0)=\frac{1}{\sqrt{2}}\begin{pmatrix}
		1 & -\im\epsilon\\
		-\im\epsilon & 1\\
		\end{pmatrix}.
\end{equation*}
More detail on the local behavior of $B_0(\lambda)$ near $\lambda=0$ will be needed later on, first
\begin{eqnarray*}
	B_0(\lambda)&=&e^{\im\frac{\pi}{4}\epsilon\sigma_3}\bigg\{\frac{1}{\sqrt{2}}\begin{pmatrix}
	1 & -1\\
	1 & 1\\
	\end{pmatrix}+\frac{\im}{2}\begin{pmatrix}
	1 & 1\\
	-1 & 1\\
	\end{pmatrix}\lambda+\frac{\sqrt{2}}{8}\begin{pmatrix}
	1 & -1\\
	1 & 1\\
	\end{pmatrix}\lambda^2+\frac{3\im}{8}\begin{pmatrix}
	1 & 1\\
	-1 & 1\\
	\end{pmatrix}\lambda^3\\
	&&\ \ \ \ \ +\frac{11\sqrt{2}}{64}\begin{pmatrix}
	1 & -1\\
	1 & 1\\
	\end{pmatrix}\lambda^4+\frac{31\im}{64}\begin{pmatrix}
	1 & 1\\
	-1 & 1\\
	\end{pmatrix}\lambda^5+\mathcal{O}\left(\lambda^6\right)\bigg\}e^{-\im\frac{\pi}{4}\epsilon\sigma_3},\ \ \lambda\rightarrow 0,
\end{eqnarray*}
followed by
\begin{eqnarray*}
	\big(B_0(\lambda)\big)^{-1}&=&e^{\im\frac{\pi}{4}\epsilon\sigma_3}\bigg\{\frac{1}{\sqrt{2}}\begin{pmatrix}
	1 & 1\\
	-1 & 1\\
	\end{pmatrix}+\frac{\im}{2}\begin{pmatrix}
	1 & -1\\
	1 & 1\\
	\end{pmatrix}\lambda+\frac{\sqrt{2}}{8}\begin{pmatrix}
	1 & 1\\
	-1 & 1\\
	\end{pmatrix}\lambda^2+\frac{3\im}{8}\begin{pmatrix}
	1 & -1\\
	1 & 1\\
	\end{pmatrix}\lambda^3\\
	&&+\frac{11\sqrt{2}}{64}\begin{pmatrix}
	1 & 1\\
	-1 & 1\\
	\end{pmatrix}\lambda^4+\frac{31\im}{64}\begin{pmatrix}
	1 & -1\\
	1 & 1\\
	\end{pmatrix}\lambda^5+\mathcal{O}\left(\lambda^6\right)\bigg\}e^{-\im\frac{\pi}{4}\epsilon\sigma_3},\ \ \lambda\rightarrow 0,
\end{eqnarray*}
and thus
\begin{equation*}
	B_0(\lambda)\sigma_+\big(B_0(\lambda)\big)^{-1}=\frac{1}{2}\begin{pmatrix}
	\im\epsilon & 1\\
	1 & -\im\epsilon
	\end{pmatrix}-\frac{\sqrt{2}}{2}\sigma_2\lambda+\frac{\im}{2}\epsilon\sigma_3\lambda^2-\frac{\sqrt{2}}{2}\sigma_2\lambda^3+\frac{3\im}{4}\epsilon\sigma_3\lambda^4-\frac{3\sqrt{2}}{4}\sigma_2\lambda^5+\mathcal{O}\left(\lambda^6\right).
\end{equation*}
We directly verify the required jump condition \eqref{orij}, and as $t\rightarrow\infty$ (subject to \eqref{upperscale}),
\begin{align}
	F(\lambda)&=\left(I-\im\epsilon\, e^{\sigma t} I\big(\z(\lambda)\big)B_0(\lambda)\sigma_+\big(B_0(\lambda)\big)^{-1}\right)\Upsilon(\lambda)\nonumber\\
	&=\left(I+\frac{\epsilon}{2\pi}\frac{e^{\sigma t}}{\z(\lambda)}B_0(\lambda)\sigma_+B_0^{-1}(\lambda)+\mathcal{O}\left(t^{-1}\right)\right)\Upsilon(\lambda)
	=\left(I+\mathcal{O}\left(t^{-\frac{1}{2}}\right)\right)\Upsilon(\lambda),\ \ \ \ \sigma_+=\begin{pmatrix}
	0 & 1\\
	0 & 0
	\end{pmatrix}\label{orimatch}
\end{align}
uniformly for $0<r_1\leq|\lambda|\leq r_2<\frac{1}{\sqrt{2}}$ with $r_j>0$ fixed. 
\subsection{Edge point parametrices near $\lambda=\pm\lambda^{\ast}$} The construction is almost identical to the one presented in Section \ref{outpara}. More precisely for the right edge point we take
\begin{equation}\label{upedge1}
	\Delta^r(\lambda)=B_{r_3}(\lambda)\widehat{A}^{RH}\big(\z(\lambda)\big)e^{\frac{2}{3}\im\z^{\frac{3}{2}}(\lambda)\sigma_3}e^{-\im\frac{\pi}{4}\epsilon\sigma_3},\ \ |\lambda-\lambda^{\ast}|<r,
\end{equation}
where $\widehat{A}^{RH}(\z)$ was introduced in \eqref{bare:5} and we have
\begin{equation*}
	B_{r_3}(\lambda)=\Upsilon(\lambda)e^{\im\frac{\pi}{4}\epsilon\sigma_3}\begin{pmatrix}
	1 & -1\\
	-\im & -\im\\
	\end{pmatrix}\alpha(\lambda)^{\sigma_3}\left(\z(\lambda)\frac{\lambda+\frac{1}{\sqrt{2}}}{\lambda-\frac{1}{\sqrt{2}}}\right)^{\frac{1}{4}\sigma_3},\ \ |\lambda-\lambda^{\ast}|<r
\end{equation*}
with the locally conformal change of variables
\begin{equation*}
	\z(\lambda)=\left(\frac{3t}{2}e^{-\im\frac{\pi}{2}}g(\lambda)\right)^{\frac{2}{3}}\sim\sqrt{2}(2t)^{\frac{2}{3}}\left(\lambda-\frac{1}{\sqrt{2}}\right),\ \ |\lambda-\lambda^{\ast}|<r.
\end{equation*}
The reader easily checks that $\Delta^r(\lambda)$ defined in \eqref{upedge1} has the same jump behavior as shown in Figure \ref{figure10} with $M$ replaced by $\lambda^{\ast}$ and $g(\lambda)$ as in \eqref{newg}. Moreover
\begin{equation*}
	\Delta^r(\lambda)=\left[I+\Upsilon(\lambda)\left\{\frac{\im}{48\z^{\frac{3}{2}}}\begin{pmatrix}
	-1 & 6\epsilon\\
	-6\epsilon & 1\\
	\end{pmatrix}+\mathcal{O}\left(\z^{-\frac{6}{2}}\right)\right\}\big(\Upsilon(\lambda)\big)^{-1}\right]\Upsilon(\lambda),\ \ \z\rightarrow\infty
\end{equation*}
so that for $t\rightarrow+\infty,|s_1|\uparrow 1$ such that $\varkappa\geq\frac{2}{3}\sqrt{2}-\frac{f_2}{t}$,
\begin{equation}\label{upedgem:1}
	\Delta^r(\lambda)=\left(I+\mathcal{O}\left(t^{-1}\right)\right)\Upsilon(\lambda)
\end{equation}
uniformly for $0<r_1\leq|\lambda-\lambda^{\ast}|\leq r_2<\frac{1}{\sqrt{2}}$.\smallskip

Near the remaining edge point, we choose
\begin{equation}\label{upedge2}
	\Delta^{\ell}(\lambda)=B_{\ell_3}(\lambda)\bar{A}^{RH}\big(\z(\lambda)\big)\Big|_{\gamma=1}e^{\frac{2}{3}\z^{\frac{3}{2}}(\lambda)\sigma_3}e^{-\im\frac{\pi}{4}\epsilon\sigma_3},\ \ |\lambda+\lambda^{\ast}|<r
\end{equation}
with $\bar{A}(\z)$ as in \eqref{bare:7} and
\begin{equation*}
	B_{\ell_3}(\lambda) = \Upsilon(\lambda)e^{\im\frac{\pi}{4}\epsilon\sigma_3}\begin{pmatrix}
	1 & -1\\
	-\im & -\im\\
	\end{pmatrix}\alpha(\lambda)^{\sigma_3}\left(\z(\lambda)\frac{\lambda-\frac{1}{\sqrt{2}}}{\lambda+\frac{1}{\sqrt{2}}}\right)^{-\frac{1}{4}\sigma_3},\ \ |\lambda+\lambda^{\ast}|<r
\end{equation*}
where we use the change of variable
\begin{equation*}
	\z(\lambda)=\left(\frac{3t}{2}g(\lambda)\right)^{\frac{2}{3}}\sim\sqrt{2}(2t)^{\frac{2}{3}}\left(\lambda+\frac{1}{\sqrt{2}}\right),\ \ |\lambda+\lambda^{\ast}|<r.
\end{equation*}
The jump behavior of \eqref{upedge2} is depicted in Figure \ref{figure12} with $-M$ replaced by $-\lambda^{\ast}$ and the $g$-function as in \eqref{newg}. We note that
\begin{equation*}
	\Delta^{\ell}(\lambda)=\left[I+\Upsilon(\lambda)\left\{\frac{1}{48\z^{\frac{3}{2}}}\begin{pmatrix}
	1 & 6\epsilon\\
	-6\epsilon & -1\\
	\end{pmatrix}+\mathcal{O}\left(\z^{-\frac{6}{2}}\right)\right\}\big(\Upsilon(\lambda)\big)^{-1}\right]\Upsilon(\lambda),\ \ \z\rightarrow\infty,
\end{equation*}
hence, as $t\rightarrow+\infty,|s_1|\uparrow 1$ such that $\varkappa\geq\frac{2}{3}\sqrt{2}-\frac{f_2}{t}$,
\begin{equation}\label{upedgem:2}
	\Delta^{\ell}(\lambda)=\left(I+\mathcal{O}\left(t^{-1}\right)\right)\Upsilon(\lambda)
\end{equation}
uniformly for $0<r_1\leq|\lambda+\lambda^{\ast}|\leq r_2<\frac{1}{\sqrt{2}}$. This completes the parametrix construction and we can move ahead with the ratio transformation.
\subsection{Final transformation - ratio problem}\label{ratioeasy}
Define
\begin{equation}\label{ratio:2}
	\xi(\lambda)=\varpi(\lambda)\begin{cases}
		\big(\Delta^r(\lambda)\big)^{-1},& |\lambda-\frac{1}{\sqrt{2}}|<\bar{r}\\
		\big(\Delta^{\ell}(\lambda)\big)^{-1},&|\lambda+\frac{1}{\sqrt{2}}|<\bar{r}\\
		\big(F(\lambda)\big)^{-1},&|\lambda|<r\\
		\big(\Upsilon(\lambda)\big)^{-1},&|\lambda\pm\frac{1}{\sqrt{2}}|>\bar{r}
		\end{cases}
\end{equation}
where $0<\bar{r}<\frac{1}{\sqrt{2}}$ and $0<r<\frac{1}{4}$ remain fixed throughout.\,
This leads to the RHP shown in the Figure below.
\begin{figure}[tbh]
\begin{center}
\resizebox{0.5\textwidth}{!}{\includegraphics{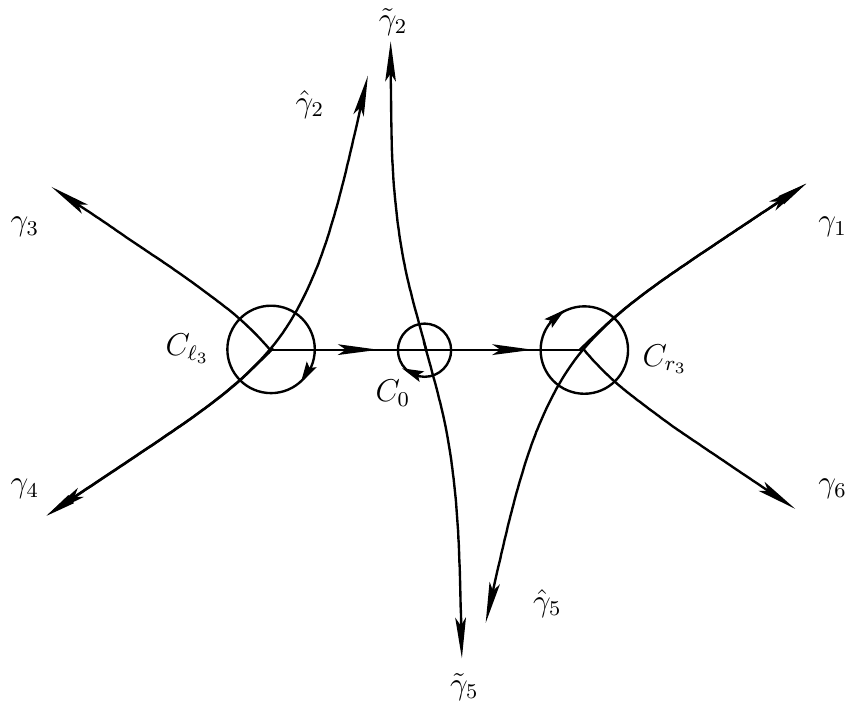}}
\caption{The jump contour $\Sigma_{\xi}$ of the ratio function $\xi(\lambda)$ defined in \eqref{ratio:2}.}
\label{ratio2}
\end{center}
\end{figure}
\begin{itemize}
	\item $\xi(\lambda)$ is analytic for $\lambda\in\mathbb{C}\backslash\Sigma_{\xi}$ where $\Sigma_{\xi}=[-\lambda^{\ast},\lambda^{\ast}]\cup C_0\cup C_{r_3}\cup C_{\ell_3}\cup\Sigma_{S_{\infty}}$
	\item The jumps are described as follows: Inside the two circles $C_{r_3}$ resp. $C_{\ell_3}$ we simply copy the corresponding expressions from Section \ref{ratiosec1}, namely the jumps inside $C_{r_2}$ resp. $C_{\ell_2}$ with the proper identification of the $g$-function \eqref{newg} and the replacement $\Omega(\lambda)\equiv 0$. The same applies to the infinite branches $\tilde{\gamma}_2\cup\tilde{\gamma}_5$ outside $C_0$ as well as for $\gamma_1\cup\gamma_3\cup\gamma_4\cup\gamma_6$, with the replacement $N(\lambda)\mapsto\Upsilon(\lambda)$ as determined through \eqref{nout:1}. On the line segments $(-\lambda^{\ast}+\bar{r},-r)\cup(r,\lambda^{\ast}-\bar{r})$, observe that
	\begin{align*}
		\xi_+(\lambda)&=\xi_-(\lambda)\Upsilon_-(\lambda)\begin{pmatrix}
		-\im\epsilon\,s_1(1+e^{-\varkappa t}) & -\im\epsilon\,e^{-t(\varkappa-\Pi(\lambda))}\\
		-\im\epsilon\,e^{-t(\varkappa+\Pi(\lambda))} & \im\epsilon\,s_3\\
		\end{pmatrix}\big(\Upsilon_-(\lambda)\big)^{-1},\ \ \lambda\in(-\lambda^{\ast}+\bar{r},-r),\\
		\xi_+(\lambda)&=\xi_-(\lambda)\Upsilon_-(\lambda)\begin{pmatrix}
		\im\epsilon\,s_3& -\im\epsilon\,e^{-t(\varkappa-\Pi(\lambda))}\\
		-\im\epsilon\,e^{-t(\varkappa+\Pi(\lambda))} & -\im\epsilon\,s_1(1+e^{-\varkappa t})
		\end{pmatrix}\big(\Upsilon_-(\lambda)\big)^{-1},\ \ \lambda\in(r,\lambda^{\ast}-\bar{r})
	\end{align*}
	and inside $C_0$, first for $\lambda\in(-r,0)$, secondly for $\lambda\in(0,r)$,
	\begin{align*}
		\xi_+(\lambda)&=\xi_-(\lambda)F_-(\lambda)\begin{pmatrix}
		-\im\epsilon\,s_1(1+e^{-\varkappa t}) & (s_1(1+e^{-\varkappa t})-\im\epsilon)e^{-t(\varkappa-\Pi(\lambda))}\\
		-\im\epsilon\,e^{-t(\varkappa+\Pi(\lambda))} & \im\epsilon\,s_3+e^{-2t\varkappa}\\
		\end{pmatrix}\big(F_-(\lambda)\big)^{-1}\\
		\xi_+(\lambda)&=\xi_-(\lambda)F_-(\lambda)\begin{pmatrix}
		\im\epsilon\,s_3 & (-\im\epsilon-s_3)e^{-t(\varkappa-\Pi(\lambda))}\\
		-\im\epsilon\,e^{-t(\varkappa+\Pi(\lambda))} & -\im\epsilon\,s_1(1+e^{-\varkappa t})+e^{-2t\varkappa}\\
		\end{pmatrix}\big(F_-(\lambda)\big)^{-1}.
	\end{align*}
	Notice that for $\lambda\in(-\lambda^{\ast}+\bar{r},-r)\cup(r,\lambda^{\ast}-\bar{r})$,
	\begin{equation*}
		t\big(\varkappa-\Pi(\lambda)\big)\geq t\big(\varkappa-\Pi(\pm r)\big)>-f_2+t\left(\frac{2}{3}\sqrt{2}-\frac{8}{3}\left(\frac{1}{2}-r^2\right)^{\frac{3}{2}}\right)>-f_2+d_8t,\ \ d_8>0
	\end{equation*}
	and for $\lambda\in(-r,r)$,
	\begin{equation*}
		\left|(-\im\epsilon-s_3)e^{-t(\varkappa-\Pi(\lambda))}\right|\leq d_9 e^{-\frac{2}{3}\sqrt{2}t},\ \ \left|(s_1(1+e^{-\varkappa t})-\im\epsilon)e^{-t(\varkappa-\Pi(\lambda))}\right|\leq d_{10}e^{-\frac{2}{3}\sqrt{2}t}
	\end{equation*}
	But since $\Upsilon_-(\lambda)$ is bounded on the full segment $[-\lambda^{\ast},\lambda^{\ast}]$ and $F_-(\lambda)$ inside $(-r,r)$, we obtain that
	\begin{equation}\label{newes:1}
		\|G_{\xi}(\cdot;t,|s_1|)-I\|_{L^2\cap L^{\infty}([-\lambda^{\ast},\lambda^{\ast}])}\leq d_{11}e^{-d_{12}t},\ \ t\rightarrow\infty,|s_1|\uparrow 1:\ \ \varkappa\geq \frac{2}{3}\sqrt{2}-\frac{f_2}{t},\ \ f_2\in\mathbb{R}.
	\end{equation}
	The jumps along $\tilde{\gamma}_2\cup\tilde{\gamma}_5$ inside $C_0$ are again exponentially close to the unit matrix, hence we now list the contributions on the circle boundaries $C_{r_3},C_{\ell_3}$ and $C_0$,
	\begin{equation*}
		\xi_+(\lambda)=\xi_-(\lambda)\Delta^r(\lambda)\big(\Upsilon(\lambda)\big)^{-1},\ \lambda\in C_{r_3};\hspace{0.5cm}\xi_+(\lambda)=\xi_-(\lambda)\Delta^{\ell}(\lambda)\big(\Upsilon(\lambda)\big)^{-1},\ \lambda\in C_{\ell_3}.
	\end{equation*}
	and
	\begin{equation*}
		\xi_+(\lambda)=\xi_-(\lambda)F(\lambda)\big(\Upsilon(\lambda)\big)^{-1},\ \ \lambda\in C_0.
	\end{equation*}
	From \eqref{upedgem:1},\eqref{upedgem:2} with fixed radius $\bar{r}$,
	\begin{equation}\label{newes:2}
		\lambda\in C_{r_3}\cup C_{\ell_3}:\ \ G_{\xi}(\lambda)=I+\mathcal{O}\left(t^{-1}\right),\ \ t\rightarrow\infty,|s_1|\uparrow 1:\ \ \varkappa\geq \frac{2}{3}\sqrt{2}-\frac{f_2}{t},\ f_2\in\mathbb{R}.
	\end{equation}
	On the other hand with \eqref{orimatch} and again fixed radius $r$,
	\begin{equation}\label{newes:3}
		\lambda\in C_0:\ \ G_{\xi}(\lambda)=I+\mathcal{O}\left(t^{-\frac{1}{2}}\right),\ \ t\rightarrow\infty,|s_1|\uparrow 1:\ \ \varkappa\geq\frac{2}{3}\sqrt{2}-\frac{f_2}{t},\ f_2\in\mathbb{R}.
	\end{equation}
	\item As $\lambda\rightarrow\infty$ we have that $\xi(\lambda)\rightarrow I$.
\end{itemize}
We solve the RHP for $\xi(\lambda)$ iteratively.
\subsection{Iterative solution}\label{iterativeeasy} The following estimation is obtained as in Section \ref{ratiosec1}, we only have to keep in mind \eqref{newes:1},\eqref{newes:2} and \eqref{newes:3},
\begin{equation*}
	\|G_{\xi}(\cdot;t,|s_1|)-I\|_{L^2\cap L^{\infty}(\Sigma_{\xi})}\leq d\,t^{-\frac{1}{2}},\ \ \ t\rightarrow+\infty,\ |s_1|\uparrow 1:\ \ \varkappa\geq\frac{2}{3}\sqrt{2}-\frac{f_2}{t},\ f_2\in\mathbb{R}.
\end{equation*}
Thus the ratio problem has a unique asymptotic solution in $L^2(\Sigma_{\xi})$ which satisfies
\begin{equation*}
	\|\xi_-(\cdot;t,|s_1|)-I\|_{L^2(\Sigma_{\xi})}\leq d\,t^{-\frac{1}{2}},\ \ \ d>0.
\end{equation*}
\subsection{Extraction of asymptotics - proof of Theorem \ref{res:4}}\label{extracteasy} Through the transformation sequence
\begin{equation*}
	Y(\lambda)\mapsto X(\lambda)\mapsto Z(\lambda)\mapsto T(\lambda)\mapsto \varpi(\lambda)\mapsto \xi(\lambda)
\end{equation*}
we obtain
\begin{equation*}
	u(x|s)=2\lim_{\lambda\rightarrow\infty}\left[\lambda\big(Y(\lambda;x)\big)_{12}\right]=2\sqrt{-x}\lim_{z\rightarrow\infty}\left[z\xi(z)\Upsilon(z)e^{-t(g(z)-\vartheta(z))\sigma_3}\right]_{12}.
\end{equation*}
Since
\begin{equation*}
	g(z)-\vartheta(z)=\frac{\im}{8z}+\mathcal{O}\left(z^{-3}\right),\ \ \ \ \ \Upsilon(z)=I-\frac{\sqrt{2}\epsilon}{4z}\sigma_1+\mathcal{O}\left(z^{-2}\right),\ \ z\rightarrow\infty
\end{equation*}
we deduce from the integral equation for $\xi(\lambda)$ the exact identity
\begin{equation*}
	u(x|s)=\sqrt{-x}\left[-\frac{\epsilon}{\sqrt{2}}+\mathcal{E}(\varkappa)\right],
\end{equation*}
where (as $t\rightarrow\infty,|s_1|\uparrow 1$ such that $\varkappa\geq\frac{2}{3}\sqrt{2}-\frac{f_2}{t},f_2\in\mathbb{R}$),
\begin{equation*}\label{nice:1}
	\mathcal{E}(\varkappa)=\frac{\im}{\pi}\int_{\Sigma_{\xi}}\big(\xi_-(w)\big(G_{\xi}(w)-I\big)\big)_{12}\,\d w=\frac{\im}{\pi}\oint_{C_0}\big(G_{\xi}(w)-I\big)_{12}\,\d w+\mathcal{O}\left(t^{-1}\right)=\frac{\epsilon}{2\pi}\frac{e^{\sigma t}}{(2\sqrt{2}\,t)^{\frac{1}{2}}}+\mathcal{O}\left(t^{-1}\right),
\end{equation*}
and the last contour integral was computed by residue theorem, compare \eqref{orimatch}. Hence, as $x\rightarrow-\infty,|s_1|\uparrow 1$,
\begin{equation}\label{finalup}
	u(x|s)=-\epsilon\sqrt{-\frac{x}{2}}\left(1-\frac{1}{2\pi}\frac{e^{\sigma t}}{(t\sqrt{2})^{\frac{1}{2}}}+J_2(x,s)\right),\ \ \sigma=\frac{2}{3}\sqrt{2}-\varkappa,\ \ t=(-x)^{\frac{3}{2}}
\end{equation}
with
\begin{prop}\label{imp:4} For any given $f_2\in\mathbb{R}$ there exist positive $t_0=t_0(f_2),v_0=v_0(f_2)$ and $c=c(f_2)$ such that
\begin{equation*}
	\big|J_2(x,s)\big|\leq ct^{-1}\ \ \ \ \forall\, t\geq t_0,\ v\geq v_0,\ \ v\geq\frac{2}{3}\sqrt{2}\,t-f_2.
\end{equation*}
\end{prop}
This completes the proof of Theorem \ref{res:4} and at the same time allows us to derive Corollary \ref{nice}. 
\subsection{Proof of Corollary \ref{nice}} Suppose both $t$ and $v$ are sufficiently large such that
\begin{equation*}
	\varkappa=-\frac{\ln(1-|s_1|^2)}{(-x)^{\frac{3}{2}}}>\frac{2}{3}\sqrt{2}.
\end{equation*}
In this case $\sigma<0$ is bounded away from zero and hence \eqref{newes:3} can be improved to
\begin{equation*}
	G_{\xi}(\lambda)=I+\mathcal{O}\left(t^{-\frac{1}{2}}e^{\sigma t}\right)=I+\mathcal{O}\left(t^{-\infty}\right),\ \ \ \lambda\in C_0,
\end{equation*}
compare \eqref{orimatch}. Therefore,
\begin{equation*}
	\|G_{\xi}(\cdot;t,|s_1|)-I\|_{L^2\cap L^{\infty}(\Sigma_{\xi})}\leq \frac{d}{t},\ \ \ t\rightarrow+\infty,|s_1|\uparrow 1:\ \varkappa>\frac{2}{3}\sqrt{2}
\end{equation*}
and \eqref{nice:1} gets replaced by
\begin{equation*}
	\mathcal{E}(\varkappa) = \frac{u_1}{t}+\frac{u_2}{t^2}+\frac{u_3}{t^3}+\mathcal{O}\left(t^{-4}\right).
\end{equation*}
with $x$ and $\gamma$ independent coefficients $u_k$. These could in principle be computed by residue theorem in \eqref{nice:1} using \eqref{upedgem:1} and \eqref{upedgem:2}, but it is easier to work with the differential equation \eqref{PII} itself: As $x\rightarrow-\infty,|s_1|\uparrow 1$,
\begin{equation}\label{improv}
	u(x|s)=-\epsilon\sqrt{-\frac{x}{2}}\,\left(1+\frac{1}{8x^3}+\mathcal{O}\left(x^{-6}\right)\right),
\end{equation}
uniformly for $\varkappa>\frac{2}{3}\sqrt{2}$ and the asymptotic expansion can be differentiated with respect to $x$.
\begin{rem} The modulus $|s_1|$ will appear explicitly in \eqref{improv}, but only contained in exponentially small error terms. All power like terms are $|s_1|$-independent and in fact identical to the ones derived in \cite{DZ2}, Theorem $1.28$.
\end{rem}
We go back to the Fredholm determinant and note that for fixed $\gamma<1$,
\begin{equation*}
	\frac{\partial}{\partial x}\ln\det\big(I-\gamma K_{\textnormal{Ai}}\big)\Big|_{L^2(x,\infty)} = \int_x^{\infty}u^2_{\scriptscriptstyle{\textnormal{AS}}}(y;\gamma)\d y=\big(u_{\scriptscriptstyle{\textnormal{AS}}}'(x;\gamma)\big)^2-xu_{\scriptscriptstyle{\textnormal{AS}}}^2(x;\gamma)-u_{\scriptscriptstyle{\textnormal{AS}}}^4(x;\gamma),\ \ x\in\mathbb{R}.
\end{equation*}
Hence with Remark \ref{connect} and \eqref{improv} after integration,
\begin{equation}\label{nice:2}
	\ln\det\big(I-\gamma K_{\textnormal{Ai}}\big)\Big|_{L^2(x,\infty)} = \frac{x^3}{12}-\frac{1}{8}\ln|x| +c(\gamma) +r(x,\gamma);
\end{equation}
where $c(\gamma)$ is $x$ independent and there exist universal $t_0,v_0>0$ and $d>0$ such that
\begin{equation*}
	\big|r(x,\gamma)\big|\leq\frac{d}{|x|^3},\ \ \ \ \forall\,t=(-x)^{\frac{3}{2}}\geq t_0,\ -\ln(1-\gamma)\geq v_0,\ \ \varkappa>\frac{2}{3}\sqrt{2}.
\end{equation*}
Now keep $t\geq t_0$ fixed and let $\gamma\uparrow 1$ in \eqref{nice:2}, so that with \eqref{TW:asy1},
\begin{equation*}
	\lim_{\gamma\uparrow 1}c(\gamma)=\ln c_0=\frac{1}{24}\ln 2+\zeta'(-1).
\end{equation*}
This, together with \eqref{nice:2}, completes the proof of Corollary \ref{nice}.
\section{Appearance of  ``Stokes-lines" - proof of Theorem \ref{res:5}}\label{stokes} We assume that $t=(-x)^{\frac{3}{2}}$ and $v=-\ln(1-|s_1|^2)$ are sufficiently large such that
\begin{equation}\label{stoke:1}
	\varkappa\geq\frac{2}{3}\sqrt{2}-f_3\frac{\ln t}{t},\ \ f_3\in\mathbb{R}.
\end{equation}
Working with this scale, several steps in the nonlinear steepest descent analysis of Section \ref{near} can still be applied. In fact, the only difference appears in the matching relation \eqref{orimatch}. Subject to \eqref{stoke:1},  we have that
\begin{equation*}
	e^{\sigma t}\leq \max\left\{1,t^{f_3}\right\}
\end{equation*}
and thus the two model functions $F(\lambda)$ and $\Upsilon(\lambda)$ are in general no longer asymptotically close to each other on $\partial D(0,r)$. In order to deal with this feature, we choose $f_3\in I_q\equiv(-\infty,\frac{1}{2}(q+1))\subset\mathbb{R}$ with $q\in\mathbb{Z}_{\geq 0}$ and notice that for $\z\notin\mathbb{R}$,
\begin{align*}
	I(\z)=\frac{1}{2\pi\im}\int_{\mathbb{R}}\frac{e^{-\mu^2}}{\mu-\z}\,\d\mu &= -\frac{1}{2\pi\im\z}\sum_{n=0}^{q-1}\z^{-n}\int_{\mathbb{R}}\mu^ne^{-\mu^2}\,\d\mu+\frac{\z^{-q}}{2\pi\im}\int_{\mathbb{R}}\frac{\mu^qe^{-\mu^2}}{\mu-\z}\,\d\mu\\
	&=-\frac{1}{2\pi\im\z}\sum_{n=0}^{\lfloor\frac{q-1}{2}\rfloor}\z^{-2n}\,\Gamma\left(\frac{1}{2}+n\right)+\frac{\z^{-q}}{2\pi\im}\int_{\mathbb{R}}\frac{\mu^qe^{-\mu^2}}{\mu-\z}\,\d\mu,
\end{align*}
which implies with Lemma \ref{princip}, as $t\rightarrow\infty$ subject to \eqref{stoke:1},
\begin{equation*}
	e^{\sigma t}I\big(\z(\lambda)\big)=-\frac{e^{\sigma t}}{2\pi\im}\sum_{n=0}^{\lfloor\frac{q-1}{2}\rfloor}\z^{-2n-1}(\lambda)\,\Gamma\left(\frac{1}{2}+n\right)+\mathcal{O}\left(t^{\max\{f_3-\frac{1}{2}-\lfloor\frac{q+1}{2}\rfloor),-\frac{1}{2}-\lfloor\frac{q+1}{2}\rfloor\}}\right),
\end{equation*}
and the last term is small because of our assumption $f_3\in I_q$. We now factorize \eqref{orimatch},
\begin{align}
	F(\lambda)&=\left(I-\im\epsilon\,e^{\sigma t}I\big(\z(\lambda)\big)B_0(\lambda)\sigma_+\big(B_0(\lambda)\big)^{-1}\right)\Upsilon(\lambda)\nonumber\\
	&=B_0(\lambda)\left(I+\frac{\epsilon}{2\pi}e^{\sigma t}\sum_{n=0}^{\lfloor\frac{q-1}{2}\rfloor}\z^{-2n-1}(\lambda)\,\Gamma\left(\frac{1}{2}+n\right)\sigma_+\right)\big(B_0(\lambda)\big)^{-1}\label{tedious}\\
	&\times B_0(\lambda)\left(I-\frac{\epsilon}{2\pi}e^{\sigma t}\z^{-q}(\lambda)\int_{\mathbb{R}}\frac{\mu^qe^{-\mu^2}}{\mu-\z(\lambda)}\,\d\mu\, \sigma_+\right)\big(B_0(\lambda)\big)^{-1}\Upsilon(\lambda)
	\equiv \big(E(\lambda)\times \hat{F}(\lambda)\big)\Upsilon(\lambda)\nonumber
\end{align}
and go back to Section \ref{near}. After the explicit transformations
\begin{equation*}
	Y(\lambda)\mapsto X(\lambda)\mapsto T(\lambda)\mapsto\varpi(\lambda)\mapsto\xi(\lambda)
\end{equation*}
we use another: With \eqref{ratio:2} and Figure \ref{ratio2}, let
\begin{equation*}
	\Theta(\lambda)=\begin{cases}
		\xi(\lambda)E(\lambda),&|\lambda|<r\\
		\xi(\lambda),&|\lambda|>r
		\end{cases}
\end{equation*}
so that by definition, the function $\Theta(\lambda)$ solves the following (singular) RHP.
\begin{problem}\label{singRHP} Determine a $2\times 2$ matrix valued function $\Theta(\lambda)$ such that
\begin{itemize}
	\item $\Theta(\lambda)$ is analytic for $\lambda\in\mathbb{C}\backslash(\Sigma_{\xi}\cup\{0\})$
	\item Along $\Sigma_{\xi}$ we have jumps, $\Theta_+(\lambda)=\Theta_-(\lambda)G_{\Theta}(\lambda)$, with
	\begin{equation}\label{unjump}
		G_{\Theta}(\lambda)=\begin{cases}
			\hat{F}(\lambda),&|\lambda|=r\\
			G_{\xi}(\lambda),& \lambda\in C_{\ell_3}\cup C_{r_3}\cup \Sigma_{\varpi_{\infty}}
			\end{cases}
	\end{equation}
	\item $\Theta(\lambda)$ has a pole of order $2\lfloor\frac{q-1}{2}\rfloor+1$ at $\lambda=0$ provided $q\in\mathbb{Z}_{\geq 1}$. In order to be more precise about the singular structure, notice that for $q\in\mathbb{Z}_{\geq 1}$, 
	\begin{equation*}
		\xi(\lambda)=\Theta(\lambda)\big(E(\lambda)\big)^{-1}=\left[\sum_{k=1}^{n(q)}\frac{\Theta_{-k}}{\lambda^k}+\sum_{\ell=0}^{n(q)-1}\left(\lambda^{n(q)}\Theta(\lambda)\right)^{(n(q)+\ell)}\bigg|_{\lambda=0}\frac{\lambda^{\ell}}{(n(q)+\ell)!}\right]\big(E(\lambda)\big)^{-1},
	\end{equation*}
	with $n(q)=2\lfloor\frac{q-1}{2}\rfloor+1$, is analytic at $\lambda=0$. Thus with \eqref{tedious} and comparison, we can express the full singular part $\{\Theta_{-k}\}$ in terms of the first $n(q)$ terms of the regular part of $\Theta(\lambda)$. This is done explicitly for $q\in\{1,2\}$,
	\begin{equation}\label{resrel}
		\Theta_{-1}\equiv\res_{\lambda=0}\Theta(\lambda)=-\big(\lambda\Theta(\lambda)\big)'\Big|_{\lambda=0}\frac{1}{2}\begin{pmatrix}
		\im\epsilon & 1\\
		1 & -\im\epsilon\\
		\end{pmatrix}\frac{p}{1-\frac{\epsilon p}{\sqrt{2}}},\ \ \ p=-\frac{\epsilon}{2\pi}\sqrt{\frac{\pi}{2}}\frac{e^{\sigma t}}{(t\sqrt{2})^{\frac{1}{2}}},
	\end{equation}
	 but becomes much more involved in case $q\in\mathbb{Z}_{\geq 3}$; we shall postpone the general case $q\in\mathbb{Z}_{\geq 3}$ to a forthcoming publication.
	\item As $\lambda\rightarrow\infty$, we have $\Theta(\lambda)\rightarrow I$.
\end{itemize}
\end{problem}
The jump matrices in RHP \ref{singRHP} are asymptotically close to the unit matrix, however $\Theta(\lambda)$ has a pole at the origin. Such type of problems have been analyzed in nonlinear steepest descent literature, see e.g. \cite{BI,BI2}, at least for the case of a first order pole. Following loc. cit., we first note that \eqref{resrel} is equivalently stated as
\begin{equation*}
	\Theta(\lambda)=\widehat{\Theta}(\lambda)\begin{pmatrix}
	1 & -\frac{\kappa}{\lambda}\\
	0 & 1
	\end{pmatrix}T^{-1},\ \ \ |\lambda|<r;\hspace{1cm}T=\frac{1}{\sqrt{2}}\begin{pmatrix}
	1 & -\im\epsilon\\
	-\im\epsilon & 1
	\end{pmatrix},\ \ \ \kappa=\frac{p}{1-\frac{\epsilon p}{\sqrt{2}}}
\end{equation*} 
where $\widehat{\Theta}(\lambda)$ is analytic at $\lambda=0$. Hence, $\Theta(\lambda)$ is indeed characterized uniquely by the properties stated in RHP \ref{singRHP}. Secondly we can resolve the singular structure  for $q\in\mathbb{Z}_{\geq 1}$ via another transformation, the {\it dressing transformation},
\begin{equation}\label{stoke:2}
	\Theta(\lambda)=\left(\lambda^{2\lfloor\frac{q-1}{2}\rfloor+1}I+\lambda^{2\lfloor\frac{q-1}{2}\rfloor}B_{2\lfloor\frac{q-1}{2}\rfloor}+\ldots+\lambda B_1+B_0\right)\Xi(\lambda)\frac{1}{\lambda^{2\lfloor\frac{q-1}{2}\rfloor+1}},\ \ \lambda\in\mathbb{C}\backslash\Sigma_{\xi}
\end{equation}
where $\{B_j\}_{j=0}^{2\lfloor\frac{q-1}{2}\rfloor}$ are $\lambda$-independent. The function $\Xi(\lambda)$ is characterized as follows.
\begin{problem}\label{dressRHP} Determine a $2\times 2$ matrix valued function $\Xi(\lambda)$ subject to the following conditions
\begin{itemize}
	\item $\Xi(\lambda)$ is analytic for $\lambda\in\mathbb{C}\backslash\Sigma_{\xi}$
	\item We have that $\Xi_+(\lambda)=\Xi_-(\lambda)G_{\Theta}(\lambda)$ for $\lambda\in\Sigma_{\xi}$, with $G_{\Theta}(\lambda)$ given in \eqref{unjump}.
	\item As $\lambda\rightarrow\infty$, 
	\begin{equation*}
		\Xi(\lambda)=I+\mathcal{O}\left(\lambda^{-1}\right).
	\end{equation*}
\end{itemize}
\end{problem}
Notice that the unknowns $B_j$ are algebraically determined from the singular structure in the $\Theta$-RHP and \eqref{stoke:2}, i.e. for $q\in\{1,2\}$ we have
\begin{equation*}
	B_0=-\frac{\kappa}{2}\,\Xi(0)\begin{pmatrix}
	\im\epsilon & 1\\
	1 & -\im\epsilon\\
	\end{pmatrix}\left(\Xi(0)+\frac{\kappa}{2}\,\Xi'(0)\begin{pmatrix}
	\im\epsilon & 1\\
	1 & -\im\epsilon\\
	\end{pmatrix}\right)^{-1}
\end{equation*}
or for more general $q\in\mathbb{Z}_{\geq 3}$ equations for the remaining unknowns in term of $\Xi(0)$ and higher derivatives thereof.
Observe that all jumps in the $\Xi$-RHP are close to the identity since
\begin{equation*}
	\|G_{\Theta}(\cdot;t,|s_1|)-I\|_{L^2\cap L^{\infty}(\Sigma_{\xi})}\leq c\,t^{\max\{f_3-\frac{1}{2}-\lfloor\frac{q+1}{2}\rfloor,-\frac{1}{2}-\lfloor\frac{q+1}{2}\rfloor,-1\}},\ \ c>0
\end{equation*}
uniformly for $t\rightarrow\infty,|s_1|\uparrow 1$ such that $\varkappa\geq\frac{2}{3}\sqrt{2}-f_3\frac{\ln t}{t}$ and $-\infty<f_3<\frac{1}{2}(q+1)$. 
\subsection{Extraction of asymptotics} The RHP \ref{dressRHP} is iteratively solvable and we have for its solution
\begin{equation*}
	\|\Xi_-(\cdot;t,|s_1|)-I\|_{L^2(\Sigma_{\xi})}\leq c\,t^{\max\{f_3-\frac{1}{2}-\lfloor\frac{q+1}{2}\rfloor,-\frac{1}{2}-\lfloor\frac{q+1}{2}\rfloor,-1\}}.
\end{equation*}
Notice that from a direct computation, by tracing back the transformations,
\begin{equation*}
	Y(\lambda)\mapsto X(\lambda)\mapsto Z(\lambda)\mapsto T(\lambda)\mapsto\varpi(\lambda)\mapsto\xi(\lambda)\mapsto\Theta(\lambda)\mapsto\Xi(\lambda),
\end{equation*}
we get
\begin{equation*}
	u(x|s)=\sqrt{-x}\left[-\frac{\epsilon}{\sqrt{2}}+2B_{2\lfloor\frac{q-1}{2}\rfloor}^{12}+\frac{\im}{\pi}\int_{\Sigma_{\Theta}}\Big(\Xi_-(w)\big(G_{\Theta}(w)-I\big)\Big)_{12}\d w\right],\ \ q\in\mathbb{Z}_{\geq 0}:\ \ B_{-2}\equiv 0.
\end{equation*}
We will now continue with the special case $q\in\{1,2\}$, i.e. we compute
\begin{equation}\label{inter:1}
	u(x|s)=\sqrt{-x}\left[-\frac{\epsilon}{\sqrt{2}}+2B_0^{12}-\frac{3p}{16\sqrt{2}\,t}+\frac{\epsilon}{8\sqrt{2}\,t^2}+\mathcal{O}\left(t^{\max\{2f_3-3,f_3-\frac{5}{2},-3\}}\right)\right].
\end{equation}
Notice that
\begin{equation*}
	\Xi(0)=I+\mathcal{O}\left(t^{\max\{f_3-\frac{3}{2},-1\}}\right),\ \ \ \Xi'(0)=\mathcal{O}\left(t^{\max\{f_3-\frac{3}{2},-1\}}\right),
\end{equation*}
and thus
\begin{equation}\label{B0value}
	B_0=-\frac{\kappa}{2}\left[\begin{pmatrix}
	\im\epsilon & 1\\
	1 & -\im\epsilon
	\end{pmatrix}+\mathcal{O}\left(t^{\max\{f_3-\frac{3}{2},-1\}}\right)\right].
\end{equation}
From this it follows in \eqref{inter:1} that
\begin{equation*}
	u(x|s)=-\epsilon\sqrt{-\frac{x}{2}}\,\left(\frac{1+\frac{\epsilon p}{\sqrt{2}}}{1-\frac{\epsilon p}{\sqrt{2}}}\right) +J_3(x,s),\ \ \ \frac{\epsilon p}{\sqrt{2}}=-\frac{1}{2\pi}\sqrt{\frac{\pi}{2}}\frac{e^{\sigma t}}{(2\sqrt{2}\,t)^{\frac{1}{2}}},\ \ \ \sigma=\frac{2}{3}\sqrt{2}-\varkappa
\end{equation*}
and
\begin{prop} For any given $f_3\in(-\infty,\frac{7}{6})$ there exist positive $t_0=t_0(f_3),v_0=v_0(f_3)$ and $c=c(f_3)$ such that
\begin{equation*}
	\big|J_3(x,s)\big|\leq ct^{-\min\{\frac{7}{6}-f_3,\frac{2}{3}\}},\ \ \ \forall\, t\geq t_0,\ v\geq v_0,\ \ \ v\geq\frac{2}{3}\sqrt{2}\,t-f_3\ln t.
\end{equation*}
\end{prop}
If we were to choose $f_3\in(-\infty,\frac{1}{6})$, then $p=\mathcal{O}(t^{f_3-\frac{1}{2}})=o(1)$ and we restore the leading order Hastings-McLeod asymptotics,
\begin{equation*}
	u(x|s)=-\epsilon\sqrt{-\frac{x}{2}}+J_3(x,s),\hspace{1cm}\big|J_3(x,s)\big|\leq ct^{-\min\{\frac{1}{6}-f_3,\frac{2}{3}\}}.
\end{equation*}
In other words,
\begin{equation*}
	\mathcal{S}_1:\ \ v=\frac{2}{3}\sqrt{2}\,t-\frac{1}{6}\ln t
\end{equation*}
forms the first Stokes line and we have completed the proof of Theorem \ref{res:5}.
\subsection{Proof of Corollary \ref{eig0}}\label{eigproof}
Observe that \eqref{inter:1} and \eqref{B0value} allow us to derive a slightly better expansion than the one given in Theorem \ref{res:5}, as $x\rightarrow-\infty,|s_1|\uparrow 1$,
\begin{equation}\label{tstok}
	u(x|s)=-\epsilon\sqrt{-\frac{x}{2}}\left[\left(\frac{1+\frac{\epsilon p}{\sqrt{2}}}{1-\frac{\epsilon p}{\sqrt{2}}}\right)+\frac{1}{8x^3}+f(x,\gamma)\right]
\end{equation}
where, uniformly for $\varkappa\geq\frac{2}{3}\sqrt{2}-f_3\frac{\ln t}{t}, f_3\in(-\infty,\frac{7}{6})$,
\begin{equation}\label{tstok2}
	\big|f(x,\gamma)\big|\leq c\begin{cases}
		\max\{t^{-2}e^{2\sigma t},t^{-3}\},&f_3\in(-\infty,\frac{1}{2}]\\
		\max\{t^{-\frac{3}{2}}e^{\sigma t},t^{-3}\},&f_3\in(\frac{1}{2},\frac{7}{6}).
		\end{cases}
\end{equation}
Next, we use identity \eqref{Meh} and combine it with \eqref{tstok}, 
\begin{equation*}
	-\frac{\partial^2}{\partial x^2}\ln\det\big(I-\gamma K_{\textnormal{Ai}}\big)\Big|_{L^2(x,\infty)} = u^2_{\scriptscriptstyle{\textnormal{AS}}}(x;\gamma) = -\frac{x}{2}\left(\frac{1+\frac{\epsilon p}{\sqrt{2}}}{1-\frac{\epsilon p}{\sqrt{2}}}\right)^2-\frac{1}{8x^2}+R_1(x,\gamma)
\end{equation*}
and from \eqref{tstok2}, noting that $e^{\sigma t}=e^{-v}e^{\frac{2}{3}\sqrt{2}t}$ with $v$ fixed,
\begin{equation}\label{tstok3}
	\iint\big|R_1(x,\gamma)\big|\d^2 x=\mathcal{O}\left(t^{-\min\{\frac{2}{3},\frac{3}{2}-f_3\}}\right),\ \ \ x\rightarrow-\infty,\ \ |s_1|\uparrow 1,\ \ \varkappa\geq\frac{2}{3}\sqrt{2}-f_3\frac{\ln t}{t},\ \ \  f_3\in\left(-\infty,\frac{7}{6}\right).
\end{equation}
However, by straightforward computation,
\begin{equation*}
	\frac{\partial^2}{\partial x^2}\ln\left(1-\frac{\epsilon p}{\sqrt{2}}\right)=\frac{x}{2}\left(\frac{1+\frac{\epsilon p}{\sqrt{2}}}{1-\frac{\epsilon p}{\sqrt{2}}}\right)^2-\frac{x}{2}+R_2(x,\gamma)
\end{equation*}
and $R_2(x,\gamma)$ satisfies also \eqref{tstok3}. Hence, we have shown that for $x\rightarrow-\infty,|s_1|\uparrow 1$,
\begin{equation*}
	\ln\det\big(I-\gamma K_{\textnormal{Ai}}\big)\Big|_{L^2(x,\infty)}=\frac{x^3}{12}-\frac{1}{8}\ln|x|+ax+b+\ln\left(1-\frac{\epsilon p}{\sqrt{2}}\right)+\mathcal{O}\left(t^{-\min\{\frac{2}{3},\frac{3}{2}-f_3\}}\right),
\end{equation*}
uniformly for $\varkappa\geq\frac{2}{3}\sqrt{2}-f_3\frac{\ln t}{t}, f_3\in(-\infty,\frac{7}{6})$ with some $x$-independent constants $a=a(\gamma),b=b(\gamma)$. These can be computed as in the proof of Corollary \ref{nice}, i.e. we use
\begin{equation*}
	\frac{\partial}{\partial x}\ln\det\big(I-\gamma K_{\textnormal{Ai}}\big)\Big|_{L^2(x,\infty)} = \int_x^{\infty}u^2_{\scriptscriptstyle{\textnormal{AS}}}(y;\gamma)\d y=\big(u_{\scriptscriptstyle{\textnormal{AS}}}'(x;\gamma)\big)^2-xu_{\scriptscriptstyle{\textnormal{AS}}}^2(x;\gamma)-u_{\scriptscriptstyle{\textnormal{AS}}}^4(x;\gamma),\ \ x\in\mathbb{R}.
\end{equation*}
together with \eqref{tstok} to conclude $a(\gamma)=0$ and the limit $\gamma\uparrow 1$, with $t\geq t_0$ fixed, yields
\begin{equation*}
	\lim_{\gamma\uparrow 1}b(\gamma)=\ln c_0=\frac{1}{24}\ln 2+\zeta'(-1).
\end{equation*}

\section{Singular transition analysis for $\varkappa\in[\delta,\frac{2}{3}\sqrt{2}-\delta]$ with $0<\delta<\frac{1}{3}\sqrt{2}$ fixed}\label{singsec:1}

Several steps of the analysis carried out below have their natural counterpart in the regular analysis. One major difference however lies in the construction of a new outer parametrix which naturally involves the exceptional set $\mathcal{Z}_n$ as defined in \eqref{exceptset} below. First we let 
\begin{equation*}
	\varkappa=\frac{v}{t}\in(0,\infty);\ \ \ \ \ t=(-x)^{\frac{3}{2}}>0,\ \ 0<v=-\ln\big(|s_1|^2-1\big)\equiv-2\pi\widehat{\beta},\ \ \ 1<|s_1|<\sqrt{2}
\end{equation*}
and keep $\varkappa\in[\delta,\frac{2}{3}\sqrt{2}-\delta]$ with $0<\delta<\frac{1}{3}\sqrt{2}$ fixed. Also here, the following expansions will prove useful.
\subsection{Preliminary expansions}
Note that
\begin{equation*}
	\Re(s_1)=-\frac{s_2}{2}e^{-\varkappa t},\hspace{1cm} |s_1|=1+\frac{1}{2}e^{-\varkappa t}+\mathcal{O}\left(e^{-2\varkappa t}\right)
\end{equation*}
and
\begin{equation*}
	s_1=e^{\im\textnormal{arg}(s_1)}\left(1+\frac{1}{2}e^{-\varkappa t}+\mathcal{O}\left(e^{-2\varkappa t}\right)\right)=\bar{s}_3,\ \ 
	s_1+s_1(1-s_1s_3)=e^{\im\textnormal{arg}(s_1)}\left(1-\frac{1}{2}e^{-\varkappa t}+\mathcal{O}\left(e^{-2\varkappa t}\right)\right).
\end{equation*}
Moreover,
\begin{equation*}
	\textnormal{arg}(s_1)=\frac{\epsilon\pi}{2}\left(1+\frac{s_2}{\pi}e^{-\varkappa t}+\mathcal{O}\left(e^{-2\varkappa t}\right)\right)
\end{equation*}
so that
\begin{prop} As $t\rightarrow\infty,|s_1|\downarrow 1$ with $e^{\im\frac{\pi}{2}\epsilon}=\im\epsilon$,
\begin{equation*}
	s_1=\im\epsilon\left(1+\frac{1}{2}\left(1+\im\epsilon s_2\right)e^{-\varkappa t}+\mathcal{O}\left(e^{-2\varkappa t}\right)\right),\ \ s_3=-\im\epsilon\left(1+\frac{1}{2}\left(1-\im\epsilon s_2\right)e^{-\varkappa t}+\mathcal{O}\left(e^{-2\varkappa t}\right)\right),
\end{equation*}
\begin{equation*}
	s_1+s_1(1-s_1s_3)=\im\epsilon\left(1-\frac{1}{2}\left(1-\im\epsilon s_2\right)e^{-\varkappa t}+\mathcal{O}\left(e^{-2\varkappa t}\right)\right).
\end{equation*}
uniformly for $\varkappa\in[\delta,\frac{2}{3}\sqrt{2}-\delta]$ with $0<\delta<\frac{1}{3}\sqrt{2}$ fixed.
\end{prop}
\subsection{g-function transformation}  The bijective correspondence between $\varkappa\in(0,\frac{2}{3}\sqrt{2})$ and $\k\in(0,1)$ has already been established in Proposition \ref{modex}, and all local expansions in Corollary \ref{cor1} can directly be applied to the singular transition analysis. With this in mind we define (formally just as in \eqref{gf:1}),
\begin{equation*}
	g(z)=4\im\int_M^z\big((\mu^2-M^2)(\mu^2-m^2)\big)^{\frac{1}{2}}\,\d\mu,\ \ \ z\in\mathbb{C}\backslash[-M,M]
\end{equation*}
where
\begin{equation*}
	0<m=\frac{1}{\sqrt{2}}\frac{\k}{\sqrt{1+\k^2}}<\frac{1}{2}<M=\frac{1}{\sqrt{2}}\frac{1}{\sqrt{1+\k^2}}<\frac{1}{\sqrt{2}},
\end{equation*}
and $\k=\k(\varkappa)$ is determined through \eqref{inteq:1}. The relevant analytical properties of $g(z)$ are summarized in Proposition \ref{gprop} and we choose in the $T$-RHP \eqref{Tfunc} (as before),
\begin{equation*}
	\lambda^{\ast}=M.
\end{equation*}
After that, introduce,
\begin{equation*}
	S(\lambda)=e^{-t\ell\sigma_3}T(\lambda)e^{t(g(\lambda)-\vartheta(\lambda))\sigma_3},\ \ \ \lambda\in\mathbb{C}\backslash\Sigma_T
\end{equation*}
where we shall use the same notation $S(\lambda)$ for the transformed function as in the regular transition analysis. Observe that, for $\lambda\in(m,M)$,
\begin{equation*}
	S_+(\lambda)=S_-(\lambda)\begin{pmatrix}
	-e^{-t(\varkappa-\Pi(\lambda))} & -s_3\\
	s_1(1-e^{-\varkappa t}) & -e^{-t(\varkappa+\Pi(\lambda))}\\
	\end{pmatrix},\ \ \varkappa-\Pi(\lambda)=8\int_m^{\lambda}\sqrt{(M^2-\mu^2)(\mu^2-m^2)}\,\d\mu>0,
\end{equation*}
so that (as before)
\begin{equation*}
	G_S(\lambda)\begin{pmatrix}	
	0&-e^{-\im\frac{\pi}{2}\epsilon} \\
	 e^{\im\frac{\pi}{2}\epsilon}& 0\\
	 \end{pmatrix}^{-1}\rightarrow I,\ \ x\rightarrow-\infty,\ |s_1|\downarrow 1:\ \ \varkappa\in\left[\delta,\frac{2}{3}\sqrt{2}-\delta\right],\ \delta>0.
\end{equation*}
Similarly, for $\lambda\in(-M,-m)$,
\begin{equation*}
	S_+(\lambda) = S_-(\lambda)\begin{pmatrix}
	-e^{-t(\varkappa-\Pi(\lambda))} & s_1(1-e^{-\varkappa t})e^{\im t\Omega(\lambda)}\\
	-s_3e^{-\im t\Omega(\lambda)} & -e^{-t(\varkappa+\Pi(\lambda))}\\
	\end{pmatrix},\ \ \varkappa-\Pi(\lambda)=8\int_{\lambda}^{-m}\sqrt{(M^2-\mu^2)(\mu^2-m^2)}\,\d\mu>0,
\end{equation*}
hence (as before)
\begin{equation*}
	G_S(\lambda)\begin{pmatrix}
	0 & e^{\im\frac{\pi}{2}\epsilon+\im t\Omega(\lambda)}\\
	-e^{-\im\frac{\pi}{2}\epsilon-\im t\Omega(\lambda)} & 0\\
	\end{pmatrix}^{-1}\rightarrow I,\ \ x\rightarrow-\infty,\ |s_1|\downarrow 1:\ \ \varkappa\in\left[\delta,\frac{2}{3}\sqrt{2}-\delta\right],\ \delta>0.
\end{equation*}
The only difference occurs in the {\it gap}, first for $\lambda\in(0,m)$,
\begin{eqnarray*}
	S_+(\lambda)&=&S_-(\lambda)\begin{pmatrix}
	-1 & -s_3e^{\im t\Omega(\lambda)}\\
	(s_1+s_1(1-s_1s_3))e^{-\im t\Omega(\lambda)} & -(s_1s_3-1)^2\\
	\end{pmatrix}\\
	&=&S_-(\lambda)\begin{pmatrix}
	1 & 0\\
	-(s_1+s_1(1-s_1s_3))e^{-\im t\Omega(\lambda)} & 1\\
	\end{pmatrix}\begin{pmatrix}
	-1 & 0\\
	0 & -1\\
	\end{pmatrix}\begin{pmatrix}
	1 & s_3e^{\im t\Omega(\lambda)}\\
	0 & 1\\
	\end{pmatrix}\\
	&\equiv& S_-(\lambda)e^{\im\frac{\pi}{2}\sigma_3}S_{L_1}(\lambda)e^{-\im\frac{\pi}{2}\sigma_3}S_D(\lambda)e^{\im\frac{\pi}{2}\sigma_3}S_{U_1}(\lambda)e^{-\im\frac{\pi}{2}\sigma_3},\ \ \lambda\in(0,m)
\end{eqnarray*}
and secondly for $\lambda\in(-m,0)$,
\begin{eqnarray*}
	S_+(\lambda)&=&S_-(\lambda)\begin{pmatrix}
	-1 & (s_1+s_1(1-s_1s_3))e^{\im t\Omega(\lambda)}\\
	-s_3e^{-\im t\Omega(\lambda)} & -(s_1s_3-1)^2\\
	\end{pmatrix}\\
	&=&S_-(\lambda)\begin{pmatrix}
	1 & 0\\
	s_3e^{-\im t\Omega(\lambda)} & 1\\
	\end{pmatrix}\begin{pmatrix}
	-1 & 0\\
	0 & -1\\
	\end{pmatrix}\begin{pmatrix}
	1 & -(s_1+s_1(1-s_1s_3))e^{\im t\Omega(\lambda)}\\
	0 & 1\\
	\end{pmatrix}\\
	&=&S_-(\lambda)e^{\im\frac{\pi}{2}\sigma_3}S_{L_2}(\lambda)e^{-\im\frac{\pi}{2}\sigma_3}S_D(\lambda)e^{\im\frac{\pi}{2}\sigma_3}S_{U_2}(\lambda)e^{-\im\frac{\pi}{2}\sigma_3},\ \ \lambda\in(-m,0).
\end{eqnarray*}	
Here, $S_{L_j}(\lambda)$ and $S_{U_j}(\lambda)$ have appeared in \eqref{ob2}, \eqref{ob3} and $S_D(\lambda)=-I$. Hence, we can use Proposition \ref{openup} in the given context as well and open lens: Compare Figure \ref{figure6} and define
\begin{equation*}
	L(\lambda)=\begin{cases}
	S(\lambda)e^{\im\frac{\pi}{2}\sigma_3}S_{U_1}^{-1}(\lambda)e^{-\im\frac{\pi}{2}\sigma_3},&\lambda\in\mathcal{L}_1^+\\
	S(\lambda)e^{\im\frac{\pi}{2}\sigma_3}S_{L_1}(\lambda)e^{-\im\frac{\pi}{2}\sigma_3},&\lambda\in\mathcal{L}_1^-\\
	S(\lambda)e^{\im\frac{\pi}{2}\sigma_3}S_{U_2}^{-1}(\lambda)e^{-\im\frac{\pi}{2}\sigma_3},&\lambda\in\mathcal{L}_2^+\\
	S(\lambda)e^{\im\frac{\pi}{2}\sigma_3}S_{L_2}(\lambda)e^{-\im\frac{\pi}{2}\sigma_3},&\lambda\in\mathcal{L}_2^-\\
	S(\lambda),&\textnormal{otherwise}.
	\end{cases}
\end{equation*}
This leads to the following RHP
\begin{problem} Determine a $2\times 2$ matrix valued function $L(\lambda)$ such that
\begin{itemize}
	\item $L(\lambda)$ is analytic for $\lambda\in\mathbb{C}\backslash\big([-M,-m]\cup[-m,m]\cup[m,M]\cup\gamma_1^{\pm}\cup\gamma_2^{\pm}\cup\Sigma_{S_{\infty}}\big)$
	\item The jumps are as follows, see Figure \ref{figure6} for orientations,
	\begin{equation}\label{newLjump}
		L_+(\lambda)=L_-(\lambda)\begin{cases}
		G_S(\lambda),&\lambda\in(-M,-m)\cup(m,M)\cup\Sigma_{S_{\infty}}\\
		e^{\im\frac{\pi}{2}\sigma_3}S_{U_1}(\lambda)e^{-\im\frac{\pi}{2}\sigma_3},&\lambda\in\gamma_1^+\\
		e^{\im\frac{\pi}{2}\sigma_3}S_{L_1}(\lambda)e^{-\im\frac{\pi}{2}\sigma_3},&\lambda\in\gamma_1^-\\
		e^{\im\frac{\pi}{2}\sigma_3}S_{U_2}(\lambda)e^{-\im\frac{\pi}{2}\sigma_3},&\lambda\in\gamma_2^+\\
		e^{\im\frac{\pi}{2}\sigma_3}S_{L_2}(\lambda)e^{-\im\frac{\pi}{2}\sigma_3},&\lambda\in\gamma_2^-\\
		S_D(\lambda),&\lambda\in(-m,m)
		\end{cases}
	\end{equation}
	\item As $\lambda\rightarrow\infty$,
	\begin{equation*}
		L(\lambda)=I+\mathcal{O}\left(\lambda^{-1}\right).
	\end{equation*}
\end{itemize}
\end{problem}
The jumps in this problem along the lens boundaries $\gamma_j^{\pm}$ and along the infinite contours are exponentially close to the unit matrix in the (singular) double scaling limit as long as we stay away from the branch points $\lambda=\pm m,\pm M$ and $\lambda=0$. Compared to the $L$-RHP in the regular transition analysis, the outer parametrix will now also have a jump in the gap $(-m,m)$.
\subsection{New outer parametrix} In case of the singular transition analysis the outer model problem is posed on the full segment $(-M,M)$ and reads as follows
\begin{problem} Determine a $2\times 2$ matrix-valued piecewise analytic function $N(\lambda)=N(\lambda;t,\epsilon)$ such that
\begin{itemize}
	\item $N(\lambda)$ is analytic for $\lambda\in\mathbb{C}\backslash[-M,M]$
	\item We have jumps
	\begin{eqnarray}
		N_+(\lambda)&=&N_-(\lambda)\begin{pmatrix}
		0 & e^{\im\frac{\pi}{2}\epsilon+\im t\Omega(\lambda)}\\
		-e^{-\im\frac{\pi}{2}\epsilon-\im t\Omega(\lambda)} & 0\\
		\end{pmatrix},\ \ \lambda\in(-M,-m)\label{ouup:1}\\
		N_+(\lambda)&=&N_-(\lambda)e^{\im\pi\sigma_3},\ \ \lambda\in(-m,m)\label{ouup:2}\\
		N_+(\lambda)&=&N_-(\lambda)\begin{pmatrix}
		0 & -e^{-\im\frac{\pi}{2}\epsilon}\\
		e^{\im\frac{\pi}{2}\epsilon} & 0\\
		\end{pmatrix},\ \ \lambda\in(m,M)\label{ouup:3}
	\end{eqnarray}
	\item $N(\lambda)$ has $L^2$ boundary values on $[M,M]$
	\item As $\lambda\rightarrow\infty$,
	\begin{equation*}
		N(\lambda)=I+\mathcal{O}\left(\lambda^{-1}\right)
	\end{equation*}
\end{itemize}
\end{problem}
Notice that the function
\begin{equation*}
	\widehat{N}(\lambda)=e^{-\im\frac{\pi}{4}\epsilon\sigma_3}N(\lambda)e^{\im\frac{\pi}{4}\epsilon\sigma_3},\ \ \ \lambda\in\mathbb{C}\backslash[-M,M]
\end{equation*}
has jumps on $(-M,M)$ described as follows
\begin{eqnarray*}
	\widehat{N}_+(\lambda)&=&\widehat{N}_-(\lambda)\begin{pmatrix}
	0 & e^{\im t\Omega(\lambda)}\\
	-e^{-\im t\Omega(\lambda)} & 0\\
	\end{pmatrix},\ \ \ \lambda\in(-M,-m)\\
	\widehat{N}_+(\lambda)&=&\widehat{N}_-(\lambda)e^{\im\pi\sigma_3},\ \ \ \lambda\in(-m,m)\\
	\widehat{N}_+(\lambda)&=&\widehat{N}_-(\lambda)\begin{pmatrix}
	0 & 1\\
	-1 & 0\\
	\end{pmatrix},\ \ \ \lambda\in(m,M).
\end{eqnarray*}
Up to the presence of the diagonal jump in the gap $(-m,m)$, the jumps are thus identical to the ones which we generated in Proposition \ref{pout:3}. To compensate for the diagonal part, we slightly modify our construction of Section \ref{outer1para}: Let $u(z)$ denote again the Abel (type) map
\begin{equation*}
	u:\mathbb{CP}^1\backslash[-M,M]\rightarrow\mathbb{C},\ \ z\mapsto u(z)=\int_M^z\omega = u(\infty)-\frac{c}{z}+\mathcal{O}\left(z^{-3}\right),\ \ z\rightarrow\infty;\ \ c=c(\varkappa)=\frac{\im}{2}\frac{M}{K'}
\end{equation*}
for which we summarized certain properties in Proposition \ref{pout:1}. We require
\begin{equation}\label{gap:1}
	\widehat{N}^{(\pm)}(z)=\left(\frac{\theta\big(u(z)+tV+\frac{\tau}{2}\pm d\big)}{\theta\big(u(z)\pm d\big)},\frac{\theta(-u(z)+tV+\frac{\tau}{2}\pm d)}{\theta(-u(z)\pm d)}\right)\equiv\left(\widehat{N}_1^{(\pm)}(z),\widehat{N}_2^{(\pm)}(z)\right),
\end{equation}
where (formally as before, see also \eqref{Vfancy}),
\begin{equation*}
	V(\varkappa)\equiv V=-\frac{4}{\pi}\int_{-m}^m\sqrt{(M^2-\mu^2)(m^2-\mu^2)}\,\d\mu\equiv\frac{1}{2\pi}\Omega(z),\ \ z\in(-M,-m);\hspace{1cm}d=-\frac{\tau}{4}
\end{equation*}
and 
\begin{equation*}
	\theta(z|\tau)\equiv \theta_3(z|\tau)=\sum_{k\in\mathbb{Z}}\exp\left[\im\pi k^2\tau+2\pi\im kz\right],\ \ \ z\in\mathbb{C}
\end{equation*}
is the third Jacobi theta function. With this, we have the following analogue to Proposition \ref{pout:3} for the current situation \eqref{ouup:1}-\eqref{ouup:3}.
\begin{prop}\label{pout:4} The function
\begin{equation}\label{singouter}
	N(\lambda) = e^{\im\frac{\pi}{4}\epsilon\sigma_3}\frac{\theta(0)}{\theta(\frac{\tau}{2}+tV)}e^{-\im\pi u(\infty)\sigma_3}\begin{pmatrix}
	\widehat{N}_1^{(+)}(\lambda)\phi(\lambda) & \widehat{N}_2^{(+)}(\lambda)\hat{\phi}(\lambda)\\
	-\widehat{N}_1^{(-)}(\lambda)\hat{\phi}(\lambda)&\widehat{N}_2^{(-)}(\lambda)\phi(\lambda)\\
	\end{pmatrix}e^{\im\pi u(\lambda)\sigma_3}e^{-\im\frac{\pi}{4}\epsilon\sigma_3}
\end{equation}
is single-valued and analytic in $\mathbb{CP}^1\backslash[-M,M]$. Its jumps are stated in \eqref{ouup:1},\eqref{ouup:2} and \eqref{ouup:3}, furthermore, as $\lambda\rightarrow\infty$,
\begin{align*}
	N(\lambda)&=I+\frac{1}{\lambda}\begin{pmatrix}
	-c\frac{\theta'(\frac{\tau}{2}+tV)}{\theta(\frac{\tau}{2}+tV)}-\im\pi c & -\frac{\theta(0)}{\theta(\frac{\tau}{2}+tV)}\frac{\theta(u(\infty)-\frac{\tau}{2}-tV-d)}{\theta(u(\infty)-d)}\frac{M-m}{2\im}e^{-2\pi\im u(\infty)}e^{\im\frac{\pi}{2}\epsilon}\\
	\frac{\theta(0)}{\theta(\frac{\tau}{2}+tV)}\frac{\theta(u(\infty)+\frac{\tau}{2}+tV-d)}{\theta(u(\infty)-d)}\frac{M-m}{2\im}e^{2\pi\im u(\infty)}e^{-\im\frac{\pi}{2}\epsilon} & c\frac{\theta'(\frac{\tau}{2}+tV)}{\theta(\frac{\tau}{2}+tV)}+\im\pi c\\
	\end{pmatrix}\\
	&+\mathcal{O}\left(\lambda^{-2}\right),
\end{align*}
provided $(x,s_1)$ are uniformly bounded away from the discrete set
\begin{equation}\label{exceptset}
	\mathcal{Z}_n=\Big\{(x,s_1):\ \ 2tV(\varkappa)=n\in\mathbb{Z}\backslash\{0\}\Big\}.
\end{equation}
\end{prop} 
\begin{proof} The stated jump behavior follows from a direct computation using Proposition \ref{pout:1} and the properties of Jacobi theta functions. For the normalization, we take into account the shift by the half period appearing in the numerators in \eqref{gap:1}, hence we have to guarantee that
\begin{equation*}
	\theta\left(\frac{\tau}{2}+tV\right)\neq 0
\end{equation*}
which is ensured for the values $(x,s_1)$ away from the set $\mathcal{Z}_n$ given in \eqref{exceptset} as the simple roots of $\theta(z)\equiv\theta_3(z|\tau)$ are located at $z\equiv\frac{1}{2}+\frac{\tau}{2}\mod \mathbb{Z}+\tau\mathbb{Z}$.
\end{proof}
\begin{remark} The latter Proposition is in sharp contrast to Proposition \ref{pout:3}. There we do not require the half-period shift and thus we always have that $\theta(tV)\neq 0$ since $tV(\varkappa)\in\mathbb{R}$, i.e. there is no exceptional set. From now on we shall always assume that $(x,s_1)$ is bounded away from the set $\mathcal{Z}_n$.
\end{remark}	
\subsection{Local parametrices}
The required model functions near $\lambda=0$ and $\lambda=\pm m$ slightly differ from the ones used in the case of the regular transition. Near $\lambda=0$ we analyze
\begin{problem} Determine a $2\times 2$ piecewise analytic function $H(\lambda)$ such that
\begin{itemize}
	\item $H(\lambda)$ is analytic for $\lambda\in\big(\mathbb{C}\cap D(0,r)\big)\backslash\big(\gamma_1^{\pm}\cup\gamma_2^{\pm}\cup(-r,r)\big)$ with $0<r<\frac{m}{2}$.
	\item For the boundary values,
	\begin{eqnarray*}
		H_+(\lambda)&=&H_-(\lambda)\begin{pmatrix}
		1 & -\im\epsilon\,e^{\im t\Omega(\lambda)}\\
		0 & 1\\
		\end{pmatrix},\hspace{0.5cm}\lambda\in D(0,r)\cap(\gamma_1^+\cup\gamma_2^+)\\
		H_+(\lambda)&=&H_-(\lambda)\begin{pmatrix}
		1 & 0\\
		-\im\epsilon\,e^{-\im t\Omega(\lambda)} & 1\\
		\end{pmatrix},\hspace{0.5cm}\lambda\in D(0,r)\cap(\gamma_1^-\cup\gamma_2^-)\\
		H_+(\lambda)&=&-H_-(\lambda),\hspace{0.5cm}\lambda\in(-r,r)
	\end{eqnarray*}
	\item  As $t\rightarrow\infty,|s_1|\downarrow 1$ such that $\varkappa\in[\delta,\frac{2}{3}\sqrt{2}-\delta]$, uniformly for $\lambda\in\partial D(0,r)$,
	\begin{equation*}
		H(\lambda)=\big(I+o(1)\big)N(\lambda).
	\end{equation*}
\end{itemize}
\end{problem}
A solution to the latter problem is derived in two steps. First define
\begin{equation*}
	H^o(\z) = \begin{cases}
		e^{\im\z\sigma_3},&\textnormal{arg}\,\z\in(0,\frac{\pi}{4})\cup(\frac{3\pi}{4},\pi)\\
		-e^{\im\z\sigma_3},&\textnormal{arg}\,\z\in(-\frac{\pi}{4},0)\cup(\pi,\frac{5\pi}{4})\\
		e^{\im\z\sigma_3}\begin{pmatrix}
		1 & -\im\epsilon\,e^{\im t\Omega(0)}\\
		0 & 1\\
		\end{pmatrix},&\textnormal{arg}\,\z\in(\frac{\pi}{4},\frac{3\pi}{4})\\
		-e^{\im\z\sigma_3}\begin{pmatrix}
		1 & 0\\
		\im\epsilon\,e^{-\im t\Omega(0)} & 1\\
		\end{pmatrix},&\textnormal{arg}\,\z\in(-\frac{3\pi}{4},-\frac{\pi}{4})
		\end{cases}
\end{equation*}
with $\Omega(0)=\pi V(\varkappa)$ from \eqref{Vfancy} and where we orient all six rays from zero to infinity. Second, the required parametrix itself: With $N(\lambda)$ from \eqref{singouter},
\begin{equation}\label{singorigin}
	H(\lambda)=N(\lambda)\begin{cases}
	+1,&\Im\lambda>0\\
	-1,&\Im\lambda<0
	\end{cases}\,\times\,H^o\big(\z(\lambda)\big)e^{-\im\z(\lambda)\sigma_3},\ \ \ \ \ \lambda\in D(0,r),
\end{equation}
where
\begin{equation*}
	\z(\lambda)=\frac{t}{2}\big(\Omega(\lambda)-\Omega(0)\big)=8tMm\lambda\left(1+\mathcal{O}\left(\lambda^2\right)\right),\ \ \lambda\in D(0,r).
\end{equation*}
The required jump behavior follows from a direct computation using a local contour deformation, moreover we have
\begin{equation*}
	H(\lambda)=\left(I+\mathcal{O}\left(t^{-\infty}\right)\right)N(\lambda)
\end{equation*}
as $t\rightarrow+\infty,|s_1|\downarrow 1$ such that $\varkappa\in[\delta,\frac{2}{3}\sqrt{2}-\delta]$ uniformly for $0<r_1\leq|\lambda|\leq r_2<\frac{m}{2}$.\bigskip

Near $\lambda=m$ we have to modify \eqref{bare:1}. Consider
\begin{equation*}
	A^{RH}(\z)=e^{\im\frac{\pi}{2}\sigma_3}A_0(\z)e^{-\im\frac{\pi}{2}\sigma_3}\begin{cases}
	I,&\textnormal{arg}\,\z\in(0,\frac{2\pi}{3})\\
	\begin{pmatrix}
	1 & 1\\
	0 & 1\\
	\end{pmatrix},&\textnormal{arg}\,\z\in(\frac{2\pi}{3},\pi)\smallskip\\
	-\begin{pmatrix}
	1 & 1\\
	0 & 1\\
	\end{pmatrix},&\textnormal{arg}\,\z\in(\pi,\frac{4\pi}{3})\smallskip\\
	-\begin{pmatrix}
	1 & 1\\
	0 & 1\\
	\end{pmatrix}\begin{pmatrix}
	1 & 0\\
	-1 & 1\\
	\end{pmatrix},&\textnormal{arg}\,\z\in(\frac{4\pi}{3},2\pi)
	\end{cases}
\end{equation*}
with $A_0(\z)$ as in \eqref{bareAiry}. This leads to the following bare RHP
\begin{itemize}
	\item $A^{RH}(\z)$ is analytic for $\z\in\mathbb{C}\backslash\{\textnormal{arg}\,\z=0,\frac{2\pi}{3},\pi,\frac{4\pi}{3}\}$
	\item We have jumps on the contours sketched in Figure \ref{Airy1sing},
	\begin{eqnarray*}
		A^{RH}_+(\z)&=&A^{RH}_-(\z)\begin{pmatrix}
		1 & -1\\
		0 & 1\\
		\end{pmatrix},\ \ \textnormal{arg}\,\z=\frac{2\pi}{3}\\
		A^{RH}_+(\z)&=&A^{RH}_-(\z)\begin{pmatrix}
		-1 & 0\\
		0 & -1\\
		\end{pmatrix},\ \ \textnormal{arg}\,\z=\pi\\
		A^{RH}_+(\z)&=&A^{RH}_-(\z)\begin{pmatrix}
		1& 0\\
		1 & 1\\
		\end{pmatrix},\ \ \textnormal{arg}\,\z=\frac{4\pi}{3}\\
		A^{RH}_+(\z)&=&A^{RH}_-(\z)\begin{pmatrix}
		-1 & 1\\
		-1 & 0\\
		\end{pmatrix},\ \ \textnormal{arg}\,\z=0
	\end{eqnarray*}
	\item As $\z\rightarrow\infty$, 
	\begin{equation*}
		A^{RH}(\z)=\begin{cases}
		+1,&\Im\z>0\\
		-1,&\Im\z<0
		\end{cases}\,\times\,\z^{-\frac{1}{4}\sigma_3}\frac{\im}{2}\begin{pmatrix}
		1 & \im\\
		1 & -\im\\
		\end{pmatrix}\left[I+\frac{1}{48\z^{\frac{3}{2}}}\begin{pmatrix}
		-1 & -6\im\\
		-6\im & 1\\
		\end{pmatrix}+\mathcal{O}\left(\z^{-\frac{6}{2}}\right)\right]e^{\frac{2}{3}\z^{\frac{3}{2}}\sigma_3}
	\end{equation*}
\end{itemize}
\begin{figure}
\begin{minipage}{0.35\textwidth} 
\begin{center}
\resizebox{1.1\textwidth}{!}{\includegraphics{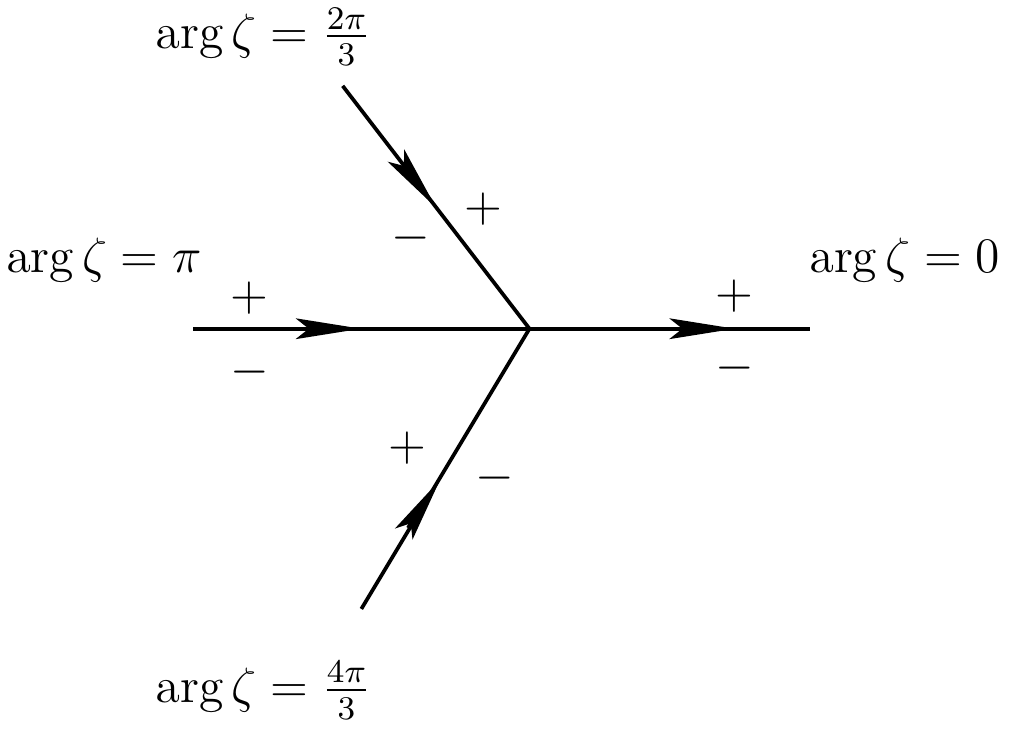}}
\caption{Jump contour for the bare parametrix $A^{RH}(\zeta)$.}
\label{Airy1sing}
\end{center}
\end{minipage}
\ \ \ \ \begin{minipage}{0.5\textwidth}
\begin{center}
\resizebox{0.92\textwidth}{!}{\includegraphics{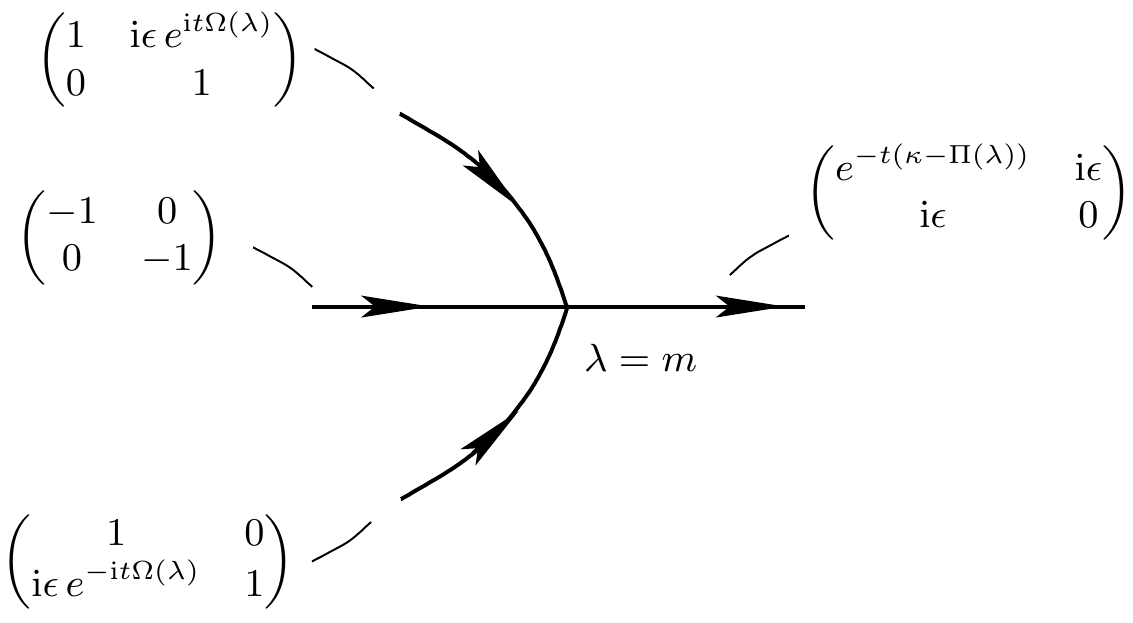}}
\caption{Jumps of the model function $U(\lambda)$ near $\lambda=m$.}
\label{figure7sing}
\end{center}
\end{minipage}
\end{figure}

Now define
\begin{equation}\label{singm}
	U(\lambda)=B_{r_3}(\lambda)A^{RH}\big(\z(\lambda)\big)e^{-\frac{2}{3}\z^{\frac{3}{2}}(\lambda)\sigma_3}e^{-\im\frac{\pi}{4}\epsilon\sigma_3},\ \ \lambda\in D(m,r)
\end{equation}
where (as before)
\begin{equation*}
	\z(\lambda)=\left[6te^{\im\frac{\pi}{2}}\int_m^{\lambda}\left(\left(M^2-\mu^2\right)\left(m^2-\mu^2\right)\right)^{\frac{1}{2}}\,\d\mu\right]^{\frac{2}{3}}\sim\left(4t\sqrt{2m(M^2-m^2)}\right)^{\frac{2}{3}}(\lambda-m),\ \ |\lambda-m|<r
\end{equation*}
and
\begin{equation*}
	B_{r_3}(\lambda)=N(\lambda)\begin{cases}
	+1,&\Im\lambda>0\\
	-1,&\Im\lambda<0
	\end{cases}\,\times e^{\im\frac{\pi}{4}\epsilon\sigma_3}\begin{pmatrix}
	-\im & -\im\\
	-1 & 1\\
	\end{pmatrix}\delta(\lambda)^{\sigma_3}\left(\z(\lambda)\frac{\lambda-M}{\lambda-m}\right)^{\frac{1}{4}\sigma_3},\ \ \delta(\lambda)=\left(\frac{\lambda-m}{\lambda-M}\right)^{\frac{1}{4}}
\end{equation*}
is analytic near $\lambda=m$. This leads to the jump behavior of $U(\lambda)$ as indicated in Figure \ref{figure7sing}, moreover we have the matching relation,
	\begin{equation*}
		U(\lambda)=\left[I+N(\lambda)\left\{\frac{1}{48\z^{\frac{3}{2}}}\begin{pmatrix}
		-1 & 6\epsilon\\
		-6\epsilon & 1\\
		\end{pmatrix} +\mathcal{O}\left(\z^{-\frac{6}{2}}\right)\right\}\big(N(\lambda)\big)^{-1}\right]N(\lambda)
	\end{equation*}
	and thus, as $t\rightarrow+\infty,|s_1|\downarrow 1$ such that $\varkappa\in[\delta,\frac{2}{3}\sqrt{2}-\delta]$ with $0<\delta<\frac{1}{3}\sqrt{2}$ fixed,
	\begin{equation*}
		U(\lambda)=\left(I+o(1)\right)N(\lambda)
	\end{equation*}
	uniformly for $0<r_1\leq|\lambda-m|\leq r_2<\min\{\frac{m}{2},\frac{1}{2}(M-m)\}$.\smallskip

For the remaining endpoint $\lambda=-m$, we put
\begin{equation*}
	\widetilde{A}^{RH}(\z)=e^{\im\frac{\pi}{2}\sigma_3}\widetilde{A}_0(\z)e^{-\im\frac{\pi}{2}\sigma_3}\begin{cases}
	-I,&\textnormal{arg}\,\z\in(-\pi,-\frac{\pi}{3})\\
	-\begin{pmatrix}
	1 & 0\\
	e^{-\im\pi(1-\gamma)} & 1\\
	\end{pmatrix},&\textnormal{arg}\,\z\in(-\frac{\pi}{3},0)\smallskip\\
	\begin{pmatrix}
	1 & 0\\
	e^{-\im\pi(1-\gamma)} & 1\\
	\end{pmatrix},&\textnormal{arg}\,\z\in(0,\frac{\pi}{3})\smallskip\\
	\begin{pmatrix}
	1 & 0\\
	e^{-\im\pi(1-\gamma)} & 1\\
	\end{pmatrix}\begin{pmatrix}
	1 & -e^{\im\pi(1-\gamma)}\\
	0 & 1\\
	\end{pmatrix},&\textnormal{arg}\,\z\in(\frac{\pi}{3},\pi)
	\end{cases}
\end{equation*}
with $\widetilde{A}_0(\z)$ as in \eqref{bareAiry2}. We obtain that
\begin{itemize}
	\item $\widetilde{A}^{RH}(\z)$ is analytic for $\z\in\mathbb{C}\backslash\{\textnormal{arg}\,\z=-\frac{\pi}{3},0,\frac{\pi}{3},\pi\}$
	\item	 Along the curves shown in Figure \ref{Airy2sing},
	\begin{eqnarray*}
		\widetilde{A}^{RH}_+(\z)&=&\widetilde{A}^{RH}_-(\z)\begin{pmatrix}
		1 & 0\\
		e^{-\im\pi(1-\gamma)} & 1\\
		\end{pmatrix},\ \ \ \textnormal{arg}\,\z=-\frac{\pi}{3}\\
		\widetilde{A}^{RH}_+(\z)&=&\widetilde{A}^{RH}_-(\z)\begin{pmatrix}
		-1 & 0\\
		0 & -1\\
		\end{pmatrix},\ \ \ \textnormal{arg}\,\z=0\\
		\widetilde{A}^{RH}_+(\z)&=&\widetilde{A}^{RH}_-(\z)\begin{pmatrix}
		1 & -e^{\im\pi(1-\gamma)}\\
		0 & 1\\
		\end{pmatrix},\ \ \ \textnormal{arg}\,\z=\frac{\pi}{3}\\
		\widetilde{A}^{RH}_+(\z)&=&\widetilde{A}^{RH}_-(\z)\begin{pmatrix}
		-1 & e^{\im\pi(1-\gamma)}\\
		-e^{-\im\pi(1-\gamma)} & 0\\
		\end{pmatrix},\ \ \ \textnormal{arg}\,\z=\pi
	\end{eqnarray*}
	\item As $\z\rightarrow\infty$,
	\begin{align*}
		\widetilde{A}^{RH}(\z)&=\begin{cases}
		+1,&\Im\z>0\\
		-1,&\Im\z<0
		\end{cases}\,\times\,\z^{\frac{1}{4}\sigma_3}\frac{\im}{2}e^{-\im\pi(1-\gamma)}\begin{pmatrix}
		1 & \im e^{\im\pi(1-\gamma)}\\
		1 & -\im e^{\im\pi(1-\gamma)}\\
		\end{pmatrix}\bigg[I+\frac{\im}{48\z^{\frac{3}{2}}}\begin{pmatrix}
		1 & -6\im\,e^{\im\pi(1-\gamma)}\\
		-6\im\,e^{-\im\pi(1-\gamma)} & -1\\
		\end{pmatrix}\\
		&+\mathcal{O}\left(\z^{-\frac{6}{2}}\right)\bigg]e^{\frac{2}{3}\im\z^{\frac{3}{2}}\sigma_3}.
	\end{align*}
\end{itemize}
\begin{figure}
\begin{minipage}{0.35\textwidth} 
\begin{center}
\resizebox{1\textwidth}{!}{\includegraphics{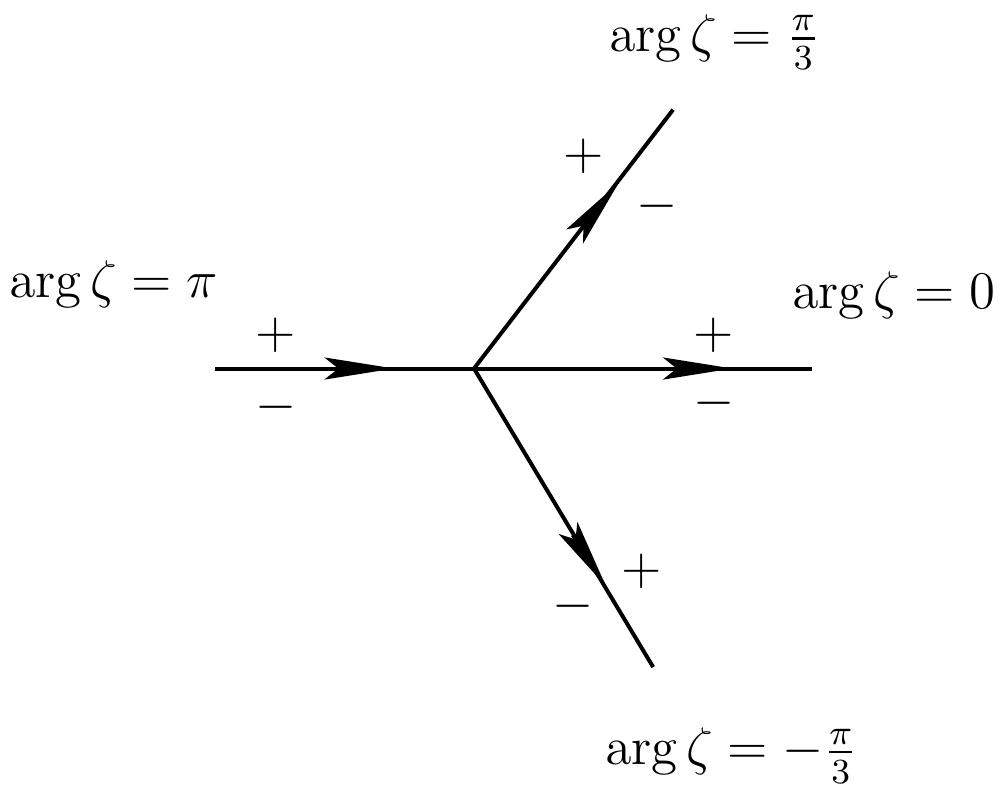}}
\caption{Jump contour for the bare parametrix $\widetilde{A}^{RH}(\zeta)$.}
\label{Airy2sing}
\end{center}
\end{minipage}
\ \ \ \ \begin{minipage}{0.5\textwidth}
\begin{center}
\resizebox{0.95\textwidth}{!}{\includegraphics{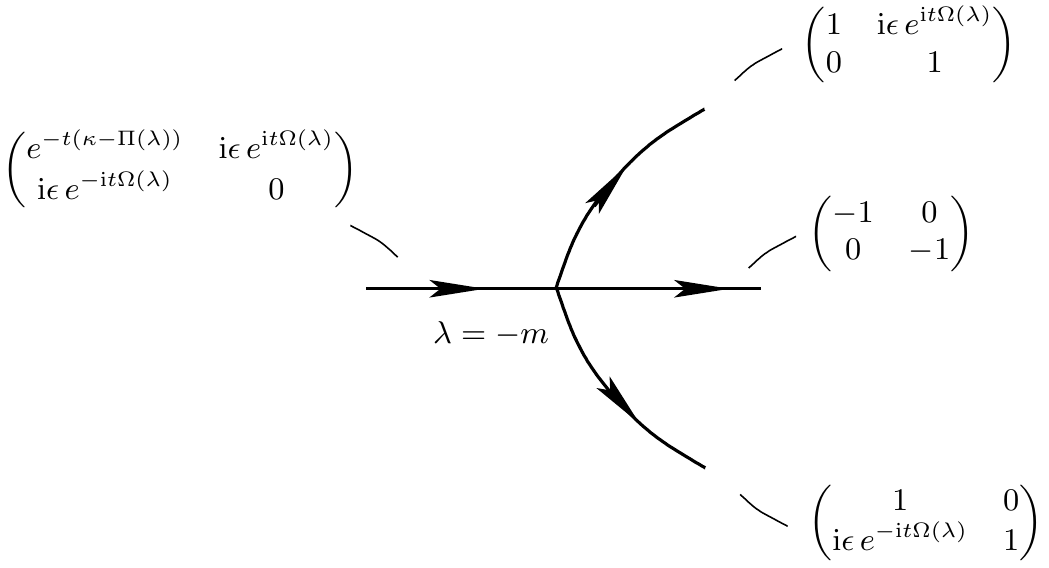}}
\caption{Jumps of the model function $V(\lambda)$ near $\lambda=-m$.}
\label{figure8sing}
\end{center}
\end{minipage}
\end{figure}

For the parametrix itself, we put
\begin{equation}\label{singnm}
	V(\lambda)=B_{\ell_3}(\lambda)\widetilde{A}^{RH}\big(\z(\lambda)\big)e^{-\frac{2}{3}\im\z^{\frac{3}{2}}(\lambda)\sigma_3}e^{-\im\frac{\pi}{4}\epsilon\sigma_3},\ \ \lambda\in D(-m,r)
\end{equation}
with
\begin{equation*}
	\z(\lambda)=\left[6t\int_{-m}^{\lambda}\left((M^2-\mu^2)(m^2-\mu^2)\right)^{\frac{1}{2}}\,\d\mu\right]^{\frac{2}{3}}\sim\left(4t\sqrt{2m(M^2-m^2)}\right)^{\frac{2}{3}}(\lambda+m),\ \ |\lambda+m|<r
\end{equation*}
and the locally analytic multiplier
\begin{equation*}
	B_{\ell_3}(\lambda)=N(\lambda)\begin{cases}
	+1,&\Im\lambda>0\\
	-1,&\Im\lambda<0
	\end{cases}\times\,e^{\im\frac{\pi}{4}\epsilon\sigma_3}\begin{pmatrix}
	-\im e^{\im\pi(1-\gamma)}&-\im e^{\im\pi(1-\gamma)}\\
	-1 & 1\\
	\end{pmatrix}\widetilde{\delta}(\lambda)^{-\sigma_3}\left(\z(\lambda)\frac{\lambda+M}{\lambda+m}\right)^{-\frac{1}{4}\sigma_3}
\end{equation*}
	\begin{equation*}
		V(\lambda)=\left[I+N(\lambda)\left\{\frac{\im}{48\z^{\frac{3}{2}}}\begin{pmatrix}
		1 & 6\epsilon\,e^{\im\pi(1-\gamma)}\\
		-6\epsilon\,e^{-\im\pi(1-\gamma)} & -1\\
		\end{pmatrix}+\mathcal{O}\left(\z^{-\frac{6}{2}}\right)\right\}\big(N(\lambda)\big)^{-1}\right]N(\lambda)
	\end{equation*}
	so that as $t\rightarrow\infty,|s_1|\downarrow 1$ such that $\varkappa\in[\delta,\frac{2}{3}\sqrt{2}-\delta]$,
	\begin{equation*}
		V(\lambda)=\left(I+o(1)\right)N(\lambda)
	\end{equation*}
	uniformly for $0<r_1\leq|\lambda+m|\leq r_2<\min\{\frac{m}{2},\frac{1}{2}(M-m)\}$.
This completes the construction of the relevant local model functions. For $\lambda$ close to $\pm M$, we can simply work with the same functions $P(\lambda)$ and $Q(\lambda)$ as in \eqref{parao} and \eqref{paraom}.
\subsection{Ratio transformation} We work with the ratio function $R(\lambda)$ defined in \eqref{ratio:1} but with $H(\lambda)$ from \eqref{singorigin}, $U(\lambda)$ from \eqref{singm}, $V(\lambda)$ from \eqref{singnm} and $N(\lambda)$ as in \eqref{singouter}. By construction, $R(\lambda)$ is characterized by the following properties.
\begin{itemize}
	\item $R(\lambda)$ is analytic for $\lambda\in\mathbb{C}\backslash\Sigma_R$ where the contour $\Sigma_R$ is shown in Figure \ref{figure13}. This follows directly from our construction of $N(\lambda)$ and $H(\lambda)$ as both functions precisely match the jump behavior of $L(\lambda)$ in the gap $(-m,m)$, compare \eqref{newLjump}.
	\item The jumps of $R(\lambda)$ are only slightly different from the ones listed in Section \ref{ratiosec1}. We choose not to list all expressions explicitly, instead only collect the important estimation
	\begin{equation}\label{singDZesti}
		R_+(\lambda)=R_-(\lambda)G_R(\lambda;t,|s_1|),\hspace{1cm}\|G_R(\cdot;t,|s_1|)-I\|_{L^2\cap L^{\infty}(\Sigma_R)}\leq\frac{c}{t},\ \ c>0
	\end{equation}
	as $t\rightarrow+\infty,|s_1|\downarrow 1$ such that $\varkappa\in[\delta,\frac{2}{3}\sqrt{2}-\delta]$ for $0<\delta<\frac{1}{3}\sqrt{2}$ fixed, provided we stay away from the singular points $(x,s_1)\in\mathcal{Z}_n$.
	\item As $\lambda\rightarrow\infty$, we have that $R(\lambda)\rightarrow I$.
\end{itemize}
With the help of estimation \eqref{singDZesti} the RHP for $R(\lambda)$ can be solved asymptotically in $L^2(\Sigma_R)$ and we have for its unique solution
\begin{equation*}
	\|R_-(\cdot;t,|s_1|)-I\|_{L^2(\Sigma_R)}\leq\frac{c}{t}.
\end{equation*}
\subsection{Singular transition asymptotics} Just as in the regular case, we have the exact identity
\begin{equation*}
	u(x|s)=2\sqrt{-x}\lim_{z\rightarrow\infty}\left[ze^{t\ell\sigma_3}N(z)R(z)e^{-t(g(z)-\vartheta(z))\sigma_3}\right]_{12}=2\sqrt{-x}\,e^{2t\ell}\left[\big(N_1\big)_{12}+\frac{1}{2}\tilde{\mathcal{E}}(\varkappa )\right]
\end{equation*}
where
\begin{equation*}
	\tilde{\mathcal{E}}(\varkappa)=\frac{\im}{\pi}\int_{\Sigma_R}\Big(R_-(w)\big(G_R(w)-I\big)\Big)_{12}\d w=\mathcal{O}\left(t^{-1}\right).
\end{equation*}
Notice that with \eqref{Abel:1} and identities of theta functions,
\begin{equation*}
	\big(N_1\big)_{12}=-\frac{\theta(0)}{\theta(\frac{\tau}{2}+tV)}\frac{\theta(u(\infty)-\frac{\tau}{2}-tV-d)}{\theta(u(\infty)-d)}\frac{M-m}{2\im}e^{-2\pi\im u(\infty)}e^{\im\frac{\pi}{2}\epsilon} = \frac{\theta_3(0|\tau)}{\theta_2(0|\tau)}\frac{\theta_3(tV|\tau)}{\theta_2(tV|\tau)}e^{\im\pi tV}\frac{M-m}{2\im}e^{\im\frac{\pi}{2}\epsilon}
\end{equation*}
and hence we obtain
\begin{equation}\label{singf1}
	u(x|s)=\sqrt{-x}\,e^{2t\ell}\left[-\epsilon(M-m)\frac{\theta_3(0|\tau)}{\theta_2(0|\tau)}\frac{\theta_3(tV|\tau)}{\theta_2(tV|\tau)}e^{\im\pi tV}+\tilde{\mathcal{E}}(\varkappa)\right].
\end{equation}
Now using Proposition \ref{helpful}, combining it with \eqref{Jell:2} and the identities
\begin{equation*}
	M-m=\frac{1}{\sqrt{2}}\frac{1-\k}{\sqrt{1+\k^2}},\ \ \ \ \left(\frac{\theta_2(0|\tau)}{\theta_3(0|\tau)}\right)^2=\frac{1-\k}{1+\k},
\end{equation*}
one deduces
\begin{equation*}
	u(x|s)=-\epsilon\sqrt{-\frac{x}{2}}\,\frac{1+\k}{\sqrt{1+\k^2}}\,\textnormal{dc}\left(2tVK\left(\frac{1-\k}{1+\k}\right),\frac{1-\k}{1+\k}\right)+J_4(x,s),
\end{equation*}
and
\begin{prop} For any given $\delta\in(0,\frac{1}{3}\sqrt{2})$ there exist positive constants $t_0=t_0(\delta),v_0=v_0(\delta)$ and $c=c(\delta)$ such that
\begin{equation*}
	\big|J_4(x,s)\big|\leq ct^{-\frac{2}{3}},\ \ \ \forall\,t\geq t_0,\ v\geq v_0:\ \ t\delta\leq v\leq t\left(\frac{2}{3}\sqrt{2}-\delta\right)
\end{equation*}
provided $(x,s_1)$ are uniformly bounded away from the exceptional set $\mathcal{Z}_n$ introduced in \eqref{exceptset}.
\end{prop}
\section{Extension at the lower end of singular transition}\label{singsec:2}
The relaxation of the lower bound for $\varkappa$ is handled via the same approach as in Section \ref{relax}, i.e. for $t=(-x)^{\frac{3}{2}}$ and $v=-\ln(|s_1|^2-1)$ sufficiently large such that
\begin{equation*}
	0<t^{-\eta}\leq\varkappa\leq \frac{2}{3}\sqrt{2}-\delta,\ \ \ \textnormal{with fixed}\ \ \ 0<\eta<1,
\end{equation*}
we obtain with the help of contracting radii the following analogue of \eqref{final2}
\begin{equation}\label{singres}
	u(x|s)=-\epsilon\sqrt{-\frac{x}{2}}\,\frac{1+\k}{\sqrt{1+\k^2}}\,\textnormal{dc}\left(2(-x)^{\frac{3}{2}}VK\left(\frac{1-\k}{1+\k}\right),\frac{1-\k}{1+\k}\right)+\mathcal{O}\left((-x)^{-1+\frac{3}{2}\eta}\right)
\end{equation}
uniformly as $x\rightarrow-\infty,|s_1|\downarrow 1$ such that $0<t^{-\eta}\leq\varkappa\leq\frac{2}{3}\sqrt{2}-\delta$ for any $0<\eta<\frac{2}{3}$ away from the exceptional set \eqref{exceptset}. After that we would use the analogue of \eqref{Ek:1}, respectively Remark \ref{bettererror}, for the singular analysis and this combined with improved $L^2$ estimations gives
\begin{equation*}
	\tilde{\mathcal{E}}(\varkappa)=\mathcal{O}\left(t^{-2+\frac{3}{2}\eta}\right)+\mathcal{O}\left(t^{-1}\right).
\end{equation*}
We summarize,
\begin{prop}\label{singimp1} For any given $\delta\in(0,\frac{2}{3}\sqrt{2}),\eta\in(0,1)$ there exist positive constants $t_0=t_0(\delta,\eta),v_0=v_0(\delta,\eta)$ and $c=c(\delta,\eta)$ such that
\begin{equation*}
	\big|J_4(x,s)\big|\leq ct^{-\min\{\frac{2}{3},\frac{5}{3}-\frac{3}{2}\eta\}},\ \ \forall\,t\geq t_0,\ \ v\geq v_0:\ \ t^{1-\eta}\leq v\leq t\left(\frac{2}{3}\sqrt{2}-\delta\right),
\end{equation*}
provided $(x,s_1)$ are bounded away from the exceptional set \eqref{exceptset}.
\end{prop}
In order to obtain estimation \eqref{AS:mat} in Corollary \ref{singmat}, we would choose $t\geq t_0$ such that $0<\varkappa\leq t^{-\frac{4}{5}}$ and again $(x,s_1)$ stay away from \eqref{exceptset}. Then, by Corollary \ref{cor1} we have
\begin{equation}\label{simple}
	-\epsilon\sqrt{-\frac{x}{2}}\,\frac{1+\k}{\sqrt{1+\k^2}}\,\textnormal{dc}\left(2(-x)^{\frac{3}{2}}VK\left(\frac{1-\k}{1+\k}\right),\frac{1-\k}{1+\k}\right)=-\frac{\epsilon\sqrt{-x}}{\cos\big(\pi tV(\varkappa)\big)}+\mathcal{O}\left((-x)^{-\frac{1}{10}}\right).
\end{equation}
Next, for $0<\varkappa\leq t^{-\frac{4}{5}}$ with \eqref{l:2}, (compare \eqref{m:2}),
\begin{equation*}
	\pi tV(\varkappa)=-\frac{2}{3}(-x)^{\frac{3}{2}}-\widehat{\beta}\ln\left(8(-x)^{\frac{3}{2}}\right)+\widehat{\beta}\ln\big|\ln\big(|s_1|^2-1\big)\big|-\widehat{\beta}(1+\ln 2\pi)+\mathcal{O}\left(t^{-\frac{3}{5}}\right)
\end{equation*}
and from Stirling's formula
\begin{equation*}
	\textnormal{arg}\,\Gamma\left(\frac{1}{2}+\im\hat{\beta}\right)=\textnormal{arg}\,\Gamma\left(\frac{1}{2}-\frac{\im}{2\pi}\varkappa t\right)=\widehat{\beta}\ln\big|\ln\big(|s_1|^2-1\big)\big|-\widehat{\beta}(1+\ln 2\pi)+\mathcal{O}\left(t^{-\frac{1}{5}}\right)
\end{equation*}
as $\varkappa t=\mathcal{O}\big(t^{\frac{1}{5}}\big)\rightarrow+\infty$. Together
\begin{equation}\label{simpler}
	\pi tV(\varkappa)=-\frac{2}{3}(-x)^{\frac{3}{2}}-\widehat{\beta}\ln\left(8(-x)^{\frac{3}{2}}\right)-\varphi-\frac{\epsilon\pi}{2}+\mathcal{O}\left(t^{-\frac{1}{5}}\right)
\end{equation}
where (compare Theorem \ref{known:3}),
\begin{equation*}
	\varphi=-\textnormal{arg}\,\Gamma\left(\frac{1}{2}+\im\widehat{\beta}\right)-\textnormal{arg}\,s_1.
\end{equation*}
We substitute \eqref{simpler} into \eqref{simple}, use the addition theorem for cosine as well as the identity 
\begin{equation*}
	\cos\left(z+\frac{\epsilon\pi}{2}\right)=-\epsilon\sin z,\ \ \ \ z\in\mathbb{C}
\end{equation*}
to obtain, as $x\rightarrow-\infty$,
\begin{align}
	-\epsilon\sqrt{-\frac{x}{2}}\,\frac{1+\k}{\sqrt{1+\k^2}}\,\textnormal{dc}\left(2(-x)^{\frac{3}{2}}VK\left(\frac{1-\k}{1+\k}\right),\frac{1-\k}{1+\k}\right)=&\frac{\sqrt{-x}}{\sin\left(\frac{2}{3}(-x)^{\frac{3}{2}}+\widehat{\beta}\ln\big(8(-x)^{\frac{3}{2}}\big)+\varphi\right)+\mathcal{O}\left((-x)^{-\frac{3}{10}}\right)}\nonumber\\
	&+\mathcal{O}\left((-x)^{-\frac{1}{10}}\right),\label{connect1}
\end{align}
uniformly for $0<\varkappa\leq t^{-\frac{4}{5}}$ and away from the zeros of the trigonometric function appearing in the last denominator. Hence, the Jacobi elliptic function leading term reproduces also here the known singular oscillatory leading order of \eqref{known:3} for $(t,v)$ such that $t\geq t_0$ and $0<\varkappa\leq t^{-\frac{4}{5}}$. Next we will also obtain control over the error $J_4(x,s)$ in the same region.

\section{Further extension at the lower end for singular transition}\label{singsec:3}
We go back to Section \ref{rloweresti} and assume throughout that both $t$ and $v=-\ln(|s_1|^2-1)$ are sufficiently large such that
\begin{equation}\label{singscale}
	t\geq v^{k+1}>0,\ \ t\geq t_0,\ \ \ k\in\mathbb{Z}_{\geq 3}.
\end{equation}
The difference between regular and singular transition asymptotics can be read off from \eqref{crucialpara},
\begin{equation}\label{easytrick}
	\nu=-\frac{1}{2\pi\im}\ln\big(1-|s_1|^2\big)=-\frac{1}{2\pi\im}\ln\big(|s_1|^2-1\big)-\frac{1}{2}\equiv\nu_0-\frac{1}{2}.
\end{equation}
Hence we can still use most of the nonlinear steepest analysis worked out in Section \ref{flower}, however a crucial differences occurs in the matching relations \eqref{mup:1} and \eqref{mup:2}, compare \cite{BI} where the same approach was first used in the derivation of \eqref{known:3} for fixed $|s_1|>1$. We shall factorize \eqref{mup:1},
\begin{align*}
	\Psi^r(\lambda)&\sim E_r(\lambda)\bigg[I+\sigma_-\big(\beta(\lambda)\big)^{-2}\left(-\frac{h_1}{s_3}\right)\sum_{m=1}^{\infty}\big((\nu+1)_{2m}-(\nu)_{2m}\big)\frac{\z(\lambda)^{-2m-1}}{m!\,2^m}+\sigma_+\big(\beta(\lambda)\big)^2\left(-\frac{s_3}{h_1}\right)\\
	&\times\,\sum_{m=1}^{\infty}(-)^m\big((-\nu)_{2m}-(-\nu-1)_{2m}\big)\frac{\z(\lambda)^{-2m+1}}{m!\,2^m}+\sum_{m=1}^{\infty}\begin{pmatrix}
	(\nu)_{2m} & 0\\
	0 & (-)^m(-\nu-1)_{2m}\\
	\end{pmatrix}\frac{\z(\lambda)^{-2m}}{m!\,2^m}\bigg]\Psi^D(\lambda)
\end{align*}
with
\begin{equation*}
	E_r(\lambda)=\begin{pmatrix}
	1 & 0\\
	-\frac{h_1}{s_3}\frac{\beta^{-2}(\lambda)}{\z(\lambda)} & 1\\
	\end{pmatrix}.
\end{equation*}
This factorization is necessary since with \eqref{easytrick}
\begin{equation*}
	\frac{\beta^{-2}(\lambda)}{\z(\lambda)} = \left(\z(\lambda)\frac{\lambda+\lambda^{\ast}}{\lambda-\lambda^{\ast}}\right)^{-2\nu_0}\left(\frac{\lambda+\lambda^{\ast}}{\lambda-\lambda^{\ast}}\right)e^{2t\vartheta(\lambda^{\ast})} = O(1),
\end{equation*}
i.e. the multiplier $E_r(\lambda)$ is not close to the unit matrix. Similarly,
\begin{align*}
	\Psi^{\ell}(\lambda)&\sim E_{\ell}(\lambda)\bigg[I+\sigma_+\big(\beta(-\lambda)\big)^{-2}\frac{h_1}{s_3}\sum_{m=1}^{\infty}\big((\nu+1)_{2m}-(\nu)_{2m}\big)\frac{\z(-\lambda)^{-2m-1}}{m!\,2^m}+\sigma_-\big(\beta(-\lambda)\big)^2\frac{s_3}{h_1}\\
	&\times\,\sum_{m=1}^{\infty}(-)^m\big((-\nu)_{2m}-(-\nu-1)_{2m}\big)\frac{\z(-\lambda)^{-2m+1}}{m!\,2^m}+\sum_{m=1}^{\infty}\begin{pmatrix}
	(-)^m(-\nu-1)_{2m} & 0\\
	0 & (\nu)_{2m}\\
	\end{pmatrix}\frac{\z(-\lambda)^{-2m}}{m!\,2^m}\bigg]\Psi^D(\lambda)
\end{align*}
with
\begin{equation*}
	E_{\ell}(\lambda)=\sigma_2 E_r(-\lambda)\sigma_2.
\end{equation*}
This means that the ratio problem \eqref{ratiolower} for $\chi(\lambda)$ cannot be solved iteratively at this point and we have to use further transformations. First, compare \cite{BI},  equation $(38)$, we let
\begin{equation*}
	\Phi^u(\lambda)=\begin{cases}
	\chi(\lambda)E_r(\lambda),&|\lambda-\lambda^{\ast}|<r\\
	\chi(\lambda)E_{\ell}(\lambda),&|\lambda+\lambda^{\ast}|<r\\
	\chi(\lambda),&|\lambda\mp\lambda^{\ast}|>r
	\end{cases}
\end{equation*}
where we use (this is in fact the only difference to \cite{BI}) a shrinking radius
\begin{equation*}
	\frac{1}{4}>r=\frac{1}{v}\geq t^{-\frac{1}{k+1}}.
\end{equation*}
The Riemann-Hilbert problem for $\Phi^u(\lambda)$ was analyzed in \cite{BI}, its jump contour is identically $\Sigma_{\chi}$ as shown in Figure \ref{ratiolowerfig} but with additional first order pole singularities at $\lambda=\pm\lambda^{\ast}$. We choose not to reproduce the analysis here, compare equations $(40)-(42)$ in \cite{BI}. Subsequently the singular structure was resolved by another transformation
\begin{equation*}
	\Phi^u(\lambda)=(\lambda I+B)\Phi^d(\lambda)\begin{pmatrix}
	(\lambda-\lambda^{\ast})^{-1} & 0\\
	0 & (\lambda+\lambda^{\ast})^{-1}
	\end{pmatrix},\ \ \lambda\in\mathbb{C}\backslash\Sigma_{\chi}
\end{equation*}
where the $\lambda$ independent matrix $B$ is determined algebraically as in equation $(46)$ in \cite{BI}. It is well defined for all values $(x,s_1)$ bounded away from the discrete set
\begin{equation}\label{exceptknown}
	\frac{2}{3}t+\widehat{\beta}\ln(8t)+\varphi=n\pi,
\end{equation}
compare again \cite{BI}, shortly before Remark $3$. If we agree to stay away from the latter exceptional points, the jumps for $\Phi^d(\lambda)$ have the following structure
\begin{equation*}
	\Phi^d_+(\lambda)=\Phi^d_-(\lambda)\begin{pmatrix}
	(\lambda-\lambda^{\ast})^{-1} & 0\\
	0 & (\lambda+\lambda^{\ast})^{-1}
	\end{pmatrix}\Psi^D(\lambda)\tilde{S}_j\big(\Psi^D(\lambda)\big)^{-1}\begin{pmatrix}
	\lambda-\lambda^{\ast} & 0\\
	0 & \lambda+\lambda^{\ast}
	\end{pmatrix},\ \lambda\in\tilde{\gamma}_j,\ j=1,\ldots,8
\end{equation*}
and
\begin{eqnarray*}
	\Phi^d_+(\lambda)&=&\Phi^d_-(\lambda)\begin{pmatrix}
	(\lambda-\lambda^{\ast})^{-1} & 0\\
	0 & (\lambda+\lambda^{\ast})^{-1}
	\end{pmatrix}E_r^{-1}(\lambda)\Psi^r(\lambda)\big(\Psi^D(\lambda)\big)^{-1}\begin{pmatrix}
	\lambda-\lambda^{\ast} & 0\\
	0 & \lambda+\lambda^{\ast}
	\end{pmatrix},\ \ \lambda\in C_r\\
	\Phi^d_+(\lambda)&=&\Phi^d_-(\lambda)\begin{pmatrix}
	(\lambda-\lambda^{\ast})^{-1} & 0\\
	0 & (\lambda+\lambda^{\ast})^{-1}
	\end{pmatrix}E_{\ell}^{-1}(\lambda)\Psi^{\ell}(\lambda)\big(\Psi^D(\lambda)\big)^{-1}\begin{pmatrix}
	\lambda-\lambda^{\ast} & 0\\
	0 & \lambda+\lambda^{\ast}
	\end{pmatrix},\ \ \lambda\in C_{\ell}.
\end{eqnarray*}
Hence we have
\begin{equation*}
	G_{\Phi^d}(\lambda)=I+\mathcal{O}\left(e^{d_1\varkappa t-d_2tr^2}\right),\ \ \lambda\in\tilde{\gamma}_j,\ \ j=1,\ldots,8
\end{equation*}
so that (just as before in \eqref{loesti:1}),
\begin{equation}\label{adhoc1}
	\|G_{\Phi^d}(\cdot;t,|s_1|)-I\|_{L^2\cap L^{\infty}(\cup\tilde{\gamma}_j)}\leq d_3e^{-d_4 t^{\frac{k-1}{k+1}}},\ \ k\in\mathbb{Z}_{\geq 3}
\end{equation}
For $\lambda\in C_r$ we estimate
\begin{equation*}
	\left|G_{\Phi^d}(\lambda;t,|s_1|)-I\right|\leq c\left\{\frac{(\varkappa t)^2}{rt}\begin{pmatrix}
	1 & e^{-\frac{3}{4}\varkappa t}e^{\widehat{\varphi}(r)} (\varkappa t)^{-\frac{3}{2}}\\
	e^{\frac{3}{4}\varkappa t}e^{-\widehat{\varphi}(r)}(\varkappa t)^{-\frac{1}{2}} & 1\\
	\end{pmatrix}+\bar{\mathcal{E}}(\lambda;t,|s_1|)\right\}
\end{equation*}
with
\begin{equation*}
	\widehat{\varphi}(r)=\frac{\varkappa t}{\pi}\textnormal{arg}\left(\z(\lambda)\frac{\lambda+\lambda^{\ast}}{\lambda-\lambda^{\ast}}\right).
\end{equation*}
Thus, using again bounds for the error term $\bar{\mathcal{E}}$ from \cite{N}, we obtain that with the contracting radius
\begin{equation}\label{adhoc2}
	\|G_{\Phi^d}(\cdot;t,|s_1|)-I\|_{L^2\cap L^{\infty}(C_r\cup C_{\ell})}\leq d_3t^{-\frac{k-2}{k+1}},\ \ \forall\,t\geq v^{k+1}>0,\ \ t\geq t_0,\ \ k\in\mathbb{Z}_{\geq 3}.
\end{equation}
Combining estimations \eqref{adhoc1} and \eqref{adhoc2} we can thus solve the $\Phi^d$-RHP asymptotically subject to the scale \eqref{singscale} as long as we stay away from the discrete set of points $(x,s_1)$ defined via \eqref{exceptknown}. The derivation of the leading order follows now as in \cite{BI}, we only have to adjust the error terms according to \eqref{adhoc2}: from equation $(51)$ in \cite{BI},
\begin{equation*}
	u(x|s)=2\sqrt{-x}\,B_{12}+\mathcal{O}\left((-x)^{-\frac{2k-7}{2(k+1)}}\right),\ \ \ k\in\mathbb{Z}_{\geq 4},
\end{equation*}
where we used \eqref{adhoc2} and the a-priori estimation for the unique solution of the $\Phi^d$-RHP. The matrix entry $B_{12}$ is then computed explicitly in \cite{BI} and we only have to adjust the error terms. Hence,
\begin{equation*}
	u(x|s)=\frac{\sqrt{-x}}{\sin\left(\frac{2}{3}(-x)^{\frac{3}{2}}+\widehat{\beta}\ln\left(8(-x)^{\frac{3}{2}}\right)+\varphi\right)+\mathcal{O}\left((-x)^{-\frac{3(k-2)}{2(k+1)}}\right)}+\mathcal{O}\left((-x)^{-\frac{2k-7}{2(k+1)}}\right)
\end{equation*}
as $x\rightarrow-\infty$ subject to the scale \eqref{singscale} for $k\in\mathbb{Z}_{\geq 4}$ provided we stay away from the points $(x,s_1)$ defined implicitly via \eqref{exceptknown}. The latter estimation, with $k=4$, appears in Corollary \ref{singmat}, expansion \eqref{AS:mat}. Moreover, with \eqref{connect1}, we have now shown that
\begin{equation*}
	u(x|s)=-\epsilon\sqrt{-\frac{x}{2}}\,\frac{1+\k}{\sqrt{1+\k^2}}\,\textnormal{dc}\left(2(-x)^{\frac{3}{2}}VK\left(\frac{1-\k}{1+\k}\right),\frac{1-\k}{1+\k}\right)+J_4(x,s),
\end{equation*}
and
\begin{prop}\label{singimp2} There exist positive constants $t_0$ and $c$ such that
\begin{equation*}
	\big|J_4(x,s)\big|\leq ct^{-\frac{1}{15}},\ \ \ \forall\,t\geq t_0,\ \ 0<v\leq t^{\frac{1}{5}}
\end{equation*}
provided $(x,s_1)$ are uniformly bounded away from the set \eqref{exceptset}.
\end{prop}
This concludes the proof of estimation \eqref{sing:est1} in Theorem \ref{bet:2}, we just have to combine Propositions \ref{singimp1} (with $\eta=\frac{4}{5}$) and \ref{singimp2}.
\section{Extension at the upper end for singular transition}\label{singsec:4} We choose $t\geq t_0$ and $v=-\ln(|s_1|^2-1)$ sufficiently large such that
\begin{equation}\label{singupext}
	\frac{2}{3}\sqrt{2}-\frac{f_1}{t}\leq\varkappa<\frac{2}{3}\sqrt{2}\ \ \ \ \Leftrightarrow\ \ \ \ 0<\sigma\equiv\frac{2}{3}\sqrt{2}-\varkappa\leq\frac{f_1}{t},\ \ f_1>0.
\end{equation}
This implies with \eqref{l:2}, 
\begin{equation*}
	0<\big|2tV(\varkappa)\big|\leq\frac{cf_4}{\ln t},\ \ \ c>0\ \ \textnormal{universal}.
\end{equation*}
In other words, we can always guarantee that for sufficiently large $t$ and $-\ln(|s_1|^2-1)$ subject to \eqref{singupext} the exceptional set $\mathcal{Z}_n$ is empty, and thus $u(x|s)$ is asymptotically pole-free. This is important as we can now go back to Section \ref{upextension} and simply repeat our line of reasoning, modulo minor adjustments. First, from \eqref{singouter},
\begin{equation*}
	N(\lambda)=e^{\im\frac{\pi}{4}\epsilon\sigma_3}\frac{\theta(0|\tau')}{\theta(\tau'(\frac{\tau}{2}+tV)|\tau')}e^{-2\pi\im\tau'u(\infty)tV\sigma_3}\begin{pmatrix}
	\widehat{N}_1'^{(+)}(\lambda)\phi(\lambda) & \widehat{N}_2'^{(+)}(\lambda)\hat{\phi}(\lambda)\\
	-\widehat{N}_1'^{(-)}(\lambda)\hat{\phi}(\lambda) & \widehat{N}_2'^{(-)}(\lambda)\phi(\lambda)
	\end{pmatrix}e^{2\pi\im\tau'u(\lambda)tV\sigma_3}e^{-\im\frac{\pi}{4}\epsilon\sigma_3},
\end{equation*}
valid for any $\lambda\in\mathbb{C}\backslash[-M,M]$ and where
\begin{equation*}
	\left(\widehat{N}_1'^{(\pm)}(z),\widehat{N}_2'^{(\pm)}(z)\right)=\left(\frac{\theta(\tau'(u(z)+tV+\frac{\tau}{2}\pm d)|\tau')}{\theta(\tau'(u(z)\pm d)|\tau')},\frac{\theta(\tau'(-u(z)+tV+\frac{\tau}{2}\pm d)|\tau')}{\theta(\tau'(-u(z)\pm d)|\tau')}\right).
\end{equation*}
But
\begin{equation*}
	\theta\left(\tau'\left(u(z)+tV+\frac{\tau}{2}\pm d\right)|\tau'\right)=\theta_4\big(\tau'\left(u(z)+tV\pm d\right)|\tau'\big),\ \ z\in\mathbb{CP}^1\backslash[-M,M],
\end{equation*}
and thus with \eqref{e:1s} and \eqref{e:2s}, subject to \eqref{singupext},
\begin{equation*}
	\widehat{N}_1'^{(\pm)}(\lambda)=1+\mathcal{O}\left(\sqrt{\frac{\sigma}{|\ln\sigma|}}\right)=\widehat{N}_2'^{(\pm)}(\lambda).
\end{equation*}
Summarizing, Proposition \ref{ou:match} also applies to the singular transition $x\rightarrow-\infty,|s_1|\downarrow 1$ subject to \eqref{singupext}. Now following the logic of Section \ref{upextension}, we would employ the sequence of transformations
\begin{equation*}
	Y(\lambda)\mapsto X(\lambda)\mapsto Z(\lambda)\mapsto T(\lambda)\mapsto S(\lambda)
\end{equation*}
but not open lens, compare \eqref{newLjump}. Instead, just as in Section \ref{upextension}, a new parametrix $\widehat{J}(\lambda)$ is introduced on the contracting segment $(-\hat{r},\hat{r})\subset(-M,M)$ with $\hat{r}$ as in \eqref{ridea}. The difference between $\widehat{J}(\lambda)$ and $J(\lambda)$ as considered in RHP \ref{nmodel} is minor, we would simply enforce
\begin{equation*}
	\widehat{J}_+(\lambda)=\widehat{J}_-(\lambda)\begin{pmatrix}
	-e^{-t(\varkappa-\Pi(\lambda))} & \im\epsilon\\
	\im\epsilon & 0
	\end{pmatrix}, \ \ \lambda\in(-\hat{r},\hat{r}).
\end{equation*}
This difference amounts to an additional sign in the $(12)$-entry of \eqref{ou:idea} which contains the Cauchy transform. However, subsequently all steps are identical to Section \ref{upextension} and we obtain
\begin{equation*}
	u(x|s)=-\epsilon\sqrt{-\frac{x}{2}}\,\frac{1+\k}{\sqrt{1+\k^2}}\textnormal{dc}\left(2(-x)^{\frac{3}{2}}VK\left(\frac{1-\k}{1+\k}\right),\frac{1-\k}{1+\k}\right)+J_4(x,s)
\end{equation*}
with
\begin{prop} For any given $f_1>0$ there exist positive constants $t_0=t_0(f_1),v_0=v_0(f_1)$ and $c=c(f_1)$ such that
\begin{equation*}
	\big|J_4(x,s)\big|\leq\frac{c}{\ln t},\ \ \ \ \forall\,t\geq t_0,\ v\geq v_0:\ \ \frac{2}{3}\sqrt{2}\,t-f_4\leq v<\frac{2}{3}\sqrt{2}\,t.
\end{equation*}
\end{prop}
This completes the proof of Theorem \ref{bet:2} and expansion \eqref{HS:mat} in Corollary \ref{singmat} is derived as in the regular transition case.

\section{Singular transition analysis near and above the separating line $\varkappa=\frac{2}{3}\sqrt{2}$}\label{singsec:5}
In order to derive Theorem \ref{res:9} we would choose identical steps as those in Section \ref{near}: We work with the same outer parametrix $\Upsilon(\lambda)$ as in \eqref{nout:1} and the same endpoint parametrices $\Delta^r(\lambda)$ in \eqref{upedge1} and $\Delta^{\ell}(\lambda)$ in \eqref{upedge2}. The only (minor) difference occurs in the definition of $F(\lambda)$, i.e. in the definition of the origin parametrix. For the singular transition we require the jump behavior
\begin{equation*}
	\widehat{F}_+(\lambda)=\widehat{F}_-(\lambda)\begin{pmatrix}
	-e^{-t(\varkappa-\Pi(\lambda))} & \im\epsilon\\
	\im\epsilon & 0\\
	\end{pmatrix},\ \ \ \lambda\in(-r,r);\hspace{0.5cm}\Pi(\lambda)=\frac{8}{3}\left(\frac{1}{2}-\lambda^2\right)^{\frac{3}{2}}
\end{equation*}
instead of \eqref{orij}. This difference is resolved by an additional sign in the $(12)$-entry containing the Cauchy transform in \eqref{FRHmodel}. After that we have to replace the matching \eqref{orimatch} through
\begin{equation*}
	\widehat{F}(\lambda)=\left(I-\frac{\epsilon}{2\pi}\frac{e^{\sigma t}}{\z(\lambda)}B_0(\lambda)\sigma_+B_0^{-1}(\lambda)+\mathcal{O}\left(t^{-1}\right)\right)\Upsilon(\lambda)
\end{equation*}
and then simply copy all subsequent steps carried out in Section \ref{near}. This leads us to
\begin{equation*}
	u(x|s)=-\epsilon\sqrt{-\frac{x}{2}}\,\left(1+\frac{1}{2\pi}\frac{e^{\sigma t}}{(t\sqrt{2})^{\frac{1}{2}}}+J_5(x,s)\right),\ \ \sigma=\frac{2}{3}\sqrt{2}-\varkappa,\ \ t=(-x)^{\frac{3}{2}}
\end{equation*}
with
\begin{prop} For any given $f_2\in\mathbb{R}$ there exist positive $t_0=t_0(f_2),v_0=v_0(f_2)$ and $c=c(f_2)$ such that
\begin{equation*}
	\big|J_5(x,s)\big|\leq ct^{-1}\ \ \ \ \forall\,t\geq t_0,\ \ v\geq v_0,\ \ \ v\geq\frac{2}{3}\sqrt{2}\,t-f_2,
\end{equation*}
\end{prop}
and thus completes the proof of Theorem \ref{res:9}.

\begin{appendix}
\section{Expansions of complete elliptic integrals}\label{appa}
We frequently use the elliptic integrals $E(\k)$ and $K(\k)$ defined as
\begin{equation*}
	E(\k)=\int_0^1\sqrt{\frac{1-\k^2t^2}{1-t^2}}\,\d t,\hspace{1cm} K(\k)=\int_0^1\frac{\d t}{\sqrt{(1-\k^2t^2)(1-t^2)}},\ \ \k\in(0,1),\ \ \k'=\sqrt{1-\k^2}.
\end{equation*}
For $\k\downarrow 0$ or $\k\uparrow 1$, the expansions of $E(\k),K(\k)$ are well known, compare for instance \cite{N}:
\begin{eqnarray*}
	K(\k)&=&\sum_{m=0}^{\infty}\frac{(\frac{1}{2})_m(\frac{1}{2})_m}{m!\,m!}(1-\k^2)^m\left(-\frac{1}{2}\ln(1-\k^2)+d(m)\right),\ \ 0<\sqrt{1-\k^2}<1;\\
	E(\k)&=&1+\frac{1}{2}\sum_{m=0}^{\infty}\frac{(\frac{1}{2})_m(\frac{3}{2})_m}{(2)_mm!}(1-\k^2)^{m+1}\left(-\frac{1}{2}\ln(1-\k^2)+d(m)-\frac{1}{(2m+1)(2m+2)}\right),
\end{eqnarray*}
with
\begin{equation*}
	d(m)=\psi(1+m)-\psi\left(\frac{1}{2}+m\right),\hspace{0.5cm}d(0)=2\ln 2;\hspace{0.5cm}(a)_n=\frac{\Gamma(a+n)}{\Gamma(a)},\ \ a\neq 0,-1,-2,\ldots
\end{equation*}
and $\psi(x)$ denoting the digamma function. Furthermore,
\begin{equation*}
	K(\k)=\frac{\pi}{2}\sum_{m=0}^{\infty}\frac{(\frac{1}{2})_m(\frac{1}{2})_m}{m!\,m!}\k^{2m},\hspace{0.7cm}
	E(\k)=\frac{\pi}{2}\sum_{m=0}^{\infty}\frac{(-\frac{1}{2})_m(\frac{1}{2})_m}{m!\,m!}\k^{2m},\ \  \ |\k|<1.
\end{equation*}

\section{Jacobi theta and elliptic functions used in the text}\label{app:theta}

The Jacobi theta functions $\theta_2(z),\theta_3(z),\theta_4(z)$ used in the text are defined by the formul\ae\, (see e.g. \cite{N}) where $\Im\tau>0$,
\begin{eqnarray*}
	\theta_3(z|\tau) &\equiv& \theta(z|\tau) = 1+2\sum_{k=1}^{\infty}e^{\im\pi k^2\tau}\cos(2\pi kz),\\
	\theta_4(z|\tau) &=&\theta_3\left(z+\frac{1}{2}\bigg|\,\tau\right) = 1+2\sum_{k=1}^{\infty}(-1)^ke^{\im\pi k^2\tau}\cos\big(2\pi kz\big),\\
	\theta_2(z|\tau) &=& e^{\im\frac{\pi}{4}\tau+\im\pi z}\theta_3\left(z+\frac{\tau}{2}\bigg|\,\tau\right)=2\sum_{k=0}^{\infty}e^{\im\pi(k+\frac{1}{2})^2\tau}\cos\big((2k+1)\pi z\big).
\end{eqnarray*}
The zeros of the theta functions are located at
\begin{equation*}
	\theta_4\left(\frac{\tau}{2}\right)=0,\hspace{0.5cm}\theta_2\left(\frac{1}{2}\right)=0,\hspace{0.5cm}\theta_3\left(\frac{1}{2}+\frac{\tau}{2}\right)=0,
\end{equation*}
up to shifts by vectors of the lattice $\Lambda=\mathbb{Z}+\tau\mathbb{Z}$; all zeros are simple. Now set $q=e^{\im\pi\tau}$ and denote $\theta_j(z|\tau)\equiv\theta_j(z,q)$. With
\begin{equation*}
	q=\exp\left[-\pi\frac{K'(\k)}{K(\k)}\right],\ \ \ \k=\left(\frac{\theta_2(0,q)}{\theta_3(0,q)}\right)^2\in(0,1),
\end{equation*}
the Jacobi elliptic functions $\textnormal{cd}(z,\k)$ and $\textnormal{dc}(z,\k)$ are given by the formul\ae
\begin{equation*}
	\textnormal{cd}(z,\k)=\frac{\theta_3(0,q)}{\theta_2(0,q)}\frac{\theta_2(\z,q)}{\theta_3(\z,q)} = \frac{1}{\textnormal{dc}(z,\k)};\hspace{0.75cm}\z=\frac{z}{2K(\k)}
\end{equation*}
for all values $z\in\mathbb{C}$ away from the zeros of $\theta_3$, resp. $\theta_2$ appearing in the denominators. Finally we note that by an application of a Landen transformation,
\begin{equation*}
	2\frac{K(\k)}{K'(\k)} = \frac{K'(\lambda)}{K(\lambda)},\ \ \ \lambda=\frac{1-\k}{1+\k},\ \ \ \lambda'=\sqrt{1-\lambda^2}.
\end{equation*}

\end{appendix}

\end{document}